\documentclass[11pt]{article}
\usepackage{xr}
\usepackage{amssymb}
\usepackage{amsmath}
\usepackage{amsmath}
\usepackage{amsthm}
\usepackage{amsfonts}
\usepackage{graphicx}
\usepackage{bbm}
\usepackage{mathrsfs}
\usepackage{enumerate}
\usepackage{dsfont}
\usepackage{natbib}
\usepackage{multirow}
\usepackage{dsfont}
\usepackage{color}
\numberwithin{equation}{section}
\newtheorem{theorem}{Theorem}

\newtheorem{algorithm}{Algorithm}

\newtheorem{corollary}{Corollary}

\newtheorem{example}{Example}

\newtheorem{remark}{Remark}

\newtheorem{assumption}{Assumption}
\newtheorem{lemma}{Lemma}
\theoremstyle{definition}

\DeclareMathOperator{\E}{\text{E}}

\newcommand{\Epsilon}{\mathcal{E}}

\newcommand{\norm}[1]{\left \|#1\right \|}
\newcommand{\abs}[1]{\left\vert#1\right\vert}

\newcommand{\argmin}{\mathop{\rm arg~min}\limits}

\newcommand{\sumg}{\sum_{g=1}^G}
\newcommand{\sumi}{\sum_{i\in C_g}}

\begin{document}

\title{Robust Uniform Inference for Quantile Treatment Effects in Regression Discontinuity Designs\footnote{First arXiv version: February 24, 2017. Code files are available upon request. We would like to thank Bill Gormley for kindly approving our use of his data for our empirical application, Brigham Frandsen for kindly providing us with a working data set, and Fangzhu Yang for excellent research assistance. We benefited from useful comments by Matias Cattaneo, Yingying Dong, Oliver Linton (the co-editor), Patrick Richard, Pedro Sant'Anna, the associate editor, anonymous referees, seminar participants at UC Davis, UC Irvine, and University of Tsukuba, and conference participants at CESG 2018, ESEM 2017, IAAE 2017, IEAS Econometric Workshop: Theories and Applications 2017, LAMES 2017, New York Camp Econometrics XII, SETA 2017, Shanghai Econometrics Workshop 2017, and University of Tokyo Conference in Advances in Econometrics. All remaining errors are ours. Yu-Chin Hsu gratefully acknowledges the research support from Ministry of Science and Technology of Taiwan (MOST107-2410-H-001-034-MY3) and Career Development Award of Academia Sinica, Taiwan. This paper was previously circulated as ``A Unified Robust Bootstrap Method for Sharp/Fuzzy Mean/Quantile Regression Discontinuity/Kink Designs.''\bigskip}}
\author{
Harold D. Chiang\thanks{Harold D. Chiang: harold.d.chiang@vanderbilt.edu. Department of Economics, Vanderbilt University.\bigskip} \qquad Yu-Chin Hsu\thanks{Yu-Chin Hsu: ychsu@econ.sinica.edu.tw. Institute of Economics, Academia Sinica; Department of Finance, National Central University; and Department of Economics, National Chengchi University.\bigskip} \qquad Yuya Sasaki\thanks{Yuya Sasaki (Corresponding Author): yuya.sasaki@vanderbilt.edu. Department of Economics, Vanderbilt University.\bigskip}
\bigskip\\
}
\date{
\today }
\maketitle
\begin{abstract}\setlength{\baselineskip}{4.5mm}
The practical importance of inference with robustness against large bandwidths for causal effects in regression discontinuity and kink designs is widely recognized.
Existing robust methods cover many cases, but do not handle uniform inference for CDF and quantile processes in fuzzy designs, despite its use in the recent literature in empirical microeconomics.
In this light, this paper extends the literature by developing a unified framework of inference with robustness against large bandwidths that applies to uniform inference for quantile treatment effects in fuzzy designs, as well as all the other cases of sharp/fuzzy mean/quantile regression discontinuity/kink designs.
We present Monte Carlo simulation studies and an empirical application for evaluations of the Oklahoma pre-K program.
\smallskip\\
{\bf Keywords:} bias correction, local Wald estimator, multiplier bootstrap, quantile, regression discontinuity design, regression kink design, robustness
\smallskip\\
{\bf JEL Codes:} C01, C14, C21
\end{abstract}
 
\newpage

\section{Introduction}

Empirical researchers have used various versions of local Wald estimators.
In widest use are the local Wald estimators for the regression discontinuity design (RDD).
More recently, researchers have also used local Wald ratios of derivative estimators for the regression kink design (RKD).
Furthermore, local Wald ratios of conditional cumulative distribution functions and their variants are used for estimation of quantile treatment effects.
In all of these variants of local Wald estimators, researchers often choose large bandwidths by alternative data-driven selectors in practice.
As such, ideal estimators and inference procedures need to be robust against large bandwidths.

Proposal of inference methods which are robust against large bandwidths has occurred relatively recently in the literature, given the history of local Wald estimators.
The framework of Calonico, Cattaneo, and Titiunik (2014) covers robust point-wise inference for mean effects in such designs as the sharp mean RDD, fuzzy mean RDD, sharp mean RKD, and fuzzy mean RKD.
The robust wild bootstrap method proposed by Bartalotti, Calhoun, and He (2017) covers the sharp mean RDD.
The pivotal method proposed by Qu and Yoon (2015b) covers robust uniform inference for the sharp quantile RDD, and is also extensible to the sharp quantile RKD (Chiang and Sasaki, 2017).

To our knowledge, these methods proposed in the existing literature do not cover another important case, namely robust uniform inference\footnote{Throughout this paper, we use the phrase ``uniform inference'' to refers to inference based on weak convergence in the uniform normed linear space $\ell^\infty$.} for quantile treatment effects in the fuzzy RDD (Frandsen, Fr\"olich, and Melly, 2012), despite its frequent use in the recent literature in applied economics, e.g., Clark and Martorell (2014) and Deshpande (2016), to list a few.\footnote{We also refer readers to Shigeoka (2014), Ito (2015), Deshpande (2016), and Bernal, Carpio, and Klein (2017) for empirical applications with quantile treatment effects in the sharp RDD.} In this light, this paper proposes a new general robust inference method and construction of bias-corrected uniform confidence bands that cover the fuzzy quantile RDD in particular.
Instead of proposing a robust inference method which specifically applies to the fuzzy quantile RDD, however, we propose one generic framework that uniformly applies to most, if not all, versions of the local Wald estimators including the sharp mean RDD, the fuzzy mean RDD, the sharp mean RKD, the fuzzy mean RKD, the sharp quantile RDD, the fuzzy quantile RDD, the sharp quantile RKD, and the fuzzy quantile RKD, to list a few most popular examples used in empirical research.
We focus on the case of the fuzzy quantile RDD for most parts of this paper, as the applicability to this particular case is new in the literature.

After assessing the performance of the proposed method through Monte Carlo simulation studies, we apply it to real data and study causal effects on test outcomes of the Oklahoma pre-K program, following the earlier work by Gormley, Gayer, Phillips and Dawson (2005) and Frandsen, Fr\"olich and Melly (2012).
While Frandsen, Fr\"olich and Melly (2012) provide point-wise confidence intervals for the quantile treatment effects in this application, we follow up and complement their earlier analysis by further providing uniform confidence bands with robustness against large data-driven bandwidths.

The rest of this paper is organized as follows.
Section \ref{sec:literature} discusses the related literature.
Section \ref{sec:overview} provides an overview of the method.
Section \ref{sec:general_framework} presents the main theoretical results.
Section \ref{sec:simulation_fqrdd} demonstrates Monte Carlo simulation studies.
Section \ref{sec:empirical_illustration_fqrdd} presents an empirical illustration.
Section \ref{sec:extensions} presents extended results, including cluster robust inference (Section \ref{sec:cluster_robust}), inference with robustness against no or weak jumps (Section \ref{sec:weak_jumps}), and inference for models with covariates (Section \ref{sec:covariates}).
Section \ref{sec:summary} concludes.
All mathematical proofs and additional details are delegated to the appendix.

\section{Relation to the Literature}\label{sec:literature}

In this section, we overview the most relevant parts of the existing literature.
Because of the breadth of the related literature, what we write below is far from being exhaustive.

{\bf Literature on Local Designs:}
The idea of the RDD is introduced by Thistlethwaite and Campbell (1960).
There is a vast literature on the RDD.
Instead of enumerating all papers, we refer the readers to a seminal paper by Hahn, Todd and van der Klaauw (2001) and a collection of surveys, including Cook (2008) contained in the special issue of \textit{Journal of Econometrics} edited by Imbens and Lemieux (2008), Imbens and Wooldridge (2009; Sec. 6.4), Lee and Lemieux (2010), and Volume 38 of \textit{Advances in Econometrics} edited by Cattaneo and Escanciano (2017), as well as the references cited therein.
For technical matters, we mainly refer to Porter (2003) in deriving the general Bahadur representation for higher-order local polynomial mean regression.
While it mostly evolved around the RDD, recent additions to this local design literature include the RKD (e.g., Nielsen, S\o rensen, and Taber, 2010; Chen and Fan, 2011; Landais, 2015; Simonsen, Skipper, and Skipper, 2015; Card, Lee, Pei, and Weber, 2016; Dong, 2016), quantile extensions (e.g., Frandsen, Fr\"olich, and Melly, 2012; Qu and Yoon, 2015b), and their combination (Chiang and Sasaki, 2017).
While we focus on the fuzzy quantile RDD for most parts of this paper, we note that all these different frameworks are uniformly encompassed by the general framework developed in this paper.

{\bf Literature on Robust Inference:}
Calonico, Cattaneo, and Titiunik (2014) introduce bias correction to achieve the robustness of asymptotic inference against large bandwidths.
This innovation paves the way for empirical practitioners to obtain valid standard errors for their estimates under popular data-driven methods of bandwidths.\footnote{Examples include Imbens and Kalyanaraman (2012), Calonico, Cattaneo and Titiunik (2014), Arai and Ichimura (2016), Calonico, Cattaneo, and Farrell (2016ab), and Arai and Ichimura (2018). Calonico, Cattaneo, and Farrell (2016ab) propose a coverage-probability optimal bandwidth selector and provide a rule of thumb adjustment method to convert MSE-optimal bandwidths into the coverage-probability optimal ones.}
Bartalotti, Calhoun, and He (2017) adapt this idea of bias correction to a wild bootstrap method of inference for the sharp mean RDD.
Qu and Yoon (2015b) adapt this idea of bias correction to a simulation method of uniform inference for the sharp quantile RDD.

Our approach is closely related to Calonico, Cattaneo, and Titiunik (2014), Qu and Yoon (2015ab), and Bartalotti, Calhoun, and He (2017).
Calonico, Cattaneo, and Titiunik (2014) analytically develop the asymptotic distribution accounting for effects of bias estimation.
We also analytically develop the limit processes for CDF and quantile processes accounting for effects of bias estimation, which can be seen as a uniform extension to the analytic asymptotic distribution of Calonico, Cattaneo, and Titiunik (2014).
Bartalotti, Calhoun, and He (2017) come up with the idea of approximating the asymptotic distribution of Calonico, Cattaneo, and Titiunik (2014) via wild bootstrap.
We propose to approximate the limit process by the multiplier bootstrap, analogously to the wild bootstrap approximation of Bartalotti, Calhoun, and He (2017).
Qu and Yoon (2015ab) analytically develop the limit processes for local quantile processes via higher-order local polynomials, which effectively account for effects of bias estimation (Calonico, Cattaneo, and Titiunik, 2014; Remark 7).
We also analytically develop the limit process for a more general classes of local Wald estimators, which can be seen as a generalization and an extension to Qu and Yoon (2015ab).
While Qu and Yoon (2015ab) -- Chiang and Sasaki (2017) likewise -- propose a simulation method of approximating the limits of conditional quantile processes for sharp designs by exploiting a pivotal property of quantile regression, we propose to approximate the limit processes via the multiplier bootstrap because it applies to fuzzy designs as well where the joint process of the numerator and the denominator in the local Wald ratio are concerned.
In summary, our contribution relies on the ideas developed in these three previous papers.

{\bf Literature on Uniform Bahadur Representation:}
The Bahadur representation is a key to asymptotic distributional results.
To cover uniform inference for local polynomial estimators over a general index set with bias correction of any arbitrary order, uniform validity of Bahadur representation over the set is essential.
For classes of nonparametric kernel regressions on which our method relies, Masry (1996), Kong, Linton, and Xia (2010) and Fan and Liu (2016) develop uniform Bahadur representations over regressors.
Furthermore, Guerre and Sabbah (2012), Qu and Yoon (2015a), Lee, Song, and Whang (2015), Fan and Guerre (2016) develop uniform validity over quantiles as well.
We take advantage of this existing idea.
In order to deal with a general class of complexity, we use a new maximal inequality (van der Vaart and Wellner, 2011; Chernozhukov, Chetverikov, and Kato, 2014a).

{\bf Literature on Multiplier Bootstrap:}
In the broad literature, the multiplier bootstrap for Donsker and other weakly convergent classes is first studied by Ledoux and Talagrand (1988) and Gin\'e and Zinn (1990). To our knowledge, use of the multiplier bootstrap in econometrics dates back to Hansen (1996).  The multiplier bootstrap for different parametric models has been extensively studied in the literature -- it is sometimes referred to as the score bootstrap. For nonparametric CDF estimators,
Barrett and Donald (2003) and Donald, Hsu, and Barrett (2012) use the multiplier bootstrap for uniform inference on unconditional and conditional CDFs, respectively, using the exact solution of their estimators.
Chernozhukov, Chetverikov and Kato (2014) demonstrate the validity of the multiplier bootstrap for inference on suprema of
certain non-Donsker processes
without using an extreme value limit distribution.
Due to the unique nature of our local Wald estimators, our results are based on a multiplier central limit theorem developed more lately by Kosorok (2003, 2008), along with our uniform Bahadur representation.

\section{An Overview}\label{sec:overview}

In this section, we present an overview of the main result, focusing on the case of the fuzzy quantile RDD, which has not been covered by the existing literature of robust inference yet despite its use in the recent literature on empirical microeconomics, such as Clark and Martorell (2014) and Deshpande (2016).
A formal and general treatment will follow in Section \ref{sec:general_framework}.

Suppose that we observe a random sample of $(Y^\ast,D^\ast,X)$, where $X$ is the running variable or the forcing variable, $D^\ast$ is the binary treatment indicator, and $Y^\ast$ is the outcome of interest.
A researcher faces a fuzzy regression discontinuity design where the cutoff location is normalized to $X=0$ without loss of generality.
Frandsen, Fr\"olich and Melly (2012) identify the conditional CDF of the potential outcome $Y^d_i$ under each treatment status $d \in \{0,1\}$ given the event $C$ of compliance locally at $X=0$ by
\begin{equation}\label{eq:ffm_cdf}
F_{Y^d | C}(y) \ = \
\frac{\lim_{x\downarrow 0}E[ \mathbbm{1}\{Y_i^\ast \leq y\} \cdot \mathbbm{1}\{D^\ast_i=d\}  |  X_i=x] - \lim_{x\uparrow 0}E[ \mathbbm{1}\{Y_i^\ast \leq y\} \cdot \mathbbm{1}\{D^\ast_i=d\}  |  X_i=x]}{\lim_{x\downarrow 0}E[ \mathbbm{1}\{D_i^\ast=d\} |  X_i=x] - \lim_{x\uparrow 0}E[ \mathbbm{1}\{D^\ast_i=d\}  |  X_i=x]},
\end{equation}
where we omit the conditioning argument $X=0$ from our notation and thus $F_{Y^d|C}$ succinctly denotes $F_{Y^d|C,X=0}$.
Consequently, the local $\theta$-th quantile treatment effect is identified by 
\begin{equation}\label{eq:ffm_qte}
\tau_{}(\theta)
=:Q_{Y^1|C}(\theta)-Q_{Y^0|C}(\theta);\qquad \text{ where $Q_{Y^d|C}(\theta):=\inf\left\{y : F_{Y^d | C}(y) \geq \theta\right\}$},
\end{equation}
where we again note that the conditioning argument $X=0$ is omitted from our notations.
Frandsen, Fr\"olich and Melly (2012) develop methods of inference for $\tau_{}$ based on the exact solutions of local linear estimation of the components of the local Wald ratio (\ref{eq:ffm_cdf}).

In order to make an inference with the local linear estimation, one would need to choose an under-smoothing bandwidth parameter $h_n$.
However, commonly available procedures choose rather large bandwidths, e.g., $h_n \propto n^{-1/5}$.
To accommodate these common procedures in the framework of Frandsen, Fr\"olich and Melly (2012), we need to estimate higher-order bias and to develop the limit process accounting for this bias estimation, as in Calonico, Cattaneo, and Titiunik (2014).
In this section, we present how to make uniform inference for $\tau_{}$ specifically based on local quadratic estimation of the components of the local Wald ratio (\ref{eq:ffm_cdf}), effectively accounting for the second-order bias estimation -- see Remark 7 of Calonico, Cattaneo, and Titiunik (2014).
Consequently, the uniform inference turns robust against large bandwidths as in the commonly available procedures, e.g., $h_n \propto n^{-1/5}$.

The right-hand limit, $\lim_{x\downarrow 0}E[ \mathbbm{1}\{Y_i^\ast \leq y\} \cdot \mathbbm{1}\{D^\ast_i=d\}  |  X_i=x]$, in the numerator of the local Wald ratio (\ref{eq:ffm_cdf}) can be estimated by $\hat\mu_{1}(0^+,y,d)$ in the local quadratic estimator
\begin{align*}
&\left(\hat\mu_{1}(0^+,y,d),\hat\mu_{1}'(0^+,y,d),\hat\mu_{1}''(0^+,y,d)\right) = \\
&\arg\min_{(\mu,\mu',\mu'')} \sum_{i:X_i>0} \left( \mathbbm{1}\left\{Y_i^\ast \leq y\right\} \cdot \mathbbm{1}\left\{D_i^\ast=d\right\} - \left\{\mu + \mu' X_i + \frac{\mu''}{2!} X_i^2 \right\} \right)^2 \cdot K\left(\frac{X_i}{h_{n}}\right),
\end{align*}
where $K$ is a kernel function and $h_{n}$ is a bandwidth parameter.
The left-hand limit, $\lim_{x\uparrow 0}E[ \mathbbm{1}\{Y_i^\ast \leq y\} \cdot \mathbbm{1}\{D^\ast_i=d\}  |  X_i=x]$, in the numerator of the local Wald ratio (\ref{eq:ffm_cdf}) can be similarly estimated by $\hat\mu_{1}(0^-,y,d)$ using the observations $\{i:X_i < 0 \}$.
Likewise, the right-hand limit, $\lim_{x\downarrow 0}E[ \mathbbm{1}\{D^\ast_i=d\}  |  X_i=x]$, in the denominator of (\ref{eq:ffm_cdf}) can be estimated by $\hat\mu_{2}(0^+,d)$ in the local quadratic estimator
\begin{align*}
&\left(\hat\mu_{2}(0^+,d),\hat\mu_{2}'(0^+,d),\hat\mu_{2}''(0^+,d)\right) = \\
&\arg\min_{(\mu,\mu',\mu'')} \sum_{i:X_i>0} \left( \mathbbm{1}\left\{D_i^\ast=d\right\} - \left\{\mu + \mu' X_i + \frac{\mu''}{2!} X_i^2 \right\} \right)^2 \cdot K\left(\frac{X_i}{h_{n}}\right),
\end{align*}
where we can use the same bandwidth $h_{n}$ as above for simplicity here.
The left-hand limit, $\lim_{x\uparrow 0}E[ \mathbbm{1}\{D^\ast_i=d\}  |  X_i=x]$, in the denominator of the local Wald ratio (\ref{eq:ffm_cdf}) can be similarly estimated by $\hat\mu_{2}(0^-,d)$ using the observations $\{i:X_i < 0 \}$.
With these component estimates, the estimand (\ref{eq:ffm_qte}) for the identified local quantile treatment effect may be estimated by
\begin{align*}
\hat\tau_{}(\theta) =
\hat Q_{Y^1|C}(\theta)
-
\hat Q_{Y^0|C}(\theta),
\end{align*}
where
$$
\hat Q_{Y^d|C}(\theta) = 
\inf\left\{y : \frac{\hat\mu_{1}(0^+,y,d) - \hat\mu_{1}(0^-,y,d)}{\hat\mu_{2}(0^+,d) - \hat\mu_{2}(0^-,d)} \geq \theta\right\}
\quad\text{for each } d \in \{0,1\}.
$$

Under suitable conditions, there exists a zero mean Gaussian process $\mathds{G}'_{}$ such that
\begin{equation}\label{eq:overview_weak_conv}
\sqrt{nh_n}[\hat{\tau}_{}(\cdot)-\tau_{}(\cdot)]\leadsto \mathds{G}'_{}(\cdot) \quad\text{as } n \rightarrow \infty.
\end{equation}
See Corollary \ref{corollary:FQRD} (i) ahead for formal and general arguments.
This result establishes the asymptotic distribution result for the fuzzy quantile RDD process with robustness against large bandwidth, e.g., $h_{n} \propto n^{-1/5}$.
In practice, it will be somewhat easier to approximate the limit process $\mathds{G}'_{}$ by the multiplier bootstrap procedure outlined below.

Let $\{\xi_i\}_{i=1}^n$ be a random sample drawn from the standard normal distribution independently from the data $\left\{(Y^*_i,D^*_i,X_i)\right\}_{i=1}^n$.
Letting $\Gamma^\pm_2=\int_{\mathds{R}_\pm}(1 \ u \ u^2)' \cdot K(u) \cdot (1 \ u \ u^2)du$ be a $3 \times 3$ matrix,
we define the estimated multiplier processes for $\hat\mu_{1}(0^+,y,d_1)$ and $\hat\mu_{2}(0^+,d_2)$ by
\begin{align*}
\hat{\nu}^+_{\xi,n}(y,d,1)=\sum_{i:X_i>0}\xi_i\frac{(1 \ 0 \ 0) \cdot (\Gamma^\pm_2)^{-1} \cdot \left(1 \ \frac{X_i}{h_n} \ \frac{X_i^2}{h_n^2}\right)' [\mathds{1}\{Y^*_i\le y,D^*_i=d\}-\tilde{\mu}_{1}(X_i,y,d)]K\left(\frac{X_i}{h_n}\right)}{\sqrt{nh_n}\hat{f}_X(0)}
\\\text{and}\qquad
\hat{\nu}^+_{\xi,n}(d,2)=\sum_{i:X_i>0}\xi_i\frac{(1 \ 0 \ 0) \cdot (\Gamma^\pm_2)^{-1} \cdot \left(1 \ \frac{X_i}{h_n} \ \frac{X_i^2}{h_n^2}\right)' [\mathds{1}\{D^*_i=d\}-\tilde{\mu}_{2}(X_i,d)]K\left(\frac{X_i}{h_n}\right)}{\sqrt{nh_n}\hat{f}_X(0)},
\end{align*}
respectively, where $\hat f_X(0)$ estimates $f_X(0)$, $\tilde\mu_{1}(X_i,y,d)$ estimates $E\left[\mathbbm{1}\left\{Y_i^\ast \leq y\right\} \cdot \mathbbm{1}\left\{D_i^\ast=d\right\} | X_i\right]$, and $\tilde\mu_{2}(X_i,d)$ estimates $E\left[ \mathbbm{1}\left\{D_i^\ast=d\right\} | X_i\right]$. 
Specifically, one can choose any consistent kernel density estimator for $\hat f_X(0)$, and concrete examples of the estimators, $\tilde\mu_{1}(X_i,y,d)$ and $\tilde\mu_{2}(X_i,d)$, are provided in Appendix \ref{sec:first_stage}.
The estimated multiplier processes for $\hat\mu_{1}(0^-,y,d_1)$ and $\hat\mu_{2}(0^-,d_2)$ are similarly defined by $\hat{\nu}^-_{\xi,n}(y,d,1)$ and $\hat{\nu}^-_{\xi,n}(d,2)$ using the observations $\{i:X_i < 0\}$.
By the Hadamard derivative, we construct the approximate estimated multiplier process
\begin{align}
\widehat{\mathbb{G}}_{n}'(\theta)
=&\frac{\hat{\mu}_{2} (0^+,1)-\hat{\mu}_{2} (0^-,1)}{\hat{f}_{Y^1|C}(\hat{Q}_{Y^1|C}(\theta))[\hat{\mu}_{2} (0^+,1)-\hat{\mu}_{2} (0^-,1)]^2} \cdot \left\{ {\hat{\nu}^+_{\xi,n}(\hat{Q}_{Y^1|C}(\theta),1,1)} - {\hat{\nu}^-_{\xi,n}(\hat{Q}_{Y^1|C}(\theta),1,1)} \right\} \nonumber\\
-&\frac{\hat{\mu}_{1}(0^+,\hat{Q}_{Y^1|C}(\theta),1)-\hat{\mu}_{1}(0^-,\hat{Q}_{Y^1|C}(\theta),1)}{\hat{f}_{Y^1|C}(\hat{Q}_{Y^1|C}(\theta))[\hat{\mu}_{2} (0^+,1)-\hat{\mu}_{2} (0^-,1)]^2} \cdot \left\{ {\hat{\nu}^+_{\xi,n}(1,2)} - {\hat{\nu}^-_{\xi,n}(1,2)} \right\}\nonumber \\
-&\frac{\hat{\mu}_{2} (0^+,0)-\hat{\mu}_{2} (0^-,0)}{\hat{f}_{Y^0|C}(\hat{Q}_{Y^0|C}(\theta))[\hat{\mu}_{2} (0^+,0)-\hat{\mu}_{2} (0^-,0)]^2} \cdot \left\{ {\hat{\nu}^+_{\xi,n}(\hat{Q}_{Y^0|C}(\theta),0,1)} - {\hat{\nu}^-_{\xi,n}(\hat{Q}_{Y^0|C}(\theta),0,1)} \right\} \nonumber\\
+&\frac{\hat{\mu}_{1}(0^+,\hat{Q}_{Y^0|C}(\theta),0)-\hat{\mu}_{1}(0^-,\hat{Q}_{Y^0|C}(\theta),0)}{\hat{f}_{Y^0|C}(\hat{Q}_{Y^0|C}(\theta))[\hat{\mu}_{2} (0^+,0)-\hat{\mu}_{2} (0^-,0)]^2} \cdot \left\{ {\hat{\nu}^+_{\xi,n}(0,2)} - {\hat{\nu}^-_{\xi,n}(0,2)} \right\}. \label{eq:approx_estimated_MP}
\end{align}
Under suitable conditions, with probability approaching one, this process $\widehat{\mathbb{G}}_{n}'$ weakly converges to the limit process, $\mathds{G}'_{}$, of interest conditionally on the data $\left\{(Y^\ast_i,D^\ast_i,X_i)\right\}_{i=1}^n$, i.e.,
\begin{equation}\label{eq:cond_overview_weak_conv}
\widehat{\mathbb{G}}_{n}' \underset{\xi}{\overset{p}{\leadsto}} \mathds{G}'_{} \quad\text{as } n \rightarrow \infty.
\end{equation}
See Corollary \ref{corollary:FQRD} (ii) ahead for formal and general arguments.
From (\ref{eq:overview_weak_conv}) and (\ref{eq:cond_overview_weak_conv}), therefore, we may use $\widehat{\mathbb{G}}_{n}'$ to approximate the limit process of $\sqrt{nh_n}[\hat{\tau}_{}-\tau_{}]$.

One of the most relevant practical applications of this result is to test the null hypothesis of uniform treatment nullity:
$$
H_0: \tau_{}(\theta) = 0 \quad\text{for all } \theta \in [a,1-a]
$$
for some $\alpha \in (0,1/2)$.
To test this hypothesis, we can use $\sup_{\theta \in [a,1-a]} \sqrt{nh_n} \abs{ \hat\tau_{}(\theta) }$ as the test statistic, and use
$$
\sup_{\theta \in [a,1-a]} \abs{\widehat{\mathbb{G}}_{n}'(\theta)}
$$
to simulate its asymptotic distribution.

Another of the most relevant practical applications of the above corollary is to test the null hypothesis of treatment homogeneity across quantiles:
$$
H_0: \tau_{}(\theta) = \tau_{}(\theta') \quad\text{for all } \theta, \theta' \in [a,1-a].
$$
To test this hypothesis, we can use $\sup_{\theta \in [a,1-a]} \sqrt{nh_n} \abs{ \hat\tau_{}(\theta) - (1-2a)^{-1} \int_{[a,1-a]} \hat\tau_{}(\vartheta)d\vartheta }$ as the test statistic, and use
\begin{align*}
\sup_{\theta \in [a,1-a]} \left\vert
\widehat{\mathbb{G}}_{n}'(\theta) \right.
\left. -
\frac{1}{1-2a} \int_{[a,1-a]} \widehat{\mathbb{G}}_{n}'(\vartheta) d\vartheta
\right\vert
\end{align*}
to simulate its asymptotic distribution.

Finally, we can use the approximate estimated multiplier process to construct uniform confidence bands for the quantile treatment effects.
To this end, we compute
\begin{align*}
\hat{\mathcal C}_{n}(a,1-a;\lambda) = \text{the $(1-\lambda)$-th quantile of } \sup_{\theta \in [a,1-a]} \abs{\widehat{\mathbb{G}}_{n}'(\theta)}.
\end{align*}
The band of the form
\begin{align*}
\left[\hat\tau_{}(\theta) \pm \frac{1}{\sqrt{nh_n}}\hat{\mathcal C}_{n}(a,1-a;\lambda) \ : \ \theta \in [a,1-a] \right]
\end{align*}
constitutes a $100(1-\lambda)$ percent uniform confidence band for the quantile treatment effects $\tau_{}$.
In summary, we provide a step-by-step procedure below.

\begin{algorithm}[Practical Guideline on Constructing Uniform Confidence Bands]
${}$
\begin{enumerate}[Step 1.]
    \item Pick a finite set $\mathcal Y^* \subset \mathcal Y$ of grid points of outcome values and a finite set $\mathcal T^*\subset [a,1-a]$ of grid points of quantiles. 
		Estimate $\hat \mu_1(0^\pm,y,d)$ and $\hat \mu_2(0^\pm,d)$ for all $y\in \mathcal Y^*$, $d\in \{0,1\}$. 
    \item Calculate $\hat Q_{Y^d|C}(\theta)$ for each $\theta\in \mathcal T^*$ by
    \begin{align*}
\hat Q_{Y^d|C}(\theta) =
\inf\left\{y \in \mathcal Y^*  : \frac{\hat\mu_{1}(0^+,y,d) - \hat\mu_{1}(0^-,y,d)}{\hat\mu_{2}(0^+,d) - \hat\mu_{2}(0^-,d)} \geq \theta\right\}
\end{align*}
for $d\in \{0,1\}$, and then compute $\hat \tau_{}=\hat Q_{Y^1|C}(\theta)-\hat Q_{Y^1|C}(\theta)$ for each $\theta \in \mathcal T^*$.
	\item Compute $\hat f_X(0)$ and $\hat f_{Y^d|C}(\hat Q_{Y^d|C}(\theta))$ -- see Appendix \ref{sec:first_stage} for formulas.
	\item For each bootstrap iteration $b=1,...,B$, generate independent standard normal $\xi^b=\{\xi^b_i\}_{i=1}^n$ independently from data, and compute $\hat \nu_{\xi^b,n}(\theta,d,1)$ and $\hat \nu_{\xi^b,n}(d,2)$ for each $\theta \in \mathcal T^*$, $d\in \{0,1\}$. 
	(Note in each iteration $b$, we compute these values for different $\theta$ and $d$ using the same $\xi^b$.)
	\item Construct $\hat{\mathbb G}_{n,b}^{\prime}$ for each $\theta\in \mathcal T^*$:
	\begin{align*}
	\widehat{\mathbb{G}}^{\prime}_{n,b}(\theta) =
&\frac{\hat{\mu}_{2} (0^+,1)-\hat{\mu}_{2} (0^-,1)}{\hat{f}_{Y^1|C}(\hat{Q}_{Y^1|C}(\theta))[\hat{\mu}_{2} (0^+,1)-\hat{\mu}_{2} (0^-,1)]^2} \cdot \left\{ {\hat{\nu}^+_{\xi^b,n}(\hat{Q}_{Y^1|C}(\theta),1,1)} - {\hat{\nu}^-_{\xi^b,n}(\hat{Q}_{Y^1|C}(\theta),1,1)} \right\} \\
-&\frac{\hat{\mu}_{1}(0^+,\hat{Q}_{Y^1|C}(\theta),1)-\hat{\mu}_{1}(0^-,\hat{Q}_{Y^1|C}(\theta),1)}{\hat{f}_{Y^1|C}(\hat{Q}_{Y^1|C}(\theta))[\hat{\mu}_{2} (0^+,1)-\hat{\mu}_{2} (0^-,1)]^2} \cdot \left\{ {\hat{\nu}^+_{\xi^b,n}(1,2)} - {\hat{\nu}^-_{\xi^b,n}(1,2)} \right\} \\
-&\frac{\hat{\mu}_{2} (0^+,0)-\hat{\mu}_{2} (0^-,0)}{\hat{f}_{Y^0|C}(\hat{Q}_{Y^0|C}(\theta))[\hat{\mu}_{2} (0^+,0)-\hat{\mu}_{2} (0^-,0)]^2} \cdot \left\{ {\hat{\nu}^+_{\xi^b,n}(\hat{Q}_{Y^0|C}(\theta),0,1)} - {\hat{\nu}^-_{\xi^b,n}(\hat{Q}_{Y^0|C}(\theta),0,1)} \right\} \\
+&\frac{\hat{\mu}_{1}(0^+,\hat{Q}_{Y^0|C}(\theta),0)-\hat{\mu}_{1}(0^-,\hat{Q}_{Y^0|C}(\theta),0)}{\hat{f}_{Y^0|C}(\hat{Q}_{Y^0|C}(\theta))[\hat{\mu}_{2} (0^+,0)-\hat{\mu}_{2} (0^-,0)]^2} \cdot \left\{ {\hat{\nu}^+_{\xi^b,n}(0,2)} - {\hat{\nu}^-_{\xi^b,n}(0,2)} \right\}.
\end{align*}
\item Set
$
\hat{\mathcal C}^B_{n}(a,1-a;\lambda) = \text{the $(1-\lambda)$-th quantile of } \Big\{\max_{\theta \in  \mathcal T^*} \abs{\widehat{\mathbb{G}}_{n,b}^{\prime}(\theta)}\Big\}^B_{b=1},
$
and construct an asymptotically valid $100(1-\lambda)$ percent uniform confidence band over $[a,1-a] $by
\begin{align*}
\left[\hat\tau_{}(\theta) \pm \frac{1}{\sqrt{nh_n}}\hat{\mathcal C}^B_{n}(a,1-a;\lambda) \ : \ \theta \in \mathcal T^* \right].
\end{align*}
\end{enumerate}
\end{algorithm}
\begin{remark}
Steps 1--5 also give the bootstrapped distribution $\{\hat{\mathbb G}_{n,b}'\}_{b=1}^B$, which can be used to construct critical values for tests of uniform treatment nullity and treatment homogeneity as well. For the null hypothesis of uniform treatment nullity, use $\max_{\theta \in [a,1-a]} \sqrt{nh_n} \abs{ \hat\tau_{}(\theta) }$ as the test statistic, and use $(1-\lambda)$-th quantile of
$$
\Big\{\max_{\theta \in [a,1-a]} \abs{\widehat{\mathbb{G}}_{n,b}'(\theta)}\Big\}^B_{b=1}
$$
as the critical value. For the null hypothesis of treatment homogeneity, let
$$\max_{\theta \in \mathcal [a,1-a]}  \sqrt{nh_n} \Big| \hat\tau_{}(\theta) - (1-2a)^{-1} \int_{[a,1-a]} \hat\tau_{}(\vartheta)d\vartheta \Big|$$ 
be the test statistic, and the $(1-\lambda)$-th quantile of
\begin{align*}
\Big\{
\max_{\theta \in \mathcal [a,1-a]} \Big|
\widehat{\mathbb{G}}_{n,b}'(\theta) 
 -
\frac{1}{1-2a} \int_{[a,1-a]} \widehat{\mathbb{G}}_{n,b}'(\vartheta) d\vartheta
 \Big| \Big\}_{b=1}^B
\end{align*}
be the critical value for the test.
\end{remark}
\begin{remark}
In practice we may set $|\mathcal T^*|=50$ and $|\mathcal Y^*|=5000$.  
We tried a local polynomial mean regression for each of the $5000$ grid points of $y\in \mathcal Y^*$, and found that it is not computationally burdensome in general, since local polynomial mean regressions are smooth convex problems and gradient decent algorithms can solve them efficiently.
\end{remark}

\section{A Unified Framework}\label{sec:general_framework}

In this section, we present a generalized framework for a broad class of local Wald estimands that encompasses not only the case of the fuzzy quantile RDD, but also the cases of the sharp mean RDD, the fuzzy mean RDD, the sharp mean RKD, the fuzzy mean RKD, the sharp CDF discontinuity design, the fuzzy CDF discontinuity design, the sharp quantile RDD, the sharp quantile RKD, and the fuzzy quantile RKD.
We revisit the case of the fuzzy quantile RDD in Section \ref{sec:a:FQRD} to provide a formal justification of the overview in Section \ref{sec:overview}.
All the other examples are relegated to Appendix \ref{sec:additional_examples} and Appendix \ref{sec:application_examples}.

\subsection{The General Framework}

Let $(Y,D,X)$ be a random vector defined on a probability space $(\Omega^x,\mathcal{F}^x,\mathds{P}^x)$, where $Y$ is a random vector containing an outcome and possibly other variables, $D$ is a random vector containing a treatment indicator and possibly others, and $X$ is a running variable or an assignment variable. We denote their supports as $\mathscr{Y}$, $\mathscr{D}$ and $\mathscr{X}$, respectively.
Suppose that a researcher observes $n$ i.i.d. copies $\{(Y_i,D_i,X_i)\}^n_{i=1}$ of $(Y,D,X)$.
Consider some subsets of some finite dimensional Euclidean spaces $\Theta_1$, $\Theta_2$, $\Theta_1'$, $\Theta_2'$, and $\Theta''$. We will use them to denote sets of indices. Let $\Theta=\Theta_1 \times \Theta_2$, and let $g_1:\mathscr{Y}\times \Theta_1 \rightarrow \mathds{R}$ and $g_2:\mathscr{D}\times \Theta_2 \rightarrow \mathds{R}$ be functions to be defined in various contexts of empirical research designs -- concrete examples are suggested in the subsections below and in the Supplementary Appendix. When we discuss the continuity of $g_k$ in $\theta_k$, $k \in \{1,2\}$, we consider $\Theta_1$ and $\Theta_2$
with the topologies they inherit from the finite dimensional Euclidean spaces they reside in.
We write $\mu_{1}(x,\theta_1)=E[g_1(Y_i,\theta_1)|X_i=x]$ and $\mu_{2}(x,\theta_2)=E[g_2(D_i,\theta_2)|X_i=x]$.
Their $v$-th order partial derivatives with respect to $x$ are denoted by $\mu^{(v)}_{1} =\frac{\partial^v}{\partial x^v}\mu_{1}$ and $\mu^{(v)}_{2} =\frac{\partial^v}{\partial x^v}\mu_{2}$. 
 For a set $T$, we denote $\mathcal{C}^1(T)$ as the collection of all real-valued functions on $T$ that are continuously differentiable, and $\ell^\infty(T)$ is the collection of all bounded real-valued functions on $T$.
With suitable operators
$\phi:\ell^\infty(\Theta_1) \rightarrow \ell^\infty(\Theta_1')$,
$\psi:\ell^\infty(\Theta_2) \rightarrow \ell^\infty(\Theta_2')$, and
$\Upsilon:\ell^\infty(\Theta_1' \times \Theta_2')\rightarrow \ell^\infty(\Theta'')$,
a general class of local Wald estimands can be expressed in the form of
\begin{align}\label{eq:wald_est}
&\tau(\theta'')=\Upsilon\Bigg(\frac{\phi\left( \lim_{x\downarrow 0}\mu^{(v)}_{1}(x,\cdot)\right)(\cdot)-\phi\left(\lim_{x\uparrow 0}\mu^{(v)}_{1}(x,\cdot)\right)(\cdot)}{\psi\left(\lim_{x\downarrow 0}\mu^{(v)}_{2}(x,\cdot)\right)(\cdot)-\psi\left(\lim_{x\uparrow 0}\mu^{(v)}_{2}(x,\cdot)\right)(\cdot)}\Bigg)(\theta'').
\end{align}
for all $\theta''\in \Theta''$.
This class of local Wald estimands encompasses a wide array of design-based estimands used by empirical practitioners.
In all examples, $v$ is either 0 (for RDD) or 1 (for RKD), and setting the order of local polynomial estimator $p=v+2$ would generally suffice in practice.
We list two examples below: 
Example \ref{ex:fuzzy_rdd} illustrates the case of the fuzzy mean RDD; 
Example \ref{ex:fuzzy_quantile_rdd} illustrates the case of the fuzzy quantile RDD for which the existing literature has proposed no robust uniform inference methods -- also overviewed in Section \ref{sec:overview}.
See Appendix \ref{sec:application_examples} for additional examples.
For convenience of writing, we introduce the notation for the intermediate local Wald estimand:
$$
W=\frac{\phi(\mu^{(v)}_{1}(0^+,\cdot))-\phi(\mu^{(v)}_{1}(0^-,\cdot))}
{\psi(\mu^{(v)}_{2}(0^+,\cdot))-\psi(\mu^{(v)}_{2}(0^-,\cdot))}.
$$

\begin{example}[Fuzzy Mean RDD]\label{ex:fuzzy_rdd}
We do not need index sets for fuzzy mean RDD, so let $\Theta_1 = \Theta_2 = \Theta_1' = \Theta_2' = \Theta'' = \{0\}$ for simplicity.
Set $g_1(Y_i,\theta_1) = Y_i$ and $g_2(D_i,\theta_2) = D_i$.
Note that $\mu_{1}(x,\theta_1) = \E[g_1(Y_i,\theta_1)  |  X_i=x] = \E[Y_i  |  X_i=x]$ and $\mu_{2}(x,\theta_2) = \E[g_2(D_i,\theta_2)  |  X_i=x] = \E[ D_i  |  X_i=x]$.
Let $\phi$ and $\psi$ be the identity operators, and for $W \in \ell^\infty (\Theta'_1 \times \Theta'_2)$ let the operator $\Upsilon$ be $\Upsilon\left( W\right)(\theta'') = W(\theta'',\theta'')$ $\forall \theta'' \in \Theta''$.
The local Wald estimand (\ref{eq:wald_est}) with $v=0$ in this setting becomes
\begin{eqnarray}
\tau(\theta'')
&=&\frac{\lim_{x\downarrow 0}\E[ Y_i  |  X_i=x]-\lim_{x\uparrow 0}\E[ Y_i  |  X_i=x]}{\lim_{x\downarrow 0}\E[ D_i  |  X_i=x]-\lim_{x\uparrow 0}\E[ D_i  |  X_i=x]}
\label{eq:fmrd}
\end{eqnarray}
for all $\theta'' \in \Theta'' = \{0\}$.
This estimand $\tau(0)$ will be denoted by $\tau_{FMRD}$ for Fuzzy Mean RD design.
\qed
\end{example}

\begin{example}[Fuzzy Quantile RDD]\label{ex:fuzzy_quantile_rdd}
Consider the estimand (\ref{eq:ffm_cdf}) proposed by Frandsen, Fr\"olich and Melly (2012) to identify the conditional CDF of potential outcome $Y^d_i$ under each treatment status $d \in \{0,1\}$ given the event $C$ of compliance, and the quantile treatment effect (\ref{eq:ffm_qte}) given this event $C$.
These estimands also fit in the general framework (\ref{eq:wald_est}).
We first fix an $a\in(0,1/2)$, $\varepsilon>0$ and let $ \mathscr{Y}_1=[Q_{Y^1|C}(a)-\varepsilon,Q_{Y^1|C}(1-a)+\varepsilon]\cup [Q_{Y^0|C}(a)-\varepsilon,Q_{Y^0|C}(1-a)+\varepsilon]$. Let $\Theta_1 = \Theta_1' = \mathscr{Y}_1 \times \mathscr{D}$ and $\Theta_2 = \Theta_2' = \mathscr{D}$ for $\mathscr{D} = \{0,1\}$, and let $\Theta'' = [a,1-a]$ for a constant $a \in (0,1/2)$.
Let $Y_i = (Y_i^\ast, D_i^\ast)$ and $D_i = D_i^\ast$.
Set $g_1((Y^\ast_i,D^\ast_i),(y,d)) = \mathbbm{1}\{Y_i^\ast \leq y\} \cdot \mathbbm{1}\{D^\ast_i=d\}$ and $g_2(D_i^\ast,d) = \mathbbm{1}\{D^\ast_i=d\}$.
Note that $\mu_{1}(x,y,d) = \E[g_1((Y_i^\ast,D_i^\ast),y,d)  |  X_i=x] = \E[ \mathbbm{1}\{Y_i^\ast \leq y\} \cdot \mathbbm{1}\{D_i^\ast=d\}  |  X_i=x]$ and $\mu_{2}(x,d) = \E[g_2(D_i^\ast,d)  |  X_i=x] = \E[ \mathbbm{1}\{D_i^\ast=d\}  |  X_i=x]$.
Let $\phi$ and $\psi$ be the identity operators, and define $\Upsilon$ for each $W\in \ell^\infty(\Theta'_1 \times \Theta'_2)=\ell^\infty(\mathscr{Y}\times \mathscr{D}^2)$ by
$
\Upsilon( W)(\theta'')
\ = \
\inf\{y \in \mathscr{Y} : W(y,1,1) \geq \theta'' \}
-
\inf\{y \in \mathscr{Y} : W(y,0,0) \geq \theta'' \}.
$
The local Wald estimand (\ref{eq:wald_est}) with $v=0$ in this setting becomes (\ref{eq:ffm_qte}), i.e.,
\begin{eqnarray}
\tau(\theta'')
&=&
Q_{Y^1 | C}(\theta'')
-
Q_{Y^0 | C}(\theta'')
\nonumber
\end{eqnarray}
for all $\theta'' \in \Theta'' = [a,1-a]$, where $Q_{Y^d | C}(\theta'') := \inf\{y \in \mathscr{Y} : F_{Y^d | C}(y) \geq \theta'' \}$ for a short-hand notation, and $F_{Y^d | C}(y)$ is given in (\ref{eq:ffm_cdf}) for all $(y,d) \in \mathscr{Y} \times \mathscr{D}$.
This estimand $\tau$ will be denoted by $\tau_{FQRD}$ for Fuzzy Quantile RD design.\qed
\end{example}

We introduce some short-hand notations for conservation of space.
Let
$$
\Epsilon_{k}(y,d,x,\theta)=g_k(y,\theta_k)-\mu_{k}(x,\theta_k)
$$
for $(\theta_1,\theta_2) \in \Theta$, $k \in \{1,2\}$, $y \in \mathscr{Y}$, $d \in \mathscr{D}$, and $x \in \mathscr{X}$.
Let
$$
\sigma_{k l}(\theta,\vartheta|x)=E[\Epsilon_k(Y_i,D_i,X_i,\theta) \ \Epsilon_{l}(Y_i,D_i,X_i,\vartheta)|X_i=x]
$$
denote the conditional covariance of residuals for each $\theta=(\theta_1,\theta_2)$, $\vartheta=(\vartheta_1,\vartheta_2) \in \Theta$, and $k,l \in \{1,2\}$. Also define the product space $\mathds{T}=\Theta \times \{1,2\}=(\Theta_1 \times \Theta_2) \times \{1,2\}$.
We will also use the following short-hand notations for functions at right- and left-hand limits: $\mu^{(v)}_{k}(0^+,\theta)=\lim_{x\downarrow 0}\mu^{(v)}_k(x,\theta)$ and $\mu^{(v)}_{k}(0^-,\theta)=\lim_{x\uparrow 0}\mu^{(v)}_k(x,\theta)$ for $k=\{1,2\}$.
The composite notation $\mu^{(v)}_{k}(0^\pm,\theta)$ is used to collectively refer to $\mu^{(v)}_{k}(0^+,\theta)$ and $\mu^{(v)}_{k}(0^-,\theta)$.
Let $r_p(x)=(1,x,...,x^p)'$.
Let $K$ denote a kernel function, and let $(h_{1,n}(\theta_1),h_{2,n}(\theta_2))$ denote bandwidth parameters that depend on $\theta = (\theta_1,\theta_2) \in \Theta$ and the sample size $n \in \mathbb{N}$.
For $p \in \mathbb{N}$, let $e_v$ denote the $v$-th standard basis element of $\mathbb{R}^p$.
We write $\Gamma^\pm_p=\int_{\mathds{R}_\pm}K(u)r_p(u)r'_p(u)du$ and $\Lambda^\pm_{p,q}=\int_{\mathds{R}_\pm}u^qK(u)r_p(u)du$. 
We use the notation $\leadsto$ to denote weak convergence, and the notation $\underset{\xi}{\overset{p}{\leadsto}}$ for conditional weak convergence
as defined in Section 2.2.3 of Korosok (2008) -- see Appendix \ref{sec:prel_lemmas} for more details.
Let $v$, $p$, $q\in \mathds{N}_+$ with $v\le p$.
We will use $v$ for the order of derivative of interest as in (\ref{eq:wald_est}), and $p$ stands for the order of local polynomial fitting to estimate (\ref{eq:wald_est}).

\subsection{The Local Wald Estimator}\label{sec:local_wald_estimator}

In this section, we develop an estimator for the nonparametric components $\mu_k^{(v)}(0^\pm, \cdot)$, $k\in\{1,2\}$, of the local Wald estimand (\ref{eq:wald_est}) based on local polynomial fitting with the bias correction approach proposed by Calonico, Cattaneo and Titiunik (2014) -- for a comprehensive treatment for local polynomial models, see Fan and Gijbels (1996).
Under proper smoothness assumptions to be formally stated below, the $p$-th order approximations
\begin{align*}
\mu_{k}(x,\theta_k)
\approx \mu_{k}(0^+,\theta_k)+\mu^{(1)}_{k}(0^+,\theta_k)x+...+\frac{\mu^{(p)}_{k}(0^+,\theta_k)}{p!}x^p
= r_p(x/h)' \alpha_{k+,p}(\theta_k)
\qquad x > 0
\\
\mu_{k}(x,\theta_k)
\approx \mu_{k}(0^-,\theta_k)+\mu^{(1)}_{k}(0^-,\theta_k)x+...+\frac{\mu^{(p)}_{k}(0^-,\theta_k)}{p!}x^p
= r_p(x/h)' \alpha_{k-,p}(\theta_k)
\qquad x < 0
\end{align*}
hold for each $k\in\{1,2\}$, where $\alpha_{k\pm,p}(\theta)=[\mu_{k}(0^\pm,\theta_k)/0!,\mu^{(1)}_{k}(0^\pm,\theta_k)h/1!,...,\mu^{(p)}_{k}(0^\pm,\theta_k)h^p/p!]$ and $h>0$.
To estimate $\alpha_{k\pm,p}(\theta_k)$, we solve the one-sided local weighted least squares problems
\begin{align}
\hat{\alpha}_{1\pm,p}(\theta_1)&=\argmin_{\alpha\in \mathds{R}^{p+1}} \sum_{i=1}^{n} \delta_i^\pm \Big(g_1(Y_i,\theta_1)-r_p\Big(\frac{X_i}{h_{1,n}(\theta_1)}\Big)'\alpha\Big)^2 K\left(\frac{X_i}{h_{1,n}(\theta_1)}\right)\\
\hat{\alpha}_{2\pm,p}(\theta_2)&=\argmin_{\alpha\in \mathds{R}^{p+1}} \sum_{i=1}^{n} \delta_i^\pm \Big(g_2(D_i,\theta_2)-r_p\Big(\frac{X_i}{h_{2,n}(\theta_2)}\Big)'\alpha\Big)^2 K\left(\frac{X_i}{h_{2,n}(\theta_2)}\right),
\label{eq:local_poly_reg}
\end{align}
where $\delta^+_i=\mathds{1}\{X_i\ge 0\}$ and $\delta^-_i=\mathds{1}\{X_i\le 0\}$.
We let the coordinates of these estimates be written by
\begin{align*}
&\hat{\alpha}_{k\pm,p}(\theta_k)=[\hat{\mu}_{k,p}(0^\pm,\theta_k)/0!,\hat{\mu}^{(1)}_{k,p}(0^\pm,\theta_k)h_{k,n}(\theta_k)/1!,...,\hat{\mu}^{(p)}_{k,p}(0^\pm,\theta)h^p_{k,n}(\theta_k)/p!]'.
\end{align*}

With these component estimates, the local Wald estimand (\ref{eq:wald_est}) is in turn estimated by the plug-in estimator.
\begin{align}\label{eq:wald_estimator}
&\hat\tau(\theta'')=\Upsilon\Bigg(\frac{\phi\left( \hat\mu^{(v)}_{1,p}(0^+,\cdot)\right)(\cdot)-\phi\left(\hat\mu^{(v)}_{1,p}(0^-,\cdot)\right)(\cdot)}{\psi\left(\hat\mu^{(v)}_{2,p}(0^+,\cdot)\right)(\cdot)-\psi\left(\hat\mu^{(v)}_{2,p}(0^-,\cdot)\right)(\cdot)}\Bigg)(\theta'')\qquad\text{for each $\theta''\in \Theta''$.}
\end{align}

\subsection{Weak Convergence}\label{sec:weak_convergence}

In this section, we establish the weak convergence result for the process $\sqrt{nh_n^{1+2v}} [\hat\tau-\tau]$ for some bandwidth $h_n$.
To this end we state the following set of assumptions.
\begin{assumption}[Uniform Bahadur Representation]\label{a:BR}
Let $\underline{x} < 0 < \overline{x}$,\\
(i) (a) $\{(Y_i,D_i,X_i)\}^n_{i=1}$ are $n$ i.i.d. copies of random vector $(Y,D,X)$ defined on a probability space $(\Omega^x,\mathcal{F}^x,\mathds{P}^x)$;
(b) $X$ has a density function $f_X$ which is continuously differentiable on $[\underline{x},\overline{x}]$, and $0<f_X(0)<\infty$.
\\
(ii) For each $k=1,2 $, (a) the collections of real-valued functions, $\{x \mapsto \mu_{k}(x,\theta_k):\theta_k \in \Theta_k\}$, $\{y\mapsto g_1(y,\theta_1):\theta_1 \in \Theta_1\}$, and $\{d\mapsto g_2(d,\theta_2):\theta_2 \in \Theta_2\}$, are of VC type with a common integrable envelope $F_\Epsilon$ such that $\int_{\mathscr{Y}\times \mathscr{D} \times[\underline{x},\overline{x}]} |F_\Epsilon (y,d,x)|^{2+\epsilon}d\mathds{P}^x(y,d,x)<\infty$ for some $\epsilon>0$;
(b) $\mu^{(j)}_{k}$ is Lipschitz on $[\underline{x},0)\times \Theta_k$ and $(0,\overline{x}]\times \Theta_k$ for $j=0,1,2,...,p+1$; (c) For any $(\theta,k)$, $(\vartheta,l)\in \mathds{T}$, we have $\sigma_{k l}(\theta,\vartheta| \ \cdot \ ) \in\mathcal{C}^1([\underline{x},\overline{x}]\setminus \{0\})$ with bounded derivatives in $x$ and $\sigma_{k l}(\theta, \vartheta|0^\pm)< \infty$; (d) For each $y\in \mathscr{Y}$, $g_1(y,\cdot)$ is left- or right-continuous in each dimension. Similarly, for each $d\in \mathscr{D}$, $g_2(d,\cdot)$ is left- or right-continuous in each dimension.\\
(iii) There exist bounded Lipschitz functions $c_1 :\Theta_1\to[\underline{c},\overline{c}]\subset(0,\infty)$ and $c_2 :\Theta_2\to[\underline{c},\overline{c}]\subset(0,\infty)$ such that $h_{1,n}(\theta_1)=c_1(\theta_1)h_n$ and $h_{2,n}(\theta_2)=c_2(\theta_2)h_n$ hold for baseline bandwidth $h_n$ satisfying $h_n\to 0$, $nh^2_n\to \infty$ and $nh^{2p+3}_n\to 0$ for some $h_0 < \infty$.\\
(iv) (a) $K:[-1,1]\to \mathds{R}^+$ is bounded and continuous;
(b) $\{K(\cdot/h):h>0\}$ is of VC type. (c) $\Gamma^\pm_p$ is positive definite.
\end{assumption}

Condition (i) requires a random sampling of $(Y,D,X)$ and sufficient data around $X=0$.
The i.i.d. condition (i) (a) is shared by most of the prior work on regression discontinuity and kink designs.
Exceptions are Bartlalotti and Brummet (2017) and Calonico, Cattaneo, Farrell and Titiunik (2018), which relax the assumption of identical distribution and study cluster-robust inference -- we will also present a method of cluster-robust inference in Section \ref{sec:cluster_robust} by extending our baseline results.
Versions of the smoothness and nonzero requirements in condition (i) (b) are shared by the assumptions made in the prior work, and relate to the absence of endogenous sorting.
This assumption is analogous to Assumption 1.1. in the closely related benchmark paper by Bartalotti, Calhoun and He (2017).
Regarding condition (ii), a sufficient condition for $\{x \mapsto \mu_{k}(x,\theta_k):\theta_k \in \Theta_k\}$ to be of VC type class is, for example, the existence of some non-negative function $M_k:\mathscr{X}\rightarrow \mathds{R}_+$ such that $|\mu_{k}(x,\bar{\theta}_k)-\mu_{k}(x,\theta_k)|\le M_k(x)|\bar{\theta}_k-\theta_k|$ for all $\bar{\theta}_k,\theta_k\in \Theta_k$ for each $k = 1,2$.
Analogous remarks apply to $\{y\mapsto g_1(y,\theta_1):\theta_1 \in \Theta_1\}$ and $\{d\mapsto g_2(d,\theta_2):\theta_2 \in \Theta_2\}$ as well.
Another sufficient condition is the case when a class of functions is of variations bounded by one, e.g. in the case of CDF estimation, $\{y \mapsto \mathds{1}\{y \le y'\}:y'\in \mathscr{Y}\}$ satisfies the VC type condition.
Also notice that the common integrable envelope $F_\Epsilon$ in condition (ii) is satisfied if all the classes of functions are uniformly bounded, but does not rule out some cases that some of these classes of functions are unbounded.
In Appendix \ref{sec:application_examples}, we will check these high-level assumptions with primitive sufficient assumptions for each of the ten specific examples presented in Examples \ref{ex:fuzzy_rdd}--\ref{ex:fuzzy_quantile_rdd} and Appendix \ref{sec:ex:sharp_rdd}--\ref{sec:ex:group}.
Parts (ii) (b) and (ii) (c) are analogous to Assumption 1.3. and Assumption 1.4., respectively, in  Bartalotti, Calhoun and He (2017).
Condition (iii) specifies admissible rates of bandwidths, which are consistent with common choice rules (e.g., Imbens and Kalyanaraman, 2012; Calonico, Cattaneo and Titiunik, 2014; Arai and Ichimura, 2016; Calonico, Cattaeneo, and Farrell, 2016ab; Arai and Ichimura, 2018; for regression discontinuity designs) -- see Appendix \ref{sec:practical_guideline}.
This condition is analogous to Assumption 2 in Bartalotti, Calhoun and He (2017) for the special case of $h \sim b$ in their notation.
Condition (iv) is satisfied by common kernel functions, such as uniform, triangular, biweight, triweight, and Epanechnikov kernels to list a few examples, while the normal kernel is obviously ruled out.

We will show below the weak convergence in $\ell^\infty(\mathbb{T})$ of the Bahadur Representation (BR) 
\begin{align*}
\nu^\pm_n(\theta,k)&=v!\sum_{i=1}^{n}\frac{e'_v(\Gamma^\pm_p)^{-1}\Epsilon_k(Y_i,D_i,X_i,\theta) r_p(\frac{X_i}{h_{k,n}(\theta_k)}) K(\frac{X_i}{h_{k,n}(\theta_k)})\delta^\pm_i}{\sqrt{nh_{k,n}(\theta_k)}f_X(0)}
\qquad
(\theta,k)\in \mathds{T}.
\end{align*}
We also write $\nu_n(\cdot)=\nu^+_n(\cdot)-\nu^-_n(\cdot)$.
By the functional delta method, the weak convergence translate into the asymptotic distribution of the process $\sqrt{nh^{1+2v}_n}[\hat{\tau}(\cdot)-\tau(\cdot)]$ for the local Wald estimator (\ref{eq:wald_estimator}).
We set the following additional assumption.
Primitive conditions for it will be discussed specifically for the fuzzy quantile RDD in Section \ref{sec:a:FQRD}, and for each of the nine examples in Appendix \ref{sec:application_examples}.

\begin{assumption}[Conditional Weak Convergence]\label{a:cond_weak_conv}
\qquad\\
(i) $\psi$, $\phi$ and $\Upsilon$ are Hadamard differentiable at $\mu^{(v)}_{1}(0^\pm,\cdot)$, $\mu^{(v)}_{2}(0^\pm,\cdot)$, and $W$, respectively, tangentially to some subspaces of their domains, with their Hadamard derivatives denoted by $\phi'_{\mu^{(v)}_{1}(0^\pm,\cdot)}$, $\psi'_{\mu^{(v)}_{2}(0^\pm,\cdot)}$, and $\Upsilon'_W$, respectively.\\
(ii) $\inf_{\theta'_2 \in \Theta_2'}|\psi(\mu^{(v)}_{2}(0^+,\cdot))(\theta'_2)- \psi(\mu^{(v)}_{2}(0^-,\cdot))(\theta'_2)| > 0$.\\
(iii) $nh^{1+2v}_n\to \infty$  as $n \to \infty$.
\end{assumption}

Condition (ii) requires the existence of a jump or a kink, which is assumed in most of the prior work as the key identification condition.
Exceptions are Otsu, Xu, and Matsushita (2015) and Feir, Lemieux, and Marmer (2016), which provide weak-identification-robust methods of inference.
We later use this idea to relax condition (ii) of this assumption in Section \ref{sec:weak_jumps}.
Condition (iii) restricts admissible rates of bandwidths, which are consistent with common choice rules (e.g., Imbens and Kalyanaraman, 2012; Calonico, Cattaneo and Titiunik, 2014; Arai and Ichimura, 2016; Calonico, Cattaeneo, and Farrell, 2016ab; Arai and Ichimura, 2018; for regression discontinuity designs) -- see Appendix \ref{sec:practical_guideline}.
The next theorem states the weak convergence result for the process $\sqrt{nh^{1+2v}_n}[\hat{\tau}-\tau]$.

\begin{theorem}[Weak Convergence]\label{theorem:weak_conv}
Under Assumptions \ref{a:BR} and \ref{a:cond_weak_conv}, we have $\nu^\pm_n\leadsto \mathds{G}_{H^\pm}$, where $\mathds{G}_{H^\pm}$ are zero mean Gaussian processes $\mathds{G}_{H^\pm}:\Omega^x\mapsto\ell^\infty(\mathds{T})$ with covariance function
\begin{align*}
H^\pm((\theta,k),(\vartheta,l))= \frac{\sigma_{k l}(\theta,\vartheta|0^\pm)e'_v(\Gamma^\pm_p)^{-1}  \Psi^\pm_p((\theta,k),(\vartheta,l)) (\Gamma^\pm_p)^{-1} e_v}{\sqrt{c_{k}(\theta_{k})c _{l}(\vartheta_{l})}f_X(0)}
\end{align*}
where
\begin{align*}
\Psi^\pm_p((\theta,k),(\vartheta,l))&=\int_{\mathds{R}_\pm}r_p(u/c_{k}(\theta_{k})) r'_p(u/c_{l}(\vartheta_{l})) K(\frac{u}{c_{k}(\theta_{1})}) K(\frac{u}{c_{l}(\vartheta_{l})})du
\end{align*}
for each $\theta=(\theta_1,\theta_2)$, $\vartheta=(\vartheta_1,\vartheta_2) \in \Theta$.
Therefore,
\begin{align*}
&\sqrt{nh^{1+2v}_n}[\hat{\tau}-\tau]
\\
\leadsto &\Upsilon'_W\Big( \frac{ [\psi(\mu^{(v)}_{2}(0+,\cdot))-\psi(\mu^{(v)}_{2}(0^-,\cdot))]\mathds{G}'(\cdot,1)-[\phi(\mu^{(v)}_{1}(0^+,\cdot))-\phi(\mu^{(v)}_{1}(0^-,\cdot))]
\mathds{G}'(\cdot,2)}{[\psi(\mu^{(v)}_{2}(0+,\cdot))-\psi(\mu^{(v)}_{2}(0^-,\cdot))]^2} \Big),
\end{align*}
where $\mathds{G}':\Omega^x \mapsto \ell^\infty(\mathds{T})$ is defined as
\begin{align*}
\begin{bmatrix}
  \mathds{G}'(\cdot,1) \\
  \mathds{G}'(\cdot,2)
\end{bmatrix}=\begin{bmatrix}
   \phi'_{\mu^{(v)}_{1}(0^+,\cdot)}\Big(\mathds{G}_{H+}(\cdot,1)/\sqrt{c^{1+2v}_1(\cdot}) \Big)( \cdot) - \phi'_{\mu^{(v)}_{1}(0^-,\cdot)}\Big(\mathds{G}_{H-}(\cdot,1)/\sqrt{c^{1+2v}_1(\cdot}) \Big)( \cdot)\\
  \psi'_{\mu^{(v)}_{2}(0^+,\cdot)}\Big(\mathds{G}_{H+}(\cdot,2)/\sqrt{c^{1+2v}_2(\cdot}) \Big)( \cdot) - \psi'_{\mu^{(v)}_{2}(0^-,\cdot)}\Big(\mathds{G}_{H-}(\cdot,2)/\sqrt{c^{1+2v}_2(\cdot}) \Big)( \cdot)
\end{bmatrix}.
\end{align*}
\end{theorem}

See Appendix \ref{sec:theorem:weak_conv} for a proof.
For the sharp mean RDD, the fuzzy mean RDD, the sharp mean RKD, and the fuzzy mean RKD, computation of the limit process is straightforward in practice.
On the other hand, for CDF and the quantile process, it is somewhat easier to approximate the limit process through the multiplier bootstrap.
The following subsection presents this additional practical consideration.

\subsection{Multiplier Bootstrap}\label{sec:multiplier_bootstrap}

To simulate the limiting process of the BR, i.e., Theorem \ref{theorem:weak_conv}, we use the pseudo random sample $\{\xi_i\}^n_{i=1}$ drawn from the standard normal distribution, independently from the data $\{(Y_i,D_i,X_i)\}_{i=1}^n$.
Precisely, $\{\xi_i\}^n_{i=1}$ is defined on $(\Omega^{\xi},\mathcal{F}^{\xi},\mathds{P}^{\xi})$, a probability space that is independent of $(\Omega^x,\mathcal{F}^x,\mathds{P}^x)$ -- this condition will be formally stated in Assumption \ref{a:multiplier} below.
With this pseudo random sample, define the multiplier processes (MP)
\begin{align*}
\nu^\pm_{\xi,n}(\theta,k)&=v!\sum_{i=1}^{n}\xi_i\frac{e'_v(\Gamma^\pm_p)^{-1}\Epsilon_k(Y_i,D_i,X_i,\theta) r_p(\frac{X_i}{h_{k,n}(\theta_k)}) K(\frac{X_i}{h_{k,n}(\theta_k)})\delta^\pm_i}{\sqrt{nh_{k,n}(\theta_k)}f_X(0)}.
\end{align*}
We also write $\nu_{\xi,n}(\cdot)=\nu^+_{\xi,n}(\cdot)-\nu^-_{\xi,n}(\cdot)$.

In practice, we need to replace $\Epsilon_k$ and $f_X$ with their estimates.
Let $\hat{f}_X$ be an estimate of ${f}_X$.
For estimation of $\Epsilon_k$, since every component in the BR is multiplied by the kernel $K$ supported on $[-1,1]$, we only need to consider $\Epsilon_{k}(Y_i,D_i,X_i,\theta)\mathds{1}\{|X_i/h_{k,n}(\theta_k)|\le 1\}$.
Write its estimate by
$
\hat{\Epsilon}_{k}(Y_i,D_i,X_i,\theta)\mathds{1}\{|X_i/h_{k,n}(\theta_k)|\le 1\}
$, which has $\mu_{k,p}$ replaced by some estimate $\tilde{\mu}_{k,p}$ of ${\mu}_{k,p}$.
Section \ref{sec:first_stage} discusses the effects of these first-stage estimates.
Substituting these estimated components in the MP, we define the estimated multiplier processes (EMP)
\begin{align*}
\hat{\nu}^\pm_{\xi,n}(\theta,k)&=v!\sum_{i=1}^{n}\xi_i\frac{e'_v(\Gamma^\pm_p)^{-1}\hat{\Epsilon}_{k}(Y_i,D_i,X_i,\theta) r_p(\frac{X_i}{h_{k,n}(\theta_k)}) K(\frac{X_i}{h_{k,n}(\theta_k)})\delta^\pm_i}{\sqrt{nh_{k,n}(\theta_k)}\hat{f}_X(0)}.
\end{align*}
We also write $\hat{\nu}_{\xi,n}(\cdot)= \hat{\nu}^+_{\xi,n}(\cdot)- \hat{\nu}^-_{\xi,n}(\cdot)$.

In order to establish the uniform validity of the multiplier bootstrap, we invoke the following two sets of assumptions.

\begin{assumption}[Multiplier]\label{a:multiplier}
$\{\xi_i\}^n_{i=1}$ is an independent standard normal random sample defined on $(\Omega^{\xi},\mathcal{F}^{\xi},\mathds{P}^{\xi})$, a probability space that is independent of $(\Omega^x,\mathcal{F}^x,\mathds{P}^x)$.
\end{assumption}

\begin{assumption}[First Stage Estimation]\label{a:first_stage}
$\tilde{\mu}_{k,p}(x,\theta_k)\mathds{1}\{|x/h_{k,n}(\theta_k)|\le 1\}$ is uniformly consistent for $\mu_k(x,\theta_k)\mathds{1}\{|x/h_{k,n}(\theta_k)|\le 1\}$ on $([\underline{x},\overline{x}] \backslash \{0\}) \times \mathds{T}$. $\hat{f}_X(0)$ is consistent for $f_X(0)$.
\end{assumption}

We remark that Assumption \ref{a:multiplier} is the standard assumption for multiplier, score, and wild bootstrap methods, cf. Kosorok (2003, 2008).
At the current level of generality, we state these high-level conditions for the first stage estimation in Assumption \ref{a:first_stage}.
However, we will propose a concrete $\tilde{\mu}_k$ in Appendix \ref{sec:first_stage} that satisfies Assumption \ref{a:first_stage} under Assumptions \ref{a:BR} and \ref{a:cond_weak_conv}.
The following theorem, together with Theorem \ref{theorem:weak_conv}, establishes the uniform validity of the multiplier bootstrap.

\begin{theorem}[Conditional Weak Convergence]\label{theorem:cond_weak_conv}
Under Assumptions \ref{a:BR}, \ref{a:cond_weak_conv}, \ref{a:multiplier}, and \ref{a:first_stage}, we have $\hat{\nu}^\pm_{\xi,n}\underset{\xi}{\overset{p}{\leadsto}}\mathds{G}_{H^\pm}$, and therefore
\begin{align*}
&\Upsilon'_W\Big( \frac{[\psi(\mu^{(v)}_{2}(0^+,\cdot))-\psi(\mu^{(v)}_{2}(0^-,\cdot))]\widehat{\mathds{X}}'_n(\cdot,1)-[\phi(\mu^{(v)}_{1}(0^+,\cdot))-\phi(\mu^{(v)}_{1}(0^-,\cdot))]
\widehat{\mathds{X}}'_n(\cdot,2)}{(\psi(\mu^{(v)}_{2}(0^+,\cdot))-\psi(\mu^{(v)}_{2}(0^-,\cdot))^2} \Big)\\
\underset{\xi}{\overset{p}{\leadsto}}
 &\Upsilon'_W\Big( \frac{[\psi(\mu^{(v)}_{2}(0^+,\cdot))-\psi(\mu^{(v)}_{2}(0^-,\cdot))]\mathds{G}'(\cdot,1)-[\phi(\mu^{(v)}_{1}(0^+,\cdot))-\phi(\mu^{(v)}_{1}(0^+,\cdot))]
\mathds{G}'(\cdot,2)}{(\psi(\mu^{(v)}_{2}(0^+,\cdot))-\psi(\mu^{(v)}_{2}(0^-,\cdot))^2} \Big),
\end{align*}
where
\begin{align*}
\begin{bmatrix}
  \widehat{\mathds{X}}'_n(\cdot,1) \\
  \widehat{\mathds{X}}'_n(\cdot,2)
\end{bmatrix}
= \begin{bmatrix}
   \phi'_{\mu^{(v)}_{1}(0^+,\cdot)}\Big( \hat{\nu}^+_{\xi,n}( \cdot,1)/\sqrt{c^{1+2v}_1(\cdot}) \Big)( \cdot) - \phi'_{\mu^{(v)}_{1}(0^-,\cdot)}\Big( \hat{\nu}^-_{\xi,n}( \cdot,1)/\sqrt{c^{1+2v}_1(\cdot}) \Big)( \cdot)\\
  \psi'_{\mu^{(v)}_{2}(0^+,\cdot)}\Big( \hat{\nu}^+_{\xi,n}( \cdot,2)/\sqrt{c^{1+2v}_2(\cdot}) \Big)( \cdot) - \psi'_{\mu^{(v)}_{2}(0^-,\cdot)}\Big( \hat{\nu}^-_{\xi,n}( \cdot,2)/\sqrt{c^{1+2v}_2(\cdot}) \Big)( \cdot)
\end{bmatrix}.
\end{align*}
\end{theorem}

See Appendix \ref{sec:theorem:cond_weak_conv} for a proof.
Theorems \ref{theorem:weak_conv} and \ref{theorem:cond_weak_conv} show that the estimated multiplier process
$$
\Upsilon'_W\Big( \frac{[\psi(\mu^{(v)}_{2}(0^+,\cdot))-\psi(\mu^{(v)}_{2}(0^-,\cdot))]\widehat{\mathds{X}}'_n(\cdot,1)-[\phi(\mu^{(v)}_{1}(0^+,\cdot))-\phi(\mu^{(v)}_{1}(0^-,\cdot))]
\widehat{\mathds{X}}'_n(\cdot,2)}{(\psi(\mu^{(v)}_{2}(0^+,\cdot))-\psi(\mu^{(v)}_{2}(0^-,\cdot))^2} \Big)
$$
can be used to approximate the limit process of $\sqrt{nh^{1+2v}_n}[\hat{\tau}-\tau]$ in practice.

\subsection{Fuzzy Quantile RDD (Example \ref{ex:fuzzy_quantile_rdd}) Revisited}\label{sec:a:FQRD}

In this section, we apply the main general results, namely Theorems \ref{theorem:weak_conv} and \ref{theorem:cond_weak_conv}, to the fuzzy quantile RDD introduced in Example \ref{ex:fuzzy_quantile_rdd}.
What we present in this section provides a formal justification of the overview in Section \ref{sec:overview}.
We present our assumptions for the case of $p=2$ as we did in the overview.
We remark that, however, using a different order $p$ of local polynomial fitting is also possible by similar arguments.

Consider $\Theta_1$, $\Theta_2$, $\Theta_1'$, $\Theta_2'$, $\Theta''$, $g_1$, $g_2$, $\phi$, $\psi$, and $\Upsilon$ defined in Example \ref{ex:fuzzy_quantile_rdd}.
Recall that we denote the local Wald estimand (\ref{eq:wald_est}) with $v=0$ in this setting by $\tau_{FQRD}$.
We also denote the analog estimator (\ref{eq:wald_estimator}) with $v=0$ in this setting by
\begin{align*}
\hat{\tau}_{FQRD}(\theta'')&=
\Upsilon(\hat{F}_{Y^\cdot|C})(\theta''),
\end{align*}
where
\begin{align*}
\hat{F}_{Y^1|C}(y)=\frac{\hat{\mu}_{1,2} (0^+,y,1)-
\hat{\mu}_{1,2} (0^-,y,1)}{\hat{\mu}_{2,2}(0^+,1)- \hat{\mu}_{2,2}(0^-,1)}
\qquad\text{and}\qquad
\hat{F}_{Y^0|C}(y)=\frac{\hat{\mu}_{1,2} (0^+,y,0)-
\hat{\mu}_{1,2} (0^-,y,0)}{\hat{\mu}_{2,2}(0^+,0)- \hat{\mu}_{2,2}(0^-,0)}.
\end{align*}
By van der Vaart and Wellner (1996; Lemma 3.9.23) and van der Vaart (1998; Theorem 20.9), $\Upsilon$ is Hadamard differentiable at $(F_{Y^\cdot|C})$ tangentially to $C(\mathscr{Y}_1\times \mathscr{D}^2)$, and the Hadamard derivative is a map that takes each $g\in C(\mathscr{Y}_1\times \mathscr{D}^2)$ to
$$
\Upsilon'_{W}(g)(\cdot)=:\Upsilon'_{F_{Y^\cdot|C}}(g)(\cdot)=-\frac{g( Q_{Y^1|C}(\cdot),1,1 )}{f_{Y^1|C}( Q_{Y^1|C}(\cdot) )}+\frac{g( Q_{Y^0|C}(\cdot) ,0,0)}{f_{Y^0|C}( Q_{Y^0|C}(\cdot) )}
$$
 under the assumptions to be stated below.
Sufficient conditions for the assumptions required for the general result, tailored to the current example, are stated as follows.

\newtheorem*{assumption_S}{Assumption S}
\begin{assumption_S}\label{a:S}
(a) $\{(Y_i,D_i,X_i)\}^n_{i=1}$ are $n$ i.i.d. copies of the random vector $(Y,D,X)$ defined on a probability space $(\Omega^x,\mathcal{F}^x,\mathds{P}^x)$. 
(b) $X$ has a density function $f_X$ which is continuously differentiable on $[\underline{x},\overline{x}]$ that contains $0$ in its interior, and $0<f_X(0)<\infty$.
\end{assumption_S}

\newtheorem*{assumption_K}{Assumption K}
\begin{assumption_K}\label{a:K}
(a) $K:[-1,1]\to \mathds{R}^+$ is bounded and continuous. (b) $\{K(\cdot/h):h>0\}$ is of VC type. (c) $\Gamma^\pm_p$ is positive definite.
\end{assumption_K}

\newtheorem*{assumption_M}{Assumption M}
\begin{assumption_M}\label{a:M}
$\{\xi_i\}^n_{i=1}$ are independent standard normal random variables defined on $(\Omega^{\xi},\mathcal{F}^{\xi},\mathds{P}^{\xi})$, a probability space that is independent of $(\Omega^x,\mathcal{F}^x,\mathds{P}^x)$.
\end{assumption_M}

\newtheorem*{assumption_FQRD}{Assumption FQRD}
\begin{assumption_FQRD}\label{a:FQRD}
\qquad\\
%
(i)
$(x,y,d) \mapsto \frac{\partial^j}{\partial x^j} E[\mathds{1}\{Y^*_i\le y,D^*_i=d\}|X_i=x]$ is Lipschitz in $x$ on $[\underline{x},0)\times \Theta_1$ and $(0,\overline{x}]\times \Theta_1$ for $j=0,1,2,3$, and $(x,d) \mapsto \frac{\partial^j}{\partial x^j} E[\mathds{1}\{D^*_i=d\}|X_i=x]$ is Lipschitz in $x$ on $[\underline{x},0)\times \Theta_2$ and $(0,\overline{x}]\times \Theta_2$ for $j=0,1,2,3$. \\
(ii) The baseline bandwidth $h_n$ satisfies $h_n\to 0$ and $nh^2_n\to \infty$, $nh^{7}_n\to 0$.
There exist bounded constants $0<c_1$, $c_2<\infty$ such that $h_{1,n}=c_1 h_n$ and $h_{2,n}=c_2h_n$.\\
%
%
(iii)  $|\mathds{P}^x(D_i=1|X_i=0^+)-\mathds{P}^x(D_i=1|X_i=0^-)| > 0$.\\
%
%
(iv) $F_{Y^1|C}, F_{Y^0|C}\in \mathcal{C}^1(\mathscr{Y}_1)$, and $f_{Y^1|C}$ and $f_{Y^0|C}$ are bounded away from $0$ on $\mathscr{Y}_1$.\\
(v) There exists $\hat f_{Y|XD^*}(y|0^\pm,d)$ such that $\sup_{(y,d) \in \mathscr{Y}_1 \times \mathscr{D}}|\hat{f}_{Y|XD^*}(y|0^\pm,d)-f_{Y|XD^*}(y|0^\pm,d)|=o^x_p(1)$.
\end{assumption_FQRD}

Assumption S (a) is assumed in most of the prior work, including the closely related benchmark by Frandsen, Fr\"olich and Melly (2012).
As emphasized after Assumption \ref{a:BR} (i) (a), a part of the literature has relaxed the assumption of identical distribution and studies cluster-robust inference -- we will also present a method of cluster-robust inference in Section \ref{sec:cluster_robust} by extending our baseline results.
Assumption S (b) is analogous to Assumption E3 of Frandsen, Fr\"olich and Melly (2012).
Assumption K is analogous to Assumption E4 of Frandsen, Fr\"olich and Melly (2012).
Assumption M is new to our paper due to our use of the multiplier bootstrap, which is not used by prior work such as Frandsen, Fr\"olich and Melly (2012).
Assumption FQRD (i), (ii), (iii), and (iv) are analogous to Assumptions E1, E5, E2, and Q, respectively, of Frandsen, Fr\"olich and Melly (2012).
We state Assumption FQRD (v) at this high level in order to accommodate a number of alternative estimators.
In Lemma \ref{lemma:FQRDD} in Appendix \ref{sec:lemma:unif_cons_est}, however, we propose one such concrete estimator which satisfies part (v).
All the other parts of this assumption are immediately interpretable.

Define the EMP by
\begin{align*}
&\hat{\nu}^\pm_{\xi,n}(y,d_1,d_2,1)=\sum_{i=1}^{n}\xi_i\frac{e'_0(\Gamma^\pm_2)^{-1}[\mathds{1}\{Y^*_i\le y,D^*_i=d_1\}-\tilde{\mu}_{1,2}(X_i,y,d_1)]r_2(\frac{X_i}{h_n})K(\frac{X_i}{h_n})\delta^\pm_i}{\sqrt{nh_n}\hat{f}_X(0)}\\
&\hat{\nu}^\pm_{\xi,n}(y,d_1,d_2,2)=\sum_{i=1}^{n}\xi_i\frac{e'_0(\Gamma^\pm_2)^{-1}[\mathds{1}\{D^*_i=d_2\}-\tilde{\mu}_{2,2}(X_i,d_2)]r_2(\frac{X_i}{h_n})K(\frac{X_i}{h_n})\delta^\pm_i}{\sqrt{nh_n}\hat{f}_X(0)}\\
&\widehat{\mathds{X}}'_n(y,d_1,d_2,k)= \hat{\nu}^+_{\xi,n}(y,d_1,d_2,k)/\sqrt{c_k}- \hat{\nu}^-_{\xi,n}(y,d_1,d_2,k)/\sqrt{c_k}
\end{align*}
for each $(y,d_1,d_2,k)\in \mathds{T}=\mathscr{Y}_1\times\mathscr{D}\times\mathscr{D}\times\{1,2\}$.
Define the following estimated process through the Hadamard derivative.\small
\begin{align*}
&\widehat{\Upsilon}'_{W}(\widehat{\mathds{Y}}_n)(\theta'')\\
=&\frac{[\hat{\mu}_{2,2} (0^+,1)-\hat{\mu}_{2,2} (0^-,1)]
\widehat{\mathds{X}}'_n(\hat{Q}_{Y^1|C}(\theta''),1,1,1)-[\hat{\mu}_{1,2}(0^+,\hat{Q}_{Y^1|C}(\theta''),1)-\hat{\mu}_{1,2}(0^-,\hat{Q}_{Y^1|C}(\theta''),1)]
\widehat{\mathds{X}}'_n(\hat{Q}_{Y^1|C}(\theta''),1,1,2)}{\hat{f}_{Y^1|C}(\hat{Q}_{Y^1|C}(\theta''))[\hat{\mu}_{2,2} (0^+,1)-\hat{\mu}_{2,2} (0^-,1)]^2}\\
-&\frac{[\hat{\mu}_{2,2} (0^+,0)-\hat{\mu}_{2,2} (0^-,0)]
\widehat{\mathds{X}}'_n(\hat{Q}_{Y^0|C}(\theta''),0,0,1)-[\hat{\mu}_{1,2}(0^+,\hat{Q}_{Y^0|C}(\theta''),0)-\hat{\mu}_{1,2}(0^-,\hat{Q}_{Y^0|C}(\theta''),0)]
\widehat{\mathds{X}}'_n(\hat{Q}_{Y^0|C}(\theta''),0,0,2)}{\hat{f}_{Y^0|C}(\hat{Q}_{Y^0|C}(\theta''))[\hat{\mu}_{2,2} (0^+,0)-\hat{\mu}_{2,2} (0^-,0)]^2}
\end{align*}
\normalsize
where 
$$
\widehat{\mathds{Y}}_n (y,d_1,d_2) =
\frac{[\hat{\mu}_{2,2} (0^+,d_2)-\hat{\mu}_{2,2} (0^-,d_2)]
\widehat{\mathds{X}}'_n(y,d_1,d_2,1)-[\hat{\mu}_{1,2}(0^+,y,d_1)-\hat{\mu}_{1,2}(0^-,y,d_1)]
\widehat{\mathds{X}}'_n(y,d_1,d_2,2)}{[\hat{\mu}_{2,2} (0^+,d_2)-\hat{\mu}_{2,2} (0^-,d_2)]^2}
$$
for $(y,d_1,d_2)\in \mathscr Y_1 \times \mathscr D^2$.
With these preparations, our general result applied to the current case yields the following corollary.

\begin{corollary}[Example: Fuzzy Quantile RDD]\label{corollary:FQRD}
Suppose that Assumptions S, K, M, and FQRD hold.
\\
(i) There exists a zero mean Gaussian process $\mathds{G}'_{FQRD}:\Omega^x \mapsto \ell^\infty([a,1-a])$ such that
$$
\sqrt{nh_n}[\hat{\tau}_{FQRD}-\tau_{FQRD}]\leadsto \mathds{G}'_{FQRD}.
$$
(ii) Furthermore, with probability approaching one,
$$
\widehat{\Upsilon}'_{W}(\widehat{\mathds{Y}}_n)\underset{\xi}{\overset{p}{\leadsto}} \mathds{G}'_{FQRD}.
$$
\end{corollary}

A proof is provided in Appendix \ref{sec:corollary:FQRD}.
This result justifies the overview in Section \ref{sec:overview}.
Specifically, the estimated multiplier process $\widehat{\Upsilon}'_{W}(\widehat{\mathds{Y}}_n)$ -- denoted by $\widehat{\mathbb{G}}_{FQRD,n}'$ in Section \ref{sec:overview} for simplicity -- can be used to approximate the limit prcess of $\sqrt{nh_n}[\hat{\tau}_{FQRD}-\tau_{FQRD}]$.
We could consider a set of functions defined on $(0,1)$ instead of $[a,1-a]$ by making a stronger assumption that the potential outcomes are compactly supported with their conditional density functions bounded away from zero, although we may not want to make such a strong assumption in general for typical applications.
This tradeoff between the globalization of the domain and the strength of the assumption is even true for simple unconditional quantile processes, e.g., Lemma 21.4 (i) versus Lemma 21.4 (ii) in van der Vaart (1998).

\begin{remark}\label{remark:bias_correction_quantile}
Our bias correction is conducted on the mean regression estimates $\hat\mu_{1,2}$ and $\hat\mu_{2,2}$ in light of Remark 7 of Calonico, Cattaneo and Titiunik (2014). Some calculation shows the biases of local Wald-ratios $\hat F_{Y^d|C}$ also have the same bias order as the regression estimates above. Finally, since the quantile estimate based on a left inverse of a CDF has the same order of bias as the CDF estimate, we can achieve bias correction for fuzzy quantile RDD estimate.
\end{remark}

\begin{remark}\label{remark:bias_correction_trick}
We emphasize that our one-step bias correction as well as Remark 7 of Calonico, Cattaneo and Titiunik (2014) relies on choosing the same bandwidth for both main estimators and higher order bias estimators. According to the simulation studies in Calonico, Cattaneo and Titiunik (2014), such bandwidth choice works rather well in their various DGPs. 
	While studies of optimal coverage probability for CDF and quantile estimation would be both useful and important, it is out of the scope of the current paper.
	
	It is also worth noting that our framework can accommodate the case when two bandwidths are different as well, since we can acquire the uniform Bahadur representations for both the main estimator and the bias estimator via our uniform Bahadur representation in Appendix \ref{sec:lemma:BR}. 
	The multiplier bootstrap can be applied to the difference of these two processes, which would be similar to Lemma 3 and Proposition 2 in Qu and Yoon (2015b). 
\end{remark}

\section{Simulation Studies}\label{sec:simulation_fqrdd}

We conduct simulation studies for the fuzzy quantile RDD -- see Sections \ref{sec:overview} and \ref{sec:a:FQRD}.
We also conduct simulation studies for the other cases -- see Appendix \ref{sec:simulation}.
We follow the procedure outlined in Appendix \ref{sec:practical_guideline} for choices of bandwidths in finite samples.
The kernel function that we use is the Epanechnikov kernel.
All the other procedures exactly follow the guideline in Section \ref{sec:a:FQRD}.

To systematically evaluate the hypothesis testing for the hypotheses of uniform treatment nullity and treatment homogeneity presented in Section \ref{sec:overview}, we consider the following data generating process.
We generate an i.i.d. sample $\{(Y_i,D_i,X_i)\}_{i=1}^n$ through the following data generating process:
\begin{align*}
&Y_i = \mu(X_i) + \beta_1 D_i + (1 + \gamma_1 D_i) \cdot U_i,
\\
&D_i = \mathbbm{1}\{ 2 \cdot \mathbbm{1}\{ X_i \geq 0\} - 1 \geq V_i \},
\\
&(X_i,U_i,V_i)' \sim N(0,\Sigma),
\end{align*}
where $\beta_1$ and $\gamma_1$ are to be varied across simulation sets, $\Sigma_{11} = \sigma_X^2 = 0.1781742^2$, $\Sigma_{22} = \sigma_U^2 = 0.1295^2$, $\Sigma_{33} = \sigma_V^2 = 0.5^2$, $\Sigma_{12} = \rho_{XU} \cdot \sigma_X \cdot \sigma_U = 0.25 \cdot 1.0^2$, $\Sigma_{13} = \rho_{XV} \cdot \sigma_X \cdot \sigma_V = 0.0 \cdot 1.0 \cdot 0.5$, and $\Sigma_{23} = \rho_{UV} \cdot \sigma_U \cdot \sigma_V = 0.25 \cdot 1.0 \cdot 0.5$.
We choose the numbers, $\sigma_X^2 = 0.1781742^2$ and $\sigma_U^2 = 0.1295^2$, so that these variances match the corresponding variances in the data generating processes in the simulation studies by Calonico, Cattaneo, and Titiunik (2014).
Furthermore, the polynomial part $\mu(X_i)$ of the outcome equation is defined due to Lee (2008):
\begin{align*}
\mu(x) = \begin{cases}
1.27x+7.18x^2+20.21x^3+21.54x^4+7.33x^5 & \text{if } x < 0\\
0.84x-3.00x^2+7.99x^3-9.01x^4+3.56x^5   & \text{if } x \geq 0\\
\end{cases},
\end{align*}
also following the simulation studies by Calonico, Cattaneo, and Titiunik (2014).
Observe under the above data generating process that
$$
\text{$\theta$-th Conditional Quantile Treatment Effect at $x=0$} = \beta_1 + \gamma_1 F_{U|X}^{-1}(\theta | 0).
$$

We set $\Theta'' = [a,1-a] = [0.20,0.80]$ as the set of quantiles on which we conduct uniform inference.
We use a grid with the interval size of 0.02 to approximate the continuum $\Theta''$ for numerical evaluation of functions defined on $\Theta''$.
Similarly, we use a grid with interval size of 0.02 to approximate the continuum $\mathscr{Y}$ for numerical evaluation of functions defined on $\mathscr{Y}$.
First, we simulate the 95\% test for the null hypothesis $H_0: \tau_{FQRD}(\theta'') = 0 \ \forall\theta'' \in [a,1-a]$ of uniform treatment nullity using the procedure described in Section \ref{sec:a:FQRD}.
Second, we simulate the 95\% test for the null hypothesis $H_0: \tau_{FQRD}(\theta'') = \tau_{FQRD}(\theta''') \ \forall \theta'',\theta''' \in [a,1-a]$ of treatment homogeneity using the procedure described in Section \ref{sec:a:FQRD}.
Third, we simulate the uniform coverage probability of the true quantile treatment effects $\tau_{FQRD}( \ \cdot \ )$ using the uniform confidence band constructed in Section \ref{sec:a:FQRD}.

Table \ref{tab:FQRD} shows simulated acceptance and coverage probabilities based on 2,500 multiplier bootstrap replications for 2,500 Monte Carlo replications for each of the sample sizes $n=500$, $1,000$, $1,500$, and $2,000$.
Panel (A) reports results for the test of uniform treatment nullity, panel (B) shows results for the test of treatment homogeneity, and panel (C) shows results for the uniform coverage probability.
The left column groups (I) present results across alternative values of $\beta_1 \in \{0.00, 0.05, 0.10, 0.15, 0.20\}$ while fixing $\gamma_1 = 0$.
The case of $\beta_1=0.00$ evaluates the size in (A), whereas the cases of $\beta_1 \in \{0.05, 0.10, 0.15, 0.20\}$ evaluate the power in (A).
All of the cases of $\beta_1 \in \{0.00, 0.05, 0.10, 0.15, 0.20\}$ evaluate the coverage in (B) and (C).
The right column groups (II) present results across alternative values of $\gamma_1 \in \{0.00, 0.25, 0.50, 0.75, 1.00\}$ while fixing $\beta_1 = 0$. 
The case of $\gamma_1=0.00$ evaluates the coverage in (A) and (B), whereas the cases of $\gamma_1 \in \{0.25, 0.50, 0.75, 1.00\}$ evaluate the power in (A) and (B).
All of the cases of $\beta_1 \in \{0.00, 0.25, 0.50, 0.75, 1.00\}$ evaluate the coverage in (C).
The nominal acceptance probability is 95\%.

In view of the columns for $\beta_1 = 0.00$ and $\gamma_1 = 0.00$, we confirm the correct size.
Furthermore, the simulated and coverage probabilities under $\beta_1 = 0$ and $\gamma_1 = 0$ approach the nominal probability as the sample size increases.
The acceptance probability in panel (A) decreases as $\beta_1$ or $\gamma_1$ deviates away from zero, which is consistent with the fact that the joint treatment nullity is violated by $\beta_1 \neq 0$ or $\gamma_1 \neq 0$.
The acceptance probability in panel (B) stays roughly constant as $\beta_1$ deviates away from zero, but it decreases as $\gamma_1$ deviates away from zero.
This result is consistent with the fact that the treatment homogeneity is retained for any value of $\beta_1$, but is violated for $\gamma_1 \neq 0$.
The uniform coverage probability in panel (C) keeps the nominal size across all the values of $\beta_1$ and $\gamma_1$, which evidences the effectiveness of the uniform confidence bands across alternative data generating processes in the presence of nontrivial and/or heterogeneous treatment effects.

To assess the sensitivity of the proposed inference procedure to weak jumps, we conduct additional simulations with a sequence of values of $\Sigma_{33} \in \{2^0, 2^1, 2^2, 2^3, 2^4, 2^5, 2^6\}$ while fixing $\beta_1=\gamma=1=0$.
Table \ref{tab:Weak_FQRD} shows simulated acceptance and coverage probabilities based on 2,500 multiplier bootstrap replications for 2,500 Monte Carlo replications for each of the sample sizes $n=500$, $1,000$, $1,500$, and $2,000$.
Panel (A) reports results for the test of uniform treatment nullity, and panel (B) shows results for the test of treatment homogeneity.
Results for (C) the uniform coverage probability are omitted because they are the same as the results of (A) under $\beta_1=\gamma_1=0$.
Notice that the size becomes smaller and deviates away from the nominal size as the jump becomes weak (i.e., as $\Sigma_{33}$ increases).
This observation of size distortions under weak jumps is consistent with prior study, e.g., Feir, Lemieux, and Marmer (2016).
Motivated by these results, we present an extended theory of inference with robustness against the possibility of no or weak jumps in Section \ref{sec:weak_jumps}.

\section{Empirical Illustration}\label{sec:empirical_illustration_fqrdd}

In this section, we apply our method of robust uniform inference for the fuzzy quantile RDD (Sections \ref{sec:overview} and \ref{sec:a:FQRD}).
Using RDD with a birth-day cutoff eligibility rule for the Oklahoma pre-K program, Gormley, Gayer, Phillips and Dawson (2005) find significant positive effects of cognitive development on average test scores.
They also find that the average effects are positive among sub-samples of students in lower socio-economic status.
Following up the latter finding, Frandsen, Fr\"olich and Melly (2012) provide additional evidence that these effects are positive among the lower end of the distribution via estimated quantile treatment effects.
We apply our method of constructing robust uniform confidence bands to complement the findings by Frandsen, Fr\"olich and Melly (2012).

The data consist of a sample of 4,710 incoming Tulsa Public Schools kindergartners and pre-K participants for the 2003-2004 school year.
The main variables used in this data are the birth date ($X$), an indicator for participation in the pre-K program in the previous year ($D^\ast$), and scores on the Woodcock-Johnson sub-tests ($Y^\ast$): Letter-word, Spelling, and Applied Problems. 
The implementation procedure follows the guideline provided in Section \ref{sec:a:FQRD} -- or the overview in Section \ref{sec:overview} -- as well as the additional first-stage estimators suggested in Appendix \ref{sec:first_stage} and the bandwidth selection procedure suggested in Appendix \ref{sec:practical_guideline}.

Figure \ref{fig:application} plots the estimated local quantile treatment effects of the pre-K programs on scores on the three sub-tests of the Woodcock-Johnson tests.
The figure shows the point estimates for each quantile $\theta$ by black curves.
It also shows the 90\% uniform confidence bands based on our proposed procedure.

The qualitative patterns of our results resemble those of Frandsen, Fr\"olich and Melly (2012) for each of the three sub-tests.
Not surprisingly, our 90\% confidence bands are wider than the 90\% point-wise confidence intervals obtained by Frandsen, Fr\"olich and Melly (2012), and statistical significance vanishes at some quantile indices.
Nonetheless, despite the generally greater widths of confidence bands than confidence intervals, the statistical significance remains for relevant quantile indices.
In particular, as in Frandsen, Fr\"olich and Melly (2012), we continue to conclude that the program succeeded in significantly raising the lower end of the distribution of test scores, especially for the Applied Problems sub-test, which is consistent with Gormley, Gayer, Phillips and Dawson's (2005) finding that estimated average effects are larger for children with disadvantaged socio-economic status.

\section{Extensions}\label{sec:extensions}

In this section, we discuss three directions for extending the baseline model and method presented in Section \ref{sec:general_framework}.
These extensions include cluster-robust inference (Section \ref{sec:cluster_robust}), inference with robustness against no or weak jumps or kinks (Section \ref{sec:weak_jumps}), and augmented models with observed covariates (Section \ref{sec:covariates}). 

\subsection{Cluster Robust Inference}\label{sec:cluster_robust}

In applications, researchers may encounter situations where data are cluster sampled and thus the i.i.d. assumption may be implausible. 
Cluster--robust standard errors for sharp/fuzzy mean RDD are studied by Bartlalotti and Brummet (2017) and Calonico, Cattaneo, Farrell and Titiunik (2018). 
In this section, we show that Theorems \ref{theorem:weak_conv} and \ref{theorem:cond_weak_conv} can be generalized to the cases of cluster sampled data. 
Suppose that a researcher observes $\{(Y_{i},D_{i},X_{i}):i\in C_g,\: g\in\{1,...,G\}\}$, where each $C_g\subset \{1,...,n\}$ is an index set with cardinality $|C_g|\le \bar N$ for an $\bar N \in \mathbb{N}$ independent of $G$ and $g$, $\sumg |C_g|=n$, and $C_g \cap C_{g'}=\emptyset$ whenever $g\ne g'$. 
Observations within the same cluster $g$ can be arbitrarily dependent, while any two observations across different clusters are independent. 
Denote $\bar f_{X}=\frac{1}{G}\sumg\sumi f_{X_i}$, where $f_{X_{i}}$ is the density function of $X_i$. 
To keep our writing and analysis simple, we let $h_{1,G}(\theta_1)=h_{2,G}(\theta_2)=h_G$. 
The local Wald estimate $\hat\tau$ is calculated according to (\ref{eq:local_poly_reg}) and (\ref{eq:wald_estimator}) in Section \ref{sec:general_framework} with the bandwidth $h_G$. 
Note that $\frac{1}{n}\sum_{i=1}^n$ can also be written as $\frac{1}{n}\sum_{g=1}^G\sum_{i\in C_g}$.  
Since $\frac{c}{\sqrt{G}}\le \frac{1}{\sqrt{n}}\le \frac{1}{\sqrt{G}}$ for some $c>0$, we can rescale estimates by $\frac{n}{G}$ without loss of generality. 
We now define the cluster-robust EMP as
\begin{align*}
\hat{\nu}^\pm_{\xi,G}(\theta,k)&=v!\sumg\xi_g\sumi\frac{e'_v(\Gamma^\pm_p)^{-1}\hat{\Epsilon}_{k}(Y_i,D_i,X_i,\theta) r_p(\frac{X_i}{h_{G}}) K(\frac{X_i}{h_{G}})\delta^\pm_i}{\sqrt{Gh_{G}}\hat{\bar f}_X(0)},
\end{align*}
where $\hat \Epsilon_k=g_k - \tilde \mu_k$ with $\tilde \mu_k$ and $\hat{\bar f}_X$ denoting first-stage estimators of $\mu_k$ and $\bar f_X$, respectively, satisfying the assumption below. 
Also denote $\hat{\nu}_{\xi,G}(\cdot)= \hat{\nu}^+_{\xi,G}(\cdot)- \hat{\nu}^-_{\xi,G}(\cdot)$.



\begin{assumption}[Cluster Sampling]\label{a:cluster_robust}
Let $\underline{x} < 0 < \overline{x}$,\\
(i) (a) For each $G\in \mathbbm N$, $\{(Y_i,D_i,X_i):i\in C_1\}$, ..., $\{(Y_i,D_i,X_i):i\in C_G\}$ are independent random vectors defined on a probability space $(\Omega^x,\mathcal{F}^x,\mathds{P}^x)$;
(b) For each $G \in \mathbbm N$ and for each $i\in C_1\cup...\cup C_G$, $X_{i}$ has a density function $f_{X_{i}}$ which is continuously differentiable and $1$-Lipschitz, and satisfies $0<\underline f\le f_{X_{i}}\le \overline f<\infty$ on $[\underline x,\overline x]$.  
\\
(ii) For each $k=1,2 $: (a) the collections of real-valued functions, $\{x \mapsto \mu_{k}(x,\theta_k):\theta_k \in \Theta_k\}$, $\{y\mapsto g_1(y,\theta_1):\theta_1 \in \Theta_1\}$, and $\{d\mapsto g_2(d,\theta_2):\theta_2 \in \Theta_2\}$, are of VC type and are uniformly bounded by $\overline M<\infty$;
(b) $\mu^{(j)}_{k}$ is Lipschitz on $[\underline{x},0)\times \Theta_k$ and $(0,\overline{x}]\times \Theta_k$ for $j=0,1,2,...,p+1$; (c) for each $y\in \mathscr{Y}$, $g_1(y,\cdot)$ is left- or right-continuous in each dimension; and similarly, for each $d\in \mathscr{D}$, $g_2(d,\cdot)$ is left- or right-continuous in each dimension. (d) For each $G\in \mathbbm N$, $i \in \{1,...,n\}$, $\theta\in \Theta$, it holds that $E[\Epsilon_k(Y_i,D_i,X_i,\theta)|(X_j: j\in C_g, i\in C_g)]=E[\Epsilon_k(Y_i,D_i,X_i,\theta)|X_i]=0$.\\
(iii) $h_G\to 0$, $Gh^2_G\to \infty$ and $Gh^{2p+3}_G\to 0$ for some $h_0 < \infty$. \\
(iv) (a) $K:[-1,1]\to \mathds{R}^+$ is bounded and continuous;
(b) $\{K(\cdot/h):h>0\}$ is of VC type. (c) $\Gamma^\pm_p$ is positive definite.\\
(v) There exist uniformly bounded functions $\Sigma^\pm:(\Theta\times \{1,2\})^2 \mapsto \mathbbm R$ such that \small
\begin{align*}
&\frac{1}{G}\sumg E\Big[\Big(\sumi\frac{e'_v(\Gamma^\pm_p)^{-1} r_p(\frac{X_i}{h_{G}})\Epsilon_k(Y_i,D_i,X_i,\theta)K(\frac{X_i}{h_{G}})\delta^\pm_i}{\sqrt{h_{G}}\bar f_X(0)}\Big)\Big(\sumi\frac{e'_v(\Gamma^\pm_p)^{-1} r_p(\frac{X_i}{h_{G}})\Epsilon_k(Y_i,D_i,X_i,\vartheta)K(\frac{X_i}{h_{G}})\delta^\pm_i}{\sqrt{h_{G}}\bar f_X(0)}\Big)'\Big]\\
&= \Sigma^\pm((\theta,k) ,(\vartheta,l))+O_p(h_G).
\end{align*} \normalsize
(vi)(a) $\phi$ and $\psi$ are identity mappings while 
$\Upsilon$ is Hadamard differentiable at $W$ tangentially to a subspace of its domain, with its Hadamard derivative denoted by $\Upsilon_W'$.
(vii) $\{\xi_g\}^G_{g=1}$ is an independent standard normal random sample defined on $(\Omega^{\xi},\mathcal{F}^{\xi},\mathds{P}^{\xi})$, a probability space that is independent of $(\Omega^x,\mathcal{F}^x,\mathds{P}^x)$. (b) $\inf_{\theta_2\in \Theta_2}|\mu_2^{(v)}(0^+,\theta_2)-\mu_2^{(v)}(0^-,\theta_2)|>0$.
(viii) There exist first stage estimators $\tilde \mu_{k,p}$ and $\hat{\bar f}_X$ such that $\tilde{\mu}_{k,p}(x,\theta_k)\mathds{1}\{|x/h_{G}|\le 1\}$ is uniformly consistent for $\mu_k(x,\theta_k)\mathds{1}\{|x/h_{G}|\le 1\}$ on $([\underline{x},\overline{x}] \backslash \{0\}) \times \mathds{T}$. $\hat{\bar f}_X(0)$ is consistent for $\bar f_X(0)$.
\end{assumption}


This Assumption \ref{a:cluster_robust} for cluster sampling corresponds to Assumptions \ref{a:BR}, \ref{a:cond_weak_conv}, \ref{a:multiplier} and \ref{a:first_stage} for the i.i.d. setting.  
In part (i)(b) of Assumption \ref{a:cluster_robust}, we restrict the density functions to be 1-Lipschitz, but this particular scale is imposed only for the sake of concise writings in proofs. 
This number can be replaced by any finite constant.
Part (vi)(a) of Assumption \ref{a:cluster_robust} imposes more restrictions (namely, the identity $\phi$ and the identity $\psi$) than the baseline case, but this assumption is made only for the sake of simplicity and concise writing in proofs, and it can be relaxed.
We make this simplifying assumption because it is already satisfied by most of the important cases anyway, including the case of FQRD, which is the main focus of this paper, as well as the cases of sharp/fuzzy mean RDD studied in prior work (Bartlalotti and Brummet, 2017; Calonico, Cattaneo, Farrell and Titiunik, 2018). 
Under this assumption for cluster sampling, we obtain the weak convergence and the conditional weak convergence results as stated in the following corollary.

\begin{corollary}\label{corollary:cluster_robust}
Suppose Assumption \ref{a:cluster_robust} holds, then there exists a zero mean Gaussian process
$\mathds G:\Omega^x \mapsto \ell^\infty(\mathds T)$ such that
\begin{align*}
\sqrt{Gh^{1+2v}_G}[\hat{\tau}(\cdot)-\tau(\cdot)]
\leadsto &\Upsilon'_W\Big( \frac{[\mu^{(v)}_{2}(0^+,\cdot)-\mu^{(v)}_{2}(0^-,\cdot)]\mathds{G}(\cdot,1)-[\mu^{(v)}_{1}(0^+,\cdot)-\mu^{(v)}_{1}(0^-,\cdot)]
\mathds{G}(\cdot,2)}{(\mu^{(v)}_{2}(0^+,\cdot)-\mu^{(v)}_{2}(0^-,\cdot))^2} \Big)(\cdot).
\end{align*}
Furthermore, for the EMP, we have
\begin{align*}
&\Upsilon'_W\Big( \frac{[\mu^{(v)}_{2}(0^+,\cdot)-\mu^{(v)}_{2}(0^-,\cdot)]\hat\nu_{\xi,G}(\cdot,1)-[\mu^{(v)}_{1}(0^+,\cdot)-\mu^{(v)}_{1}(0^-,\cdot)]
\hat\nu_{\xi,G}(\cdot,2)}{(\mu^{(v)}_{2}(0^+,\cdot)-\mu^{(v)}_{2}(0^-,\cdot))^2} \Big)(\cdot)\\
\overset{p}{\underset{\xi}{\leadsto} }&\Upsilon'_W\Big( \frac{[\mu^{(v)}_{2}(0^+,\cdot)-\mu^{(v)}_{2}(0^-,\cdot)]\mathds{G}(\cdot,1)-[\mu^{(v)}_{1}(0^+,\cdot)-\mu^{(v)}_{1}(0^-,\cdot)]
\mathds{G}(\cdot,2)}{(\mu^{(v)}_{2}(0^+,\cdot)-\mu^{(v)}_{2}(0^-,\cdot))^2} \Big)(\cdot).
\end{align*}
\end{corollary} 

A proof is found in Appendix \ref{sec:corollary:cluster_robust}.

\subsection{Inference with Robustness against No or Weak Jump or Kink}\label{sec:weak_jumps}

Another practically relevant issue to consider as an extension to our baseline model is the issue of no and weak jumps in regression discontinuity designs.
Feir, Lemieux, and Marmer (2016) report significant size distortions in the event of weak jumps in the mean RDD, similarly to weak instrument problems.
In addition, we also find size distortions in the event of weak jumps in the FQRD designs -- see Section \ref{sec:simulation_fqrdd}.
Otsu, Xu, and Matsushita (2015) and Feir, Lemieux, and Marmer (2016) provide weak-identification-robust methods of inference in regression discontinuity designs.
In the spirit of Anderson and Rubin (1949) -- also see Kleibergen (2002), Moreira (2003), and Andrews, Moreira, and Stock (2006) -- a null restriction can provide valid standard errors robustly against the possibility of weak and no jumps (Feir, Lemieux, and Marmer, 2016).
Adopting this idea, we apply our baseline results of the weak convergence and the conditional weak convergence to non- and weak-identification-robust inference in this section.

Under the null hypothesis $H_0: \tau = \tau^\ast$, we can rewrite (\ref{eq:wald_est}) as
\begin{align*}
\phi\left( \lim_{x\downarrow 0}\mu^{(v)}_{1}(x,\cdot)\right)(\cdot)-\phi\left(\lim_{x\uparrow 0}\mu^{(v)}_{1}(x,\cdot)\right)(\cdot) =
\left(\Upsilon^{-1}\tau^\ast\right)(\cdot) \left[\psi\left(\lim_{x\downarrow 0}\mu^{(v)}_{2}(x,\cdot)\right)(\cdot)-\psi\left(\lim_{x\uparrow 0}\mu^{(v)}_{2}(x,\cdot)\right)(\cdot)\right],
\end{align*}
provided that the operator $\Upsilon$ is invertible.
Notice that this characterization allows for a statistical inference without assuming that a nonzero jump or kink exists.
To see this, note that $\psi\left(\lim_{x\downarrow 0}\mu^{(v)}_{2}(x,\cdot)\right)(\cdot)-\psi\left(\lim_{x\uparrow 0}\mu^{(v)}_{2}(x,\cdot)\right)(\cdot)$ no longer appears as a denominator, unlike the original expression (\ref{eq:wald_est}).
We use this equality restriction for a statistical inference.
A sample-analog statistic can be defined as
\begin{align*}
\hat \Xi_\tau (\cdot) = \phi\left( \hat\mu^{(v)}_{1,p}(0^+,\cdot)\right)(\cdot)-\phi\left(\hat\mu^{(v)}_{1,p}(0^-,\cdot)\right)(\cdot) - \left(\Upsilon^{-1} \tau\right)(\cdot) \left[ \psi\left(\hat\mu^{(v)}_{2,p}(0^+,\cdot)\right)(\cdot)-\psi\left(\hat\mu^{(v)}_{2,p}(0^-,\cdot)\right)(\cdot) \right].
\end{align*}
Our identification-robust inference procedure is based on the idea that $\hat\Xi_{\tau^\ast}$ should be uniformly close to zero under the null hypothesis $H_0: \tau = \tau^\ast$.

\newtheorem*{assumption_cond_weak_conv_weak}{Assumption \ref{a:cond_weak_conv}$'$}
\begin{assumption_cond_weak_conv_weak}[Conditional Weak Convergence]
\qquad\\
(i) $\psi$ and $\phi$ are Hadamard differentiable at $\mu^{(v)}_{1}(0^\pm,\cdot)$ and $\mu^{(v)}_{2}(0^\pm,\cdot)$, respectively, tangentially to some subspaces of their domains, with their Hadamard derivatives denoted by $\phi'_{\mu^{(v)}_{1}(0^\pm,\cdot)}$ and $\psi'_{\mu^{(v)}_{2}(0^\pm,\cdot)}$, respectively. $\Upsilon$ is invertible.\\
(ii) $nh^{1+2v}_n\to \infty$  as $n \to \infty$.
\end{assumption_cond_weak_conv_weak}

Compared with Assumption \ref{a:cond_weak_conv}, this assumption is weaker.
Specifically, part (ii) of Assumption \ref{a:cond_weak_conv} has been dropped, and hence we do allow for a possibility of no jump or no kink.
Similar lines of arguments to those of Theorems \ref{theorem:weak_conv} and \ref{theorem:cond_weak_conv} yield the following weak and conditional weak convergence results robustly against possible non-identification due to a lack of jumps or kinks.

\begin{corollary}[Weak and Conditional Weak Convergence]
Under Assumptions \ref{a:BR} and \ref{a:cond_weak_conv}$'$ and the null hypothesis $H_0: \tau=\tau^\ast$, we have 
\begin{align*}
\sqrt{nh^{1+2v}_n} \hat\Xi_{\tau^\ast}
\leadsto \mathbb{G}'(\cdot,1) - (\Upsilon^{-1}\tau^\ast) (\cdot) \mathbb{G}'(\cdot,2),
\end{align*}
where $\mathds{G}'$ is defined in the statement of Theorem \ref{theorem:weak_conv}.
If Assumptions \ref{a:multiplier} and \ref{a:first_stage} hold in addition, then
\begin{align*}
\widehat{\mathds{X}}'_n(\cdot,1) - (\Upsilon^{-1}\tau^\ast) (\cdot) \widehat{\mathds{X}}'_n(\cdot,2)
\underset{\xi}{\overset{p}{\leadsto}}
\mathbb{G}'(\cdot,1) - (\Upsilon^{-1}\tau^\ast) (\cdot) \mathbb{G}'(\cdot,2),
\end{align*}
where $\widehat{\mathds{X}}'_n$ is defined in the statement of Theorem \ref{theorem:cond_weak_conv}.
\end{corollary}

\subsection{Covariates}\label{sec:covariates}
Many applied researchers augment their RDD estimating equations with additional predetermined covariates. 
Formal justification and asymptotic properties of such a practice for mean sharp/fuzzy RDD/RKD are shown in Calonico, Cattaneo, Farrell and Titiunik (2018). 
In this section, we demonstrate that such a practice is also justified in the case of FQRD designs, and then provide a generalization of Corollary \ref{corollary:FQRD} for models with covariates. 

Suppose that we have a $d$-dimensional vector of covariates denoted by $Z_i$.
Define
\begin{align*}
&\left(\check\mu_{1}(0^+,y,d),\check\mu_{1}'(0^+,y,d),\check\mu_{1}''(0^+,y,d), \check\gamma_{1+}(y,d)\right) = \\
&\arg\min_{(\mu,\mu',\mu'',\gamma)\in \mathbbm R^{3+d}} \sum_{i:X_i>0} \left( \mathbbm{1}\left\{Y_i^\ast \leq y ,D_i^\ast=d\right\} - \left\{\mu + \mu' X_i + \frac{\mu''}{2!} X_i^2 + Z_i' \gamma \right\} \right)^2 \cdot K\left(\frac{X_i}{h_{n}}\right),\\
&\left(\check\mu_{2}(0^+,d),\check\mu_{2}'(0^+,d),\check\mu_{2}''(0^+,d), \check\gamma_{2+}(d)\right) = \\
&\arg\min_{(\mu,\mu',\mu'',\gamma)\in \mathbbm R^{3+d}} \sum_{i:X_i>0} \left( \mathbbm{1}\left\{D_i^\ast=d\right\} - \left\{\mu + \mu' X_i + \frac{\mu''}{2!} X_i^2 + Z_i' \gamma \right\} \right)^2 \cdot K\left(\frac{X_i}{h_{n}}\right), \text{ and }\\
&\check F_{Y^d|C} (y) = \frac{\check\mu_{1+}(y,d) - \check\mu_{1-}(y,d) }{ \check\mu_{2+}(d) -  \check\mu_{2-}(d)  }.
\end{align*}
The left-hand-side counterparts of them are defined analogously. 
The FQRD estimator based on these covariate-augmented local linear estimators is defined by
\begin{align*}
\check\tau_{FQRD}(\theta) =
\check Q_{Y^1|C}(\theta)
-
\check Q_{Y^0|C}(\theta),
\end{align*}
where
$$
\check Q_{Y^d|C}(\theta) = 
\inf\left\{y : \check F_{Y^d|C}(y) \geq \theta\right\}
\quad\text{for each } d \in \{0,1\}.
$$
We next introduce the short-hand notations: $\gamma_{1+}(y,d)=\sigma^{-1}_{Z+} E[(Z_i-\mu_{Z+}(X_i))\mathbbm 1\{ Y_i\le y,D_i=d \}|X_i=0^+ ]$, $ \gamma_{2+}(d)=\sigma^{-1}_{Z+} E[(Z_i-\mu_{Z+}(X_i))\mathbbm 1\{ D_i=d \}|X_i=0^+ ]$, $\mu_{Z_+}(x)=[\mu_{Z_1+}(x),...,\mu_{Z_d+}(x)]'$, $\mu_{Z_l+}(x)=E[Z_{il}|X_i=x]$,  and $\sigma^2_{Z+}=V(Z_i|X_i=0^+)$.
Their left-hand-side counterparts are also defined analogously.
Consider the following assumptions.

\newtheorem*{assumption_C}{Assumption C}
\begin{assumption_C}\label{a:covariates}
(i) $\gamma_{1}=\gamma_{1+}=\gamma_{1-}$ and $\gamma_{2}=\gamma_{2+}=\gamma_{2-}$ on $ \mathscr Y_1 \times \{0,1\}$.
(ii) $\mu_{Z+}=\mu_{Z-}  $ on $(\underline{x},0]$ and $[0,\overline{x})$. 
(iii)  For all $(y,d)\in \mathscr Y_1 \times \{0,1\}$, $E[Z_i[\mathbbm 1\{Y_i\le y ,D_i=d\},D_i]'|X=\cdot]$ are continuously differentiable on $(\underline{x},0]$ and $[0,\overline{x})$. (iv) $\mu_{Z+}(\cdot)$ and $\mu_{Z-}(\cdot)$ are three-times continuously differentiable on $(\underline{x},0]$ and $[0,\overline{x})$. (v) $\sigma^2_{Z\pm}(\cdot)$ are continuous and invertible on $(\underline{x},0]$ and $[0,\overline{x})$. 
\end{assumption_C}

\newtheorem*{assumption_first_stage_estimation_covariates}{Assumption \ref{a:first_stage}$'$}
\begin{assumption_first_stage_estimation_covariates}[First Stage Estimation]
$\ddot{\mu}_{k,p}(x,\theta_k)\mathds{1}\{|x/h_{k,n}(\theta_k)|\le 1\}$ is uniformly consistent for $\mu_k(x,\theta_k)\mathds{1}\{|x/h_{k,n}(\theta_k)|\le 1\}$ on $([\underline{x},\overline{x}] \backslash \{0\}) \times \mathds{T}$. $\hat{f}_X(0)$ is consistent for $f_X(0)$.
\end{assumption_first_stage_estimation_covariates}

With a first-stage estimator satisfying this assumption, we define $\check U_{1i}=\mathbbm 1\{Y_i^*\le y,D_i^*=d\} -\check\gamma'_{1+}(y,d)\hat \mu_{Z+}(X_i) - \ddot  \mu_1(X_i,y,d)$ and $\check U_{2i}=\mathbbm 1\{D_i^*=d\} -\check\gamma'_{2+}(d)\hat \mu_{Z+}(X_i) - \ddot  \mu_2(X_i,d)$. 
Define the estimated multiplier processes
\begin{align*}
\check{\nu}^+_{\xi,n}(y,d,1)=\sum_{i:X_i>0}\xi_i\frac{(1 \ 0 \ 0) \cdot (\Gamma^\pm_2)^{-1} \cdot \left(1 \ \frac{X_i}{h_n} \ \frac{X_i^2}{h_n^2}\right)' \check U_{1i}K\left(\frac{X_i}{h_n}\right)}{\sqrt{nh_n}\hat{f}_X(0)}
\\\text{and}\qquad
\check{\nu}^+_{\xi,n}(d,2)=\sum_{i:X_i>0}\xi_i\frac{(1 \ 0 \ 0) \cdot (\Gamma^\pm_2)^{-1} \cdot \left(1 \ \frac{X_i}{h_n} \ \frac{X_i^2}{h_n^2}\right)' \check U_{2i}K\left(\frac{X_i}{h_n}\right)}{\sqrt{nh_n}\hat{f}_X(0)}.
\end{align*}
We now define $\check{\mathds{X}}_n'$, $\check{\mathds{Y}}_n$ and $\check{\Upsilon}'_W$ as $\widehat{\mathds{X}}_n'$, $\widehat{\mathds{Y}}_n$ and $\widehat{\Upsilon}'_W$, respectively, but with $\check \nu^\pm_{\xi,n}$ in place of $\hat\nu^\pm_{\xi,n}$, $\check Q_{Y^d|C}$ in place of $\hat Q_{Y^d|C}$, and $\check \mu_{k,2}$ in place of $\hat \mu_{k,2}$.

\begin{corollary}\label{corollary:covariates}
Suppose that Assumptions S, K, M, FQRD, C and \ref{a:first_stage}$'$ hold.
Corollary \ref{corollary:FQRD} holds with $\check \tau_{FQRD}$ and $\check{\Upsilon}'_W(\check{\mathds Y}_n)$ in place of $\hat \tau_{FQRD}$ and $\widehat{\Upsilon}'_W(\widehat{\mathds Y}_n)$, respectively.
\end{corollary}

A proof is found in Appendix \ref{sec:corollary:covariates}.

\begin{remark}\label{remark:covariates}
While many cases covered in our general framework can also be extended to covariate-augmented estimating equations \textit{\`a la} Calonico, Cattaneo, Farrell and Titiunik (2018), not all the cases can be.
The covariate extension is not compatible with the designs in which the operators, $\phi(\cdot)$ and $\psi(\cdot)$, are not linear, and hence we did not work with the general unified framework in this subsection. 
Since the cases of mean sharp/fuzzy RDD/RKD are covered in prior work (Calonico, Cattaneo, Farrell and Titiunik, 2018), we focused on the case of the FQRD design, which is the main design of interest in this paper.
\end{remark}

\section{Summary}\label{sec:summary}

The existing literature on robust inference for causal effects in RDD and RKD covers major important cases, including the sharp mean RDD, the fuzzy mean RDD, the sharp mean RKD, the sharp fuzzy RKD, the sharp quantile RDD, and the sharp quantile RKD.
These existing methods, however, do not cover uniform inference for CDF and quantile cases in fuzzy designs.
Particularly, the existing methods are not able to handle robust uniform inference for quantile treatment effects in the fuzzy RDD, despite its practical relevance, e.g., Clark and Martorell (2014) and Deshpande (2016).
In this light, this paper proposes a new general robust inference method that covers the fuzzy quantile RDD in particular, but also all the other cases covered by the existing methods of robust inference.
We provide high-level statements for the general result, and also provide more primitive conditions and detailed discussions focusing on the fuzzy quantile RDD.

Monte Carlo simulation studies confirm the theoretical properties for data generating processes calibrated to match real applications.
Applying the proposed method to real data, we study causal effects on test outcomes of the Oklahoma pre-K program, following the earlier work by Gormley, Gayer, Phillips and Dawson (2005) and Frandsen, Fr\"olich and Melly (2012).
Despite the generally larger lengths of uniform confidence bands than point-wise confidence intervals, we obtain results qualitatively similar to those of Frandsen, Fr\"olich and Melly (2012).


Finally, we conclude this paper with a guide for practitioners.
Our method applies to most, if not all, of the commonly used local Wald estimators, but practitioners may want to know which method they should consider in empirical research.
We are not aware yet of theoretical benefits of the multiplier bootstrap method over the analytic method -- which is left for future research -- and hence analytic methods may be preferred for their computational advantage for applications where such an analytic method is applicable.
As such, the analytic method proposed by Calonico, Cattaneo, and Titiunik (2014) is probably a superior option if a practitioner is interested in mean-regression-based designs.
As their method is limited to mean-regression-based designs, one may need to seek alternative methods for quantile-regression-based designs.
The pivotal methods of Qu and Yoon (2015b) and Chiang and Sasaki (2017) are probably superior to our method if a practitioner is interested in sharp quantile RDD and sharp quantile RKD, respectively.
The bootstrap method proposed in this paper, to the best of our knowledge, is the only option if a practitioner is interested in the fuzzy quantile RDD and the fuzzy CDF discontinuity design.


\appendix
\section{Mathematical Appendix}\label{sec:math_appenidx}
The relation $a\lesssim b$ means that there exists $C$, $0<C<\infty$, such that $a\le C b$. 
For the definitions and notations of what follows, we refer the reader to Gin\'e and Nickl (2016), Kosorok (2003) and van der Vaart and Wellner (1996).
For an arbitrary semimetric space $(T,d)$, define the covering number $N(\epsilon,T,d)$ to be the minimal number of closed $d$-balls of radius $\epsilon$ required to cover $T$. For any function $f:  T\times \Omega\mapsto \mathbb R$ and any $w\in \Omega$, denote $\|f(\cdot,w)\|_T=\sup_{t\in T}|f (t,w)|$. We further define the Uniform Entropy Integral $J(\delta,\mathscr{F},F)=\sup_{Q}\int_{0}^{\delta}  \sqrt{1+\log N (\epsilon\norm{F}_{Q,2}, \mathscr{F}, \norm{\cdot}_{Q,2} )} d\epsilon$, where $\norm{\cdot}_{Q,2}$ is the $L2$ norm with respect to measure $Q$ and the supremum is taken over all probability measure over $(\Omega^x,\mathcal{F}^x)$.
A class of measurable functions $\mathscr{F}$ is called a VC (Vapnik-Chervonenkis) type with respect to a measurable envelope $F$ of $\mathscr{F}$ if there exist finite constants $A\ge 1$, $V\ge 1$ such that for all probability measures $Q$ on $(\Omega^x,\mathcal{F}^x)$, we have $N(\epsilon\norm{F},\mathscr{F},{\norm{\cdot}_{Q,2}})\le (\frac{A}{\epsilon})^V$, for $0<\epsilon\le 1$.
We use $C$, $C_1$, $C_2$,... to denote constants that are positive and independent of $n$.
The values of $C$ may change at each appearance but $C_1$, $C_2$,... are fixed.

\subsection{Uniform Bahadur Representation}\label{sec:lemma:BR}

We obtain the uniform BR of the local polynomial estimators $\hat\mu_{1,p}^{(v)}(0^\pm,\cdot)$ and $\hat\mu_{2,p}^{(v)}(0^\pm,\cdot)$ presented in the following lemma -- see Section \ref{sec:lemma:BR} for a proof.

\begin{lemma}[Uniform Bahadur Representation]\label{lemma:BR}
Under Assumption \ref{a:BR}, we have:
\begin{align*}
&\sqrt{nh^{1+2v}_{k,n}(\theta_k)}\big( \hat{\mu}^{(v)}_{k,p}(0^\pm,\theta_k)- \mu^{(v)}_{k}(0^\pm,\theta_k)-h^{p+1-v}_{k,n}(\theta_k)\frac{e'_v(\Gamma^{\pm}_p)^{-1}\Lambda^\pm_{p,p+1}}{(p+1)!}\mu^{(p+1)}_{k}(0^\pm,\theta_k) \big)\\
= \ &  v!\sum_{i=1}^{n}\frac{e'_v(\Gamma^\pm_p)^{-1}\Epsilon_k(Y_i,D_i,X_i,\theta) r_p(\frac{X_i}{h_{k,n}(\theta_k)}) K(\frac{X_i}{h_{k,n}(\theta_k)}) \delta^\pm_i}{\sqrt{nh_{k,n}(\theta_k)}f_X(0)}+o^x_p(1)
\end{align*}
uniformly for all $\theta_k\in \Theta_k$ for each $k \in \{1,2\}$.
\end{lemma}

Notice that the leading bias terms on the left-hand side of the equations in this lemma are of the $(p+1)$-st order.
Thus, the asymptotic distributions of the Bahadur representation take into account the $p$-th order bias reduction.
This property is the key to develop a method of inference which is robust against large bandwidth parameter choices.

\begin{remark}
It is worth remarking on the relation between this lemma and Theorems 2 and 3 in Frandsen, Fr\"olich and Melly (2012).
The proofs of Theorem 2 and 3 in Frandsen, Fr\"olich and Melly (2012) rely on the exact solution of the local linear smoother (see their equation (11)), while ours makes use of the uniform Bahadur representation derived in our Lemma 1. 
As illustrated in Section 4.2 of Dony, Einmahl and Mason (2006), if we let 
\begin{align*}
\tilde f_{n,h,j}&=\frac{1}{nh}\sum_{i=1}^n \Big(\frac{X_i-x}{h}\Big)^j K\Big(\frac{x-X_i}{h}\Big)\\
\tilde r_{n,h,j}&=\frac{1}{nh}\sum_{i=1}^n Y_i\Big(\frac{X_i-x}{h}\Big)^j K\Big(\frac{x-X_i}{h}\Big),
\end{align*}
then the exact solution for the local linear mean regression estimator ($p=1$) has the form
\begin{align*}
\hat \mu (x)=\frac{\tilde f_{n,h,2}\tilde r_{n,h,0}-\tilde f_{n,h,1}\tilde r_{n,h,1}}{\tilde f_{n,h,0}\tilde f_{n,h,2}-\tilde f^2_{n,h,1}}.
\end{align*}
For the local quadratic mean regression estimator ($p=2$), we have the exact solution
\begin{align*}
\hat \mu (x)=\frac{(\tilde f_{n,h,2}\tilde f_{n,h,4}-\tilde f_{n,h,3}^2)\tilde r_{n,h,0} + (\tilde f_{n,h,2} \tilde f_{n,h,3}-\tilde f_{n,h,1}\tilde f_{n,h,4})\tilde r_{n,h,1}+(\tilde f_{n,h,1}\tilde f_{n,h,3}-\tilde f^2_{n,h,2})\tilde r_{n,h,2}}{\tilde f_{n,h,0}\tilde f_{n,h,2}\tilde f_{n,h,4}-\tilde f_{n,h,0}\tilde f_{n,h,3}^2-\tilde f_{n,h,1}^2\tilde f_{n,h,4}+\tilde f_{n,h,1}\tilde f_{n,h,2}\tilde f_{n,h,3}-\tilde f_{n,h,2}^2}.
\end{align*}
Therefore, for kernel estimators with the squared loss, even though it is in principle possible to establish results by working with the exact solution for any order of polynomial, the expression of exact solution becomes rather complex as the order $p$ increases due to the presence of an inverse factor, whose determinant gives the denominator terms of the above expressions. 
Further, this illustration is only for mean regressions, and it would add even more complications if we have an index $\theta$ to keep track of in this inverse factor. 
By working with its limit $\Gamma^\pm_p$, our uniform Bahadur representation holds with the same unified general expression for any arbitrary order, with any index set that satisfies Assumption \ref{a:BR} (ii) (a).
This helps us to allow for robust bias correction of any order in a unified manner, in comparison with the exact solution approach taken by Frandsen, Fr\"olich and Melly (2012).
\end{remark}

\begin{proof}
In this proof, we will show the first result for the case of $k=1$ and $\pm = +$.
All the other results can be shown by similar lines of proof.
As in Section 1.6 of Tsybakov (2003), the solution to (\ref{eq:local_poly_reg}) can be computed explicitly as
\begin{align*}
\hat{\alpha}_{1+,p}(\theta_1)=&[\hat{\mu}_{1,p} (0^\pm,\theta_1),\hat{\mu}^{(1)}_{1,p} (0^\pm,\theta_1)h_{1,n}(\theta_1),...,\hat{\mu}^{(p)}_{1,p} (0^\pm,\theta_1)h^p_{1,n}(\theta_1)/p!]'\\
=&\Big[ \frac{1}{nh_{1,n}(\theta_1)}\sum_{i=1}^{n}\delta^+_i K(\frac{X_i}{h_{1,n}(\theta_1)}) r_p(\frac{X_i}{h_{1,n}(\theta_1)}) r'_p(\frac{X_i}{h_{1,n}(\theta_1)}) \Big]^{-1} \cdot
\\
&\quad\Big[ \frac{1}{nh_{1,n}(\theta_1)}\sum_{i=1}^{n}\delta^+_i K(\frac{X_i}{h_{1,n}(\theta_1)}) r_p(\frac{X_i}{h_{1,n}(\theta_1)})g_1(Y_i,\theta_1)\Big]
\end{align*}
For each data point $X_i>0$, a mean value expansion by Assumption \ref{a:BR} (ii)(b) gives
\begin{align*}
g_1(Y_i,\theta_1)&=\mu_{1}(X_i,\theta)+\Epsilon_{1}(Y_i,D_i,X_i,\theta)\\
&=\mu_{1}(0,\theta)+\mu^{(1)}_{1} (0^+,\theta_1)X_i+...+\mu^{(p)}_{1} (0^+,\theta_1)\frac{X_i^p}{p!}+\mu^{(p+1)}_{1}(x^*_{ni},\theta_1)\frac{X_i^{(p+1)}}{(p+1)!}+\Epsilon_{1}(Y_i,D_i,X_i,\theta)\\
&=r'_p(\frac{X_i}{h_{1,n}(\theta_1)})\alpha_{1+,p}(\theta_1)+\mu^{(p+1)}_{1}(x^*_{ni},\theta_1)h^{p+1}_{1,n}(\theta_1)\frac{(\frac{X_i}{h_{1,n}(\theta_1)})^{(p+1)}}{(p+1)!}+\Epsilon_{1}(Y_i,D_i,X_i,\theta)
\end{align*}
for an $x^*_{ni}\in (0, X_i]$. Substituting this expansion in the equation above and multiplying both sides by $\sqrt{nh_{1,n}(\theta_1)}e'_v$, we obtain
\begin{align*}
 &\sqrt{nh_{1,n}(\theta_1)}\hat\mu^{(v)}_1(0^+,\theta_1)h^v_{1,n}(\theta_1)/v!\\
 =&\sqrt{nh_{1,n}(\theta_1)}e'_v\Big[ \frac{1}{nh_{1,n}(\theta_1)}\sum_{i=1}^{n}\delta^+_i K(\frac{X_i}{h_{1,n}(\theta_1)}) r_p(\frac{X_i}{h_{1,n}(\theta_1)})r'_p(\frac{X_i}{h_{1,n}(\theta_1)}) \Big]^{-1} \cdot
\\
&\quad\Big[ \frac{1}{nh_{1,n}(\theta_1)}\sum_{i=1}^{n} \delta^+_i K(\frac{X_i}{h_{1,n}(\theta_1)}) r_p(\frac{X_i}{h_{1,n}(\theta_1)}) g_1(Y_i,t)\Big]\\
 =&\sqrt{nh_{1,n}(\theta_1)}e'_v\Big[ \frac{1}{nh_{1,n}(\theta_1)}\sum_{i=1}^{n}\delta^+_i K(\frac{X_i}{h_{1,n}(\theta_1)}) r_p(\frac{X_i}{h_{1,n}(\theta_1)})r'_p(\frac{X_i}{h_{1,n}(\theta_1)}) \Big]^{-1} \cdot
\\
&\quad\Big[ \frac{1}{nh_{1,n}(\theta_1)}\sum_{i=1}^{n} \delta^+_i K(\frac{X_i}{h_{1,n}(\theta_1)}) r_p(\frac{X_i}{h_{1,n}(\theta_1)})\Big( r'_p(\frac{X_i}{h_{1,n}(\theta_1)})\alpha_{1+,p}(\theta_1)\\
 \qquad&+\mu^{(p+1)}_{1}(x^*_{ni},\theta_1)h^{p+1}_{l,n}(\theta_1)\frac{(\frac{X_i}{h_{1,n}(\theta_1)})^{p+1}}{(p+1)!}+\Epsilon_{1}(Y_i,D_i,X_i,\theta) \Big)\Big]\\
 =&\sqrt{nh_{1,n}(\theta_1)}e'_v\alpha_{1+,p}(\theta_p)\\
 +&e'_v\Big[ \frac{1}{nh_{1,n}(\theta_1)}\sum_{i=1}^{n}\delta^+_i K(\frac{X_i}{h_{1,n}(\theta_1)}) r_p(\frac{X_i}{h_{1,n}(\theta_1)}) r'_p(\frac{X_i}{h_{1,n}(\theta_1)}) \Big]^{-1} \cdot
\\
&\quad \frac{1}{\sqrt{nh_{1,n}(\theta_1)}}\sum_{i=1}^{n} \delta^+_i K(\frac{X_i}{h_{1,n}(\theta_1)}) r_p(\frac{X_i}{h_{1,n}(\theta_1)})\mu^{(p+1)}_{1}(x^*_{ni},\theta_1)h^{p+1}_{l,n}(\theta_1)\frac{(\frac{X_i}{h_{1,n}(\theta_1)})^{p+1}}{(p+1)!}\\
  +& e'_v\Big[ \frac{1}{nh_{1,n}(\theta_1)}\sum_{i=1}^{n}\delta^+_i K(\frac{X_i}{h_{1,n}(\theta_1)}) r_p(\frac{X_i}{h_{1,n}(\theta_1)}) r'_p(\frac{X_i}{h_{1,n}(\theta_1)})  \Big]^{-1} \cdot
	\\
	&\quad \frac{1}{\sqrt{nh_{1,n}(\theta_1)}}\sum_{i=1}^{n} \delta^+_i K(\frac{X_i}{h_{1,n}(\theta_1)}) r_p(\frac{X_i}{h_{1,n}(\theta_1)})\Epsilon_{1}(Y_i,D_i,X_i,\theta)\\
 =&\sqrt{nh_{1,n}(\theta_1)}\mu^{(v)}_1(0^+,\theta_1)h^v_{1,n}(\theta_1)/v!+(a)+(b)
\end{align*}
where
\begin{align*}
(a)&=e'_v\Big[ \frac{1}{nh_{1,n}(\theta_1)}\sum_{i=1}^{n}\delta^+_i K(\frac{X_i}{h_{1,n}(\theta_1)}) r_p(\frac{X_i}{h_{1,n}(\theta_1)}) r'_p(\frac{X_i}{h_{1,n}(\theta_1)}) \Big]^{-1} \cdot
\\
&\quad \frac{1}{\sqrt{nh_{1,n}(\theta_1)}}\sum_{i=1}^{n} \delta^+_i K(\frac{X_i}{h_{1,n}(\theta_1)}) r_p(\frac{X_i}{h_{1,n}(\theta_1)})\mu^{(p+1)}_{1}(x^*_{ni},\theta_1)h^{p+1}_{1,n}(\theta_1)\frac{(\frac{X_i}{h_{1,n}(\theta_1)})^{p+1}}{(p+1)!}\\
\end{align*}
and
\begin{align*}
(b)&= e'_v\Big[ \frac{1}{nh_{1,n}(\theta_1)}\sum_{i=1}^{n}\delta^+_i K(\frac{X_i}{h_{1,n}(\theta_1)}) r_p(\frac{X_i}{h_{1,n}(\theta_1)}) r'_p(\frac{X_i}{h_{1,n}(\theta_1)})  \Big]^{-1} \cdot
	\\
	&\quad \frac{1}{\sqrt{nh_{1,n}(\theta_1)}}\sum_{i=1}^{n} \delta^+_i K(\frac{X_i}{h_{1,n}(\theta_1)}) r_p(\frac{X_i}{h_{1,n}(\theta_1)})\Epsilon_{1}(Y_i,D_i,X_i,\theta)
\end{align*}
We will show stochastic limits of the (a) and (b) terms above.

\textbf{Step 1}
First, we consider their common inverse factor.
Specifically, we show that
\begin{align}\label{eq:gamma_unif_conv}
[ \frac{1}{nh_{1,n}(\theta_1)}\sum_{i=1}^{n}\delta^+_i K(\frac{X_i}{h_{1,n}(\theta_1)}) r_p(\frac{X_i}{h_{1,n}(\theta_1)}) r'_p(\frac{X_i}{h_{1,n}(\theta_1)}) ]^{-1} \xrightarrow{p} (\Gamma^{+}_p)^{-1}/f_X(0)
\end{align}
uniformly in $\theta_1$.
Note that by Minkowski's inequality
\begin{align*}
&\Big|[ \frac{1}{nh_{1,n}(\theta_1)}\sum_{i=1}^{n}\delta^+_i K(\frac{X_i}{h_{1,n}(\theta_1)}) r_p(\frac{X_i}{h_{1,n}(\theta_1)}) r'_p(\frac{X_i}{h_{1,n}(\theta_1)}) ]^{-1}
-
(\Gamma^{+}_p)^{-1}/f_X(0)\Big|_{\Theta_1}\\
\le &\Big|[ \frac{1}{nh_{1,n}(\theta_1)}\sum_{i=1}^{n}\delta^+_i K(\frac{X_i}{h_{1,n}(\theta_1)}) r_p(\frac{X_i}{h_{1,n}(\theta_1)}) r'_p(\frac{X_i}{h_{1,n}(\theta_1)}) ]^{-1}\\
&-E[[ \frac{1}{nh_{1,n}(\theta_1)}\sum_{i=1}^{n}\delta^+_i K(\frac{X_i}{h_{1,n}(\theta_1)}) r_p(\frac{X_i}{h_{1,n}(\theta_1)}) r'_p(\frac{X_i}{h_{1,n}(\theta_1)}) ]^{-1}]\Big|_{\Theta_1}\\
+&
\Big|E[[ \frac{1}{nh_{1,n}(\theta_1)}\sum_{i=1}^{n}\delta^+_i K(\frac{X_i}{h_{1,n}(\theta_1)}) r_p(\frac{X_i}{h_{1,n}(\theta_1)}) r'_p(\frac{X_i}{h_{1,n}(\theta_1)}) ]^{-1}]-(\Gamma^{+}_p)^{-1}/f_X(0)\Big|_{\Theta_1}
\end{align*}
where the first term on the right hand side is stochastic, while the second term is deterministic.
First, regarding the deterministic part, we have
\begin{align*}
E[\frac{1}{nh_{1,n}(\theta_1)}\sum_{i=1}^{n}\delta^+_i K(\frac{X_i}{h_{1,n}(\theta_1)}) r_p(\frac{X_i}{h_{1,n}(\theta_1)}) r'_p(\frac{X_i}{h_{1,n}(\theta_1)})]=f_X(0)\Gamma^+_p + O(h_{1,n}(\theta_1))
\end{align*}
uniformly by Assumption \ref{a:BR} (i), (iii), and (iv).
For the stochastic part, we will show that each entry of such a matrix converges in probability with respect to $\mathds{P}^x$ uniformly. We may write
\begin{align*}
&\mathscr{F}_r=\{ x \mapsto \mathds{1}\{x\ge0\}K(ax)(ax)^{r}\mathds{1}\{ax\in [-1,1]\}:a>1/h_0 \}\\
&\mathscr{F}_{n,r}=\{ x \mapsto \mathds{1}\{x\ge0\}K(x/h_{1,n}(\theta_1))(x/h_{1,n}(\theta_1))^{r}:\theta_1\in \Theta_1 \}
\end{align*}
for each integer $r$ such that $0\le r \le 2p$.
By Lemma 
\ref{lemma:VC type_stability}, each  $\mathscr{F}_r$ is of VC type (Euclidean) with envelope $F=\norm{K}_{\infty}$ under Assumption \ref{a:BR} (iii) and (iv)(a) and (b), i.e., there exist constants $k$, $v < \infty$ such that $ \sup_Q \log N (\epsilon\norm{F}_{Q,2} , \mathscr{F}_r, \norm{\cdot}_{Q,2} )  \le (\frac{k}{\epsilon  })^v$ for $0< \epsilon \le 1 $ and for all probability measures $Q$ supported on $[\underline x, \overline x]$.
This implies $J(1,\mathscr{F}_r,F)=\sup_{Q}\int_{0}^{1}  \sqrt{1+\log N (\epsilon\norm{F}_{Q,2}, \mathscr{F}_r, \norm{\cdot}_{Q,2} )} d\epsilon  < \infty $.
Since $F \in L_2 (P)$, we can apply Theorem 5.2 of Chernozhukov, Chetverikov and Kato (2014) to obtain
\begin{align*}
E\bigg[\sup_{f \in \mathscr{F}_r} \bigg|\frac{1}{\sqrt{n}} \sum_{i=1}^{n}  \Big(f(X_i) - E f \Big) \bigg| \bigg] \le C\{ J(1,\mathscr{F}_r,F)\norm{F}_{P,2} + \frac{\norm{K}_{\infty} J^2(1,\mathscr{F}_r,F)}{\delta^2 \sqrt{n}}\} < \infty
\end{align*}
for a universal constant $C>0$. Note that $\mathscr{F}_{n,r}\subset\mathscr{F}_r$ for all $n\in \mathds{N}$. Multiplying both sides by $[\sqrt{n}h_{1,n}(\theta_1)]^{-1}$ yields
\begin{align*}
& E\bigg[\sup_{\theta_1 \in \Theta_1} \bigg|\frac{1}{nh_{1,n}(\theta_1)}\sum_{i=1}^{n}\delta^+_i K(\frac{X_i}{h_{1,n}(\theta_1)}) (\frac{X_i}{h_{1,n}(\theta_1)})^{s}
\\
&\qquad
-\frac{1}{nh_{1,n}(\theta_1)}\sum_{i=1}^{n}E[\delta^+_i K(\frac{X_i}{h_{1,n}(\theta_1)}) (\frac{X_i}{h_{1,n}(\theta_1)})^{s})] \bigg| \bigg]\\
\le&
\frac{1}{\sqrt{n}h_{1,n}(\theta_1)} C\{ J(1,\mathscr{F}_r,F)\norm{F}_{P,2} + \frac{B J^2(1,\mathscr{F}_r,F)}{\delta^2 \sqrt{n}}\} \\
=&O(\frac{1}{\sqrt{n}h_n})
\end{align*}
The last line goes to zero uniformly under Assumption \ref{a:BR}(iii).
Finally, Markov's inequality gives the uniform convergence of the stochastic part at the rate $O^x_p(\frac{1}{\sqrt{n}h_n})$.
Consequently, we have the uniform convergence in probability for each $r\in \{0,...,2p\}$. Assumption \ref{a:BR} (iv)(c) and the continuous mapping theorem concludes (\ref{eq:gamma_unif_conv}).

\textbf{Step 2}
For term (a), we may again use Minkowski's inequality under the supremum norm as in Step 1 to decompose
\begin{align*}
&\Big|\frac{1}{\sqrt{nh_{1,n}(\theta_1)}}\sum_{i=1}^{n} \delta^+_i K(\frac{X_i}{h_{1,n}(\theta_1)}) r_p(\frac{X_i}{h_{1,n}(\theta_1)})\mu^{(p+1)}_{1}(x^*_{ni},\theta_1)h^{p+1}_{1,n}(\theta_1)\frac{(\frac{X_i}{h_{1,n}(\theta_1)})^{p+1}}{(p+1)!}-0\Big|_{\Theta_1}\\
\le&\Big| \frac{1}{\sqrt{nh_{1,n}(\theta_1)}}\sum_{i=1}^{n} \delta^+_i K(\frac{X_i}{h_{1,n}(\theta_1)}) r_p(\frac{X_i}{h_{1,n}(\theta_1)})\mu^{(p+1)}_{1}(x^*_{ni},\theta_1)h^{p+1}_{1,n}(\theta_1)\frac{(\frac{X_i}{h_{1,n}(\theta_1)})^{p+1}}{(p+1)!}\\
&-E[\frac{1}{\sqrt{nh_{1,n}(\theta_1)}}\sum_{i=1}^{n} \delta^+_i K(\frac{X_i}{h_{1,n}(\theta_1)}) r_p(\frac{X_i}{h_{1,n}(\theta_1)})\mu^{(p+1)}_{1}(x^*_{ni},\theta_1)h^{p+1}_{1,n}(\theta_1)\frac{(\frac{X_i}{h_{1,n}(\theta_1)})^{p+1}}{(p+1)!}] \Big|_{\Theta_1}\\
+&
\Big|E[\frac{1}{\sqrt{nh_{1,n}(\theta_1)}}\sum_{i=1}^{n} \delta^+_i K(\frac{X_i}{h_{1,n}(\theta_1)}) r_p(\frac{X_i}{h_{1,n}(\theta_1)})\mu^{(p+1)}_{1}(x^*_{ni},\theta_1)h^{p+1}_{1,n}(\theta_1)\frac{(\frac{X_i}{h_{1,n}(\theta_1)})^{p+1}}{(p+1)!}]-0 \Big|_{\Theta_1}
\end{align*}
Under Assumption \ref{a:BR}(i),(ii)(b),(iii),(iv)(a), standard calculations show that the deterministic part
\begin{align*}
&E[\frac{1}{\sqrt{nh_{1,n}(\theta_1)}}\sum_{i=1}^{n} \delta^+_i K(\frac{X_i}{h_{1,n}(\theta_1)}) r_p(\frac{X_i}{h_{1,n}(\theta_1)})\mu^{(p+1)}_{1}(x^*_{ni},\theta_1)h^{p+1}_{1,n}(\theta_1)\frac{(\frac{X_i}{h_{1,n}(\theta_1)})^{p+1}}{(p+1)!}]
\\
=&\frac{h^{p+1}_{1,n}(\theta_1)\Lambda^+_{p,p+1}}{\sqrt{nh_{1,n}(\theta_1)}(p+1)!}f_X(0)\mu^{(p+1)}_{1}(0^+,\theta_1)+O(\frac{h^{p+2}_n}{\sqrt{nh_{1,n}}})\\
=&O(\sqrt{\frac{h^{2p+1}_n}{n}})
\end{align*}
uniformly in $\theta_1$.

As for the stochastic part, first note that under Assumption \ref{a:BR}(ii), we know that for a Lipschitz constant $L$ such that $0\le L<\infty$, it holds uniformly in $\theta_1$ that
\begin{align*}
&\Big|\frac{1}{\sqrt{nh_{1,n}(\theta_1)}}\sum_{i=1}^{n} \delta^+_i K(\frac{X_i}{h_{1,n}(\theta_1)}) r_p(\frac{X_i}{h_{1,n}(\theta_1)})h^{p+1}_{1,n}(\theta_1)\frac{(\frac{X_i}{h_{1,n}(\theta_1)})^{p+1}}{(p+1)!}[\mu^{(p+1)}_{1}(x^*_{ni},\theta_1)-\mu^{(p+1)}_{1}(0^+,\theta_1)]\Big|\\
\lesssim & n \max_{1\le i \le n}\Big|\frac{1}{\sqrt{nh_{1,n}(\theta_1)}} \delta^+_i K(\frac{X_i}{h_{1,n}(\theta_1)}) r_p(\frac{X_i}{h_{1,n}(\theta_1)})h^{p+1}_{1,n}(\theta_1)\frac{(\frac{X_i}{h_{1,n}(\theta_1)})^{p+1}}{(p+1)!}[\mu^{(p+1)}_{1}(x^*_{ni},\theta_1)-\mu^{(p+1)}_{1}(0^+,\theta_1)]\Big| \\
\lesssim &\frac{n \max_{1\le i \le n}|\mu^{(p+1)}_{1}(x^*_{ni},\theta_1)-\mu^{(p+1)}_{1}(0^+,\theta_1)|h^{p+1}_{1,n}(\theta_1)}{\sqrt{nh_n}}\\
\le& \frac{nLh^{p+2}_n}{\sqrt{nh_n}}=O^x_p(\sqrt{nh^{2p+3}_n})
\end{align*}
The second inequality holds since $\delta^+_i K(\frac{X_i}{h_{1,n}(\theta_1)}) r_p(\frac{X_i}{h_{1,n}(\theta_1)})\frac{(\frac{X_i}{h_{1,n}(\theta_1)})^{p+1}}{(p+1)!}$ is bounded under Assumption \ref{a:BR}(iii) and (iv)(a), while the third one is due to Assumption \ref{a:BR}(ii).
It is then sufficient to consider the asymptotic behavior of
\begin{align*}
\frac{1}{\sqrt{nh_{1,n}(\theta_1)}}\sum_{i=1}^{n} \delta^+_i K(\frac{X_i}{h_{1,n}(\theta_1)}) r_p(\frac{X_i}{h_{1,n}(\theta_1)})\mu^{(p+1)}_{1}(0^+,\theta_1)h^{p+1}_{1,n}(\theta_1)\frac{(\frac{X_i}{h_{1,n}(\theta_1)})^{p+1}}{(p+1)!}
\end{align*}
Note that $r_p(x/h)(x/h)^{p+1}=[(x/h)^{p+1},...,(x/h)^{2p+1}]'$. So we may let
\begin{align*}
&\mathscr{F}_{s}=\{ x \mapsto \mathds{1}\{x\ge0\}K(ax)(ax)^{s+p+1}\mu^{p+1}_1(0^+,\theta_1) \mathds{1}\{ax\in [-1,1]\}:a\ge 1/h_0 , \theta_1 \in \Theta_1 \}\\
&\mathscr{F}_{n,s}=\{ x \mapsto \mathds{1}\{x\ge0\}K(x/h_{1,n}(\theta_1))(x/h_{1,n}(\theta_1))^{s+p+1}\mu^{p+1}_1(0^+,\theta_1)  :\theta_1 \in  \Theta_1 \}
\end{align*}
for each integer $s$ such that $0\le s \le p$.
Since $(ax)^{s+p+1}\mathds{1}\{ax \in [-1,1]\}$ is Lipschitz continuous for each $a\ge 1/h_0$ and bounded by $1$, it is of VC type by Lemma \ref{lemma:VC type}. We then apply Lemma \ref{lemma:VC type_stability} to show that for each $0 \le s \le p$, $\mathscr{F}_{s}$ is a VC type class with envelope $F_s(x)=\norm{K }_\infty  \int_{\mathcal{Y}\times \mathcal{D}} F_\epsilon(y,d',x)d\mathds{P}^x(y,D=d'|x)$, which is integrable under Assumption \ref{a:BR}(ii)(b) and (iv)(a).
By Assumption \ref{a:BR} (iii) and (iv)(a), $\mathscr{F}_{n,s}\subset\mathscr{F}_s$ for all $n\in \mathds{N}$, thus an argument similar to the one above with Theorem 5.2 of Chernozhukov, Chetverikov and Kato (2014) shows that for each $0 \le s \le p$
\begin{align}
E[\sup_{f\in \mathscr{F}_s}|\frac{1}{\sqrt{n}}\sum_{i=1}^{n}(f(X_i)-Ef(X_i))|]=O(1) \label{eq:Lemma_1}
\end{align}
To prove the uniform convergence of the part of (a) outside the inverse sign, it suffices to show that for each $s$
\begin{align*}
\sup_{f\in \mathscr{F}_s}|\frac{1}{\sqrt{nh_n}}\sum_{i=1}^{n}\Big(f(X_i)h^{p+1}_n-E[f(X_i)h^{p+1}_n]\Big)|\underset{x}{\overset{p}{\to}}0.
\end{align*}
Multiplying both sides of equation (\ref{eq:Lemma_1}) by $h^{p+1}_n/\sqrt{h_n}$ and applying Markov's inequality as in Step 1, we have
\begin{align}
\sup_{f\in \mathscr{F}_s}|\frac{1}{\sqrt{nh_n}}\sum_{i=1}^{n}(f(X_i)h^{p+1}_n-Ef(X_i)h^{p+1}_n)|=O^x_p(\sqrt{\frac{h^{2p+1}_n}{n}})
\end{align}
which converges to zero in probability ($\mathds{P}^x$) under Assumption \ref{a:BR}(iii).
To conclude, we have shown
\begin{align*}
&\frac{1}{\sqrt{nh_{1,n}(\theta_1)}}\sum_{i=1}^{n} \delta^+_i K(\frac{X_i}{h_{1,n}(\theta_1)}) r_p(\frac{X_i}{h_{1,n}(\theta_1)})\mu^{(p+1)}_{1}(x^*_{ni},\theta_1)h^{p+1}_{1,n}(\theta_1)\frac{(X_i/h_{l,n}(\theta_1))^{p+1}}{(p+1)!}
\\
=& O(\sqrt{\frac{h^{2p+1}_n}{n}})+O^x_p(\sqrt{nh^{2p+3}_n})+O^x_p(\sqrt{\frac{h^{2p+1}_n}{n}})
\end{align*}
uniformly in $\theta_1$.
Finally, the continuous mapping theorem gives $(a)\xrightarrow{p}\Big( O(h_n)+O^x_p(\frac{1}{\sqrt{n}h_n})\Big)\Big(O(\sqrt{\frac{h^{2p+1}_n}{n}})+O^x_p(\sqrt{nh^{2p+3}_n})+O^x_p(\sqrt{\frac{h^{2p+1}_n}{n}}) \Big)=o^x_p(1)$ uniformly in $\theta_1$.

\textbf{Step 3}
For term $(b)$, Minkowski's inequality under the supremum norm implies
\begin{align*}
&\Big|\frac{1}{\sqrt{nh_{1,n}(\theta_1)}}\sum_{i=1}^{n} \delta^+_i K(\frac{X_i}{h_{1,n}(\theta_1)}) r_p(\frac{X_i}{h_{1,n}(\theta_1)})\Epsilon_{1}(Y_i,D_i,X_i,\theta)\Big|_{\Theta_1}\\
\le
&\Big| \frac{1}{\sqrt{nh_{1,n}(\theta_1)}}\sum_{i=1}^{n} \delta^+_i K(\frac{X_i}{h_{1,n}(\theta_1)}) r_p(\frac{X_i}{h_{1,n}(\theta_1)})\Epsilon_{1}(Y_i,D_i,X_i,\theta)\\
&- E[\frac{1}{\sqrt{nh_{1,n}(\theta_1)}}\sum_{i=1}^{n} \delta^+_i K(\frac{X_i}{h_{1,n}(\theta_1)}) r_p(\frac{X_i}{h_{1,n}(\theta_1)})\Epsilon_{1}(Y_i,D_i,X_i,\theta)] \Big|_{\Theta_1}\\
+&
\Big| E[\frac{1}{\sqrt{nh_{1,n}(\theta_1)}}\sum_{i=1}^{n} \delta^+_i K(\frac{X_i}{h_{1,n}(\theta_1)}) r_p(\frac{X_i}{h_{1,n}(\theta_1)})\Epsilon_{1}(Y_i,D_i,X_i,\theta)]-0 \Big|_{\Theta_1}
\end{align*}

By construction, local polynomial regression satisfies $E[\Epsilon_{1}(Y_i,D_i,X_i,\theta)|X]=E[g(Y_i,\theta_1)-\mu_1(X_i,\theta_1)|X]=0$, thus by the law of iterated expectations, we have
\begin{align*}
&E[\frac{1}{\sqrt{nh_{1,n}(\theta_1)}}\sum_{i=1}^{n} \delta^+_i K(\frac{X_i}{h_{1,n}(\theta_1)}) r_p(\frac{X_i}{h_{1,n}(\theta_1)})\Epsilon_{1}(Y_i,D_i,X_i,\theta)]\\
=&E[\frac{1}{\sqrt{nh_{1,n}(\theta_1)}}\sum_{i=1}^{n} \delta^+_i K(\frac{X_i}{h_{1,n}(\theta_1)}) r_p(\frac{X_i}{h_{1,n}(\theta_1)})E[\Epsilon_{1}(Y_i,D_i,X_i,\theta)|X]]=0.
\end{align*}
Therefore, in light of (\ref{eq:gamma_unif_conv}), in order to show $(b)\xrightarrow{p} \sum_{i=1}^{n}\frac{(\Gamma^+_p)^{-1}\delta^+_iK(\frac{X_i}{h_{1,n}(\theta_1)})r_p(\frac{X_i}{h_{1,n}(\theta_1)})\Epsilon_{1}(Y_i,D_i,X_i,\theta)  }{\sqrt{nh_{1,n}(\theta_1)}f_X(0)}$ uniformly in $\theta_1$, it remains to be shown that for each coordinate $0\le s \le p$
\begin{align*}
\sup_{\theta_1 \in \Theta_1}\Big|\frac{e'_s}{\sqrt{nh_{1,n}(\theta_1)}}\sum_{i=1}^{n} \delta^+_i K(\frac{X_i}{h_{1,n}(\theta_1)}) r_p(\frac{X_i}{h_{1,n}(\theta_1)})\Epsilon_{1}(Y_i,D_i,X_i,\theta)\Big|=O^x_p(1)
\end{align*}
First note that, under Assumptions \ref{a:BR} (i), (ii) (c), (iii), (iv),
\begin{align*}
&\sup_{\theta_1 \in \Theta_1}E[\frac{e'_s}{\sqrt{nh_{1,n}(\theta_1)}}\sum_{i=1}^{n} \delta^+_i K(\frac{X_i}{h_{1,n}(\theta_1)}) r_p(\frac{X_i}{h_{1,n}(\theta_1)})\Epsilon_{1}(Y_i,D_i,X_i,\theta)]^2\\
=&\sup_{\theta_1 \in \Theta_1}E[\frac{1}{h_{1,n}(\theta_1)} \delta^+_i E[\Epsilon^2_{1}(Y_i,D_i,X_i,\theta)|X_i]K^2(\frac{X_i}{h_{1,n}(\theta_1)})e'_s r_p(\frac{X_i}{h_{1,n}(\theta_1)})r'_p(\frac{X_i}{h_{1,n}(\theta_1)})e_s]\\
=&\sup_{\theta_1 \in \Theta_1}\int_{\mathds{R}_+}\sigma_{11}(\theta_1,\theta_1|uh_{1,n}(\theta_1))K^2(u)e'_s r_p(u)r'_p(u)e_s f_X(uh_{1,n}(\theta_1))du\\
=&\sup_{\theta_1 \in \Theta_1}\int_{\mathds{R}_+}\sigma_{11}(\theta_1,\theta_1|0^+)K^2(u)e'_s r_p(u)r'_p(u)e_s f_X(0^+)du+O^x_p(h_n)\\
\le&f_X(0) e'_s\Psi^+_p e_s\sup_{\theta_1 \in \Theta_1}\sigma_{11}(\theta_1,\theta_1|0^+)+O^x_p(h_n)\\
\lesssim& f_X(0) e'_s\Psi^+_p e_s +O^x_p(h_n),
\end{align*}
where the right hand side is bounded and does not depend on $\theta_1$. Using Markov's inequality, we know that for some constant that doesn't depend on $n$, for each $M>0$
\begin{align*}
\mathds{P}^x\Big(\Big|\frac{e'_s}{\sqrt{nh_{1,n}(\theta_1)}}\sum_{i=1}^{n} \delta^+_i K(\frac{X_i}{h_{1,n}(\theta_1)}) r_p(\frac{X_i}{h_{1,n}(\theta_1)})\Epsilon_{1}(Y_i,D_i,X_i,\theta)\Big|>M\Big)
\le
 \frac{C [f_X(0) e'_s\Psi^+_p e_s +h_n]}{M^2}
\end{align*}
Since the right hand side can be made arbitrarily small by choosing $M$, it is bounded in probability uniformly in $\theta_1$, which concludes the proof.
\end{proof}
\begin{remark}\label{remark:compare_EM2000}
As an anonymous referee pointed out, it is worth discussing the proof of Lemma 1 in comparison to Einmahl and Mason (2000). The main difference between their results and our convergence of stochastic parts is that the rates are different; the framework of Einmahl and Mason (2000) is uniform in both $x$ and a set of functions. 
	Therefore, their rate has an extra $\log h_n$ factor and thus the estimator based on their result does not converge weakly to a Gaussian process. 

	In addition, their results are about the uniform convergence of a kernel type estimator to its expectation, which can be considered counterparts of the stochastic parts in Steps 2 and 3 in our proof of Lemma 1. 
	This is insufficient, and we also need to handle the asymptotic bias terms since our uniform Bahadur representation is only an approximation.  
	Note that much of our proof for Lemma 1 is about establishing the uniform validity of such an approximation. 
\end{remark}

\subsection{Auxiliary Lemmas for the General Result}\label{sec:prel_lemmas}

Since we are working on two probability spaces, $(\Omega^x,\mathcal{F}^x,\mathds{P}^x)$ and $(\Omega^{\xi},\mathcal{F}^{\xi},\mathds{P}^{\xi})$, we use the following notations to clarify the sense of various modes of convergence and expectations.
We let $ \underset{\bullet}{\overset{p}{\to}}$ denote the convergence in probability with respect to the probability measure $\mathds{P}^\bullet$, let $E_{\xi|x}$ denote the conditional expectation with respect to the product probability measure $\mathds{P}^x \times \mathds{P}^\xi$ given the events in $\mathcal{F}^x$, and let $E_{x}$ denote the expectation with respect to the probability measure $\mathds{P}^x$.
Following Section 1.13 of van der Vaart and Wellner (1996), we define the conditional weak convergence in probability, or convergence of the conditional limit laws of bootstraps, denoted by $X_n\underset{\xi}{\overset{p}{\leadsto}}X$, by $\sup_{h\in BL_1}|E_{\xi|x}h(X_n)-Eh(X)|\underset{x}{\overset{p}{\to}}0$, where $BL_1$ is the set of functions with Lipschitz constant and supremum norm bounded by 1.
We state and prove the following lemma, which can be seen as a conditional weak convergence analogy of Theorem 18.10 (iv) of van der Vaart (1998).

The following lemma is used for the purpose of bounding estimation errors in the EMP to approximate the MP.

\begin{lemma}\label{lemma:cond_weak_conv_and in prob}
Let $(\Omega^x \times \Omega^\xi, \mathcal{F}^x\otimes \mathcal{F}^\xi,\mathds{P}^{x \times \xi})$ be the product probability space of $(\Omega^x , \mathcal{F}^x,P^x)$ and $( \Omega^\xi,  \mathcal{F}^\xi, P^\xi)$, where $\mathcal{F}^x\otimes \mathcal{F}^\xi$ stands for the product sigma field of $\mathcal{F}^x$ and $\mathcal{F}^\xi$.
For a metric space $(\mathds{T},d)$, consider $X_n$, $Y_n$, $X:\Omega^x \times \Omega^\xi\rightarrow \mathds{T}$, $n=1,2,...$.
If $X_n\underset{\xi}{\overset{p}{\leadsto}}X$ and $d(Y_n,X_n)\underset{x\times \xi}{\overset{p}{\to}}0$, then $Y_n\underset{\xi}{\overset{p}{\leadsto}}X$.
\end{lemma}

\begin{proof}
For each $h\in BL_1$, we can write
\begin{align*}
&|E_{\xi|x}h(Y_n)-Eh(X)|
\le |E_{\xi|x}h(Y_n)-E_{\xi|x}h(X_n)|+|E_{\xi|x}h(X_n)-Eh(X)|.
\end{align*}
For the second term on the right-hand side, $|E_{\xi|x}h(X_n)-Eh(X)| \underset{x}{\overset{p}{\to}} 0$ by the assumption $X_n\underset{\xi}{\overset{p}{\leadsto}}X$ and the definition of $BL_1$.
To analyze the first term on the right-hand side, note that for any $\varepsilon\in(0,1)$, we have
\begin{align*}
&|E_{\xi|x}h(Y_n)-E_{\xi|x}h(X_n)|
\le \varepsilon E_{\xi|x} \mathds{1}\{d(X_n,Y_n)\le \varepsilon\}+2 E_{\xi|x} \mathds{1}\{d(X_n,Y_n)> \varepsilon\}.
\end{align*}
The first part can be set arbitrarily small by letting $\varepsilon \to 0$.
To bound the second part, note first that the assumption of $d(Y_n,X_n)\underset{x\times \xi}{\overset{p}{\to}}0$ yields
$
\lim_{n\to \infty} \mathds{P}^{x \times \xi}(d(X_n,Y_n)> \varepsilon)=\lim_{n\to \infty} E[\mathds{1}\{d(X_n,Y_n)> \varepsilon\}]=0.
$
By the law of iterated expectations and the dominated convergence theorem, we obtain
\begin{align*}
&\lim_{n\to \infty} E[\mathds{1}\{d(X_n,Y_n)> \varepsilon\}]
=\lim_{n\to \infty} E_x[E_{\xi|x}[\mathds{1}\{d(X_n,Y_n)> \varepsilon\}]]
=E_x[\lim_{n\to \infty} E_{\xi|x}[\mathds{1}\{d(X_n,Y_n)> \varepsilon\}]]=0
\end{align*}
In other words, $\lim_{n\to \infty} E_{\xi|x}[\mathds{1}\{d(X_n,Y_n)> \varepsilon\}]=0$ $\mathds{P}^x$-almost surely.
\end{proof}

The following lemma will be used for deriving the weak convergence of Wald type statistics from joint weak convergence results for the numerator and denominator processes.
It can be easily checked by using definition of Hadamard differentiation, and so we omit a proof.

\begin{lemma}\label{lemma:hadamard_fraction}
Let $(A(\cdot),B(\cdot))\in \ell^\infty(\Theta)\times  \ell^\infty(\Theta)$.
If $B(\cdot)>C>0$ on $\Theta$, then $(F,G)\stackrel{\Phi}{\mapsto}{F}/{G}$ is Hadamard differentiable at $(A,B)$ tangentially to $\ell^\infty(\Theta)$ with the Hadamard derivative $\Phi'_{(A,B)}$ given by $\Phi'_{(A,B)}(g,h)={(Bg-Ah)}/{B^2}$.
\end{lemma}

We restate the Functional Central Limit Theorem of Pollard (1990) as the following lemma, which plays a pivotal role in the proof of our main Theorem. To cope with some measurability issues, we present the version with sufficient conditions for measurability by Kosorok (Lemma 1; 2003). See also Theorem 10.6 of Pollard (1990).
\begin{lemma}[Pollard (1990); Kosorok (2003)]\label{lemma:FCLT}
Denote outer expectation, as defined in Section 1.2 of van der Vaart and Wellner (1996), by $E^*$. Let a triangular array of almost measurable Suslin (AMS) stochastic processes $\{f_{ni}(t):i=1,...n,t\in T\}$ be row independent, and define $\nu_n(t)=\sum_{i=1}^{n}[f_{ni}(t)-Ef_{ni}(\cdot,t)]$. Define $\rho_n(s,t)=(\sum_{i=1}^{n} [f_{ni}(s)-f_{ni}(t)]^2)^{1/2}$. Suppose that the following conditions are satisfied.
\begin{enumerate}
  \item the $\{f_{ni}\}$ are manageable, with envelope $\{F_{ni}\}$ which are also independent within rows;
  \item $H(s,t)=\lim_{n\to \infty}E\nu_n(s) \nu_n(t)$ exists for every $s,t\in T$;
  \item $\limsup_{n\to \infty} \sum_{i=1}^{n}E^*F^2_{ni}<\infty$;
  \item $\lim_{n\to \infty} \sum_{i=1}^{n}E^*F^2_{ni}\mathds{1}\{F_{ni}>\epsilon\}=0$ for each $\epsilon>0$;
  \item $\rho(s,t)=\lim_{n\to\infty} \rho_n (s,t)$ exists for every $s,t\in T$, and for all deterministic sequences $\{s_n\}$ and $\{t_n\}$ in $T$, if $\rho(s_n,t_n)\to 0$ then $\rho_n(s_n,t_n)\to 0$.
	
\end{enumerate}
Then $T$ is totally bounded under the $\rho$ pseudometric and $X_n$ converges weakly to a tight mean zero Gaussian process $X$ concentrated on $\{z\in \ell^\infty(T):\text{$z$ is uniformly $\rho$-continuous}\}$, with covariance $H(s,t)$.
\end{lemma}
\begin{remark}
The AMS condition is technical and thus we refer the readers to Kosorok (2003). In this paper, we will make use of the following separability as a sufficient condition for AMS (Lemma 2; Kosorok (2003)):

Denote $\mathds{P}^*$ as outer probability, as defined in Section 1.2 of van der Vaart and Wellner (1996). A triangular array of stochastic processes $\{f_{ni}(t):i=1,...n,t\in T\}$ is said to be separable if for every $n\ge 1$, there exists a countable subset $T_n\subset T$ such that
\begin{align*}
\mathds{P}^*\Big(\sup_{t \in T} \inf_{s \in T_n} \sum_{i=1}^{n}(f_{ni}(s)-f_{ni}(t))^2>0\Big)=0
\end{align*}
\end{remark}

Checking the manageability in condition 1 above is usually not straightforward.
In practice, we use VC type as a sufficient condition.
We state Proposition 3.6.12 of Gin\'e and Nickl (2016) as a lemma below, which is used for establishing the VC type of functions we encounter.

\begin{lemma}\label{lemma:VC type}
Let $f$ be a function of bounded $p$-variation, $p\ge 1$.
Then, the collection $\mathscr{F}$ of translations and dilation of $f$,
$
\mathscr{F}=\{x\mapsto f(tx-s):t>0, s\in \mathds{R}\}
$
is of VC type.
\end{lemma}

We also cite some results of Chernozhukov, Chetverikov and Kato (2014) as the following lemma, which shows the stability of VC type classes under element-wise addition and multiplication.

\begin{lemma}\label{lemma:VC type_stability}
Let $\mathscr{F}$ and $\mathscr{G}$ be of VC type with envelopes $F$ and $G$ respectively. Then the collection of element-wise sums $\mathscr{F}+\mathscr{G}$ and the collection of element-wise products $\mathscr{F}\mathscr{G}$ are of VC type with envelope $F+G$ and $FG$, respectively.
\end{lemma}

The first one is a special case of Lemma A.6 of Chernozhukov, Chetverikov and Kato (2014). The second one is proven in Corollary A.1 of Chernozhukov, Chetverikov and Kato (2014).

\subsection{Proof of Theorem \ref{theorem:weak_conv} (The Main Result)}\label{sec:theorem:weak_conv}
\begin{proof}
\qquad\\
\textbf{Part (i)}  For $(\theta,k)\in \mathds{T}$, we define
\begin{align*}
f_{ni}(\theta,k)=&\frac{e'_v(\Gamma^+_p)^{-1} r_p(\frac{X_i}{h_{k,n}(\theta_k)})}{\sqrt{n h_{k,n}(\theta_k)}f_X(0)}\Epsilon_k(Y_i,D_i,X_i,\theta)K(\frac{X_i}{h_{k,n}(\theta_k)})\delta^+_i,\\
=& \frac{a_0+a_1(\frac{X_i}{c_k(\theta_k)h_n})+...+a_{p}(\frac{X_i}{c_k(\theta_k)h_n})^p }{\sqrt{ n c_k(\theta_k)h_{n}}f_X(0)}\Epsilon_k(Y_i,D_i,X_i,\theta)K(\frac{X_i}{c_k(\theta_k)h_n})\delta^+_i \qquad\text{and}\\
\nu^+_n(\theta,k)=&\sum_{i=1}^{n}[f_{ni}(\theta,k)-Ef_{ni}(\theta,k)].
\end{align*}
By Assumption \ref{a:BR} (i)(a), the triangular array $\{f_{ni}(\theta,k)\}$ is row independent.
The separability follows from the same argument as in the proof of Theorem 4 in Kosorok (2003) and the left or right continuity (in both $\theta_1$ and $\theta_2$) of the process $f_{ni}(\theta,k)$, which followings from Assumption \ref{a:BR} (ii)(d),(iii) and (iv).
We claim that it satisfies the conditions required by Lemma \ref{lemma:FCLT}.


For condition 1, we note that $\Epsilon_{k}(Y_i,X_i,\cdot)$ is a VC type (Euclidean) class with envelope $2F_\Epsilon$ by Assumption \ref{a:BR}(ii)(a) and Lemma \ref{lemma:VC type_stability}.
Notice that for a fixed $n$, denote $\delta^+_x=\mathds{1}\{x>0\}$.
Both
\begin{align*}
&\Big\{ x \mapsto \frac{(a_0+a_1(\frac{x}{c_k(\theta_k)h_n})+...+a_{p}(\frac{x}{c_k(\theta_k)h_n})^p)\mathds{1}\{|x|\le \overline{c}h_n\}\delta^+_x }{\sqrt{ n c_k(\theta_k)h_{n}}f_X(0)}:(\theta,k)\in \mathds{T} \Big\} \qquad\text{and}\\
&\Big\{x \mapsto K(\frac{x}{c_k(\theta_k)h_n}):(\theta,k)\in \mathds{T}\Big\}
\end{align*}
are of VC type with envelopes $\frac{C_1}{\sqrt{nh_n}}\mathds{1}\{|x|\le \overline{c}h_n\} $ and $\mathds{1}\{|x|\le \overline{c}h_n\}\norm{K}_\infty$, respectively, under Assumptions \ref{a:BR}(i),(iii) and (iv) and Lemma \ref{lemma:VC type}.
By Lemma \ref{lemma:VC type_stability}, their product is a VC type class with envelope
\begin{equation*}
F_{ni}(y,d,x)=\frac{C_3}{\sqrt{nh_n}}F_\Epsilon(y,d,x)\mathds{1}\{C_2\frac{x}{h_n}\in [-1,1]\}.
\end{equation*}
Applying Lemma 9.14 (iii) and Theorem 9.15 of Kosorok (2008), we obtain that $\{f_{ni}\}$ is a bounded uniform entropy integral class with row independent envelopes $F_{ni}$. Theorem 1 of Andrews (1994) then implies that $\{f_{ni}\}$ is manageable with respect to the envelope $\{F_{ni}\}$, and therefore condition 1 is satisfied.


To check condition 2, notice that
$$E[\nu^+_n(\theta,k)\nu^+_n(\vartheta,l)]=\sum_{i=1}^{n}Ef_{ni}(\theta,k)f_{ni}(\vartheta,l)-(\sum_{i=1}^{n}Ef_{ni}(\theta,k))(\sum_{i=1}^{n}Ef_{ni}(\vartheta,l)).$$
Under Assumptions \ref{a:BR}(i)(b),(ii)(c),(iii),(iv)(a) we can write
\begin{align*}
  &\sum_{i=1}^{n} Ef_{ni}(\theta,k)f_{ni}(\vartheta,l) \\
  =& E[\frac{e'_v(\Gamma^+_p)^{-1} r_p(x/c_{k}(\theta_{k})h_n) r'_p(x/c_{l}(\vartheta_{l})h_n) (\Gamma^+_p)^{-1} e_v}{ \sqrt{c_{k}(\theta_k)c_{l}(\vartheta_{l})} h_n f^2_X (0)}\sigma_{k l}(\theta,\vartheta|X_i) K(\frac{X_i}{c_{k}(\theta_k)h_n})K(\frac{X_i}{c_{l}(\vartheta_{l})h_n})\delta^+_i]\\
  =&\int_{\mathds{R}^+}\frac{e'_v (\Gamma^+_p)^{-1} r_p(u/c_{k}(\theta_1)) r'_p(u/c_{l}(\vartheta_{l})) (\Gamma^+_p)^{-1} e_v}{\sqrt{c_{k}(\theta_k)c_{l}(\vartheta_{l})}f^2_X(0)}\sigma_{k l}(\theta,\vartheta|uh_n)
  K(\frac{u}{c_{k}(\theta_k)}) K(\frac{u}{c_{l}(\vartheta_{l})})f_X(u h_n)du \\
  =& \frac{\sigma_{k l}(\theta,\vartheta|0^+)e'_v(\Gamma^+_p)^{-1} \Psi^+_p((\theta,k),(\vartheta,l)) (\Gamma^+_p)^{-1} e_v}{\sqrt{c_{k}(\theta_k)c_{l}(\vartheta_{l})}f_X(0)} + O(h_n).
\end{align*}
$\Psi^+_p((\theta,k),(\vartheta,l)) $ exists under Assumptions \ref{a:BR}(iii) and (iv)(a).
All entries in the matrix part are bounded under Assumption \ref{a:BR}(iii),(iv)(a)(c).
In the last line, $n$ enters only through $O(h_n)$.
Therefore, by Assumption \ref{a:BR}(iii), the limit exists and is finite.
Thus, $\lim_{n\to \infty}\sum_{i=1}^{n} Ef_{ni}(\theta_1,l_1)f_{ni}(\theta_2,l_2)$ exists.
Since $Ef_{ni}(\theta_1,l_1)=0$ implies $\lim_{n\to \infty}(\sum_{i=1}^{n}Ef_{ni}(\theta,k))(\sum_{i=1}^{n}Ef_{ni}(\vartheta,l))=0$, and condition 2 is satisfied.


Under Assumption \ref{a:BR} (i)(a), (ii)(a), (iii), and (iv)(a), it is clear that
$$
\sum_{i=1}^{n}E^* F^2_{ni}=\sum_{i=1}^{n}E F^2_{ni}\lesssim\int_{\mathscr{Y}\times\mathscr{D}\times \mathscr{X}}F^2_\Epsilon(y,d,uh_n) \mathds{1}\{C_2 u \in [-1,1]\}d\mathds{P}^x(y,d,uh_n) +o(h_n)< \infty
$$
as $n \to \infty$. This shows condition 3.


To show condition 4, note that for any $\epsilon>0$
 \begin{align*}
 & \lim_{n\to \infty}\sum_{i=1}^{n}E^* F^2_{ni}\mathds{1}\{F_{ni}>\epsilon\}\\
=&\lim_{n\to \infty}\sum_{i=1}^{n}E F^2_{ni}\mathds{1}\{F_{ni}>\epsilon\}\\
 \lesssim&\lim_{n\to \infty}
 \int_{\mathscr{Y}\times \mathscr{D}\times \mathscr{X}}F^2_\Epsilon(y,d,uh_n)\mathds{1}\{\frac{C_3 }{\sqrt{nh_{n}}}F_\Epsilon(y,d,uh_n)\mathds{1}\{C_2 u\in[-1,1]\} > \epsilon \}d\mathds{P}^x(y,d,uh_n)\\
=&
 \int_{\mathscr{Y}\times \mathscr{D}\times \mathscr{X}}F^2_\Epsilon(y,d,uh_n)\lim_{n\to \infty} \mathds{1}\{\frac{C_3 }{\sqrt{nh_{n}}}F_\Epsilon(y,d,uh_n)\mathds{1}\{C_2 u\in[-1,1]\} > \epsilon \}d\mathds{P}^x(y,d,uh_n)\\
=& 0
 \end{align*}
by the dominated convergence theorem under Assumption \ref{a:BR}(ii)(a), (iii).


To show condition 5, note that we can write
\begin{align*}
\rho^2_n((\theta,k),(\vartheta,l))
&=\sum_{i=1}^{n}E[f_{ni}(\theta,k)-f_{ni}(\vartheta,l)]^2\\
&=nE[f^2_{ni}(\theta,k)+f^2_{ni}(\vartheta,l)-2f_{ni}(\theta,k)f_{ni}(\vartheta,l)].
\end{align*}
From our calculations on the way to show condition 2, we know that each term exists on the right-hand side.
Since $n$ enters the expression only through the $O(h_n)$ part, for all deterministic sequences $\{s_n\}$ and $\{t_n\}$ in $\mathds{T}$, $\rho^2(s_n,t_n)\to 0$ implies $\rho^2_n(s_n,t_n)\to 0.$


Now, applying Lemma \ref{lemma:FCLT}, we have $\nu^+_{n}(\cdot)$ converging weakly to a tight mean-zero Gaussian process $\mathds{G}_{H^+}(\cdot)$ with covariance function
\begin{align*}
H^+((\theta,k),(\vartheta,l))= \frac{\sigma(\theta,\vartheta|0)e'_v(\Gamma^+_p)^{-1}  \Psi^+((\theta,k),(\vartheta,l)) (\Gamma^+_p)^{-1} e_v}{\sqrt{c_{k}(\theta_k)c_{l}(\vartheta_{l})}f_X(0)}.
\end{align*}
Slutsky's Theorem and Assumption \ref{a:BR}(iv) then give
\begin{align*}
\sqrt{nh^{1+2v}_n}
\begin{bmatrix}
  \hat{\mu}^{(v)}_{1,p}(0^+,\cdot)-\mu^{(v)}_{1}(0^+,\cdot)\\
  \hat{\mu}^{(v)}_{2,p} (0^+,\cdot)-\mu^{(v)}_{2} (0^+,\cdot)
\end{bmatrix}
\leadsto
\begin{bmatrix}
  \mathds{G}_{H^+}(\cdot,1)/\sqrt{c^{1+2v}_1(\cdot)} \\
  \mathds{G}_{H^+}(\cdot,2)/\sqrt{c^{1+2v}_2(\cdot)}
\end{bmatrix}
\end{align*}
Applying the functional delta method under Assumption \ref{a:cond_weak_conv}(i), we then have
\begin{align*}
\sqrt{nh^{1+2v}_n}
\begin{bmatrix}
  \phi(\hat{\mu}^{(v)}_{1,p}(0^+,\cdot))(\cdot)-\phi(\mu^{(v)}_{1}(0^+,\cdot))(\cdot)\\
  \psi(\hat{\mu}^{(v)}_{2,p}(0^+,\cdot))(\cdot)-\psi(\mu^{(v)}_{2}(0^+,\cdot))(\cdot)
\end{bmatrix}
\leadsto
\begin{bmatrix}
  \phi'_{\mu^{(v)}_{1}(0^+,\cdot)}\Big(\mathds{G}_{H^+}( \cdot,1)/\sqrt{c^{1+2v}_1(\cdot)} \Big)(\cdot)\\
  \psi'_{\mu^{(v)}_{2}(0^+,\cdot)}\Big(\mathds{G}_{H^+}( \cdot,2)/\sqrt{c^{1+2v}_2(\cdot)} \Big)(\cdot)
\end{bmatrix}
\end{align*}


All arguments above can be replicated for the left limit objects, and thus by Assumption \ref{a:BR}(i)(a), we obtain
\begin{align*}
&\sqrt{nh^{1+2v}_n}
\begin{bmatrix}
  \Big(\phi(\hat{\mu}^{(v)}_{1,p}(0^+,\cdot))-\phi(\hat{\mu}^{(v)}_{1,p}(0^-,\cdot))\Big)(\cdot)
  -\Big((\phi(\mu^{(v)}_{1}(0^+,\cdot))-\phi(\mu^{(v)}_{1}(0^-,\cdot))\Big)(\cdot)\\
  \Big(\psi(\hat{\mu}^{(v)}_{2,p}(0^+,\cdot))-\psi(\hat{\mu}^{(v)}_{2,p}(0^-,\cdot))\Big)(\cdot)
  -\Big((\psi(\mu^{(v)}_{2}(0^+,\cdot))-\psi(\mu^{(v)}_{2}(0^-,\cdot))\Big)(\cdot)\\
\end{bmatrix}\\
\leadsto&
\begin{bmatrix}
   \phi'_{\mu^{(v)}_{1}(0^+,\cdot)}\Big(\mathds{G}_{H^+}( \cdot,1)/\sqrt{c^{1+2v}_1(\cdot}) \Big)(\cdot)- \phi'_{\mu^{(v)}_{1}(0^-,\cdot)}\Big(\mathds{G}_{H^-}( \cdot,1)/\sqrt{c^{1+2v}_1(\cdot}) \Big)(\cdot)\\
  \psi'_{\mu^{(v)}_{2}(0^+,\cdot)}\Big(\mathds{G}_{H^+}( \cdot,2)/\sqrt{c^{1+2v}_2(\cdot}) \Big)(\cdot) - \psi'_{\mu^{(v)}_{2}(0^-,\cdot)}\Big(\mathds{G}_{H^-}( \cdot,2)/\sqrt{c^{1+2v}_2(\cdot}) \Big)(\cdot)
\end{bmatrix}
=
\begin{bmatrix}
  \mathds{G}'(\cdot,1) \\
  \mathds{G}'(\cdot,2)
\end{bmatrix}.
\end{align*}
Finally, by another application of the functional delta method, the chain rule for the functional delta method (Lemma 3.9.3 of van der Vaart and Wellner(1996)), and Lemma \ref{lemma:hadamard_fraction} under Assumption \ref{a:cond_weak_conv}(i) and (ii), we obtain
\begin{align*}
&\sqrt{nh^{1+2v}_n}[\hat{\tau}(\cdot)-\tau(\cdot)]=
\\
&\sqrt{nh^{1+2v}_n}[\Upsilon\Big(\frac{\phi(\hat{\mu}^{(v)}_{1,p}(0^+,\cdot))-\phi(\hat{\mu}^{(v)}_{1,p}(0^-,\cdot))}{\psi(\hat{\mu}^{(v)}_{2,p}(0^+,\cdot))-\psi(\hat{\mu}^{(v)}_{2,p}(0^-,\cdot))}\Big)(\cdot)
-
\Upsilon\Big(\frac{\phi(\mu^{(v)}_{1}(0^+,\cdot))-\phi(\mu^{(v)}_{1}(0^-,\cdot))}
{\psi(\mu^{(v)}_{2}(0^+,\cdot))-\psi(\mu^{(v)}_{2}(0^-,\cdot))}\Big)(\cdot)]\\
\leadsto &\Upsilon'_W\Big( \frac{[\psi(\mu^{(v)}_{2}(0^+,\cdot))-\psi(\mu^{(v)}_{2}(0^-,\cdot))]\mathds{G}'(\cdot,1)-[\phi(\mu^{(v)}_{1}(0^+,\cdot))-\phi(\mu^{(v)}_{1}(0^-,\cdot))]
\mathds{G}'(\cdot,2)}{(\psi(\mu^{(v)}_{2}(0^+,\cdot))-\psi(\mu^{(v)}_{2}(0^-,\cdot)))^2} \Big)(\cdot).
\end{align*}
\end{proof}

\subsection{Proof of Theorem \ref{theorem:cond_weak_conv}}\label{sec:theorem:cond_weak_conv}
\begin{proof}
\qquad\\
\textbf{Part (ii)} We introduce the following notations.
\begin{align*}
\nu^+_{\xi,n} (\theta,k)=&\sum_{i=1}^{n}\xi_i\frac{e'_v(\Gamma^+_p)^{-1}r_p(\frac{X_i}{h_{k,n}(\theta_k)})}{\sqrt{n c_k(\theta_k)h_n}f_X(0)}\Epsilon_{k}(Y_i,D_i,X_i,\theta)K(\frac{X_i}{c_k(\theta_k)h_n})\delta^+_i\\
 \hat{\nu}^+_{\xi,n} (\theta,k)=&\sum_{i=1}^{n}\xi_i\frac{e'_v(\Gamma^+_p)^{-1}r_p(\frac{X_i}{h_{k,n}(\theta_k)})}{\sqrt{n c_k(\theta_k)h_n}\hat{f}_X(0)}\hat{\Epsilon}_{k}(Y_i,D_i,X_i,\theta)K(\frac{X_i}{c_k(\theta_k)h_n})\delta^+_i
\end{align*}
Applying Theorem 2 of Kosorok (2003) (which is also the same as Theorem 11.19 of Kosorok (2008)), we have $\nu^+_{\xi,n} \underset{\xi}{\overset{p}{\leadsto}} \mathds G_{H+}$.
In order to apply Lemma \ref{lemma:cond_weak_conv_and in prob}, we need to show
\begin{align*}
\sup_{(\theta,k)\in \mathds{T}}|\nu^+_{\xi,n}(\theta,k)- \hat{\nu}^+_{\xi,n} (\theta,k)| \underset{x \times \xi}{\overset{p}{\to}} 0.
\end{align*}

We will focus on the case of $k=1$, and same argument applies to the case of $k=2$. 
Note that under Assumption \ref{a:first_stage}, $|\hat f_X(0)-f_X(0)|=o^{x \times \xi}_p(1)$. Thus under Assumption \ref{a:BR}(i)(b),
\begin{align*}
&\nu^+_{\xi,n}(\theta,1)- \hat{\nu}^+_{\xi,n} (\theta,1)\\
=&\frac{1}{f_X(0)\hat f_X(0)}\sum_{i=1}^{n}\xi_i\frac{e'_v(\Gamma^+_p)^{-1}r_p(\frac{X_i}{h_{1,n}(\theta_1)})}{\sqrt{n c_1(\theta_1)h_n}}K(\frac{X_i}{c_1(\theta_1)h_n})\delta^+_i[\hat{\Epsilon}_{1}(Y_i,D_i,X_i,\theta)f_X(0)-\Epsilon_{1}(Y_i,D_i,X_i,\theta)\hat{f}_X(0)]\\
=&\frac{1}{f^2_X(0)+o^{x \times \xi}_p(1)}\sum_{i=1}^{n}\xi_i\frac{e'_v(\Gamma^+_p)^{-1}r_p(\frac{X_i}{h_{1,n}(\theta_1)})}{\sqrt{n c_1(\theta_1)h_n} }K(\frac{X_i}{c_1(\theta_1)h_n})\delta^+_i[\hat{\Epsilon}_{1}(Y_i,D_i,X_i,\theta)f_X(0)-\Epsilon_{1}(Y_i,D_i,X_i,\theta)\hat{f}_X(0)]\\
=&\frac{1}{f^2_X(0)+o^{x \times \xi}_p(1)}\sum_{i=1}^{n} Z_i(\theta_1)[\hat{\Epsilon}_{1}(Y_i,D_i,X_i,\theta)f_X(0)-\Epsilon_{1}(Y_i,D_i,X_i,\theta)\hat{f}_X(0)]\\
=&\frac{1}{f^2_X(0)+o^{x \times \xi}_p(1)}\sum_{i=1}^{n} Z_i(\theta_1)[\hat{\Epsilon}_{1}(Y_i,D_i,X_i,\theta)f_X(0)-\Epsilon_{1}(Y_i,D_i,X_i,\theta)f_X(0)\\
&+\Epsilon_{1}(Y_i,D_i,X_i,\theta)f_X(0)-\Epsilon_{1}(Y_i,D_i,X_i,\theta)\hat{f}_X(0)]\\
=&\frac{1}{f^2_X(0)+o^{x \times \xi}_p(1)}\sum_{i=1}^{n} Z_i(\theta_1)[\hat{\Epsilon}_{1}(Y_i,D_i,X_i,\theta)-\Epsilon_{1}(Y_i,D_i,X_i,\theta)]f_X(0)\\
&+\frac{o^{x \times \xi}_p(1)}{f^2_X(0)+o^{x \times \xi}_p(1)}\sum_{i=1}^{n} Z_i(\theta_1)\Epsilon_{1}(Y_i,D_i,X_i,\theta)\\
=&(1)+(2)
\end{align*}
where $Z_i(\theta_1)=\xi_i\frac{e'_v(\Gamma^+_p)^{-1}r_p(\frac{X_i}{h_{1,n}(\theta_1)})}{\sqrt{n c_1(\theta_1)h_n}}K(\frac{X_i}{c_1(\theta_1)h_n})\delta^+_i$. It can be shown following the same procedures in the proof of Theorem \ref{theorem:weak_conv} that under Assumption \ref{a:BR}, \ref{a:cond_weak_conv}, $\sum_{i=1}^{n}Z_i(\theta_1)\leadsto \mathds{G}_1$ and $\sum_{i=1}^{n}Z_i(\theta_1)\Epsilon_1(Y_i,D_i,X_i,\theta) \leadsto \mathds{G}_2$ for some zero mean Gaussian processes $\mathds{G}_1$, $\mathds{G}_2:\Omega^x \times \Omega^\xi \mapsto \ell^\infty(\Theta)$. By Prohorov's Theorem, the weak convergence implies asymptotic tightness and therefore implies that $\sum_{i=1}^{n}Z_i(\theta_1)=O^{x \times \xi}_p(1)$ uniformly on $\Theta$ and $\sum_{i=1}^{n}Z_i(\theta_1)\Epsilon_1(Y_i,D_i,X_i,\theta)=O^{x \times \xi}_p(1)$ uniformly on $\Theta$. Thus $(2)=\frac{o^{x \times \xi}_p(1)}{f^2_X(0)+o^{x \times \xi}_p(1)}$ uniformly on $\Theta$.
We then control $(1)$. Assumption \ref{a:first_stage} implies
\begin{align*}
&\sum_{i=1}^{n} Z_i(\theta_1)[\hat{\Epsilon}_{1}(Y_i,D_i,X_i,\theta)-\Epsilon_{1}(Y_i,D_i,X_i,\theta)]f_X(0)\\
=&\sum_{i=1}^{n} Z_i(\theta_1)[o^{x \times \xi}_p(1)]f_X(0)\\
=&f_X(0)[O^{x \times \xi}_p(1)]\sum_{i=1}^{n} Z_i(\theta_1)
\end{align*}
uniformly on $\Theta$.
Therefore, we have
\begin{align*}
\sup_{(\theta,k)\in \mathds{T}}|\nu^+_{\xi,n}(\theta,k)- \hat{\nu}^+_{\xi,n} (\theta,k)| \underset{x \times \xi}{\overset{p}{\to}} 0,
\end{align*}
And thus we can apply Lemma \ref{lemma:cond_weak_conv_and in prob} to conclude $ \hat{\nu}^+_{\xi,n}\underset{\xi}{\overset{p}{\leadsto}}  \mathds{G}_{H^+}$.
With similar arguments, we can derive that $ \hat{\nu}^-_{\xi,n}\underset{\xi}{\overset{p}{\leadsto}}  \mathds{G}_{H^-}$.

The continuous mapping theorem for bootstrap (Kosorok, 2008; Proposition 10.7) and the continuity of the Hadamard derivatives imply
\begin{align*}
\begin{bmatrix}
  \widehat{\mathds{X}}'_n(\cdot,1) \\
  \widehat{\mathds{X}}'_n(\cdot,2)
\end{bmatrix}=\begin{bmatrix}
   \phi'_{\mu^{(v)}_{1}(0^+,\cdot)}\Big( \hat{\nu}^+_{\xi,n}( \cdot,1)/\sqrt{c^{1+2v}_1(\cdot}) \Big)( \cdot) - \phi'_{\mu^{(v)}_{1}(0^-,\cdot)}\Big( \hat{\nu}^+_{\xi,n}( \cdot,1)/\sqrt{c^{1+2v}_1(\cdot}) \Big)( \cdot)\\
  \psi'_{\mu^{(v)}_{2}(0^+,\cdot)}\Big( \hat{\nu}^+_{\xi,n}( \cdot,2)/\sqrt{c^{1+2v}_2(\cdot}) \Big)( \cdot) - \psi'_{\mu^{(v)}_{2}(0^-,\cdot)}\Big( \hat{\nu}^+_{\xi,n}( \cdot,2)/\sqrt{c^{1+2v}_2(\cdot}) \Big)( \cdot)
\end{bmatrix}\underset{\xi}{\overset{p}{\leadsto}} \begin{bmatrix}
  \mathds{G}'(\cdot,1) \\
  \mathds{G}'(\cdot,2)
\end{bmatrix}
\end{align*}
Recursively applying Functional Delta for Bootstrap (Theorem 2.9 of Kosorok (2008)) then gives
\begin{align*}
&\Upsilon'_W\Big( \frac{[\psi(\mu^{(v)}_{2}(0^+,\cdot))-\psi(\mu^{(v)}_{2}(0^-,\cdot))]\widehat{\mathds{X}}'_n(\cdot,1)-[\phi(\mu^{(v)}_{1}(0^+,\cdot))-\phi(\mu^{(v)}_{1}(0^-,\cdot))]
\widehat{\mathds{X}}'_n(\cdot,2)}{[\psi(\mu^{(v)}_{2}(0^+,\cdot))-\psi(\mu^{(v)}_{2}(0^+,\cdot)]^2} \Big)\\
\underset{\xi}{\overset{p}{\leadsto}}
 &\Upsilon'_W\Big( \frac{[\psi(\mu^{(v)}_{2}(0^+,\cdot))-\psi(\mu^{(v)}_{2}(0^+,\cdot))]\mathds{G}'(\cdot,1)-[\phi(\mu^{(v)}_{1}(0^+,\cdot))-\phi(\mu^{(v)}_{1}(0^-,\cdot))]
\mathds{G}'(\cdot,2)}{[\psi(\mu^{(v)}_{2}(0^+,\cdot))-\psi(\mu^{(v)}_{2}(0^-,\cdot)]^2} \Big)
\end{align*}
This completes the proof.
\end{proof}

\subsection{Uniform Consistency}

Theorem \ref{theorem:weak_conv} implies the following useful asymptotic identities.
They state the uniform consistency of the numerator and the denominator of Wald ratios.
Because these implications themselves prove useful when we apply the main theorem to the specific example, we state them as corollaries below for convenience of reference.

\begin{corollary}\label{corollary:mu_unif_consist}
Under Assumptions \ref{a:BR} and \ref{a:cond_weak_conv}, $\hat{\mu}^{(v)}_{l,p}(0^\pm,\cdot)-\mu^{(v)}_{k}(0^\pm,\cdot)\underset{x}{\overset{p}{\to}}0$ uniformly.
\end{corollary}

\begin{corollary}\label{corollary:operator_unif_consist}
Under Assumption \ref{a:BR} and \ref{a:cond_weak_conv}, $\phi(\hat{\mu}^{(v)}_{l,p}(0^\pm,\cdot))-\phi(\mu^{(v)}_{k}(0^\pm,\cdot))\underset{x}{\overset{p}{\to}}0$ uniformly.
\end{corollary}

Assumption \ref{a:cond_weak_conv}(iv) further implies that the mode $\underset{x}{\overset{p}{\to}}$ of convergence in the above corollaries can be replaced by the mode $\underset{x \times \xi}{\overset{p}{\to}}$ of convergence.

\subsection{First Stage Estimators}\label{sec:first_stage}
To estimate MP, we replace $\mu_{k}(x,\theta)\mathds{1}\{|x/h_{k,n}(\theta_k)|\le1\}$ by its estimate $\tilde{\mu}_{k,p}(x,\theta)\mathds{1}\{|x/h_{k,n}(\theta_k)|\le 1\}$, which is uniformly consistent across $(x,\theta)$.
Lemma \ref{lemma:epsilon_unif_const} below proposes such a uniformly consistent estimator without requiring to solve an additional optimization problem; by a mean-value expansion and uniform boundedness of $\mu^{(p+1)}_k$, we can reuse the first stage local polynomial estimates of $\hat{\mu}^{(v)}_{l,p}(0^\pm,\theta)$ for all $v\le p$.
This auxiliary result will prove useful when we apply Theorems \ref{theorem:weak_conv} and \ref{theorem:cond_weak_conv} to specific examples. In fact, we will prove a more general result that allows us to use any first $r$ terms of our $p$-th order polynomial estimator for a $t$ such that $0 \le t \le p$.

\begin{lemma}\label{lemma:epsilon_unif_const}
Fix an integer $t$ such that $0 \le t \le p$. Suppose that Assumptions \ref{a:BR} and \ref{a:cond_weak_conv} are satisfied.
Let $\Epsilon_{1}(y,d,x,\theta)=g_1(y,\theta_1)-\mu_{1}(x,\theta_1)$, $\Epsilon_{2}(y,d,x,\theta)=g_2(d,\theta_2)-\mu_{2}(d,\theta_2)$, $\delta^+_x=\mathds{1}\{x\ge0\}$ and $\delta^-_x=\mathds{1}\{x\le0\}$.
Define
\begin{align*}
\tilde{\mu}_{1,t}(x,\theta_1)&=r_t(x/h_{1,n}(\theta_1))'\hat{\alpha}_{1+,t}(\theta_1)\delta^+_x
+r_t(x/h_{1,n}(\theta_1))'\hat{\alpha}_{1-,t}(\theta_1)\delta^-_x
\quad\text{and}\\
\tilde{\mu}_{2,t}(x,\theta_2)&=r_t(x/h_{2,n}(\theta_2))'\hat{\alpha}_{2+,t}(\theta_2)\delta^+_x
+r_t(x/h_{2,n}(\theta_2))'\hat{\alpha}_{2-,t}(\theta_2)\delta^-_x
\end{align*}
Then we have
\begin{align*}
\hat{\Epsilon}_{1}(y,d,x,\theta)&=[g_1(y,\theta_1)-\tilde{\mu}_{1,t}(x,\theta_1)]\mathds{1}\{|x /h_{1,n}(\theta_1)|\le 1\}
\quad\text{and}\\
\hat{\Epsilon}_{2}(y,d,x,\theta)&=[g_2(d,\theta_2)-\tilde{\mu}_{2,t}(x,\theta_2)]\mathds{1}\{|x/h_{2,n}(\theta_2)|\le 1\}
\end{align*}
are uniformly consistent for $\Epsilon_{1}(y,d,x,\theta)\mathds{1}\{|x/h_{1,n}(\theta_1)|\le 1\}$ and $\Epsilon_{2}(y,d,x,\theta)\mathds{1}\{|x/h_{2,n}(\theta_2)|\le 1\}$ on $[\underline{x},\overline{x}]\times \mathscr{Y} \times \mathscr{D} $, respectively.
\end{lemma}

\begin{proof}
We will show the $1+$ part only, since the other parts can be shown similarly.
Recall that $\Epsilon_{1}(y,d,x,\theta)=g_1(y,\theta_1)-\mu_{1}(x,\theta_1)$.
If $x>0$,
\begin{align*}
& \mu_{1}(x,\theta_1)\mathds{1}\{|x/h_{1,n}(\theta_1)|\le 1\}=
\\
& \Big(\mu_{1} (0^+,\theta_1)+\mu^{(1)}_{1} (0^+,\theta_1)x+...+\mu^{(t)}_{1} (0^+,\theta_1)\frac{x^t}{t!}+\mu^{(t+1)}_{1}(x^*_{ni},\theta_1)\frac{x^{(t+1)}}{(t+1)!}\Big)\mathds{1}\{|x/c_1(\theta)h_n|\le 1\}.
\end{align*}
By Corollary \ref{corollary:mu_unif_consist},  $\hat{\mu}^{(v)}_{1,t}(0^\pm,\theta_1)$ is uniformly consistent for $\mu^{(v)}_{1} (0^\pm,\theta_1)$, $v=0,1,...,t$.
Thus,
\begin{align*}
&[\hat{\Epsilon}_{1}(y,d,x,\theta)-\Epsilon_{1}(y,d,x,\theta)]\mathds{1}\{|x/h_{1,n}(\theta)|\le 1\}\\
=&\Big(\hat{\mu}_{1,t} (0^+,\theta_1)+\hat{\mu}^{(1)}_{1,t} (0^+,\theta_1)x+...+\hat{\mu}^{(t)}_{1,t} (0^+,\theta_1)\frac{x^t}{t!}\Big)\mathds{1}\{|x/h_{1,n}(\theta)|\le 1\}-\\
&\Big(\mu_{1} (0^+,\theta_1)+\mu^{(1)}_{1} (0^+,\theta_1)x+...+\mu^{(t)}_{1} (0^+,\theta_1)\frac{x^t}{t!}+\mu^{(t+1)}_{1}(x^*,\theta_1)\frac{x^{(t+1)}}{(t+1)!}\Big)\mathds{1}\{|x/h_{1,n}(\theta_1)|\le 1\}\\
=&o^x_p(1)-\mu^{(t+1)}_{1}(x^*,\theta_1)\frac{x^{(t+1)}}{(t+1)!}\mathds{1}\{|x|\le h_{1,n}(\theta)\}=o^x_p(1)+O(h_n),
\end{align*}
where the last equality is by Assumption \ref{a:BR}(iii) and by the uniform boundedness of $\mu^{(t+1)}$ under Assumption \ref{a:BR} (ii)(a).
\end{proof}

Finally, we also provide consistent estimators, $\hat f_X(0)$ and $\hat f_{Y^d|C}(\cdot)$.
For $f_X(0)$, we propose to use the kernel density estimator
$$
\hat f_X(0) = \frac{1}{n b_n} \sum_{i=1}^n K\left(\frac{X_i}{b_n}\right)
$$
with $b_n\to 0$ and $nb_n \to \infty$.
For $f_{Y^d|C}(\cdot)$, we propose 
\begin{align*}
\hat f_{Y^1|C}( y)
&= \frac{\hat{f}_{Y^*|XD^*}(y|x^+_0,1)\hat{P}(D^*_i=1|X_i=x^+_0)-
\hat{f}_{Y^*|XD^*}(y|x^-_0,1)\hat{P}(D^*_i=1|X_i=x^-_0)}
{\hat \mu_{2,2}(0^+,1)-\hat \mu_{2,2}(0^-,1)}
\qquad\text{and}
\\
\hat f_{Y^0|C}(y)
&=\frac{\hat{f}_{Y^*|XD^*}(y|x^+_0,0)\hat{P}(D^*_i=0|X_i=x^+_0)-
\hat{f}_{Y^*|XD^*}(y|x^-_0,0)\hat{P}(D^*_i=0|X_i=x^-_0)}
{-(\hat \mu_{2,2}(0^+,0)-\hat \mu_{2,2}(0^-,0))},
\end{align*}
where 
\begin{align*}
\hat{f}_{Y^*|XD^*}(y|x^\pm_0,1)&=\frac{\frac{1}{na_n^2}\sum_{i=1}^{n} K(\frac{X_i}{a_n})K(\frac{Y^*_i-y}{a_n}) D^*_i \delta^\pm_i }{\frac{1}{na_n}\sum_{i=1}^{n}K(\frac{X_i}{a_n})  D^*_i \delta^\pm_i},
\\
\hat{f}_{Y^*|XD^*}(y|x^\pm_0,0)&=\frac{\frac{1}{na_n^2}\sum_{i=1}^{n} K(\frac{X_i}{a_n})K(\frac{Y^*_i-y}{a_n}) (1-D^*_i) \delta^\pm_i }{\frac{1}{na_n}\sum_{i=1}^{n}K(\frac{X_i}{a_n}) (1- D^*_i) \delta^\pm_i},
\\
\hat{P}(D^*_i=1|X_i=x^\pm_0)&=\frac{ \sum_{i=1}^{n}K(\frac{X_i}{c_n})D^*_i \delta^\pm_i}{\sum_{i=1}^{n}K(\frac{X_i}{c_n})\delta^\pm_i}, 
\qquad\text{and}
\\
\hat{P}(D^*_i=0|X_i=x^\pm_0)&=\frac{ \sum_{i=1}^{n}K(\frac{X_i}{c_n})(1-D^*_i) \delta^\pm_i}{\sum_{i=1}^{n}K(\frac{X_i}{c_n})\delta^\pm_i}.
\end{align*}
\begin{lemma}\label{lemma:density_estimators}
Suppose Assumptions S, K, M and FQRD (i)-(iv) hold. 
Assume that the conditional CDF $f_{Y^*|XD^*}(\cdot|0^+,d)$ is continuously differentiable and its derivative is uniformly bounded on $ \mathcal Y_1$ for $d\in\{0,1\}$. 
If $a_n$ satisfies $a_n\to 0$, $na_n^2/|\log a_n| \to\infty$, $a_n\le c a_{2n}$ for some $c>0$ and $|\log a_n|/\log\log n\to \infty$ and $c_n$ satisfies $c_n\to 0$ and $nc_n \to \infty$, then $\sup_{y\in \mathcal{Y}_1}|\hat{f}_{Y^1|C}(y)-f_{Y^1|C}(y)|=o^x_p(1)$ and $\sup_{y\in \mathcal{Y}_1}|\hat{f}_{Y^0|C}(y)-f_{Y^0|C}(y)|=o^x_p(1)$.
\begin{proof}
We will prove this lemma for $f_{Y^1|C}$, and the corresponding results for $f_{Y^0|C}(y)$ follow similarly. 
From (\ref{eq:ffm_cdf}) and Assumption FQRD (i), (vi)
\begin{align*}
f_{Y^1|C}(y)&=\frac{\partial }{\partial y} F_{Y^1|C}(y)\\
&=\frac{\frac{\partial}{\partial y}\lim_{x\downarrow 0}E[ \mathbbm{1}\{Y_i^\ast \leq y\} \cdot \mathbbm{1}\{D^\ast_i=d\}  |  X_i=x] - \frac{\partial}{\partial y}\lim_{x\uparrow 0}E[ \mathbbm{1}\{Y_i^\ast \leq y\} \cdot \mathbbm{1}\{D^\ast_i=d\}  |  X_i=x]}{\lim_{x\downarrow 0}E[ \mathbbm{1}\{D_i^\ast=d\} |  X_i=x] - \lim_{x\uparrow 0}E[ \mathbbm{1}\{D^\ast_i=d\}  |  X_i=x]}.
\end{align*}
Theorem \ref{theorem:weak_conv} then implies that the denominator terms can be consistently estimated by $\hat \mu_{2,2}(0^+,1)-\hat \mu_{2,2}(0^-,1)$ at the rate of $O^x_p(1/\sqrt{nh_n})$ and the limit is bounded away from zero by Assumption FQRD (iii).

As for the numerator terms, the boundedness of the integrand ensures that we can interchange the expectation and limit.
Thus, we can write
\begin{align*}
\frac{\partial}{\partial y}E[ \mathbbm{1}\{Y_i^\ast \leq y\} \cdot \mathbbm{1}\{D^\ast_i=d\}  |  X_i=0^+]
=&\frac{\partial}{\partial y}E[ \mathbbm{1}\{Y_i^\ast \leq y\}   |  X_i=0^+,D^*_i=1]\mathds P^x(D^*_i=1|X_i=0^+) + 0\\
=&\frac{\partial}{\partial y}F_{Y^*|XD^*}(y|0^+,1)\mathds P^x(D^*_i=1|X_i=0^+)\\
=&f_{Y^*|XD^*}(y|0^+,1)\mathds P^x(D^*_i=1|X_i=0^+).
\end{align*}
The standard point-wise convergence result under the stated assumptions for Nadaraya-Watson estimator and Assumption FQRD (i) imply the consistency of $\hat P(D^*_i=1|X_i=0^+)$ for $\mathds P^x(D^*_i=1|X_i=0^+)$. The uniform consistency of $\hat f_{Y^*|XD^*}(\cdot|0^+,1)$ follows from Assumption FQRD (i), (ii), continuous differentiability, uniform boundedness of derivatives of $\hat f_{Y^*|XD^*}(\cdot|0^+,1)$ on $\mathcal Y_1$ and Theorem 2.3 of Gin\'e and Guillou (2002), which is applicable under the above bandwidth assumptions for $a_n$ and Assumption K. 
\end{proof}
\end{lemma}

\subsection{Proof of Corollary \ref{corollary:cluster_robust}}\label{sec:corollary:cluster_robust}
\begin{proof}
We first show that Lemma \ref{lemma:BR} holds under cluster sampling. 
Note that Steps 1 and 2 of the Proof for Lemma \ref{lemma:BR} follow through. 
To see this, note that for the inverse factor in Step 1, the deterministic part is now  
\begin{align*}
&E[\frac{1}{Gh_{G}}\sumg\sumi\delta^+_i K(\frac{X_i}{h_{G}}) r_p(\frac{X_i}{h_{G}}) r'_p(\frac{X_i}{h_{G}})]\\
=&\frac{1}{G}\sumg\sumi E[\frac{1}{h_{G}}\delta^+_i K(\frac{X_i}{h_{G}}) r_p(\frac{X_i}{h_{G}}) r'_p(\frac{X_i}{h_{G}})].
\end{align*}
Assumption \ref{a:cluster_robust}(i), (iii), and (iv) imply that the above becomes
\begin{align*}
&\frac{1}{G}\sumg\sumi f_{X_i}(0)\Big(\int_{\mathbbm R^+} K(u)r_p(u)r'_p(u)+ O(h_{G})\Big)
=\bar f_X(0)\Gamma^+_p + O(h_{G}).
\end{align*}
For the stochastic part, under Assumption \ref{a:cluster_robust}(ii)(a),
\begin{align*}
\frac{1}{G}\sumg\Big(\frac{1}{h_{G}}\sumi\delta^+_i K(\frac{X_i}{h_{G}}) r_p(\frac{X_i}{h_{G}}) r'_p(\frac{X_i}{h_{G}})\Big)-E\Big[\frac{1}{G}\sumg\Big(\frac{1}{h_{G}}\sumi\delta^+_i K(\frac{X_i}{h_{G}}) r_p(\frac{X_i}{h_{G}}) r'_p(\frac{X_i}{h_{G}})\Big)\Big]=o^x_p(1)
\end{align*}
follows from the WLLN for the row-wise independent triangular array since nothing depends on $\theta$.

For Step 2, the deterministic part follows the same argument under Assumption  \ref{a:cluster_robust}(i), (ii)(b), (iii), and (iv)(a), and thus we have
\begin{align*}
&E\Big[\frac{1}{\sqrt{Gh_{G}}}\sumg\sumi \delta^+_i K(\frac{X_i}{h_{G}}) r_p(\frac{X_i}{h_{G}})\mu^{(p+1)}_{1}(x^*_{ni},\theta_1)h^{p+1}_{G}\frac{(\frac{X_i}{h_{G}})^{p+1}}{(p+1)!}\Big]
=O(\sqrt{\frac{h^{2p+1}_G}{G}})
\end{align*}
uniformly in $\theta_1$. For the stochastic part, by the same Lipschitz argument under Assumption \ref{a:cluster_robust} (ii), (iii) and (iv)(a), it is sufficient to consider 
\begin{align*}
\frac{1}{\sqrt{Gh_{G}}}\sumg\sumi \delta^+_i K(\frac{X_i}{h_{G}}) r_p(\frac{X_i}{h_{G}})\mu^{(p+1)}_{1}(0^+,\theta_1)h^{p+1}_{G}\frac{(\frac{X_i}{h_{G}})^{p+1}}{(p+1)!}.
\end{align*}
For each cluster $g$, we relabel the subscript $i$ as $tg$, where $t\in\{1,...,n_g\}$ indicates the index in cluster $g$.  
Write $W_g=(W_{1g},...,W_{\bar N g})$, where $W_{tg}=(Y_{tg},D_{tg},X_{tg})$ if $t\le n_g$ and $W_{tg}=0$ if $t > n_g$. Let
\begin{align*}
&\mathscr{F}^t_{s}=\{ W_g \mapsto \mathds{1}\{X_{tg}> 0\}K(aX_{tg})(aX_{tg})^{s+p+1}\mu^{p+1}_1(0^+,\theta_1) \mathds{1}\{X_{tg}\in [-1,1]\}:a\ge 1/h_0 , \theta_1 \in \Theta_1 \}\\
&\mathscr{F}_{G,s}=\{ W_g \mapsto \sum_{t=1}^{\bar N}\mathds{1}\{X_{tg}> 0\}K(X_{tg}/h_{G})(X_{tg}/h_{G})^{s+p+1}\mu^{p+1}_1(0^+,\theta_1)  :\theta_1 \in  \Theta_1 \}
\end{align*}
for each integer $s$ such that $0\le s \le p$. 
Note that $\mathds 1 \{X_{tg}>0\}=0$ for $t>n_g$. 
Under Assumption \ref{a:cluster_robust} (ii), (iii), and (iv), $\mathscr F_{G,s}\subset \mathscr F_s:=\sum_{t=1}^{\bar N} \mathscr F^t_{s}$ for all $G$. 
$\mathscr F_s$ is a VC type class with an integrable envelope $F_s=\bar N  \overline M \|K\|_\infty$ by Lemma \ref{lemma:VC type_stability}. Thus, applying Lemma 7 of Chiang (2018) in place of Theorem 5.2 of Chernozhukov, Chetverikov and Kato (2014) yields
\begin{align*}
\sup_{f\in \mathscr F_s}\Big|\frac{1}{\sqrt{Gh_G}}\sumg (f(W_g)-Ef(W_g))h_G^{p+1}\Big|=O^x_p(\frac{h_G^{2P+1}}{G}),
\end{align*}
and thus the same conclusion can be made.

For Step 3, under Assumption \ref{a:cluster_robust} (i)(a), (ii)(d) and (iv)(a), its deterministic part now becomes
\begin{align*}
&E[\frac{1}{\sqrt{Gh_{G}}}\sumg\sumi \delta^+_i K(\frac{X_i}{h_{G}}) r_p(\frac{X_i}{h_{G}})\Epsilon_{1}(Y_i,D_i,X_i,\theta)]\\
=&E[\frac{1}{\sqrt{Gh_{G}}}\sumg\sumi\delta^+_i K(\frac{X_i}{h_{G}}) r_p(\frac{X_i}{h_{G}})E[\Epsilon_{1}(Y_i,D_i,X_i,\theta)|(X_j:j\in C_g,i\in C_g)]\,]\\
=&E[\frac{1}{\sqrt{Gh_{G}}}\sumg\sumi \delta^+_i K(\frac{X_i}{h_{G}}) r_p(\frac{X_i}{h_{G}})E[\Epsilon_{1}(Y_i,D_i,X_i,\theta)|X_i]\,]=0.
\end{align*}
The stochastic part follows from Markov's inequality, Assumption \ref{a:cluster_robust} (i), (ii)(c), (iii), (iv) and (v), and the calculation:
\begin{align*}
&\sup_{\theta_1 \in \Theta_1}E[\frac{e'_s}{\sqrt{Gh_{G}}}\sumg\sumi \delta^+_i K(\frac{X_i}{h_{G}}) r_p(\frac{X_i}{h_{G}})\Epsilon_{1}(Y_i,D_i,X_i,\theta)]^2\\
=&\sup_{\theta_1 \in \Theta_1}\sumg \frac{1}{G} E\Big[\frac{e_v'}{h_G}\sumi \delta^+_i K(\frac{X_i}{h_{G}}) r_p(\frac{X_i}{h_{G}})\Epsilon_{1}(Y_i,D_i,X_i,\theta)\Big]^2\\
\le &\overline N^2 \overline M^2 \|K\|_\infty^2=O(1)
\end{align*}
for each $1\le v \le p$.
Combining all three steps, we have
\begin{align*}
&\sqrt{Gh^{1+2v}_{G}}\big( \hat{\mu}^{(v)}_{k,p}(0^\pm,\theta_k)- \mu^{(v)}_{k}(0^\pm,\theta_k)-h^{p+1-v}_{G}\frac{e'_v(\Gamma^{\pm}_p)^{-1}\Lambda^\pm_{p,p+1}}{(p+1)!}\mu^{(p+1)}_{k}(0^\pm,\theta_k) \big)\\
= \ &  v!\sumg\sumi\frac{e'_v(\Gamma^\pm_p)^{-1}\Epsilon_k(Y_i,D_i,X_i,\theta) r_p(\frac{X_i}{h_{G}}) K(\frac{X_i}{h_{G}}) \delta^\pm_i}{\sqrt{nh_{G}}\bar f_X(0)}+o^x_p(1)
\end{align*}
uniformly for all $\theta_k\in \Theta_k$ for each $k \in \{1,2\}$.

Second, we show that Theorem \ref{theorem:weak_conv} holds under cluster sampling. This is straightforward since Lemma \ref{lemma:FCLT} holds for non-identically distributed row-wise independent processes. Specifically, for $(\theta,k)\in \mathds{T}$, we define
\begin{align*}
f_{Gg}(\theta,k)=&\sumi\frac{e'_v(\Gamma^+_p)^{-1} r_p(\frac{X_i}{h_{G}})}{\sqrt{G h_{G}}\bar f_X(0)}\Epsilon_k(Y_i,D_i,X_i,\theta)K(\frac{X_i}{h_{G}})\delta^+_i \qquad\text{and}\\
\nu^+_G(\theta,k)=&\sumg[f_{Gg}(\theta,k)-Ef_{Gg}(\theta,k)].
\end{align*}
Then, using the closure-under-summation property for VC type classes from Lemma \ref{lemma:VC type_stability} and following the same argument as in proof of Theorem \ref{theorem:weak_conv}, Condition 1 required by Lemma \ref{lemma:FCLT} can be established. 
Condition 2 required by Lemma \ref{lemma:FCLT} holds under Assumption \ref{a:cluster_robust} (v). 
Conditions 3, 4 and 5 follow from similar arguments to those in the proof of Theorem \ref{theorem:weak_conv}. 
By an application of Lemma \ref{lemma:FCLT}, $\nu^+$ converges weakly to a tight mean-zero Gaussian process $\mathds G_{\Sigma^+}$ with covariance function $\Sigma^+$ defined in Assumption \ref{a:cluster_robust} (v). 
Finally, by an application of the functional delta method, the chain rule for the functional delta method (Lemma 3.9.3 of van der Vaart and Wellner(1996)), and Lemma \ref{lemma:hadamard_fraction} under Assumption \ref{a:cond_weak_conv} (i) and (ii), we obtain
\begin{align*}
&\sqrt{Gh^{1+2v}_G}[\hat{\tau}(\cdot)-\tau(\cdot)]=
\\
&\sqrt{Gh^{1+2v}_G}[\Upsilon\Big(\frac{\hat{\mu}^{(v)}_{1,p}(0^+,\cdot)-\hat{\mu}^{(v)}_{1,p}(0^-,\cdot)}{\hat{\mu}^{(v)}_{2,p}(0^+,\cdot)-\hat{\mu}^{(v)}_{2,p}(0^-,\cdot)}\Big)(\cdot)
-
\Upsilon\Big(\frac{\mu^{(v)}_{1}(0^+,\cdot)-\mu^{(v)}_{1}(0^-,\cdot)}
{\mu^{(v)}_{2}(0^+,\cdot)-\mu^{(v)}_{2}(0^-,\cdot)}\Big)(\cdot)]\\
\leadsto &\Upsilon'_W\Big( \frac{[\mu^{(v)}_{2}(0^+,\cdot)-\mu^{(v)}_{2}(0^-,\cdot)]\mathds{G}(\cdot,1)-[\mu^{(v)}_{1}(0^+,\cdot)-\mu^{(v)}_{1}(0^-,\cdot)]
\mathds{G}(\cdot,2)}{(\mu^{(v)}_{2}(0^+,\cdot)-\mu^{(v)}_{2}(0^-,\cdot))^2} \Big)(\cdot),
\end{align*}
where $\mathds G:=\mathds G_{\Sigma^+} - \mathds G_{\Sigma^-}$.

Finally, we argue that Theorem \ref{theorem:cond_weak_conv} holds under cluster sampling. 
This is also straightforward since Theorem 2 of Kosorok (2003) only requires row-wise independence as well.
Let us denote
\begin{align*}
\nu^+_{\xi,G} (\theta,k)=&\sumg \xi_g\sumi\frac{e'_v(\Gamma^+_p)^{-1}r_p(\frac{X_i}{h_{G}})}{\sqrt{G h_G}\bar f_X(0)}\Epsilon_{k}(Y_i,D_i,X_i,\theta)K(\frac{X_i}{h_G})\delta^+_i.
\end{align*}
Applying Theorem 2 of Kosorok (2003), we have $\nu^+_{\xi,G} \underset{\xi}{\overset{p}{\leadsto}} \mathds G_{\Sigma^+}$ and similar results hold that $\nu^-_{\xi,G} \underset{\xi}{\overset{p}{\leadsto}} \mathds G_{\Sigma^-}$.
To apply Lemma \ref{lemma:cond_weak_conv_and in prob}, note that it follows from Assumption \ref{a:cluster_robust} (viii) that
\begin{align*}
\sup_{(\theta,k)\in \mathds{T}}|\nu^+_{\xi,n}(\theta,k)- \hat{\nu}^+_{\xi,n} (\theta,k)| \underset{x \times \xi}{\overset{p}{\to}} 0.
\end{align*}
The result then follows from the proof of Theorem \ref{theorem:cond_weak_conv}.
\end{proof}

\subsection{Auxiliary Lemma for the Extended Result with Covariates}\label{sec:prel_lemmas_covariates}

\begin{lemma}\label{lemma:covariates}
Suppose Assumptions S, K, M, FQRD and C hold, then
\begin{align*}
&\sup_{(y,d)\in\mathscr Y_1 \times \{0,1\}}|\check\gamma_{1+}(y,d)-\gamma_{1+}(y,d)| + \sup_{(y,d)\in\mathscr Y_1 \times \{0,1\}}|\check\gamma_{1-}(y,d)-\gamma_{1-}(y,d)|=o_p(1) \text{ and }\\
&\sup_{(y,d)\in\mathscr Y_1 \times \{0,1\}}|\check F_{Y^d|C} (y) -F_{Y^d|C} (y) |=o_p(1).
\end{align*}
\begin{proof}
Section 2.1 and Lemma SA-14 of the Supplementary Appendix of Calonico, Cattaneo, Farrell and Titiunik (2018) imply that for each fixed $(y,d)$, under Assumption FQRD(i)(iii) and Assumption C, one has 
\begin{align*}
\frac{\check\mu_{1}(0^+,y,d) - \check\mu_{1}(0^-,y,d) }{ \check\mu_{2}(0^+,d) -  \check\mu_{2}(0^-,d)  } \overset{p }{\to}  \frac{\mu_{1}(0^+,y,d) -\mu_{1}(0^-,y,d) - [\gamma'_{1+}(y,d) \mu_Z(0^+) -\gamma'_{1-}(y,d) \mu_Z(0^-)   ] }{ \mu_{2+}(d) -  \mu_{2}-(d)  - [\gamma'_{2+}(d) \mu_Z(0^+) -\gamma'_{2-}(d) \mu_Z(0^-)   ]  }.
\end{align*}
Assumptions C (i) and C (ii) further imply that for all $(y,d)\in \mathscr Y_1 \times \{0,1\}$, it holds that $\gamma'_{1+}(y,d) \mu_Z(0^+) -\gamma'_{1-}(y,d) \mu_Z(0^-) =0 $ and $\gamma'_{2+}(d) \mu_Z(0^+) -\gamma'_{2-}(d) \mu_Z(0^-) =0 $. Thus
\begin{align}
\frac{\check\mu_{1}(0^+,y,d) - \check\mu_{1}(0^-,y,d) }{ \check\mu_{2}(0^+,d) -  \check\mu_{2}(0^-,d)  } \overset{p }{\to}  F_{Y^d|C} (y). \label{eq:covariates_model_limit}
\end{align}
This convergence can be made uniformly over $\mathscr Y_1 \times \{0,1\}$. 
To see this, note that Section 8 of Supplementary Appendix of Calonico, Cattaneo, Farrell and Titiunik (2018) implies that
\begin{align*}
\check\gamma_{1+}(y,d)=&\Big[ \textbf{Z}'\textbf{K}_{+} \textbf{Z}/n - \Upsilon_{Z+}'\Gamma^{-1}_{+} \Upsilon_{Z+}    \Big]^{-1} \Big[\textbf{Z}'\textbf{K}_{+} \textbf{Y}(y,d)/n - \Upsilon_{Z+}'\Gamma^{-1}_{+} \Upsilon_{Y+}(y,d) \Big],
\end{align*}
where $\textbf{Z}=[Z_1,...,Z_n]'$, $\textbf{K}_{+}=\text{diag}\{K(X_1/h_n)\delta^+_1,...,K(X_n/h_n)\delta^+_n\}$, $\textbf{Y}(y,d)=[\mathbbm 1\{Y_1\le y, D_1=d\},...,\mathbbm 1\{Y_n\le y, D_n=d\}]'$, $\Upsilon_{Z+}=\frac{1}{n}\sum_{i=1}^n r(X_i/h_n)K(X_i/h_n) \delta^+_i Z_i$, and $\Upsilon_{Y+}(y,d) = \frac{1}{n}\sum_{i=1}^n$ $r(X_i/h_n)K(X_i/h_n) \delta^+_i \mathbbm 1\{Y_i\le y,D_i=d\}$. 
Note also that $\textbf{Z}'\textbf{K}_{+} \textbf{Z}/n $ and $ \Upsilon_{Z+}   $ converge to their respective probability limits following Lemmas SA-2 and SA-3 of Supplementary Appendix of Calonico, Cattaneo, Farrell and Titiunik (2018). 
Convergence of $\Gamma^{-1}_{+}$ follows from proof of Lemma \ref{lemma:BR}, and uniform convergence of $\textbf{Y}(y,d)$ and $\Upsilon_{Y+}(y,d)$ follows the same arguments as those for the convergence of $\Gamma^{-1}_{+}$ in Lemma \ref{a:BR}. 
Therefore, $\check\gamma_{1+}(y,d)$ converges uniformly in probability to its probability limit $\gamma_{1+}(y,d)$. 
This shows that (\ref{eq:covariates_model_limit}) holds uniformly over $\mathscr Y_1 \times \{0,1\}$.
\end{proof}
\end{lemma}

\subsection{Proof of Corollary \ref{corollary:covariates}}\label{sec:corollary:covariates}
\begin{proof}
Let $\check g_1(Y_i,\theta)=\mathbbm 1\{Y_i^*\le y , D_i^*=d\}-Z'_i\check\gamma_{1+}(y,d)$ and $\check g_2(D_i,\theta)=\mathbbm 1\{ D_i^*=d\}-Z'_i\check\gamma_{1+}(d)$. 
Then, Lemma \ref{lemma:covariates} implies that 
 $g_1(Y_i,\theta)=\mathbbm 1\{Y_i^*\le y , D_i^*=d\}-Z'_i\gamma_{1}(y,d)+o_p(1)$ and $g_2(D_i,\theta)=\mathbbm 1\{ D_i^*=d\}-Z'_i\gamma_{1}(d)+o_p(1)$ uniformly over $\mathscr Y_1 \times \{0,1\}$. 
Lemma \ref{lemma:BR} holds for such $\check g_1(Y_i,\theta)$ and $\check g_2(D_i,\theta)$ and using $\check g_1$ and $\check g_2$ in place leads to asymptotically equivalent expressions. 
Theorem \ref{theorem:weak_conv} thus follows. 
Furthermore, the proof of Theorem \ref{theorem:cond_weak_conv} follows by replacing $\gamma_{1}$ by $\check \gamma_{1\pm}$ and replacing $\gamma_{2}$ by $\check \gamma_{2\pm}$. 
From these results follows Corollary \ref{corollary:FQRD}.
\end{proof}

\begin{center}
{\large\bf SUPPLEMENTARY MATERIAL}
\end{center}

\begin{description}
\item[Supplementary material:] The supplementary material includes additional examples, additional proofs, and additional simulation results. 
\end{description}

\section*{Tables and Figures}
\clearpage
\thispagestyle{empty}
\begin{table}[t]
	\centering
	  \caption{(A) Simulated acceptance probabilities for uniform treatment nullity, (B) simulated acceptance probabilities for treatment homogeneity under the fuzzy quantile RDD, and (C) uniform coverage probability of the true quantile treatment effects by the uniform confidence bands. The left column groups (I) present results across alternative values of $\beta_1 \in \{0.00, 0.05, 0.10, 0.15, 0.20\}$ while fixing $\gamma_1 = 0$. The right column groups (II) present results across alternative values of $\gamma_1 \in \{0.00, 0.25, 0.50, 0.75, 1.00\}$ while fixing $\beta_1 = 0$. The nominal acceptance probability under the null hypothesis is 95\%.}
		\begin{tabular}{cc|c|ccccccc|c|cccc}
		\hline\hline
		\multicolumn{7}{c}{(I) (A) Joint Treatment Nullity} && \multicolumn{7}{c}{(II) (A) Joint Treatment Nullity} \\
			\cline{1-7}\cline{9-15}
		  \multicolumn{2}{c}{$n$} & \multicolumn{5}{c}{$\beta_1$} && \multicolumn{2}{c}{$n$} & \multicolumn{5}{c}{$\gamma_1$} \\
			\cline{3-7}\cline{11-15}
			       && 0.00  & 0.05  & 0.10  & 0.15  & 0.20  &&        && 0.00  & 0.25  & 0.50  & 0.75  & 1.00  \\
			\cline{1-7}\cline{9-15}
			 500   && 0.969 & 0.874 & 0.709 & 0.532 & 0.362 &&  500   && 0.969 & 0.956 & 0.942 & 0.934 & 0.922 \\
			1000   && 0.968 & 0.823 & 0.559 & 0.302 & 0.131 && 1000   && 0.968 & 0.941 & 0.916 & 0.867 & 0.800 \\
			1500   && 0.957 & 0.784 & 0.422 & 0.149 & 0.045 && 1500   && 0.957 & 0.944 & 0.877 & 0.802 & 0.708 \\
			2000   && 0.957 & 0.725 & 0.318 & 0.082 & 0.014 && 2000   && 0.957 & 0.927 & 0.852 & 0.747 & 0.617 \\
		\hline\hline
		\\
		\hline\hline
		\multicolumn{7}{c}{(I) (B) Treatment Homogeneity} && \multicolumn{7}{c}{(II) (B) Treatment Homogeneity} \\
			\cline{1-7}\cline{9-15}
		  \multicolumn{2}{c}{$n$} & \multicolumn{5}{c}{$\beta_1$} && \multicolumn{2}{c}{$n$} & \multicolumn{5}{c}{$\gamma_1$} \\
			\cline{3-7}\cline{11-15}
			       && 0.00  & 0.05  & 0.10  & 0.15  & 0.20  &&        && 0.00  & 0.25  & 0.50  & 0.75  & 1.00  \\
			\cline{1-7}\cline{9-15}
			 500   && 0.989 & 0.984 & 0.988 & 0.986 & 0.987 &&  500   && 0.989 & 0.966 & 0.928 & 0.873 & 0.828 \\
			1000   && 0.980 & 0.979 & 0.979 & 0.979 & 0.985 && 1000   && 0.980 & 0.939 & 0.842 & 0.729 & 0.591 \\
			1500   && 0.968 & 0.973 & 0.972 & 0.977 & 0.977 && 1500   && 0.968 & 0.913 & 0.770 & 0.606 & 0.438 \\
			2000   && 0.963 & 0.962 & 0.966 & 0.970 & 0.967 && 2000   && 0.963 & 0.884 & 0.705 & 0.503 & 0.311 \\
		\hline\hline
		\\
		\hline\hline
		\multicolumn{7}{c}{(I) (C) Uniform Coverage Probability} && \multicolumn{7}{c}{(II) (C) Uniform Coverage Probability} \\
			\cline{1-7}\cline{9-15}
		  \multicolumn{2}{c}{$n$} & \multicolumn{5}{c}{$\beta_1$} && \multicolumn{2}{c}{$n$} & \multicolumn{5}{c}{$\gamma_1$} \\
			\cline{3-7}\cline{11-15}
			       && 0.00  & 0.05  & 0.10  & 0.15  & 0.20  &&        && 0.00  & 0.25  & 0.50  & 0.75  & 1.00  \\
			\cline{1-7}\cline{9-15}
			 500   && 0.969 & 0.966 & 0.969 & 0.966 & 0.970 &&  500   && 0.969 & 0.962 & 0.962 & 0.969 & 0.970 \\
			1000   && 0.968 & 0.965 & 0.965 & 0.969 & 0.976 && 1000   && 0.968 & 0.953 & 0.953 & 0.952 & 0.960 \\
			1500   && 0.957 & 0.967 & 0.966 & 0.962 & 0.967 && 1500   && 0.957 & 0.958 & 0.952 & 0.944 & 0.951 \\
			2000   && 0.957 & 0.953 & 0.965 & 0.968 & 0.971 && 2000   && 0.957 & 0.950 & 0.950 & 0.941 & 0.944 \\
		\hline\hline
		\\
		\end{tabular}
	\label{tab:FQRD}
\end{table}
\clearpage
\thispagestyle{empty}
\begin{table}[t]
	\centering
	  \caption{(A) Simulated acceptance probabilities for uniform treatment nullity, (B) simulated acceptance probabilities for treatment homogeneity under the fuzzy quantile RDD for a sequence of weakening jumps. The results are displayed for $\Sigma_{33} \in \{2^0, 2^1, 2^2, 2^3, 2^4, 2^5, 2^6\}$. The parameter values are fixed at $\beta_1=\gamma_1=0$. The nominal acceptance probability under the null hypothesis is 95\%. The results for (C) uniform coverage probability of the true quantile treatment effects by the uniform confidence bands are omitted because they produce the same numbers as (A) under $\beta_1=\gamma_1=0$.}
		\begin{tabular}{cc|ccccccc}
		\hline\hline
		\multicolumn{9}{c}{(A) Joint Treatment Nullity}\\
		\hline
		\multicolumn{2}{c}{$n$} & \multicolumn{7}{c}{$\Sigma_{33}$} \\
		\cline{3-9}
		       && $2^0$ & $2^1$ & $2^2$ & $2^3$ & $2^4$ & $2^5$ & $2^6$\\ 
		\hline
		 500   && 0.968	& 0.984	& 0.984	& 0.985	& 0.986	& 0.987	& 0.989\\
		1000   && 0.962	& 0.969	& 0.979	& 0.984	& 0.984	& 0.983	& 0.988\\
		1500   && 0.956	& 0.962	& 0.967	& 0.976	& 0.980	& 0.982	& 0.980\\
		2000   && 0.957	& 0.955	& 0.957	& 0.972	& 0.977	& 0.979	& 0.981\\
		\hline\hline
		\\
		\hline\hline
		\multicolumn{9}{c}{(B) Treatment Homogeneity}\\
		\hline
		\multicolumn{2}{c}{$n$} & \multicolumn{7}{c}{$\Sigma_{33}$} \\
		\cline{3-9}
		       && $2^0$ & $2^1$ & $2^2$ & $2^3$ & $2^4$ & $2^5$ & $2^6$\\ 
		\hline
		 500   && 0.990 &	0.996	&	0.992	&	0.990	&	0.988	&	0.990	&	0.992\\
		1000   && 0.982	&	0.991	&	0.995	&	0.993	&	0.992	&	0.991	&	0.995\\
		1500   && 0.974	&	0.992	&	0.994	&	0.992	&	0.991	&	0.993	&	0.990\\
		2000   && 0.977	&	0.986	&	0.991	&	0.993	&	0.993	&	0.991	&	0.993\\
		\hline\hline
		\\
		\end{tabular}
	\label{tab:Weak_FQRD}
\end{table}
\clearpage
\thispagestyle{empty}
\begin{figure}[t]
	\centering
		\includegraphics[width=0.67\textwidth]{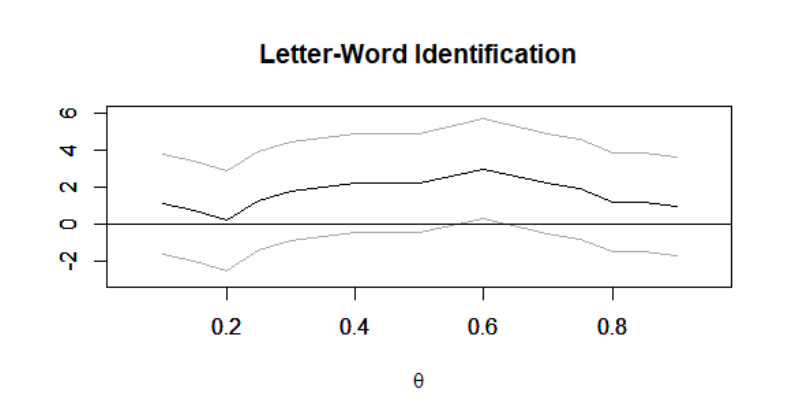}
		\\
		\includegraphics[width=0.67\textwidth]{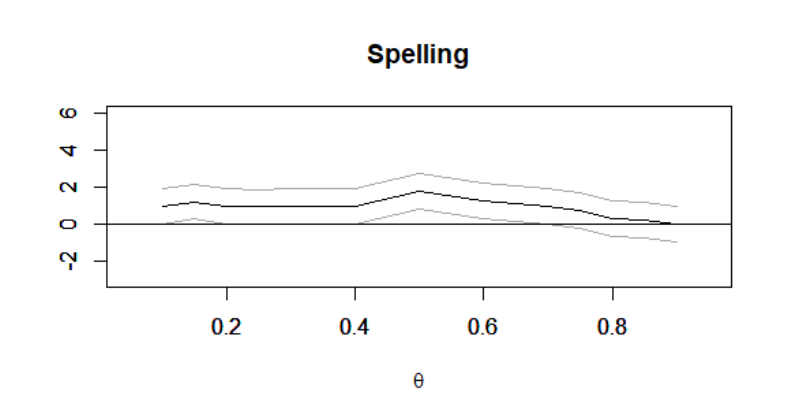}
		\\
		\includegraphics[width=0.67\textwidth]{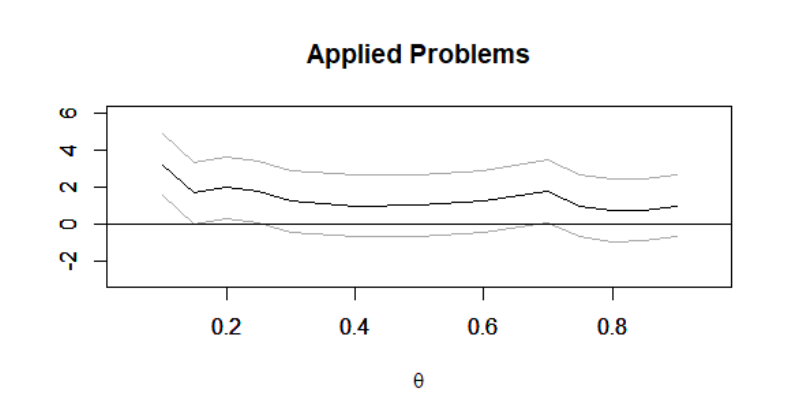}
	\caption{The estimated local quantile treatment effects of the pre-K programs on scores on the three sub-tests of the Woodcock-Johnson tests, and their 90\% uniform confidence bands.}
	\label{fig:application}
\end{figure}
\clearpage

\newpage
\setcounter{page}{1}
\begin{center}
\Large
Supplementary Material for ``Robust Uniform Inference for Quantile Treatment Effects in Regression Discontinuity Designs''
\\${}$\\
\normalsize
Harold D. Chiang \qquad Yu-Chin Hsu \qquad Yuya Sasaki
\\${}$\\
\date{ \today }
\end{center}
${}$\\
\begin{abstract}
This supplementary material contains additional examples of the general framework (Section \ref{sec:additional_examples}), applications of the general results to the additional examples (Section \ref{sec:simulation}), simulation results for the additional examples (Section \ref{sec:simulation}), additional mathematical proofs (Section \ref{sec:additional_mathematical_appendix}), and a guide to bandwidth choice procedures in practice (Section \ref{sec:practical_guideline}).
\end{abstract}

\newpage

\appendix
\setcounter{table}{1}
\setcounter{corollary}{1}
\setcounter{lemma}{7}
\setcounter{section}{1}
\section{Additional Examples of the General Framework}\label{sec:additional_examples}

\subsection{Example: Sharp Mean RDD}\label{sec:ex:sharp_rdd}
Sharp mean RDD is a special case of fuzzy mean RDD, where $ D_i = \mathbbm{1}\{X_i \geq 0\}$.
Thus, we can write $\mu_2(x,\theta_2) = \E[ D_i  |  X_i=x] = \mathbbm{1}\{x \geq 0\}$, and the local Wald estimand (\ref{eq:wald_est}) of the form (\ref{eq:fmrd}) further reduces to
\begin{eqnarray*}
\tau(\theta'')
&=&\lim_{x\downarrow 0}\E[ Y_i  |  X_i=x]-\lim_{x\uparrow 0}\E[ Y_i  |  X_i=x]
\end{eqnarray*}
for all $\theta'' \in \Theta'' = \{0\}$.
This estimand $\tau(0)$ will be denoted by $\tau_{SMRD}$ for Sharp Mean RD design.
For this estimand, Bartalotti, Calhoun, and He (2016) propose a robust bootstrap method of inference, and hence we are not the first to propose a robust bootstrap method for $\tau_{SMRD}$.
The benefit of our method is its applicability not only to $\tau_{SMRD}$, but also to many other estimands of the form (\ref{eq:wald_est}).

\subsection{Example: Fuzzy Mean RKD}\label{sec:ex:fuzzy_rkd}
Define $\Theta_1, \Theta_2, \Theta_1', \Theta_2', \Theta''$, $g_1$, $g_2$, $\phi$, $\psi$, and $\Upsilon$ as in Example \ref{ex:fuzzy_rdd}.
The local Wald estimand (\ref{eq:wald_est}) with $v=1$ in this setting becomes
\begin{eqnarray}
\tau(\theta'')
&=&\frac{\lim_{x\downarrow 0}\frac{\partial}{\partial x}\E[ Y_i  |  X_i=x]-\lim_{x\uparrow 0}\frac{\partial}{\partial x}\E[ Y_i  |  X_i=x]}{\lim_{x\downarrow 0}\frac{\partial}{\partial x}\E[ D_i  |  X_i=x]-\lim_{x\uparrow 0}\frac{\partial}{\partial x}\E[ D_i  |  X_i=x]}
\label{eq:fmrk}
\end{eqnarray}
for all $\theta'' \in \Theta'' = \{0\}$.
This estimand $\tau(0)$ will be denoted by $\tau_{FMRK}$ for Fuzzy Mean RK design.
See Card, Lee, Pei, and Weber (2016) for a causal interpretation of this estimand.

\subsection{Example: Sharp Mean RKD}\label{sec:ex:sharp_rkd}
Sharp mean RKD is a special case of fuzzy mean RKD where the treatment is defined by $ E[g_2(D_i,\theta_2)|X_i] = b(X_i)$ through a known function $b$.
Thus, the local Wald estimand (\ref{eq:wald_est}) of the form (\ref{eq:fmrk}) further reduces to
\begin{eqnarray*}
\tau(\theta'')
&=&\frac{\lim_{x\downarrow 0}\frac{\partial}{\partial x}\E[ Y_i  |  X_i=x]-\lim_{x\uparrow 0}\frac{\partial}{\partial x}\E[ Y_i  |  X_i=x]}{\lim_{x\downarrow 0}b^{(1)}(x) - \lim_{x\uparrow 0}b^{(1)}(x)}
\end{eqnarray*}
for all $\theta'' \in \Theta'' = \{0\}$.
This estimand $\tau(0)$ will be denoted by $\tau_{SMRK}$ for Sharp Mean RK design.
See Card, Lee, Pei, and Weber (2016) for a causal interpretation of this estimand.

\subsection{Example: CDF Discontinuity and Test of Stochastic Dominance}\label{sec:ex:cdf}
Let $\Theta_1 = \Theta_1' = \Theta'' = \mathscr{Y}$ for $\mathscr{Y} \subset \mathbb{R}$, and let $\Theta_2 = \Theta_2' = \{0\}$.
Set $g_1(Y_i,\theta_1) = \mathbbm{1}\{ Y_i \leq \theta_1\}$ and $g_2(D_i,\theta_2) =  D_i$, where $D_i = \mathbbm{1}\{X_i \geq 0\}$ holds under the sharp RD design.
Note that $\mu_{1}(x,\theta_1) = \E[g_1(Y_i,\theta_1)  |  X_i=x] = F_{Y  |  X}(\theta_1  |  x)$ and $\mu_{2}(x,\theta_2) = \E[g_2(D_i,\theta_2)  |  X_i=x] = \E[ D_i  |  X_i=x] = \mathbbm{1}\{x \geq 0\}$.
Let $\phi$ and $\psi$ be the identity operators, and for $W \in \ell^\infty (\Theta'_1 \times \Theta'_2)$, define $\Upsilon$ as $\Upsilon\left( W\right)(\theta'') = W(\theta'',0)$ $\forall \theta'' \in \Theta''$.
The local Wald estimand (\ref{eq:wald_est}) with $v=0$ in this setting becomes
\begin{eqnarray}
\tau(\theta'')
&=& \lim_{x\downarrow 0} F_{Y  |  X}(\theta''  |  x) - \lim_{x\uparrow 0} F_{Y  |  X}(\theta''  |  x)
\nonumber
\end{eqnarray}
for all $\theta'' \in \Theta'' = \mathscr{Y} \subset \mathbb{R}$.
This estimand $\tau$ will be denoted by $\tau_{SCRD}$ for Sharp CDF RD design.
This estimand may be useful to test the hypothesis of stochastic dominance ($\tau(\theta'') \leq 0$ for all $\theta'' \in \Theta''$).
See Shen and Zhang (2016) for this estimand and hypothesis testing.


\subsection{Example: Sharp Quantile RDD}\label{sec:ex:sharp_quantile_rdd}
Denote $Q_{Y|X}(\theta''):=\inf\{y\in \mathscr{Y}: F_{Y|X}(y)\ge \theta''\}$, fix an $a\in(0,1/2)$, $\varepsilon>0$ and let $ \mathscr{Y}_1=[Q_{Y|X}(a|0^-)-\varepsilon,Q_{Y|X}(1-a|0^-)+\varepsilon]\cup [Q_{Y|X}(a|0^+)-\varepsilon,Q_{Y|X}(1-a|0^+)+\varepsilon]$.

Let $\Theta_1 = \mathscr{Y}_1$, $\Theta_1' = \Theta'' = [a,1-a]$ and $\Theta_2 = \Theta_2' = \{0\}$.
Set $g_1(Y_i,\theta_1) = \mathbbm{1}\{ Y_i \leq \theta_1\}$ and $g_2(D_i,\theta_2) =  D_i$, where $D_i = \mathbbm{1}\{X_i \geq 0\}$ holds under the sharp RDD.
Note that $\mu_{1}(x,\theta_1) = \E[g_1(Y_i,\theta_1)  |  X_i=x] = F_{Y  |  X}(\theta_1  |  x)$ and $\mu_{2}(x,\theta_2) = \E[g_2(D_i,\theta_2)  |  X_i=x] = \E[ D_i  |  X_i=x] = \mathbbm{1}\{x \geq 0\}$.
Let $\phi(F_{Y  |  X}(\cdot|x))(\theta') = \inf\{\theta_1\in \Theta_1:F_{Y|X}(\theta_1|x)\ge \theta'\}$ $\forall \theta' \in \Theta_1'$,
$\psi(\mathbbm{1}\{x \geq 0\})(\theta') = \mathbbm{1}\{x \geq 0\}$ $\forall \theta' \in \Theta_2' = \{0\}$, and for $W \in \ell^\infty (\Theta'_1 \times \Theta'_2)$ let the operator $\Upsilon$ be the mapping $\Upsilon\left( W\right)(\theta'') = W(\theta'',0)$ $\forall \theta'' \in \Theta''$.
The local Wald estimand (\ref{eq:wald_est}) with $v=0$ in this setting becomes
\begin{eqnarray}
\tau(\theta'')
&=& \lim_{x\downarrow 0} Q_{Y  |  X}(\theta''  |  x) - \lim_{x\uparrow 0}Q_{Y  |  X}(\theta''  |  x)
\nonumber
\end{eqnarray}
for all $\theta'' \in \Theta'' = [a,1-a]$, where $Q_{Y  |  X}(\theta''  |  x) := \inf\{\theta_1\in \Theta_1 :F_{Y|X}(\theta_1|x)\ge \theta''\}$ for a short-hand notation.
This estimand $\tau$ will be denoted by $\tau_{SQRD}$ for Sharp Quantile RD design.
For this estimand, Qu and Yoon (2015b) propose a method of uniform inference based on uniform random sampling, and hence we are not the first to propose a bootstrap method for $\tau_{SQRD}$.
Our method adds the property of robustness to this existing method, besides the main result that it applies to other estimands too.

\subsection{Example: Fuzzy Quantile RKD}\label{sec:ex:fuzzy_quantile_rkd}
Let $ \mathscr{Y}_1=[Q_{Y|X}(a|0)-\varepsilon,Q_{Y|X}(1-a|0)+\varepsilon]$, $\Theta_1 = \mathscr{Y}_1$
, $\Theta_1' = \Theta'' = [a,1-a]$ for a constant $a \in (0,1/2)$, and $\Theta_2 = \Theta_2' = \{0\}$.
We set $g_1(Y_i,\theta_1) = \mathbbm{1}\{ Y_i \leq \theta_1\}$ and $g_2(D_i,\theta_2) =  D_i$.
Note that $\mu_{1}(x,\theta_1) = \E[g_1(Y_i,\theta_1)  |  X_i=x] = F_{Y  |  X}(\theta_1  |  x)$ and $\mu_{2}(x,\theta_2) = \E[g_2(D_i,\theta_2)  |  X_i=x] = \E[ D_i  |  X_i=x]$.
With the short-hand notations $f_{Y  |  X} = \frac{\partial}{\partial y}F_{Y  |  X}$ and $F_{Y  |  X}^{(1)} := \frac{\partial}{\partial x}F_{Y  |  X}$, let
$$
\phi(F_{Y  |  X}^{(1)}(\cdot|x))(\theta'):= -\frac{ F^{(1)}_{Y  |  X}(\inf\{\theta\in \Theta_1:F_{Y|X}(\theta|0)\ge \theta'\}  |  x) }{ f_{Y  |  X}( \inf\{\theta\in \Theta_1:F_{Y|X}(\theta|0)\ge \theta'\}  |  x)}
\qquad\forall \theta' \in \Theta_1',
$$
let $\psi$ be the identity operator, and for $W \in \ell^\infty (\Theta'_1 \times \Theta'_2)$, let the operator $\Upsilon$ be $\Upsilon\left( W\right)(\theta'') = W(\theta'',0)$ $\forall \theta'' \in \Theta''$.
We emphasize that $F^{(1)}_{Y  |  X}(\cdot|x)$ does \textit{not} map to $f_{Y  |  X}( \cdot  |  0)$ or $F_{Y  |  X}( \cdot  |  0)$ in the definition of $\phi$; instead $f_{Y  |  X}( \cdot  |  0)$ and $F_{Y  |  X}( \cdot  |  0)$ are embedded in the definition of $\phi$. It will be shown that  $\phi(F_{Y  |  X}^{(1)}(\cdot|x))(\theta'')=\frac{\partial }{\partial x} Q_{Y|X}(\theta''|x)$.
The local Wald estimand (\ref{eq:wald_est}) with $v=1$ in this setting becomes
\begin{eqnarray}
\label{eq:fuzzy_qrkd_estimand}
\tau(\theta'')
&=& \frac{ \lim_{x\downarrow 0} \frac{\partial }{\partial x} Q_{Y|X}(\theta''|x) - \lim_{x\uparrow 0}\frac{\partial }{\partial x} Q_{Y|X}(\theta''|x) }{ \lim_{x\downarrow 0} \frac{d}{dx}\E[D_i | X_i=x] - \lim_{x\uparrow 0} \frac{d}{dx}\E[D_i | X_i=x] }
\end{eqnarray}
for all $\theta'' \in \Theta'' = [a,1-a]$.
This estimand $\tau$ will be denoted by $\tau_{FQRK}$ for Fuzzy Quantile RK design.
See Chiang and Sasaki (2017) for its causal interpretation.

\subsection{Example: Sharp Quantile RKD}\label{sec:ex:sharp_quantile_rkd}
Sharp quantile RKD is a special case of fuzzy quantile RKD where the treatment is defined by $ D_i = b(X_i)$ as a known function $b$ of $X_i$.
Thus, the local Wald estimand (\ref{eq:wald_est}) of the form (\ref{eq:fuzzy_qrkd_estimand}) further reduces to
\begin{eqnarray}
\tau(\theta'')
&=& \frac{ \lim_{x\downarrow 0} \frac{\partial }{\partial x} Q_{Y|X}(\theta''|x) - \lim_{x\uparrow 0}\frac{\partial }{\partial x} Q_{Y|X}(\theta''|x) }{ \lim_{x\downarrow 0} b^{(1)}(x) - \lim_{x\uparrow 0} b^{(1)}(x) }
\nonumber
\end{eqnarray}
for all $\theta'' \in \Theta'' = [a,1-a]$. 
This estimand $\tau$ will be denoted by $\tau_{SQRK}$ for Sharp Quantile RK design.
See Chiang and Sasaki (2017) for its causal interpretation.

\subsection{Example: Group Covariate and Test of Heterogeneous Treatment Effects}\label{sec:ex:group}
Suppose that a researcher wants to make a joint inference for the average causal effects across observed heterogeneous groups $G_i$ taking categorical values in $\Theta_1 = \Theta_2 = \Theta_1' = \Theta_2' = \Theta'' = \{1,\cdots,K\}$ by RDD.
Let $Y_i=( Y_i^\ast, G_i)$ and $D_i = ( D_i^\ast, G_i)$.
Set $g_1((Y_i^\ast,G_i),\theta_1) = Y_i^\ast \cdot \mathbbm{1}\{G_i=\theta_1\}$ and $g_2((D_i^\ast,G_i),\theta_2) =  D_i^\ast \cdot \mathbbm{1}\{G_i=\theta_2\}$.
Note that $\mu_{1}(x,\theta_1) = \E[g_1((Y_i^\ast,G_i),\theta_1)  |  X_i=x] = \E[ Y_i^\ast \cdot \mathbbm{1}\{G_i=\theta_1\}  |  X_i=x]$ and $\mu_{2}(x,\theta_2) = \E[g_2((D_i^\ast,G_i),\theta_2)  |  X_i=x] = \E[ D_i^\ast \cdot \mathbbm{1}\{G_i=\theta_2\}  |  X_i=x]$.
Let $\phi$ and $\psi$ be the identity operators, and for $W \in \ell^\infty (\Theta'_1 \times \Theta'_2)$, define the operator $\Upsilon$ be $\Upsilon\left( W\right)(\theta'') = W(\theta'',\theta'')$ for all $ \theta'' \in \Theta''$.
The local Wald estimand (\ref{eq:wald_est}) with $v=0$ in this setting becomes
\begin{eqnarray}
\tau(\theta'')
&=&\frac{\lim_{x\downarrow 0}E[ Y_i^\ast \cdot \mathbbm{1}\{G_i=\theta''\}  |  X_i=x] - \lim_{x\uparrow 0}E[ Y_i^\ast \cdot \mathbbm{1}\{G_i=\theta''\}  |  X_i=x]}{\lim_{x\downarrow 0}E[ D_i^\ast \cdot \mathbbm{1}\{G_i=\theta''\}  |  X_i=x] - \lim_{x\uparrow 0}E[ D_i^\ast \cdot \mathbbm{1}\{G_i=\theta''\}  |  X_i=x]}
\nonumber
\end{eqnarray}
for all $\theta'' \in \Theta'' = \{1,\cdots,K\}$.
This estimand $\tau$ will be denoted by $\tau_{GFMRD}$ for Group Fuzzy Mean RD design.
This estimand may be useful to test the hypotheses of heterogeneous treatment effects ($\tau(\theta_1'') \neq \tau(\theta_2'')$ for some $\theta_1'',\theta_2'' \in \Theta''$) or unambiguous treatment significance ($\tau(\theta'') > 0$ for all $\theta'' \in \Theta''$).
While we introduced this group estimand for the fuzzy mean regression discontinuity design, we remark that a similar estimand can be developed for any combinations of sharp/fuzzy mean/quantile regression discontinuity/kink designs.

\section{Applications of the General Results to the Ten Examples}\label{sec:application_examples}

In this section, we apply the general results to the additional examples introduced in Example \ref{ex:fuzzy_rdd} as well as Sections \ref{sec:ex:sharp_rdd}--\ref{sec:ex:group}.
Throughout this section, we present our assumptions for the case of $p=2$.
We remark that, however, using a different order $p$ of local polynomial fitting is also possible by similar arguments.

\subsection{Example: Fuzzy Mean RDD}\label{sec:a:FMRD}

Consider $\Theta_1$, $\Theta_2$, $\Theta_1'$, $\Theta_2'$, $\Theta''$, $g_1$, $g_2$, $\phi$, $\psi$, and $\Upsilon$ defined in Example \ref{ex:fuzzy_rdd}.
Recall that we denote the local Wald estimand (\ref{eq:wald_est}) with $v=1$ in this setting by $\tau_{FMRD}$.
We also denote the analog estimator (\ref{eq:wald_estimator}) with $v=1$ in this setting by
\begin{align*}
\hat\tau_{FMRD}=\frac{\hat{\mu}_{1,2}(0^+,0)-\hat{\mu}_{1,2}(0^-,0)}{\hat{\mu}_{2,2}(0^+,0)-\hat{\mu}_{2,2}(0^-,0)}.
\end{align*}
For this application, we consider the following set of assumptions.

\newtheorem*{assumption_FMRD}{Assumption FMRD}
\begin{assumption_FMRD}\label{a:FMRD}
\qquad\\
%
%
(i) (a) $E[|Y|^{2+\epsilon}|X=\cdot]<\infty$ on $[\underline{x},\overline{x}] \backslash \{0\}$ for some $\epsilon>0$.
(b) $\frac{\partial^j}{\partial x^j}E[Y|X=\cdot]$ and $\frac{\partial^j}{\partial x^j}E[D|X=\cdot]$ are Lipschitz on $[\underline{x},0)$ and $(0,\overline{x}]$ for $j=0,1,2,3$.
(d) $E[D|X=0^+]\ne E[D|X=0^-]$. \\
(ii) The baseline bandwidth $h_n$ satisfies $h_n\to 0$, $nh^2_n\to \infty$, $nh^7_n\to 0$.
There exist constants $c_1$, $c_2$ such that $h_{1,n}=c_1h_n$ and $h_{2,n}=c_2 h_n$.\\
%
%
(iii) $V(Y|X=\cdot)$, $V(D|X=\cdot)\in \mathcal{C}^{1}([\underline{x},\overline{x}]\setminus\{0\})$ with bounded derivatives in $x$ and $0<V(Y|X=0^\pm)<\infty$ \\
%
\end{assumption_FMRD}

For $k \in \{1,2\}$, define
\begin{align*}
&\widehat{\mathds{X}}'_n(0,k)=\frac{1}{\sqrt{c_k}}[ \hat{\nu}^+_{\xi,n}(0,k) -  \hat{\nu}^-_{\xi,n}(0,k)],
\end{align*}
where the EMP is given by
\begin{align*}
&\hat{\nu}^\pm_{\xi,n}(0,1)=\sum_{i=1}^{n}\xi_i\frac{e'_0(\Gamma^\pm_2)^{-1}[Y_i-\tilde \mu_{1,2}(X_i,0)] r_2(\frac{X_i}{h_{1,n}}) K(\frac{X_i}{h_{1,n}})\delta_i^\pm}{\sqrt{nh_{1,n}}\hat{f}_X(0)}\\
&\hat{\nu}^\pm_{\xi,n}(0,2)=\sum_{i=1}^{n}\xi_i\frac{e'_0(\Gamma^\pm_2)^{-1}[D_i-\tilde \mu_{2,2}(X_i,0)]  r_2(\frac{X_i}{h_{2,n}}) K(\frac{X_i}{h_{2,n}})\delta_i^\pm}{\sqrt{nh_{2,n}}\hat{f}_X(0)}
\end{align*}
and $\tilde \mu_{k,2}(0^\pm,0)$ is defined in Lemma \ref{lemma:epsilon_unif_const}.
Our general result applied to the current case yields the following corollary.

\begin{corollary}[Example: Fuzzy Mean RDD]\label{corollary:FMRD}
Suppose that Assumptions S, K, M, and FMRD hold.\\
(i) There exists $\sigma_{FMRD}>0$ such that
\begin{align*}
&\sqrt{nh_n}[\hat{\tau}_{FMRD}-\tau_{FMRD}]\leadsto N(0,\sigma^2_{FMRD}).\textsc{}
\end{align*}
(ii)
Furthermore, with probability approaching one,
\begin{align*}
&\frac{(\hat{\mu}_{2,2}(0^+,0)-\hat{\mu}_{2,2}(0^-,0))\widehat{\mathds{X}}'_n(0,1)-(\hat{\mu}_{1,2}(0^+,0)-\hat{\mu}_{1,2}(0^-,0))\widehat{\mathds{X}}'_n(0,2)}{(\hat{\mu}_{2,2}(0^+,0)-\hat{\mu}_{2,2}(0^-,0))^2}\underset{\xi}{\overset{p}{\leadsto}} N(0,\sigma^2_{FMRD}).
\end{align*}
\end{corollary}

A proof is provided in Section \ref{sec:corollary:FMRD}.
Perhaps the most practically relevant application of this corollary is the test of the null hypothesis of treatment nullity:
$$
H_0: \tau_{FMRD} = 0.
$$
To test this hypothesis, we can use $\sqrt{nh_n} \abs{ \hat\tau_{FMRD} }$ as the test statistic, and use
$$
\abs{\frac{(\hat{\mu}_{2,2}(0^+,0)-\hat{\mu}_{2,2}(0^-,0))\widehat{\mathds{X}}'_n(0,1)-(\hat{\mu}_{1,2}(0^+,0)-\hat{\mu}_{1,2}(0^-,0))\widehat{\mathds{X}}'_n(0,2)}{(\hat{\mu}_{2,2}(0^+,0)-\hat{\mu}_{2,2}(0^-,0))^2}}
$$
to simulate its asymptotic distribution.

\subsection{Example: Sharp Mean RDD}\label{sec:a:SMRD}

Consider $\Theta_1$, $\Theta_2$, $\Theta_1'$, $\Theta_2'$, $\Theta''$, $g_1$, $g_2$, $\phi$, $\psi$, and $\Upsilon$ defined in Section \ref{sec:ex:sharp_rdd}.
Recall that we denote the local Wald estimand (\ref{eq:wald_est}) with $v=0$ in this setting by $\tau_{SMRD}$.
We also denote the analog estimator (\ref{eq:wald_estimator}) with $v=0$ in this setting by
\begin{align*}
\hat\tau_{SMRD}=\hat{\mu}_{1,2}(0^+,0)-\hat{\mu}_{1,2}(0^-,0).
\end{align*}
For this application, we consider the following set of assumptions.

\newtheorem*{assumption_SMRD}{Assumption SMRD}
\begin{assumption_SMRD}\label{a:SMRD}
\qquad\\
%
%
(i) (a) $E[|Y|^{2+\epsilon}|X=\cdot]<\infty$ on $[\underline{x},\overline{x}] \backslash \{0\}$ for some $\epsilon>0$.
(b) $\frac{\partial^j}{\partial x^j}E[Y|X=\cdot]$ is Lipschitz on $[\underline{x},0)$ and $(0,\overline{x}]$ for $j=0,1,2,3$. \\
(ii) $h_n$ satisfies $h_n\to 0$, $nh^7_n\to 0$ and $nh^2_n\to \infty$.\\
%
%
(iii) $V(Y|X=\cdot)\in \mathcal{C}^{1}([\underline{x},\overline{x}]\setminus \{0\})$ with bounded derivative in $x$ and $0<V(Y|X=0^\pm)<\infty$ \\
%
\end{assumption_SMRD}

Define the EMP
\begin{align*}
&\hat{\nu}^\pm_{\xi,n}=\sum_{i=1}^{n}\xi_i\frac{e'_0(\Gamma^\pm_2)^{-1}[Y_i-\tilde{\mu}_{1,2}(X_i,0)] r_2(\frac{X_i}{h_n}) K(\frac{X_i}{h_n})\delta_i^\pm}{\sqrt{nh_{n}}\hat{f}_X(0)},
\end{align*}
where $\tilde{\mu}_{1,2}$ is defined in the statement of Lemma \ref{lemma:epsilon_unif_const}.
Our general result applied to the current case yields the following corollary.

\begin{corollary}[Example: Sharp Mean RDD]\label{corollary:SMRD}
Suppose that Assumptions S, K, M, and SMRD hold.
\\
(i) There exists $\sigma_{SMRD}>0$ such that
\begin{align*}
&\sqrt{nh_n}[\hat{\tau}_{SMRD}-\tau_{SMRD}]\leadsto N(0,\sigma^2_{SMRD})
\end{align*}
(ii)
Furthermore, with probability approaching one,
\begin{align*}
& \hat{\nu}^+_{\xi,n}- \hat{\nu}^-_{\xi,n}\underset{\xi}{\overset{p}{\leadsto}}N(0,\sigma^2_{SMRD}).
\end{align*}
\end{corollary}

A proof is provided in Section \ref{sec:corollary:SMRD}.
Perhaps the most practically relevant application of this corollary is the test of the null hypothesis of treatment nullity:
$$
H_0: \tau_{SMRD} = 0.
$$
To test this hypothesis, we can use $\sqrt{nh_n} \abs{ \hat\tau_{SMRD} }$ as the test statistic, and use $\abs{ \hat{\nu}^+_{\xi,n}- \hat{\nu}^-_{\xi,n} }$ to simulate its asymptotic distribution.

\subsection{Example: Fuzzy Mean RKD}\label{sec:a:FMRK}

Consider $\Theta_1$, $\Theta_2$, $\Theta_1'$, $\Theta_2'$, $\Theta''$, $g_1$, $g_2$, $\phi$, $\psi$, and $\Upsilon$ defined in Section \ref{sec:ex:fuzzy_rkd}.
Recall that we denote the local Wald estimand (\ref{eq:wald_est}) with $v=1$ in this setting by $\tau_{FMRK}$.
We also denote the analog estimator (\ref{eq:wald_estimator}) with $v=1$ in this setting by
\begin{align*}
\hat\tau_{FMRK}=\frac{\hat{\mu}^{(1)}_{1,2}(0^+,0)-\hat{\mu}^{(1)}_{1,2}(0^-,0)}{\hat{\mu}^{(1)}_{2,2}(0^+,0)-\hat{\mu}^{(1)}_{2,2}(0^-,0)}
\end{align*}
For this application, we consider the following set of assumptions.

\newtheorem*{assumption_FMRK}{Assumption FMRK}
\begin{assumption_FMRK}\label{a:FMRK}
\qquad\\
%
(i) (a) $E[|Y|^{2+\epsilon}|X=\cdot]<\infty$ and $E[|D|^{2+\epsilon}|X=\cdot]<\infty$ on $[\underline{x},\overline{x}] \backslash \{0\}$ for some $\epsilon>0$.
(b) $\frac{\partial^j}{\partial x^j}E[Y|X=\cdot]$ and $\frac{\partial^j}{\partial x^j}E[D|X=\cdot]$ are Lipschitz on $[\underline{x},0)$ and $(0,\overline{x}]$ for $j=0,1,2,3$.
(c) $\frac{\partial}{\partial x}E[D|X=0^+]\ne \frac{\partial}{\partial x}E[D|X=0^-]$. \\
(ii) The baseline bandwidth $h_n$ satisfies $h_n\to 0$, $nh^3_n\to \infty$, $nh^7_n\to 0$.
There exist constant $c_1$, $c_2$ such that $h_{1,n}=c_1h_n$ and $h_{2,n}=c_2 h_n$.\\
%
%
(iii) $V(Y|X=\cdot)$, $V(D|X=\cdot)\in \mathcal{C}^{1}([\underline{x},\overline{x}]\setminus\{0\})$ with bounded derivatives in $x$ and $0<V(Y|X=0^\pm)<\infty$ \\
%
\end{assumption_FMRK}

For $k \in \{1,2\}$, define
\begin{align*}
&\widehat{\mathds{X}}'_n(0,k)=\frac{1}{\sqrt{c^{3}_k}}[ \hat{\nu}^+_{\xi,n}(0,k) -  \hat{\nu}^-_{\xi,n}(0,k)],
\end{align*}
where the EMP is given by
\begin{align*}
&\hat{\nu}^\pm_{\xi,n}(0,1)=\sum_{i=1}^{n}\xi_i\frac{e'_1(\Gamma^\pm_2)^{-1}[Y_i-\tilde \mu_{1,2}(X_i,0)] r_2(\frac{X_i}{h_{1,n}}) K(\frac{X_i}{h_{1,n}})\delta_i^\pm}{\sqrt{nh_{1,n}}\hat{f}_X(0)}\\
&\hat{\nu}^\pm_{\xi,n}(0,2)=\sum_{i=1}^{n}\xi_i\frac{e'_1(\Gamma^\pm_2)^{-1}[D_i-\tilde \mu_{2,2}(X_i,0)]  r_2(\frac{X_i}{h_{2,n}}) K(\frac{X_i}{h_{2,n}})\delta_i^\pm}{\sqrt{nh_{2,n}}\hat{f}_X(0)}
\end{align*}
and $\tilde \mu_{k,2}(0^\pm,0)$ is defined in Lemma \ref{lemma:epsilon_unif_const}.
Our general result applied to the current case yields the following corollary.

\begin{corollary}[Example: Fuzzy Mean RKD]\label{corollary:FMRK}
Suppose that Assumptions S, K, M, and FMRK hold.
\\
(i) There exists $\sigma_{FMRK}>0$ such that
\begin{align*}
&\sqrt{nh^3_n}[\hat{\tau}_{FMRK}-\tau_{FMRK}]\leadsto N(0,\sigma^2_{FMRK}).
\end{align*}
(ii)
Furthermore, with probability approaching one,
\begin{align*}
&\frac{(\hat{\mu}^{(1)}_{2,2}(0^+,0)-\hat{\mu}^{(1)}_{2,2}(0^-,0))\widehat{\mathds{X}}'_n(0,1)-(\hat{\mu}^{(1)}_{1,2}(0^+,0)-\hat{\mu}^{(1)}_{1,2}(0^-,0))\widehat{\mathds{X}}'_n(0,2)}{(\hat{\mu}^{(1)}_{2,2}(0^+,0)-\hat{\mu}^{(1)}_{2,2}(0^-,0))^2}\underset{\xi}{\overset{p}{\leadsto}}N(0,\sigma^2_{FMRK}).
\end{align*}
\end{corollary}

This corollary can be proved similarly to Corollary \ref{corollary:FMRD}.
Perhaps the most practically relevant application of this corollary is the test of the null hypothesis of treatment nullity:
$$
H_0: \tau_{FMRK} = 0.
$$
To test this hypothesis, we can use $\sqrt{nh_n^3} \abs{ \hat\tau_{FMRK} }$ as the test statistic, and use
$$
\abs{\frac{(\hat{\mu}^{(1)}_{2,2}(0^+,0)-\hat{\mu}^{(1)}_{2,2}(0^-,0))\widehat{\mathds{X}}'_n(0,1)-(\hat{\mu}^{(1)}_{1,2}(0^+,0)-\hat{\mu}^{(1)}_{1,2}(0^-,0))\widehat{\mathds{X}}'_n(0,2)}{(\hat{\mu}^{(1)}_{2,2}(0^+,0)-\hat{\mu}^{(1)}_{2,2}(0^-,0))^2}}
$$
to simulate its asymptotic distribution.

\subsection{Example: Sharp Mean RKD}\label{sec:a:SMRK}

Consider $\Theta_1$, $\Theta_2$, $\Theta_1'$, $\Theta_2'$, $\Theta''$, $g_1$, $g_2$, $\phi$, $\psi$, and $\Upsilon$ defined in Section \ref{sec:ex:sharp_rkd}.
Recall that we denote the local Wald estimand (\ref{eq:wald_est}) with $v=1$ in this setting by $\tau_{SMRK}$.
We also denote the analog estimator (\ref{eq:wald_estimator}) with $v=1$ in this setting by
\begin{align*}
\hat\tau_{SMRK}=\frac{\hat{\mu}^{(1)}_{1,2}(0^+,0)-\hat{\mu}^{(1)}_{1,2}(0^-,0)}{b^{(1)}(0^+)-b^{(1)}(0^-)}.
\end{align*}
For this application, we consider the following set of assumptions.

\newtheorem*{assumption_SMRK}{Assumption SMRK}
\begin{assumption_SMRK}\label{a:SMRK}
\qquad\\
%
%
(i) (a) $E[|Y|^{2+\epsilon}|X=\cdot]<\infty$ on $[\underline{x},\overline{x}] \backslash \{0\}$ for some $\epsilon>0$.
(b) $\frac{\partial^j}{\partial x^j}E[Y|X=\cdot]$ is Lipschitz on $[\underline{x},0)$ and $(0,\overline{x}]$ for $j=0,1,2,3$. \\
(ii) $h_n$ satisfies $h_n\to 0$, $nh^3_n\to \infty$ and $nh^7_n\to 0$.\\
%
%
(iii) $V(Y|X=x)\in \mathcal{C}^{1}([\underline{x},\overline{x}]\setminus \{0\})$ with bounded derivative in $x$ and $0<V(Y|X=0^\pm)<\infty$ \\
%
%
(iv) $b$ is continuously differentiable on $I\setminus\{0\}$ and $b^{(1)}(0^+)-b^{(1)}(0^-)\ne 0$.
\end{assumption_SMRK}

Define the EMP
\begin{align*}
&\hat{\nu}^\pm_{\xi,n}=\sum_{i=1}^{n}\xi_i\frac{e'_1(\Gamma^\pm_2)^{-1}[Y_i-\tilde{\mu}_{1,2}(X_i,0)] r_2(\frac{X_i}{h_n}) K(\frac{X_i}{h_n})\delta_i^\pm}{\sqrt{nh_{n}}\hat{f}_X(0)},
\end{align*}
where $\tilde{\mu}_{1,2}$ is defined in the statement of Lemma \ref{lemma:epsilon_unif_const}.
Our general result applied to the current case yields the following corollary.

\begin{corollary}[Example: Sharp Mean RKD]\label{corollary:SMRK}
Suppose that Assumptions S, K, M, and SMRK hold.
\\
(i) There exists $\sigma_{SMRK}>0$ such that
\begin{align*}
&\sqrt{nh^3_n}[\hat{\tau}_{SMRK}-\tau_{SMRK}]\leadsto N(0,\sigma^2_{SMRK}).
\end{align*}
(ii)
Furthermore, with probability approaching one,
\begin{align*}
&\frac{ \hat{\nu}^+_{\xi,n}- \hat{\nu}^-_{\xi,n}}{b^{(1)}(0^+)-b^{(1)}(0^-)}\underset{\xi}{\overset{p}{\leadsto}}N(0,\sigma^2_{SMRK}).
\end{align*}
\end{corollary}

This corollary can be proved similarly to Corollary \ref{corollary:SMRD}.
Perhaps the most practically relevant application of this corollary is the test of the null hypothesis of treatment nullity:
$$
H_0: \tau_{SMRK} = 0.
$$
To test this hypothesis, we can use $\sqrt{nh_n^3} \abs{ \hat\tau_{SMRK} }$ as the test statistic, and use $\abs{ \frac{ \hat{\nu}^+_{\xi,n}- \hat{\nu}^-_{\xi,n}  }{ b^{(1)}(0^+) - b^{(1)}(0^-)}}$ to simulate its asymptotic distribution.

\subsection{Example: CDF Discontinuity and Test of Stochastic Dominance}\label{sec:a:SCRD}

Consider $\Theta_1$, $\Theta_2$, $\Theta_1'$, $\Theta_2'$, $\Theta''$, $g_1$, $g_2$, $\phi$, $\psi$, and $\Upsilon$ defined in Section \ref{sec:ex:cdf}.
Recall that we denote the local Wald estimand (\ref{eq:wald_est}) with $v=1$ in this setting by $\tau_{SCRD}$.
We also denote the analog estimator (\ref{eq:wald_estimator}) with $v=1$ in this setting by
\begin{align*}
\hat\tau_{SCRD}(\theta'')=
\hat F_{Y|X}(\theta''|0^+)-\hat F_{Y|X}(\theta''|0^-)=
\hat\mu_{1,2}(0^+,\theta'')-\hat\mu_{1,2}(0^-,\theta'').
\end{align*}
For this application, we consider the following set of assumptions.

\newtheorem*{assumption_SCRD}{Assumption SCRD}
\begin{assumption_SCRD}\label{a:SCRD}
\quad\\
%
%
(i) $\frac{\partial^j}{\partial x^j}F_{Y|X}$ is Lipschitz in $x$ on $\mathscr{Y}\times[\underline{x},0)$ and $\mathscr{Y}\times(0,\overline{x}]$ for $j=0,1,2,3$.\\
(ii) $h_n$ satisfies $h_n\to 0$, $nh^7_n\to 0$, and $nh^2_n\to \infty$.\\
%
%
\end{assumption_SCRD}

Let the EMP be given by
\begin{align*}
\hat{\nu}^\pm_{\xi,n}(\theta'')&=\sum_{i=1}^{n}\xi_i\frac{e'_0(\Gamma^\pm_2)^{-1}[\mathds{1}\{Y_i\le \theta''\}-\tilde{F}_{Y|X}(\theta''|X_i)]r_2(\frac{X_i}{h_n})K(\frac{X_i}{h_n})\delta^\pm_i}{\sqrt{nh_{n}}\hat{f}_X(0)},
\end{align*}
where $\tilde{F}_{Y|X}(\theta''|x)=\tilde{\mu}_{1,2}(x,\theta'')\mathds{1}\{|x/h_n| \in [-1,1]\}$, and $\tilde{\mu}_{1,2}$ is defined in the statement of Lemma \ref{lemma:epsilon_unif_const}.
Our general result applied to the current case yields the following corollary.
A proof is provided in Section \ref{sec:corollary:SCRD}.

\begin{corollary}[Example: CDF Discontinuity]\label{corollary:SCRD}
Suppose that Assumptions S, K, M, and SCRD hold.
\\
(i) There exists a zero mean Gaussian process $\mathds{G}'_{SCRD}:\Omega^x \mapsto \ell^\infty(\mathscr{Y})$ such that
$$
\sqrt{nh_{n}}[\hat{\tau}_{SCRD}-\tau_{SCRD}]\leadsto \mathds{G}'_{SCRD}.
$$
(ii) Furthermore, with probability approaching one,
$$
 \hat{\nu}^+_{\xi,n}- \hat{\nu}^-_{\xi,n} \underset{ \xi}{\overset{p}{\leadsto}} \mathds{G}'_{SCRD}.
$$
\end{corollary}

One of the most common applications of weak convergence results for CDFs as stated in this corollary is the test of the stochastic dominance:
$$
H_0: \tau_{SCRD}(\theta'') \leq 0 \qquad\forall \theta'' \in \Theta''.
$$
See McFadden (1989).
To test this hypothesis, we can use $\sup_{\theta'' \in \Theta''} \sqrt{nh_n} \max\{\hat\tau_{SCRD}(\theta''),0\} $ as the test statistic, and use
$
\sup_{\theta'' \in \Theta''} \sqrt{nh_n}  \max\{\hat{\nu}^+_{\xi,n}(\theta'') - \hat{\nu}^-_{\xi,n}(\theta''),0\}
$
to simulate its asymptotic distribution.

\subsection{Example: Sharp Quantile RDD}\label{sec:a:SQRD}

Consider $\Theta_1$, $\Theta_2$, $\Theta_1'$, $\Theta_2'$, $\Theta''$, $g_1$, $g_2$, $\phi$, $\psi$, and $\Upsilon$ defined in Section \ref{sec:ex:sharp_quantile_rdd}.
Recall that we denote the local Wald estimand (\ref{eq:wald_est}) with $v=1$ in this setting by $\tau_{SQRD}$.
We also denote the analog estimator (\ref{eq:wald_estimator}) with $v=1$ in this setting by
\begin{align*}
\hat{\tau}_{SQRD}(\theta'')
=\hat Q_{Y|X}(\theta'' | 0^+) - \hat Q_{Y|X}(\theta'' | 0^-)
=\phi(\hat{\mu}_{1,p}(0^+,\cdot))(\theta'')-\phi(\hat{\mu}_{1,p}(0^-,\cdot))(\theta'')
\end{align*}
for $\theta''\in [a,1-a]\subset (0,1)$.
For this application, we consider the following set of assumptions.

\newtheorem*{assumption_SQRD}{Assumption SQRD}
\begin{assumption_SQRD}\label{a:SQRD}
\quad\\
%
%
(i) (a)$\frac{\partial^j}{\partial x^j}F_{Y|X}$ is Lipschitz in $x$ on $\mathscr{Y}_1\times[\underline{x},0)$ and $\mathscr{Y}_1\times(0,\overline{x}]$ for $j=0,1,2,3$.
(b) $f_{Y|X}(y|x)$ is Lipschitz in $x$ and $0<C<f_{Y|X}<C'<\infty$ on $\mathscr{Y}_1\times[\underline{x},0)$ and $\mathscr{Y}_1\times (0,\overline{x}]$. \\
(ii) $h_n$ satisfies $h_n\to 0$, $nh^7_n\to 0$, and $nh^2_n\to \infty$.\\
%
%
%
(iii) There exists $\hat f_{Y|X} (y|0^\pm)$ such that $\sup_{y\in \mathscr{Y}_1}|\hat f_{Y|X} (y|0^\pm)-f_{Y|X} (y|0^\pm)|=o^{x}_p(1)$.
\end{assumption_SQRD}
 We state (iii) as a high level assumption to accommodate a number of alternative estimators. An example and sufficient conditions for (iii) is given above Lemma \ref{lemma:unif_cons_fYX} in section \ref{sec:lemma:unif_cons_est} in the Mathematical Appendix.

Define
\begin{align*}
&\widehat{\phi}'_{F_{Y|X}(\cdot|0^\pm)}(\hat{\nu}^\pm_{\xi,n})(\theta'')
=-\frac{\hat{\nu}^\pm_{\xi,n}(\hat{Q}_{Y|X}(\theta''|0^\pm))}{\hat{f}_{Y|X}(\hat{Q}_{Y|X}(\theta''|0^\pm)|0^\pm)}
\\
& \qquad =-\sum_{i=1}^{n}\xi_i\frac{e'_0(\Gamma^+_2)^{-1}[\mathds{1}\{Y_i\le \hat{Q}_{Y|X}(\theta''|0^\pm)\}-\tilde{F}_{Y|X}(\hat{Q}_{Y|X}(\theta''|0^\pm)|X_i)]r_2(\frac{X_i}{h_n})K(\frac{X_i}{h_n})\delta^\pm_i}{\sqrt{nh_{n}}\hat{f}_X(0)\hat{f}_{Y|X}(\hat Q_{Y|X}(\theta''|0^\pm)|0^\pm)}
\end{align*}
where $\tilde{F}_{Y|X}(y|x)=\tilde{\mu}_{1,2}(x,y)\mathds{1}\{|x/h_n|\le 1\}$, and
$\tilde{\mu}_{1,2}$ is defined in the statement of Lemma \ref{lemma:epsilon_unif_const}.
Our general result applied to the current case yields the following corollary.

\begin{corollary}[Example: Sharp Quantile RDD]\label{corollary:SQRD}
Suppose that Assumptions S, K, M, and SQRD hold.\\
(i) There exists a zero mean Gaussian process $\mathds{G}'_{SQRD}:\Omega^x \mapsto \ell^\infty([a,1-a])$ such that
$$
\sqrt{nh_{n}}[\hat{\tau}_{SQRD}-\tau_{SQRD}]\leadsto \mathds{G}'_{SQRD}.
$$
(ii) Furthermore, with probability approaching one,
$$
\widehat{\phi}'_{F_{Y|X}(\cdot|0^+)}( \hat{\nu}^+_{\xi,n})-\widehat{\phi}'_{F_{Y|X}(\cdot|0^-)}( \hat{\nu}^-_{\xi,n})
\underset{ \xi}{\overset{p}{\leadsto}} \mathds{G}'_{SQRD}.
$$
\end{corollary}

A proof is provided in Section \ref{sec:corollary:SQRD}.
One of the practically most relevant applications of this corollary is the test of the null hypothesis of uniform treatment nullity:
$$
H_0: \tau_{SQRD}(\theta'') = 0 \quad\text{for all } \theta'' \in \Theta'' = [a,1-a].
$$
See Koenker and Xiao (2002), Chernozhukov and Fern\'andez-Val (2005), and Qu and Yoon (2015b).
To test this hypothesis, we can use $\sup_{\theta'' \in [a,1-a]} \sqrt{nh_n} \abs{ \hat\tau_{SQRD}(\theta'') }$ as the test statistic, and use
$$
\sup_{\theta'' \in [a,1-a]} \abs{\widehat{\phi}'_{F_{Y|X}(\cdot|0^+)}( \hat{\nu}^+_{\xi,n})(\theta'')-\widehat{\phi}'_{F_{Y|X}(\cdot|0^-)}( \hat{\nu}^-_{\xi,n})(\theta'')}
$$
to simulate its asymptotic distribution.

Another one of the practically most relevant applications of the above corollary is the test of the null hypothesis of treatment homogeneity across quantiles:
$$
H_0: \tau_{SQRD}(\theta'') = \tau_{SQRD}(\theta''') \quad\text{for all } \theta'', \theta''' \in \Theta'' = [a,1-a].
$$
We again refer to the list of references in the previous paragraph.
To test this hypothesis, we can use $\sup_{\theta'' \in [a,1-a]} \sqrt{nh_n} \abs{ \hat\tau_{SQRD}(\theta'') - (1-2a)^{-1} \int_{[a,1-a]} \hat\tau_{SQRD}(\theta''')d\theta''' }$ as the test statistic, and use
\begin{align*}
\sup_{\theta'' \in [a,1-a]}
&\left\vert \widehat{\phi}'_{F_{Y|X}(\cdot|0^+)}( \hat{\nu}^+_{\xi,n})(\theta'')-\widehat{\phi}'_{F_{Y|X}(\cdot|0^-)}( \hat{\nu}^-_{\xi,n})(\theta'') \right.
\\
&\left. -
\frac{1}{1-2a} \int_{[a,1-a]} \left( \widehat{\phi}'_{F_{Y|X}(\cdot|0^+)}( \hat{\nu}^+_{\xi,n})(\theta''')-\widehat{\phi}'_{F_{Y|X}(\cdot|0^-)}( \hat{\nu}^-_{\xi,n})(\theta''') \right) d\theta'''
\right\vert
\end{align*}
to simulate its asymptotic distribution.

\begin{remark}
It may happen that $\hat F_{Y|X}(\cdot|0^\pm)$ is not monotone increasing in finite sample. We may monotonize the estimated CDFs by re-arrangements following Chernozhukov, Fernandez-Val, Galichon (2010).
This does not affect the asymptotic properties of the estimators, while allowing for inversion of the CDF estimators.
\end{remark}

\subsection{Example: Fuzzy Quantile RKD}\label{sec:a:FQRK}
Consider $\Theta_1$, $\Theta_2$, $\Theta_1'$, $\Theta_2'$, $\Theta''$, $g_1$, $g_2$, $\phi$, $\psi$, and $\Upsilon$ defined in Section \ref{sec:ex:fuzzy_quantile_rkd}.
Recall that we denote the local Wald estimand (\ref{eq:wald_est}) with $v=1$ in this setting by $\tau_{FQRK}$.
We also denote the analog estimator (\ref{eq:wald_estimator}) with $v=1$ in this setting by
\begin{align*}
\hat\tau_{FQRK}(\theta'')=&\frac{\widehat \phi(\hat{F}^{(1)}_{Y|X}(\cdot|0^+))(\theta'')-\widehat \phi(\hat{F}^{(1)}_{Y|X}(\cdot|0^-))(\theta'')}{\hat{\mu}^{(1)}_{2,2}(0^+,0)-\hat{\mu}^{(1)}_{2,2}(0^-,0)},
\end{align*}
where $\widehat \phi(\hat{F}^{(1)}_{Y|X}(\cdot|0^\pm))(\theta''):=F^{(1)}_{Y|X}(\hat{Q}_{Y|X}(\cdot|0)|0^\pm)/\hat{f}_{Y|X}(\hat{Q}_{Y|X}(\cdot|0)|0)$.
We further define
\begin{align*}
&\widehat{\phi}'_{F^{(1)}_{Y|X}(\cdot|0^\pm)}(\hat{\nu}^\pm_{\xi,n}(\cdot,1)/\sqrt{c^3_1})(\cdot)
\\
& \qquad\qquad =-\sum_{i=1}^{n}\xi_i\frac{e'_1(\Gamma^\pm_2)^{-1}[\mathds{1}\{Y_i\le \hat{Q}_{Y|X}(\cdot|0)\}-\tilde{F}_{Y|X}(\hat{Q}_{Y|X}(\cdot|0)|X_i)]r_2(\frac{X_i}{h_{1,n}})K(\frac{X_i}{h_{1,n}})\delta^\pm_i}{\sqrt{c^3_1 nh_{1,n}}\hat{f}_X(0)\hat{f}_{Y|X}(\hat{Q}_{Y|X}(\cdot|0)|0)},\\
&\widehat{\psi}'_{\mu^{(1)}_{2}(0^\pm,0)}((\cdot,2)/\sqrt{c^3_2})(\cdot)
=\sum_{i=1}^{n}\xi_i\frac{e'_1(\Gamma^\pm_2)^{-1}[D_i-\tilde{\mu}_{2}(X_i,0)]r_2(\frac{X_i}{h_{2,n}})K(\frac{X_i}{h_{2,n}})\delta^\pm_i}{\sqrt{c^3_2 nh_{2,n}}\hat{f}_X(0)},\\
&\widehat{\mathds{X}}'_n(\cdot,1)=\widehat{\phi}'_{F_{Y|X}(\cdot|0^+)}( \hat{\nu}^+_{\xi,n}(\cdot,1)/\sqrt{c^3_1})(\cdot)-\widehat{\phi}'_{F_{Y|X}(\cdot|0^-)}( \hat{\nu}^-_{\xi,n}(\cdot,1)/\sqrt{c^3_1})(\cdot), \qquad\text{and}\\
&\widehat{\mathds{X}}'_n(\cdot,2)=\widehat{\psi}'_{\mu^{(1)}_{2}(0^+,0)}( \hat{\nu}^+_{\xi,n}(\cdot,2)/\sqrt{c^3_2})(\cdot)-\widehat{\psi}'_{\mu^{(1)}_{2}(0^-,0)}( \hat{\nu}^-_{\xi,n}(\cdot,2)/\sqrt{c^3_2})(\cdot),
\end{align*}
where $\tilde{F}_{Y|X}(y|x)=\tilde{\mu}_{1,2}(x,y)\mathds{1}\{|x/h_{1,n}|\le 1\}$.
For this application, we consider the following set of assumptions.

\newtheorem*{assumption_FQRK}{Assumption FQRK}
\begin{assumption_FQRK}\label{a:FQRK}
\qquad\\
%
(i)(a) $\frac{\partial^j}{\partial x^j}F_{Y|X}$ is Lipschitz on $\mathscr{Y}_1\times [\underline{x},0)$ and $\mathscr{Y}_1\times(0,\overline{x}]$ in $x$ for $j=0,1,2,3$. $\frac{\partial^j}{\partial x^j}E[D|X=\cdot]$ is Lipschitz continuous on $ [\underline{x},0)$ and $(0,\overline{x}]$ in $x$ for $j=0,1,2,3$.  (b) $f_{Y|X}$ is Lipschitz in $x$ and $0<C<f_{Y|X}(y|x)<C'<\infty$ on $\mathscr{Y}_1\times ([\underline{x},\overline{x}]\setminus \{0\})$. \\
(ii) The baseline bandwidth $h_n$ satisfies $h_n\to 0$, $nh^7_n\to 0$, and $nh^3_n\to \infty$, and there exist constants $c_1$, $c_2\in(0,\infty)$ such that $h_{k,n}=c_kh_n$. \\
%
%
%
(iii) $E[D|X=0^+]\ne E[D|X=0^-]$.\\
(iv) $\sup_{y\in \mathscr{Y}_1}|\sqrt{nh^3_n}[\hat{f}_{Y|X}(y|0)-f_{Y|X}(y|0)]|\underset{x}{\overset{p}{\to}}0$
\end{assumption_FQRK}

Our general result applied to the current case yields the following corollary.

\begin{corollary}[Example: Fuzzy Quantile RKD]\label{corollary:FQRK}
Suppose that Assumptions S, K, M, and FQRK hold.
\\
(i) There exists a zero mean Gaussian process $\mathds{G}'_{FQRK}:\Omega^x \mapsto \ell^\infty([a,1-a])$ such that
$$
\sqrt{nh^3_n}[\hat{\tau}_{FQRK}-\tau_{FQRK}]\leadsto \mathds{G}'_{FQRK}
$$
(ii) Furthermore, with probability approaching one,
$$
\frac{[\hat{\mu}^{(1)}_{2,2}(0^+,0)-\hat{\mu}^{(1)}_{2,2}(0^-,0)]\widehat{\mathds{X}}'_n(\cdot,1)-[\widehat \phi(\hat{F}^{(1)}_{Y|X}(\cdot|0^+))(\cdot)-\widehat \phi(\hat{F}^{(1)}_{Y|X}(\cdot|0^-))(\cdot)]
\widehat{\mathds{X}}'_n(\cdot,2)}{[\hat{\mu}^{(1)}_{2,2}(0^+,0)-\hat{\mu}^{(1)}_{2,2}(0^-,0)]^2}
\underset{\xi}{\overset{p}{\leadsto}} \mathds{G}'_{FQRK}(\cdot)
$$
\end{corollary}

This corollary can be proved similarly to Corollary \ref{corollary:SQRK}.
One of the practically most relevant applications of this corollary is the test of the null hypothesis of uniform treatment nullity:
$$
H_0: \tau_{FQRK}(\theta'') = 0 \quad\text{for all } \theta'' \in \Theta'' = [a,1-a].
$$
To test this hypothesis, we can use $\sup_{\theta'' \in [a,1-a]} \sqrt{nh_n^3} \abs{ \hat\tau_{FQRK}(\theta'') }$ as the test statistic, and use
$$
\sup_{\theta'' \in [a,1-a]} \abs{\frac{[\hat{\mu}^{(1)}_{2,2}(0^+,0)-\hat{\mu}^{(1)}_{2,2}(0^-,0)]\widehat{\mathds{X}}'_n(\theta'',1)-[\widehat \phi(\hat{F}^{(1)}_{Y|X}(\cdot|0^+))(\cdot)-\widehat \phi(\hat{F}^{(1)}_{Y|X}(\cdot|0^-))(\cdot)]
\widehat{\mathds{X}}'_n(\theta'',2)}{[\hat{\mu}^{(1)}_{2,2}(0^+,0)-\hat{\mu}^{(1)}_{2,2}(0^-,0)]^2}}
$$
to simulate its asymptotic distribution.

Another one of the practically most relevant applications of the above corollary is the test of the null hypothesis of treatment homogeneity across quantiles:
$$
H_0: \tau_{SQRK}(\theta'') = \tau_{SQRK}(\theta''') \quad\text{for all } \theta'', \theta''' \in \Theta'' = [a,1-a].
$$
To test this hypothesis, we can use $\sup_{\theta'' \in [a,1-a]} \sqrt{nh_n^3} \abs{ \hat\tau_{SQRK}(\theta'') - (1-2a)^{-1} \int_{[a,1-a]} \hat\tau_{SQRK}(\theta''')d\theta''' }$ as the test statistic, and use
\begin{align*}
&\sup_{\theta'' \in [a,1-a]}
\left\vert \frac{[\hat{\mu}^{(1)}_{2,2}(0^+,0)-\hat{\mu}^{(1)}_{2,2}(0^-,0)]\widehat{\mathds{X}}'_n(\theta'',1)-[\widehat \phi(\hat{F}^{(1)}_{Y|X}(\cdot|0^+))(\cdot)-\widehat \phi(\hat{F}^{(1)}_{Y|X}(\cdot|0^-))(\cdot)]
\widehat{\mathds{X}}'_n(\theta'',2)}{[\hat{\mu}^{(1)}_{2,2}(0^+,0)-\hat{\mu}^{(1)}_{2,2}(0^-,0)]^2} \right.
\\
&\left. -
\frac{1}{1-2a} \int_{[a,1-a]} \frac{[\hat{\mu}^{(1)}_{2,2}(0^+,0)-\hat{\mu}^{(1)}_{2,2}(0^-,0)]\widehat{\mathds{X}}'_n(\theta''',1)-[\widehat \phi(\hat{F}^{(1)}_{Y|X}(\cdot|0^+))(\cdot)-\widehat \phi(\hat{F}^{(1)}_{Y|X}(\cdot|0^-))(\cdot)]
\widehat{\mathds{X}}'_n(\theta''',2)}{[\hat{\mu}^{(1)}_{2,2}(0^+,0)-\hat{\mu}^{(1)}_{2,2}(0^-,0)]^2} d\theta'''
\right\vert
\end{align*}
to simulate its asymptotic distribution.

\subsection{Example: Sharp Quantile RKD}\label{sec:a:SQRK}

Consider $\Theta_1$, $\Theta_2$, $\Theta_1'$, $\Theta_2'$, $\Theta''$, $g_1$, $g_2$, $\phi$, $\psi$, and $\Upsilon$ defined in Section \ref{sec:ex:sharp_quantile_rkd}.
Recall that the operator $\phi$ is defined by
\begin{equation}\label{eq:operator_phi_qrkd}
\phi( F^{(1)}_{Y|X}(\cdot|x))(\theta'')=-\frac{F^{(1)}_{Y|X}(Q_{Y|X}(\theta''|x)|x)}{f_{Y|X}(Q_{Y|X}(\theta''|x)|x)}
\end{equation}
where  $F^{(1)}_{Y|X}(y|x)=\frac{\partial}{\partial x}F(y|x)$ and $Q_{Y|X}(\theta'|x)=\inf\{\theta\in \Theta_1:F_{Y|X}(\theta|x)\ge \theta'\}$.
Also recall that the local Wald estimand (\ref{eq:wald_est}) with $v=1$ in this setting is denoted by $\tau_{SQRK}$.
We denote the analog `intermediate' estimator (\ref{eq:wald_estimator}) with $v=1$ in this setting by
\begin{align*}
\tilde{\tau}_{SQRK}(\theta'')=&\frac{\phi(\hat{F}^{(1)}_{Y|X}(\cdot|0^+))(\theta'')-\phi(\hat{F}^{(1)}_{Y|X}(\cdot|0^-))(\theta'')}{b^{(1)}(0^+)-b^{(1)}(0^-)}
\end{align*}
for $\theta''\in [a,1-a]\subset (0,1)$. Note that $\hat{F}^{(1)}_{Y|X}$ denotes for the slope estimate from a local polynomial CDF estimation.
This is intermediate because the operator $\phi$ contains unknowns, $f_{Y|X}(\cdot|x)$ and $Q_{Y|X}(\cdot|x)$.
In practice, we need to also estimate this operator $\phi$ by replacing these unknowns by uniformly consistent estimators.
Thus, a feasible analog estimator is denoted by
\begin{align*}
\hat{\tau}_{SQRK}(\theta''):=&\frac{\widehat\phi(\hat{F}^{(1)}_{Y|X}(\cdot|0^+))(\theta'')-\widehat\phi(\hat{F}^{(1)}_{Y|X}(\cdot|0^-))(\theta'')}{b^{(1)}(0^+)-b^{(1)}(0^-)}
\end{align*}
where
\begin{align*}
\widehat \phi(\hat F^{(1)}_{Y|X}(\cdot|0^\pm))(\theta'')=-\frac{\hat F^{(1)}_{Y|X}(\hat Q_{Y|X}(\theta''|0)|0^\pm)}{\hat{f}_{Y|X}(\hat Q_{Y|X}(\theta''|0)|0)}
\end{align*}
for $\hat Q_{Y|X}(\theta''|x)=\inf\{\theta \in \Theta_1: \hat F_{Y|x}(\theta|x)\ge \theta''\}$ and
 $
\widehat \phi(F^{(1)}_{Y|X}(\cdot|x))(\theta''):=-\frac{F^{(1)}_{Y|X}(\hat Q_{Y|X}(\theta''|x)|x)}{\hat{f}_{Y|X}(\hat Q_{Y|X}(\theta''|x)|x)}
$.
For this application, we consider the following set of assumptions.

\newtheorem*{assumption_SQRK}{Assumption SQRK}
\begin{assumption_SQRK}\label{a:SQRK}
\qquad\\
%
(i) (a) $\frac{\partial^j}{\partial x^j}F_{Y|X}$ is Lipschitz in $x$ on $\mathscr{Y}_1\times [\underline{x},0)$ and $\mathscr{Y}_1\times (0,\overline{x}]$ for $j=0,1,2,3$.
(b) $f_{Y|X}$ is Lipschitz in $x$ and $0<C<f_{Y|X}(y|x)<C'<\infty$ on $\mathscr{Y}_1\times [\underline{x},\overline{x}]$. \\
(ii) $h_n$ satisfies $h_n\to 0$, $nh^7_n\to 0$, and $nh^3_n\to \infty$.\\
%
%
%
(iii) $b^{(1)}$ is continuous on $[\underline{x},\overline{x}] \backslash \{0\}$ and $\lim_{x \downarrow 0}b^{(1)}(x)\ne\lim_{x \uparrow 0}b^{(1)}(x)$.\\
(iv) $\sup_{y\in \mathscr{Y}_1}|\sqrt{nh^3_n}[\hat{f}_{Y|X}(y|0)-f_{Y|X}(y|0)]|\underset{x}{\overset{p}{\to}}0$
\end{assumption_SQRK}

We state part (iv) of this assumption at this high-level in order to accommodate a number of alternative estimators.
Above Lemma \ref{lemma:super_consist_SQRK} in Section \ref{sec:lemma:unif_cons_est}, we propose one particular such estimator which satisfies (iv).

We defined the operator $\phi$ for $\tau_{SQRK}$ with $\phi$ in (\ref{eq:operator_phi_qrkd}) without a formal justification.
Now that Assumption SQRK is stated, we can now provide the following lemma for a justification.

\begin{lemma}\label{lemma:quantile_deriv}
Suppose that Assumptions S and SQRK (i) hold.
Then, we have
\begin{align*}
\frac{\partial }{\partial x}Q_{Y|X}(\theta''|0^\pm)=-\frac{ F^{(1)}_{Y|X}(Q_{Y|X}(\theta''|0)|0^\pm)}{f_{Y|X}(Q_{Y|X}(\theta''|0)|0)}=\phi(F^{(1)}_{Y|X}(\cdot|0^\pm)).
\end{align*}
\end{lemma}

To state the result of this subsection, define the following objects.
\begin{align*}
&\phi'_{F_{Y|X}(\cdot|0^\pm)}(\hat{\nu}^\pm_{\xi,n})(\theta'')=-\frac{\hat{\nu}^\pm_{\xi,n}(Q_{Y|X}(\theta''|0))}{f_{Y|X}(Q_{Y|X}(\theta''|0)|0)}\\
& \qquad  =-\sum_{i=1}^{n}\xi_i\frac{e'_1(\Gamma^\pm_2)^{-1}[\mathds{1}\{Y_i\le Q_{Y|X}(\theta''|0)\}-\tilde{F}_{Y|X}(Q_{Y|X}(\theta''|0)|X_i)]r_2(\frac{X_i}{h_n})K(\frac{X_i}{h_n})\delta^\pm_i}{\sqrt{nh_{n}}\hat{f}_X(0)f_{Y|X}( Q_{Y|X}(\theta''|0)|0)}\\
&\widehat{\phi}'_{F_{Y|X}(\cdot|0^\pm)}(\hat{\nu}^\pm_{\xi,n})(\theta'')=-\frac{\hat{\nu}^\pm_{\xi,n}(\hat{Q}_{Y|X}(\theta''|0))}{\hat{f}_{Y|X}(\hat{Q}_{Y|X}(\theta''|0)|0)}\\
& \qquad =-\sum_{i=1}^{n}\xi_i\frac{e'_1(\Gamma^\pm_2)^{-1}[\mathds{1}\{Y_i\le \hat{Q}_{Y|X}(\theta''|0)\}-\tilde{F}_{Y|X}(\hat{Q}_{Y|X}(\theta''|0)|X_i)]r_2(\frac{X_i}{h_n})K(\frac{X_i}{h_n})\delta^\pm_i}{\sqrt{nh_{n}}\hat{f}_X(0)\hat{f}_{Y|X}(\hat Q_{Y|X}(\theta''|0)|0)}
\end{align*}
where $\tilde{F}(y|x)=\tilde{\mu}_{1,2}(x,y)\mathds{1}\{|x/h_n|\le 1\}$ with $\tilde{\mu}_{1,2}$ defined in the statement of Lemma \ref{lemma:epsilon_unif_const}.
Our general result applied to the current case yields the following corollary.

\begin{corollary}[Example: Sharp Quantile RKD]\label{corollary:SQRK}
Suppose that Assumptions S, K, M, and SQRK hold.
\\
(i) There exists a zero mean Gaussian process $\mathds{G}'_{SQRK}:\Omega^x \mapsto\ell^\infty([a,1-a])$ such that
$$
\sqrt{nh^3_n}[\tilde{\tau}_{SQRK}-\tau_{SQRK}]\leadsto \mathds{G}'_{SQRK},
$$
and thus
$$
\sqrt{nh^3_n}[\hat{\tau}_{SQRK}-\tau_{SQRK}]\leadsto \mathds{G}'_{SQRK}.
$$
(ii) Furthermore, with probability approaching one,
\begin{align*}
\frac{\phi'_{F_{Y|X}(\cdot|0^+)}( \hat{\nu}^+_{\xi,n})-\phi'_{F_{Y|X}(\cdot|0^-)}( \hat{\nu}^-_{\xi,n})}{b^{(1)}(0^+)-b^{(1)}(0^-)}\underset{\xi}{\overset{p}{\leadsto}}\mathds{G}'_{SQRK},
\end{align*}
and thus
\begin{align*}
\frac{\widehat{\phi}'_{F_{Y|X}(\cdot|0^+)}( \hat{\nu}^+_{\xi,n})-\widehat{\phi}'_{F_{Y|X}(\cdot|0^-)}( \hat{\nu}^-_{\xi,n})}{b^{(1)}(0^+)-b^{(1)}(0^-)}\underset{\xi}{\overset{p}{\leadsto}}\mathds{G}'_{SQRK}.
\end{align*}
\end{corollary}

A proof is provided in Section \ref{sec:corollary:SQRK}.
One of the practically most relevant applications of this corollary is the test of the null hypothesis of uniform treatment nullity:
$$
H_0: \tau_{SQRK}(\theta'') = 0 \quad\text{for all } \theta'' \in \Theta'' = [a,1-a].
$$
To test this hypothesis, we can use $\sup_{\theta'' \in [a,1-a]} \sqrt{nh_n^3} \abs{ \hat\tau_{SQRK}(\theta'') }$ as the test statistic, and use
$$
\sup_{\theta'' \in [a,1-a]} \abs{\frac{\widehat{\phi}'_{F_{Y|X}(\cdot|0^+)}( \hat{\nu}^+_{\xi,n})(\theta'')-\widehat{\phi}'_{F_{Y|X}(\cdot|0^-)}( \hat{\nu}^-_{\xi,n})(\theta'')}{b^{(1)}(0^+)-b^{(1)}(0^-)}}
$$
to simulate its asymptotic distribution.

Another one of the practically most relevant applications of the above corollary is the test of the null hypothesis of treatment homogeneity across quantiles:
$$
H_0: \tau_{SQRK}(\theta'') = \tau_{SQRK}(\theta''') \quad\text{for all } \theta'', \theta''' \in \Theta'' = [a,1-a].
$$
To test this hypothesis, we can use $\sup_{\theta'' \in [a,1-a]} \sqrt{nh_n^3} \abs{ \hat\tau_{SQRK}(\theta'') - (1-2a)^{-1} \int_{[a,1-a]} \hat\tau_{SQRK}(\theta''')d\theta''' }$ as the test statistic, and use
\begin{align*}
\sup_{\theta'' \in [a,1-a]}
&\left\vert \frac{\widehat{\phi}'_{F_{Y|X}(\cdot|0^+)}( \hat{\nu}^+_{\xi,n})(\theta'')-\widehat{\phi}'_{F_{Y|X}(\cdot|0^-)}( \hat{\nu}^-_{\xi,n})(\theta'')}{b^{(1)}(0^+)-b^{(1)}(0^-)} \right.
\\
&\left. -
\frac{1}{1-2a} \int_{[a,1-a]} \frac{\widehat{\phi}'_{F_{Y|X}(\cdot|0^+)}( \hat{\nu}^+_{\xi,n})(\theta''')-\widehat{\phi}'_{F_{Y|X}(\cdot|0^-)}( \hat{\nu}^-_{\xi,n})(\theta''')}{b^{(1)}(0^+)-b^{(1)}(0^-)} d\theta'''
\right\vert
\end{align*}
to simulate its asymptotic distribution.

\subsection{Example: Group Covariate and Test of Heterogeneous Treatment Effects}\label{sec:a:GFMRD}

Consider $\Theta_1$, $\Theta_2$, $\Theta_1'$, $\Theta_2'$, $\Theta''$, $g_1$, $g_2$, $\phi$, $\psi$, and $\Upsilon$ defined in Section \ref{sec:ex:group}.
Recall that we denote the local Wald estimand (\ref{eq:wald_est}) with $v=1$ in this setting by $\tau_{GFMRD}$.
We also denote the analog estimator (\ref{eq:wald_estimator}) with $v=1$ in this setting by $\hat\tau_{GFMRD}$.
\begin{align*}
\hat\tau_{GFMRD}(\theta'')=\frac{\hat{\mu}_{1,2}(0^+,\theta'')-\hat{\mu}_{1,2}(0^-,\theta'')}{\hat{\mu}_{2,2}(0^+,\theta'')-\hat{\mu}_{2,2}(0^-,\theta'')}
\end{align*}
For this application, we consider the following set of assumptions.

\newtheorem*{assumption_GFMRD}{Assumption GFMRD}
\begin{assumption_GFMRD}\label{a:GFMRD}
\qquad\\
%
%
(i) (a) $E[|Y|^{\ast 2+\epsilon} \cdot \mathbbm{1}\{G={\theta''}\} |X=\cdot]<\infty$ 
on $[\underline{x},\overline{x}] \backslash \{0\}$ for some $\epsilon>0$ for ${\theta''} \in \{1,...,K\}$.
(b) $\frac{\partial^j}{\partial x^j}E[Y^\ast  \cdot \mathbbm{1}\{G={\theta''}\} |X=\cdot]<C$ and $\frac{\partial^j}{\partial x^j}E[D^\ast  \cdot \mathbbm{1}\{G={\theta''}\} |X=\cdot]$ are Lipschitz on $[\underline{x},0)$ and $(0,\overline{x}]$ for $j=0,1,2,3$.
(c) $E[D^\ast  \cdot \mathbbm{1}\{G={\theta''}\} |X=0^+]\ne E[D^\ast  \cdot \mathbbm{1}\{G={\theta''}\} |X=0^-]$ for ${\theta''} \in \{1,...,K\}$. (d) $E[\mathds{1}\{G=\theta''\}|X=\cdot]>0$ for all $\theta''\in \Theta''$ \\
(ii) The baseline bandwidth $h_n$ satisfies $h_n\to 0$, $nh^2_n\to \infty$, $nh^7_n\to 0$.
There exist functions $c_1$, $c_2 :\{1,...,K\}\to[\underline{c},\overline{c}]\subset(0,\infty)$ such that $h_{1,n}({\theta''}_1)=c_1({\theta''}_1)h_n$ and $h_{2,n}({\theta''}_2)=c_2({\theta''}_2) h_n$.\\
%
%
(iii) $V(Y^\ast  \cdot \mathbbm{1}\{G={\theta''}\} |X=x)$, $V(D^\ast  \cdot \mathbbm{1}\{G={\theta''}\} |X=x)\in \mathcal{C}^{1}([\underline{x},\overline{x}]\setminus\{0\})$ and $0<V(Y^\ast  \cdot \mathbbm{1}\{G={\theta''}\} |X=0^\pm)$ with derivatives in $x$ all bounded on $[\underline{x},\overline{x}]\setminus\{0\}$. Furthermore, $0<V(D^\ast  \cdot \mathbbm{1}\{G={\theta''}\} |X=0^\pm)$ for ${\theta''}\in \{1,...,K\}$. \\
%
\end{assumption_GFMRD}

For $k \in \{1,2\}$ and ${\theta''}\in \{1,...,K\}$, define
\begin{align*}
&\widehat{\mathds{X}}'_n({\theta''},k)=\frac{1}{\sqrt{c_k}}[ \hat{\nu}^+_{\xi,n}({\theta''},k)-  \hat{\nu}^-_{\xi,n}({\theta''},k)],
\end{align*}
where the EMP is given by
\begin{align*}
&\hat{\nu}^\pm_{\xi,n}(\theta''_1,1)=\sum_{i=1}^{n}\xi_i\frac{e'_0(\Gamma^\pm_2)^{-1}[Y_i-\tilde \mu_{1,2}(X_i,\theta''_1)] r_2(\frac{X_i}{h_{1,n}(\theta''_1)}) K(\frac{X_i}{h_{1,n}(\theta''_1)})}{\sqrt{nh_{1,n}(\theta''_1)}\hat{f}_X(0)}\\
&\hat{\nu}^\pm_{\xi,n}(\theta''_2,2)=\sum_{i=1}^{n}\xi_i\frac{e'_0(\Gamma^\pm_2)^{-1}[D_i-\tilde \mu_{2,2}(X_i,\theta''_2)]  r_2(\frac{X_i}{h_{2,n}(\theta''_2)}) K(\frac{X_i}{h_{2,n}(\theta''_2)})}{\sqrt{nh_{2,n}(\theta''_2)}\hat{f}_X(0)}
\end{align*}
and $\tilde \mu_{k,2}(0^\pm,\theta''_k)$ are defined in Lemma \ref{lemma:epsilon_unif_const}.
Our general result applied to the current case yields the following corollary.

\begin{corollary}[Example: Group Covariate]\label{corollary:GFMRD}
Suppose that Assumptions S, K, M, and GFMRD hold.
\\
(i) There exists a symmetric positive definite $K$-by-$K$ matrix $\Sigma_{GFMRD}$ such that
\begin{align*}
&\sqrt{nh_n}[\hat{\tau}_{GFMRD}-\tau_{GFMRD}]\leadsto N(0,\Sigma_{GFMRD})
\end{align*}
(ii)
Furthermore, with probability approaching one,
\begin{align*}
&\frac{(\hat{\mu}_{2,2} (0^+,\cdot)-\hat{\mu}_{2,2} (0^-,\cdot))\widehat{\mathds{X}}'_n(\cdot,1)-(\hat{\mu}_{1,2}(0^+,\cdot)-\hat{\mu}_{1,2}(0^-,\cdot))\widehat{\mathds{X}}'_n(\cdot,2)}{(\hat{\mu}_{2,2} (0^+,\cdot)-\hat{\mu}_{2,2} (0^-,\cdot))^2}\underset{\xi}{\overset{p}{\leadsto}}N(0,\Sigma_{GFMRD}).
\end{align*}
\end{corollary}

A proof is provided in Section \ref{sec:corollary:GFMRD}.
One of the practically most relevant applications of this corollary is the test of the null hypothesis of joint treatment nullity:
$$
H_0: \tau_{FMRD}(\theta'') = 0 \quad\text{for all } \theta'' \in \{1,...,K\}.
$$
To test this hypothesis, we can use $\max_{\theta'' \in \{1,...,K\}} \sqrt{nh_n} \abs{ \hat\tau_{GFMRD}(\theta'') }$ as the test statistic, and use
$$
\max_{\theta'' \in \{1,...,K\}} \abs{\frac{(\hat{\mu}_{2,2} (0^+,\theta'')-\hat{\mu}_{2,2} (0^-,\theta''))\widehat{\mathds{X}}'_n(\theta'',1)-(\hat{\mu}_{1,2}(0^+,\theta'')-\hat{\mu}_{1,2}(0^-,\theta''))\widehat{\mathds{X}}'_n(\theta'',2)}{(\hat{\mu}_{2,2} (0^+,\theta'')-\hat{\mu}_{2,2} (0^-,\theta''))^2}}
$$
to simulate its asymptotic distribution.

Another one of the practically most relevant applications of the above corollary is the test of the null hypothesis of treatment homogeneity across the covariate-index groups:
$$
H_0: \tau_{FMRD}(\theta'') = \tau_{FMRD}(\theta''') \quad\text{for all } \theta'', \theta''' \in \{1,...,K\}.
$$
To test this hypothesis, we can use $\max_{\theta'' \in \{1,...,K\}} \sqrt{nh_n} \abs{ \hat\tau_{GFMRD}(\theta'') - K^{-1} \sum_{\theta'''=1}^K \hat\tau_{GFMRD}(\theta''') }$ as the test statistic, and use
\begin{align*}
\max_{\theta'' \in \{1,...,K\}}
&\left\vert \frac{(\hat{\mu}_{2,2} (0^+,\theta'')-\hat{\mu}_{2,2} (0^-,\theta''))\widehat{\mathds{X}}'_n(\theta'',1)-(\hat{\mu}_{1,2}(0^+,\theta'')-\hat{\mu}_{1,2}(0^-,\theta''))\widehat{\mathds{X}}'_n(\theta'',2)}{(\hat{\mu}_{2,2} (0^+,\theta'')-\hat{\mu}_{2,2} (0^-,\theta''))^2} \right.
\\
&\left. -
\frac{1}{K} \sum_{\theta''' = 1}^K \frac{(\hat{\mu}_{2,2} (0^+,\theta''')-\hat{\mu}_{2,2} (0^-,\theta'''))\widehat{\mathds{X}}'_n(\theta''',1)-(\hat{\mu}_{1,2}(0^+,\theta''')-\hat{\mu}_{1,2}(0^-,\theta'''))\widehat{\mathds{X}}'_n(\theta''',2)}{(\hat{\mu}_{2,2} (0^+,\theta''')-\hat{\mu}_{2,2} (0^-,\theta'''))^2}
\right\vert
\end{align*}
to simulate its asymptotic distribution.

\section{Additional Simulation Studies}\label{sec:simulation}

We conduct simulation studies to demonstrate the robustness and unified applicability of our general multiplier bootstrap method.
Each of the ten examples covered in Example \ref{ex:fuzzy_rdd} as well as Sections \ref{sec:ex:sharp_rdd}--\ref{sec:ex:group} and Sections \ref{sec:a:FMRD}--\ref{sec:a:GFMRD} are tested, except for the example of CDF discontinuity -- this example is omitted because it is a less complicated form of the sharp Quantile RDD without CDF inversions.
The concrete bootstrap procedures outlined in Section \ref{sec:application_examples} are used in the respective subsections below.
We follow the procedure outlined in Section \ref{sec:practical_guideline} for choices of bandwidths in finite sample.
The kernel function that we use is the Epanechnikov kernel.
Other details are discussed in each example subsection below.

\subsection{Example: Fuzzy Mean RDD}

Consider the case of fuzzy RDD presented in Example \ref{ex:fuzzy_rdd} and Section \ref{sec:a:FMRD}.
We generate an i.i.d. sample $\{(Y_i,D_i,X_i)\}_{i=1}^n$ through the following data generating process:
\begin{align*}
&Y_i = \alpha_0 + \alpha_1 X_i + \alpha_2 X_i^2 + \beta_1 D_i + U_i,
\\
&D_i = \mathbbm{1}\{ 2 \cdot \mathbbm{1}\{ X_i \geq 0\} - 1 \geq V_i \},
\\
&(X_i,U_i,V_i)' \sim N(0,\Sigma),
\end{align*}
where $\alpha_0=1.00$, $\alpha_1=0.10$, $\alpha_2=0.01$, $\beta_1$ is to be varied across simulation sets, $\Sigma_{11} = \sigma_X^2 = 1.0^2$, $\Sigma_{22} = \sigma_U^2 = 1.0^2$, $\Sigma_{33} = \sigma_V^2 = 0.5^2$, $\Sigma_{12} = \rho_{XU} \cdot \sigma_X \cdot \sigma_U = 0.5 \cdot 1.0^2$, $\Sigma_{13} = \rho_{XV} \cdot \sigma_X \cdot \sigma_V = 0.0 \cdot 1.0 \cdot 0.5$, and $\Sigma_{23} = \rho_{UV} \cdot \sigma_U \cdot \sigma_V = 0.5 \cdot 1.0 \cdot 0.5$.
In this setup, we have
$$
\text{Treatment Effect} = \beta_1.
$$

We simulate the 95\% test for the null hypothesis $H_0: \tau_{FMRD} = 0$ of treatment nullity using the procedure described in  Section \ref{sec:a:FMRD}.
Table \ref{tab:FMRD} shows simulated acceptance probabilities based on 2,500 multiplier bootstrap replications for 2,500 Monte Carlo replications for each of the sample sizes $n=1,000$, $2,000$, and $4,000$.
The results, exhibiting the same qualitative features as those in the previous subsections, evidence the power as well as the size correctness.

\subsection{Example: Sharp Mean RDD}

Consider the case of sharp RDD presented in Sections \ref{sec:ex:sharp_rdd} and \ref{sec:a:SMRD}.
We generate an i.i.d. sample $\{(Y_i,D_i,X_i)\}_{i=1}^n$ through the following data generating process:
\begin{align*}
&Y_i = \alpha_0 + \alpha_1 X_i + \alpha_2 X_i^2 + \beta_1 D_i + U_i,
\\
&D_i = \mathbbm{1}\{ X_i \geq 0\},
\\
&(X_i,U_i)' \sim N(0,\Sigma),
\end{align*}
where $\alpha_0=1.00$, $\alpha_1=0.10$, $\alpha_2=0.01$, $\beta_1$ is to be varied across simulation sets, $\Sigma_{11} = \sigma_X^2 = 1.0^2$, $\Sigma_{22} = \sigma_U^2 = 1.0^2$, and $\Sigma_{12} = \rho_{XU} \cdot \sigma_X \cdot \sigma_U = 0.5 \cdot 1.0^2$.
In this setup, we have
$$
\text{Treatment Effect} = \beta_1.
$$

We simulate the 95\% test for the null hypothesis $H_0: \tau_{SMRD} = 0$ of treatment nullity using the procedure described in  Section \ref{sec:a:SMRD}.
Table \ref{tab:SMRD} shows simulated acceptance probabilities based on 2,500 multiplier bootstrap replications for 2,500 Monte Carlo replications for each of the sample sizes $n=1,000$, $2,000$, and $4,000$.
The first column under $\beta_1 = 0.00$ shows that simulated acceptance probabilities are close to the designed nominal probability, 95\%.
The next four columns show that the acceptance probability decreases in $\beta_1$, and the rate of decrease is higher for the larger sample sizes.
These results evidence the power as well as the size correctness.

In addition to the above data generating process, we also consider the data generating processes employed by Calonico, Cattaneo and Titiunik (2014) for the sake of comparisons under the sharp mean RDD.
An i.i.d. sample of size $n=500$ is generated in the following manner:
\begin{align*}
&Y_i = \mu(X_i) + U_i,
\\
&D_i = \mathbbm{1}\{ X_i \geq 0\},
\\
&X_i \sim 2\mathcal{B}(2,4)-1,
\\
&U_i \sim N(0,0.1295^2),
\end{align*}
where $\mathcal{B}(2,4)$ denotes the beta distribution with parameters $(2,4)$.
Two models of $\mu$ are considered.
The first one is due to Lee (2008):
\begin{align*}
\mu(x) = \begin{cases}
0.48+1.27x+7.18x^2+20.21x^3+21.54x^4+7.33x^5 & \text{if } x < 0\\
0.52+0.84x-3.00x^2+7.99x^3-9.01x^4+3.56x^5   & \text{if } x \geq 0\\
\end{cases}
\end{align*}
The second one is due to Ludwig and Miller (2007):
\begin{align*}
\mu(x) = \begin{cases}
3.71+2.30x+3.28x^2+1.45x^3+0.23x^4+0.03x^5     & \text{if } x < 0\\
0.26+18.49x-54.81x^2+74.30x^3-45.02x^4+9.83x^5 & \text{if } x \geq 0\\
\end{cases}
\end{align*}

We simulate the coverage probability of the true treatment effect for the nominal probability of 95\%.
Table \ref{tab:SMRD_CCT} shows simulated coverage probabilities based on 2,500 multiplier bootstrap replications for 5,000 Monte Carlo replications -- we run 5,000 iterations to have our results comparable with those of Calonico, Cattaneo and Titiunik (2014) who also ran 5,000 iterations.
The first two columns indicate coverage probabilities under the conventional non-robust approach with the the fixed-neighborhood standard error estimators (FN) and the plug-in residual standard error estimators (PI), copied from Table I of Calonico, Cattaneo and Titiunik (2014).
The next two columns indicate coverage probabilities under the robust approach with the the fixed-neighborhood standard error estimators (FN) and the plug-in residual standard error estimators (PI), copied from Table I of Calonico, Cattaneo and Titiunik (2014).
Finally, the last column report the coverage probabilities under the robust multiplier bootstrap approach.
The numbers are similar to each other.
As a robust approach, the robust MB is certainly better than the conventional approaches.
MB is slightly better than PI, but is slightly worse than FN.
As such, for cases like the sharp mean RDD for which the existing FN method is available, we recommend that the users employ FN.

\subsection{Example: Fuzzy Mean RKD}

Consider the case of fuzzy RKD presented in Sections \ref{sec:ex:fuzzy_rkd} and \ref{sec:a:FMRK}.
We generate an i.i.d. sample $\{(Y_i,D_i,X_i)\}_{i=1}^n$ through the following data generating process:
\begin{align*}
&Y_i = \alpha_0 + \alpha_1 X_i + \alpha_2 X_i^2 + \beta_1 D_i + U_i,
\\
&D_i = X_i \cdot \left( 2 \cdot \mathbbm{1}\{ X_i \geq 0 \} - 1 \right) + V_i,
\\
&(X_i,U_i,V_i)' \sim N(0,\Sigma),
\end{align*}
where $\alpha_0=1.00$, $\alpha_1=0.10$, $\alpha_2=0.01$, $\beta_1$ is to be varied across simulation sets, $\Sigma_{11} = \sigma_X^2 = 1.0^2$, $\Sigma_{22} = \sigma_U^2 = 1.0^2$, $\Sigma_{33} = \sigma_V^2 = 0.1^2$, $\Sigma_{12} = \rho_{XU} \cdot \sigma_X \cdot \sigma_U = 0.5 \cdot 1.0^2$, $\Sigma_{13} = \rho_{XV} \cdot \sigma_X \cdot \sigma_V = 0.0 \cdot 1.0 \cdot 0.1$, and $\Sigma_{23} = \rho_{UV} \cdot \sigma_U \cdot \sigma_V = 0.5 \cdot 1.0 \cdot 0.1$.
In this setup, we have
$$
\text{Treatment Effect} = \beta_1.
$$

We simulate the 95\% test for the null hypothesis $H_0: \tau_{FMRK} = 0$ of treatment nullity using the procedure described in  Section \ref{sec:a:FMRK}.
Table \ref{tab:FMRK} shows simulated acceptance probabilities based on 2,500 multiplier bootstrap replications for 2,500 Monte Carlo replications for each of the sample sizes $n=1,000$, $2,000$, and $4,000$.
The results, exhibiting the same qualitative features as those in the previous subsections, evidence the power as well as the size correctness.

\subsection{Example: Sharp Mean RKD}

Consider the case of sharp RKD presented in Sections \ref{sec:ex:sharp_rkd} and \ref{sec:a:SMRK}.
We generate an i.i.d. sample $\{(Y_i,D_i,X_i)\}_{i=1}^n$ through the following data generating process:
\begin{align*}
&Y_i = \alpha_0 + \alpha_1 X_i + \alpha_2 X_i^2 + \beta_1 D_i + U_i,
\\
&D_i = X_i \cdot \left( 2 \cdot \mathbbm{1}\{ X_i \geq 0 \} - 1 \right),
\\
&(X_i,U_i)' \sim N(0,\Sigma),
\end{align*}
where $\alpha_0=1.00$, $\alpha_1=0.10$, $\alpha_2=0.01$, $\beta_1$ is to be varied across simulation sets, $\Sigma_{11} = \sigma_X^2 = 1.0^2$, $\Sigma_{22} = \sigma_U^2 = 1.0^2$, and $\Sigma_{12} = \rho_{XU} \cdot \sigma_X \cdot \sigma_U = 0.5 \cdot 1.0^2$.
In this setup, we have
$$
\text{Treatment Effect} = \beta_1.
$$

We simulate the 95\% test for the null hypothesis $H_0: \tau_{SMRK} = 0$ of treatment nullity using the procedure described in  Section \ref{sec:a:SMRK}.
Table \ref{tab:SMRK} shows simulated acceptance probabilities based on 2,500 multiplier bootstrap replications for 2,500 Monte Carlo replications for each of the sample sizes $n=1,000$, $2,000$, and $4,000$.
The results, exhibiting the same qualitative features as those in the previous subsection, evidence the power as well as the size correctness.


\subsection{Example: Sharp Quantile RDD}

Consider the case of sharp quantile RDD presented in Sections \ref{sec:ex:sharp_quantile_rdd} and \ref{sec:a:SQRD}.
We generate an i.i.d. sample $\{(Y_i,D_i,X_i)\}_{i=1}^n$ through the following data generating process:
\begin{align*}
&Y_i = \alpha_0 + \alpha_1 X_i + \alpha_2 X_i^2 + \beta_1 D_i + (\gamma_0 + \gamma_1 D_i) \cdot U_i,
\\
&D_i = \mathbbm{1}\{ X_i \geq 0\},
\\
&(X_i,U_i)' \sim N(0,\Sigma),
\end{align*}
where $\alpha_0=1.00$, $\alpha_1=0.10$, $\alpha_2=0.01$, $\beta_1$ is to be varied across simulation sets, $\gamma_0 = 1$, $\gamma_1$ is to be varied across simulation sets, $\Sigma_{11} = \sigma_X^2 = 1.0^2$, $\Sigma_{22} = \sigma_U^2 = 1.0^2$, and $\Sigma_{12} = \rho_{XU} \cdot \sigma_X \cdot \sigma_U = 0.5 \cdot 1.0^2$.
In this setup, we have
$$
\text{$\theta$-th Conditional Quantile Treatment Effect at $x=0$} = \beta_1 + \gamma_1 F_{U|X}^{-1}(\theta | 0).
$$

We set $\Theta'' = [a,1-a] = [0.20,0.80]$ as the set of quantiles on which we conduct inference.
We use a grid with the interval size of 0.02 to approximate the continuum $\Theta''$ for numerical evaluation of functions defined on $\Theta''$.
First, we simulate the 95\% test for the null hypothesis $H_0: \tau_{SQRD}(\theta'') = 0 \ \forall\theta'' \in [a,1-a]$ of uniform treatment nullity using the procedure described in Section \ref{sec:a:SQRD}.
Next, we simulate the 95\% test for the null hypothesis $H_0: \tau_{SQRD}(\theta'') = \tau_{SQRD}(\theta''') \ \forall \theta'',\theta''' \in [a,1-a]$ of treatment homogeneity using the procedure described in  Section \ref{sec:a:SQRD}.

Table \ref{tab:SQRD} show simulated acceptance probabilities based on 2,500 multiplier bootstrap replications for 2,500 Monte Carlo replications for each of the sample sizes $n=1,000$, $2,000$, and $4,000$.
Part (A) reports results for the test of uniform treatment nullity and part (B) shows results for the test of treatment homogeneity.
The top panel (I) presents results across alternative values of $\beta_1 \in \{0.00, 0.25, 0.50, 0.75, 1.00\}$ while fixing $\gamma_1 = 0$.
The bottom panel (II) presents results across alternative values of $\gamma_1 \in \{0.00, 0.25, 0.50, 0.75, 1.00\}$ while fixing $\beta_1 = 0$. The nominal acceptance probability is 95\%.

The first column of each part of the table shows that simulated acceptance probabilities are close to the designed nominal probability, 95\%.
The next four columns in parts (I) (A) and (II) (A) of the table show that the acceptance probability decreases in $\beta_1$ and $\gamma_1$, respectively, and the rate of decrease is higher for the larger sample sizes.
These results evidence the power as well as the size correctness for the test of uniform treatment nullity.
In part (I) (B), all the simulated acceptance probabilities are close to the designed nominal probability, 95\%.
This is consistent with the fact that $\beta_1$ does not contribute to treatment heterogeneity.
On the other hand, part (II) (A) of the table shows that the acceptance probability decreases in $\gamma_1$, and the rate of decrease is higher for the larger sample sizes.
These results evidence the power as well as the size correctness for the test of treatment homogeneity.

\subsection{Example: Fuzzy Quantile RKD}

Consider the case of fuzzy quantile RKD presented in Sections \ref{sec:ex:fuzzy_quantile_rkd} and \ref{sec:a:FQRK}.
We generate an i.i.d. sample $\{(Y_i,D_i,X_i)\}_{i=1}^n$ through the following data generating process:
\begin{align*}
&Y_i = \alpha_0 + \alpha_1 X_i + \alpha_2 X_i^2 + \beta_1 D_i + (\gamma_0 + \gamma_1 D_i) \cdot U_i,
\\
&D_i = X_i \cdot \left( 2 \cdot \mathbbm{1}\{ X_i \geq 0 \} - 1 \right) + V_i,
\\
&(X_i,U_i,V_i)' \sim N(0,\Sigma),
\end{align*}
where $\alpha_0=1.00$, $\alpha_1=0.10$, $\alpha_2=0.01$, $\beta_1$ is to be varied across simulation sets, $\gamma_0 = 1$, $\gamma_1$ is to be varied across simulation sets, $\Sigma_{11} = \sigma_X^2 = 1.0^2$, $\Sigma_{22} = \sigma_U^2 = 1.0^2$, $\Sigma_{33} = \sigma_V^2 = 0.1^2$, $\Sigma_{12} = \rho_{XU} \cdot \sigma_X \cdot \sigma_U = 0.5 \cdot 1.0^2$, $\Sigma_{13} = \rho_{XV} \cdot \sigma_X \cdot \sigma_V = 0.0 \cdot 1.0 \cdot 0.1$, and $\Sigma_{23} = \rho_{UV} \cdot \sigma_U \cdot \sigma_V = 0.5 \cdot 1.0 \cdot 0.1$.
$$
\text{$\theta$-th Conditional Quantile Treatment Effect at $x=0$} = \beta_1 + \gamma_1 F_{U|X}^{-1}(\theta | 0).
$$

We set $\Theta'' = [a,1-a] = [0.20,0.80]$ as the set of quantiles on which we conduct inference.
We use a grid with the interval size of 0.02 to approximate the continuum $\Theta''$ for numerical evaluation of functions defined on $\Theta''$.
First, we simulate the 95\% test for the null hypothesis $H_0: \tau_{FQRK}(\theta'') = 0 \ \forall\theta'' \in [a,1-a]$ of uniform treatment nullity using the procedure described in Section \ref{sec:a:FQRK}.
Next, we simulate the 95\% test for the null hypothesis $H_0: \tau_{FQRK}(\theta'') = \tau_{FQRK}(\theta''') \ \forall \theta'',\theta''' \in [a,1-a]$ of treatment homogeneity using the procedure described in  Section \ref{sec:a:FQRK}.

Table \ref{tab:FQRK} show simulated acceptance probabilities based on 2,500 multiplier bootstrap replications for 2,500 Monte Carlo replications for each of the sample sizes $n=1,000$, $2,000$, and $4,000$.
Part (A) reports results for the test of uniform treatment nullity and part (B) shows results for the test of treatment homogeneity.
The top panel (I) presents results across alternative values of $\beta_1 \in \{0.00, 0.25, 0.50, 0.75, 1.00\}$ while fixing $\gamma_1 = 0$.
The bottom panel (II) presents results across alternative values of $\gamma_1 \in \{0.00, 0.25, 0.50, 0.75, 1.00\}$ while fixing $\beta_1 = 0$. The nominal acceptance probability is 95\%.
The results, exhibiting the same qualitative features as those in the previous three subsections, evidence the power as well as the size correctness for both of the tests of uniform treatment nullity and treatment homogeneity.

\subsection{Example: Sharp Quantile RKD}

Consider the case of sharp quantile RKD presented in Sections \ref{sec:ex:sharp_quantile_rkd} and \ref{sec:a:SQRK}.
We generate an i.i.d. sample $\{(Y_i,D_i,X_i)\}_{i=1}^n$ through the following data generating process:
\begin{align*}
&Y_i = \alpha_0 + \alpha_1 X_i + \alpha_2 X_i^2 + \beta_1 D_i + (\gamma_0 + \gamma_1 D_i) \cdot U_i,
\\
&D_i = X_i \cdot \left( 2 \cdot \mathbbm{1}\{ X_i \geq 0 \} - 1 \right),
\\
&(X_i,U_i)' \sim N(0,\Sigma),
\end{align*}
where $\alpha_0=1.00$, $\alpha_1=0.10$, $\alpha_2=0.01$, $\beta_1$ is to be varied across simulation sets, $\gamma_0 = 1$, $\gamma_1$ is to be varied across simulation sets, $\Sigma_{11} = \sigma_X^2 = 1.0^2$, $\Sigma_{22} = \sigma_U^2 = 1.0^2$, and $\Sigma_{12} = \rho_{XU} \cdot \sigma_X \cdot \sigma_U = 0.5 \cdot 1.0^2$.
$$
\text{$\theta$-th Conditional Quantile Treatment Effect at $x=0$} = \beta_1 + \gamma_1 F_{U|X}^{-1}(\theta | 0).
$$

We set $\Theta'' = [a,1-a] = [0.20,0.80]$ as the set of quantiles on which we conduct inference.
We use a grid with the interval size of 0.02 to approximate the continuum $\Theta''$ for numerical evaluation of functions defined on $\Theta''$.
First, we simulate the 95\% test for the null hypothesis $H_0: \tau_{SQRK}(\theta'') = 0 \ \forall\theta'' \in [a,1-a]$ of uniform treatment nullity using the procedure described in Section \ref{sec:a:SQRK}.
Next, we simulate the 95\% test for the null hypothesis $H_0: \tau_{SQRK}(\theta'') = \tau_{SQRK}(\theta''') \ \forall \theta'',\theta''' \in [a,1-a]$ of treatment homogeneity using the procedure described in  Section \ref{sec:a:SQRK}.

Table \ref{tab:SQRK} show simulated acceptance probabilities based on 2,500 multiplier bootstrap replications for 2,500 Monte Carlo replications for each of the sample sizes $n=1,000$, $2,000$, and $4,000$.
Part (A) reports results for the test of uniform treatment nullity and part (B) shows results for the test of treatment homogeneity.
The top panel (I) presents results across alternative values of $\beta_1 \in \{0.00, 0.25, 0.50, 0.75, 1.00\}$ while fixing $\gamma_1 = 0$.
The bottom panel (II) presents results across alternative values of $\gamma_1 \in \{0.00, 0.25, 0.50, 0.75, 1.00\}$ while fixing $\beta_1 = 0$. The nominal acceptance probability is 95\%.
The results, exhibiting the same qualitative features as those in the previous subsection, evidence the power as well as the size correctness for both of the tests of uniform treatment nullity and treatment homogeneity.

\subsection{Example: Group Covariate and Test of Heterogeneous Treatment Effects}

Consider the case of fuzzy RDD with heterogeneous groups presented in Sections \ref{sec:ex:group} and \ref{sec:a:GFMRD}.
We generate an i.i.d. sample $\{(Y^\ast_i,D^\ast_i,G_i,X_i)\}_{i=1}^n$ through the following data generating process:
\begin{align*}
&Y^\ast_i = \alpha_0 + \alpha_1 X_i + \alpha_2 X_i^2 + \beta_1 D^\ast_i \cdot \mathbbm{1}\{G_i=1\} + \beta_2 D^\ast_i \cdot \mathbbm{1}\{G_i=2\} + U_i,
\\
&D^\ast_i = \mathbbm{1}\{ 2 \cdot \mathbbm{1}\{ X_i \geq 0\} - 1 \geq V_i \},
\\
&G_i \sim Bernoulli(\pi) + 1,
\\
&(X_i,U_i,V_i)' \sim N(0,\Sigma),
\end{align*}
where $\alpha_0=1.00$, $\alpha_1=0.10$, $\alpha_2=0.01$, $\beta_1$ or $\beta_2$ is to be varied across simulation sets, $\pi = 0.5$, $\Sigma_{11} = \sigma_X^2 = 1.0^2$, $\Sigma_{22} = \sigma_U^2 = 1.0^2$, $\Sigma_{33} = \sigma_V^2 = 0.5^2$, $\Sigma_{12} = \rho_{XU} \cdot \sigma_X \cdot \sigma_U = 0.5 \cdot 1.0^2$, $\Sigma_{13} = \rho_{XV} \cdot \sigma_X \cdot \sigma_V = 0.0 \cdot 1.0 \cdot 0.5$, and $\Sigma_{23} = \rho_{UV} \cdot \sigma_U \cdot \sigma_V = 0.5 \cdot 1.0 \cdot 0.5$.
In this setup, we have
$$
\text{Treatment Effect} =
\begin{cases}
\beta_1 & \text{if } G_i=1\\
\beta_2 & \text{if } G_i=2
\end{cases}.
$$

First, we simulate the 95\% test for the null hypothesis $H_0: \tau_{GFMRD}(1) = \tau_{GFMRD}(2) = 0$ of joint treatment nullity using the procedure described in Section \ref{sec:a:GFMRD}.
Part (A) of Table \ref{tab:GFMRD} shows simulated acceptance probabilities based on 2,500 multiplier bootstrap replications for 2,500 Monte Carlo replications for each of the sample sizes $n=1,000$, $2,000$, and $4,000$.
The first column under $\beta_1 = 0.00$ shows that simulated acceptance probabilities are close to the designed nominal probability, 95\%.
The next four columns show that the acceptance probability decreases in $\beta_1$, and the rate of decrease is higher for the larger sample sizes.
These results evidence the power as well as the size correctness for the test of joint treatment nullity.

We next simulate the 95\% test for the null hypothesis $H_0: \tau_{GFMRD}(1) = \tau_{GFMRD}(2)$ of treatment homogeneity using the procedure described in Section \ref{sec:a:GFMRD}.
Part (B) of Table \ref{tab:GFMRD} shows simulated acceptance probabilities based on 2,500 multiplier bootstrap replications for 2,500 Monte Carlo replications for each of the sample sizes $n=1,000$, $2,000$, and $4,000$.
The first column under $\beta_1 = 0.00$ shows that simulated acceptance probabilities are close to the designed nominal probability, 95\%.
The next four columns show that the acceptance probability decreases in $\beta_1$, and the rate of decrease is higher for the larger sample sizes.
These results evidence the power as well as the size correctness for the test of treatment homogeneity.

\section{Additional Mathematical Appendix}\label{sec:additional_mathematical_appendix}

\subsection{Auxiliary Lemmas for the Ten Examples}\label{sec:prel_lemmas_for_examples}

\subsubsection{Lemmas Related to Asymptotic Equicontinuity}\label{sec:lemma:asymptotic_rho-conti}
The following Lemma establish the relationship between convergence in probability in supremum norm and convergence in probability with respect to the semi-metric $\rho$ induced by the limiting Gaussian process.
We use this lemma in the corollary below to ensure that the asymptotically equicontinuous process, $\hat{\nu}^\pm_{\xi,n}$, evaluating $\hat Q_{Y | X}( \cdot | 0^\pm)$ can nicely approximate this process evaluating $Q_{Y | X}( \cdot | 0^\pm)$.

\begin{lemma}\label{lemma:asymptotic_rho-conti}
Suppose that Assumptions S, K, and SQRD hold.
Define
\begin{align*}
\rho(\hat{Q}_{Y|X}(\theta''|0^\pm),Q_{Y|X}(\theta''|0^\pm))=\lim_{n\to \infty}\big(\sum_{i=1}^{n} E|f_{ni}(\hat{Q}_{Y|X}(\theta''|0^\pm))-f_{ni}(Q_{Y|X}(\theta''|0^\pm))|^2  \big)^{1/2},
\end{align*}
where
\begin{align*}
f_{ni}(y)&=\frac{e'_0(\Gamma^\pm_2)^{-1}r_2(\frac{X_i}{h_n})K(\frac{X_i}{h_n})[\mathds{1}\{Y_i\le y\}-F_{Y|X}(y|X_i)]\delta^\pm_i}{\sqrt{nh_n}f_{X}(0)}.
\end{align*}
Then,
$\norm{\hat{Q}_{Y|X}(\cdot|0^\pm)-Q_{Y|X}(\cdot|0^\pm)}_{[a,1-a]}\overset{p}{\underset{x}{\to}}0$ implies
$\sup_{\theta'' \in [a,1-a]}\rho(\hat{Q}_{Y|X}(\theta''|0^\pm),Q_{Y|X}(\theta''|0^\pm))\overset{p}{\underset{x}{\to}}0$.
\end{lemma}
\begin{proof}
We will show the claim for the $+$ side only.
The case of the $-$ side can be similarly proved. 
Notice that by Law of Iterated Expectations and calculations under Assumption S, K, and SQRD (ii),
\begin{align*}
&\rho^2(\hat{Q}_{Y|X}(\theta''|0^+),Q_{Y|X}(\theta''|0^+))\\
=&\lim_n\sum_{i=1}^{n}E[E[\Big(\frac{e'_0(\Gamma^+_2)^{-1}r_2(\frac{X_i}{h_n})K(\frac{X_i}{h_n})\delta^+_i}{\sqrt{nh_n}f_{X}(0)}\\
& [\mathds{1}\{Y_i\le \hat{Q}_{Y|X}(\theta''|0^+)\}-F_{Y|X}(\hat{Q}_{Y|X}(\theta''|0^+)|X_i)-\mathds{1}\{Y_i\le Q_{Y|X}(\theta''|0^+)\}+F_{Y|X}(Q_{Y|X}(\theta''|0^+)|X_i)]\Big)^2|X_i]]\\
=&\lim_n E[\frac{e'_0 (\Gamma^+_2)^{-1} \Psi^+ (\Gamma^+_2)^{-1} e_0}{f_X(0)} E[(\mathds{1}\{Y_i\le \hat{Q}_{Y|X}(\theta''|0^+)\} -F_{Y|X}(\hat{Q}_{Y|X}(\theta''|0^+)|X_i)\\
&\qquad -\mathds{1}\{Y_i\le Q_{Y|X}(\theta''|0^+)\}+F_{Y|X}(Q_{Y|X}(\theta''|0^+)|X_i))^2|X_i]]
\end{align*}
It then suffices to show that
\begin{align*}
E[(\mathds{1}\{Y_i\le \hat{Q}_{Y|X}(\theta''|0^+)\}-F_{Y|X}(\hat{Q}_{Y|X}(\theta''|0^+)|X_i)-\mathds{1}\{Y_i\le Q_{Y|X}(\theta''|0^+)\}+F_{Y|X}(Q_{Y|X}(\theta''|0^+)|X_i))^2|X_i]
\\
\underset{x}{\overset{p}{\to}}0
\end{align*}
uniformly in $\theta''$.
We write
\begin{align*}
&(\mathds{1}\{Y_i\le \hat{Q}_{Y|X}(\theta''|0^+)\}
-F_{Y|X}(\hat{Q}_{Y|X}(\theta''|0^+)|X_i)
-\mathds{1}\{Y_i\le Q_{Y|X}(\theta''|0^+)\}
+F_{Y|X}(Q_{Y|X}(\theta''|0^+)|X_i))^2\\
=&([\mathds{1}\{Y_i\le \hat{Q}_{Y|X}(\theta''|0^+)\}-\mathds{1}\{Y_i\le Q_{Y|X}(\theta''|0^+)\}]
-[F_{Y|X}(\hat{Q}_{Y|X}(\theta''|0^+)|X_i)-F_{Y|X}(Q_{Y|X}(\theta''|0^+)|X_i)])^2\\
=&
[\mathds{1}\{Y_i\le \hat{Q}_{Y|X}(\theta''|0^+)\}-\mathds{1}\{Y_i\le Q_{Y|X}(\theta''|0^+)\}]^2
+[F_{Y|X}(\hat{Q}_{Y|X}(\theta''|0^+)|X_i)-F_{Y|X}(Q_{Y|X}(\theta''|0^+)|X_i)]^2\\
&-2[\mathds{1}\{Y_i\le \hat{Q}_{Y|X}(\theta''|0^+)\}-\mathds{1}\{Y_i\le Q_{Y|X}(\theta''|0^+)\}]
[F_{Y|X}(\hat{Q}_{Y|X}(\theta''|0^+)|X_i)-F_{Y|X}(Q_{Y|X}(\theta''|0^+)|X_i)]
\\
=&(1)+(2)-(3).
\end{align*}
The conditional expectation of part (1) is
\begin{align*}
E[(1)|X_i]= F_{Y|X}(\hat{Q}_{Y|X}(\theta''|0^+)|X_i)+F_{Y|X}(Q_{Y|X}(\theta''|0^+)|X_i)-2F_{Y|X}( \hat{Q}_{Y|X}(\theta''|0^+)\wedge Q_{Y|X}(\theta''|0^+) |X_i)
\\
\underset{x}{\overset{p}{\to}}0
\end{align*}
uniformly by the uniform consistency $\norm{\hat{Q}_{Y|X}(\cdot|0^\pm)-Q_{Y|X}(\cdot|0^\pm)}_{[a,1-a]} \underset{x}{\overset{p}{\to}}0$ and the continuous mapping theorem under Assumption SQRD (i).
The uniform convergence in probability of other parts, (2) and (3), can be concluded similarly.
\end{proof}
Similar results hold in the cases of Sharp FQRK, Fuzzy FQRK and Fuzzy FQRD as the following
\begin{lemma}\label{lemma:asymptotic_rho-conti_QRK}
Suppose that Assumption SQRK or FQRK holds, in addition to Assumptions S and K.
Define
\begin{align*}
\rho(\hat{Q}_{Y|X}(\theta''|0),Q_{Y|X}(\theta''|0))=\lim_{n\to \infty}\big(\sum_{i=1}^{n} E|f_{ni}(\hat{Q}_{Y|X}(\theta''|0))-f_{ni}(Q_{Y|X}(\theta''|0))|^2  \big)^{1/2},
\end{align*}
where
\begin{align*}
f_{ni}(y)&=\frac{e'_1(\Gamma^\pm_2)^{-1}r_2(\frac{X_i}{h_n})K(\frac{X_i}{h_n})[\mathds{1}\{Y_i\le y\}-F_{Y|X}(y|X_i)]\delta^\pm_i}{\sqrt{nh_n}f_{X}(0)}.
\end{align*}
Then,
$\norm{\hat{Q}_{Y|X}(\cdot|0)-Q_{Y|X}(\cdot|0)}_{[a,1-a]} \overset{p}{\underset{x}{\to}}0$ implies
$\sup_{\theta'' \in [a,1-a]}\rho(\hat{Q}_{Y|X}(\theta''|0),Q_{Y|X}(\theta''|0))\overset{p}{\underset{x}{\to}}0$.
\end{lemma}

\begin{lemma}\label{lemma:asymptotic_rho-conti_FQRD}
Suppose that Assumptions S, K, and FQRD hold.
Define
\begin{align*}
\rho(\hat{Q}_{Y^d|C}(\theta''),Q_{Y^d|C}(\theta''))=\lim_{n\to \infty}\big(\sum_{i=1}^{n} E|f_{ni}(\hat{Q}_{Y^d|C}(\theta''),d)-f_{ni}(Q_{Y^d|C}(\theta''),d)|^2  \big)^{1/2},
\end{align*}
where
\begin{align*}
f_{ni}(y,d)&=\frac{e'_1(\Gamma^\pm_2)^{-1}r_2(\frac{X_i}{h_n})K(\frac{X_i}{h_n})[\mathds{1}\{Y^*_i\le y,D^*_i=d\}-\mu_1(X_i,y,d)]\delta^\pm_i}{\sqrt{nh_n}f_{X}(0)}.
\end{align*}
Then,
$\norm{\hat{Q}_{Y^d|C}(\theta'')-Q_{Y^d|C}(\theta'')}_{[a,1-a]\times \{0,1\}} \overset{p}{\underset{x}{\to}}0$ implies
$\sup_{(\theta'',d) \in [a,1-a]\times \{0,1\}}\rho(\hat{Q}_{Y^d|C}(\theta''),Q_{Y^d|C}(\theta''))\overset{p}{\underset{x}{\to}}0$.
\end{lemma}
Since the proofs are mostly identical to Lemma \ref{lemma:asymptotic_rho-conti}, we omit them.

\subsubsection{Lemmas for Uniform Consistency of Estimators for Conditional Densities}\label{sec:lemma:unif_cons_est}
Recall that in Section \ref{sec:a:SQRD}, the existence of an uniform consistency estimator is assumed in Assumption SQRD (iii)
An example of estimators $\hat f_{Y|X}(y,0^\pm)$ that satisfy Assumption SQRD (iii) are
\begin{align*}
\hat f_{Y|X}(y|0^\pm)=\frac{\frac{1}{na^2_n}\sum_{i=1}^{n}K(\frac{Y_i-y}{a_n})K(\frac{X_i}{a_n})\delta^\pm_i}{\frac{1}{na_n}
\sum_{i=1}^{n}K(\frac{X_i}{a_n})\delta^\pm_i}=\frac{\hat g_{YX}(y,0^\pm)}{\hat g_X(0^\pm)}
\end{align*}
The following result gives the sufficient conditions for $f_{Y|X}(y|0^\pm)$ to satisfy Assumption SQRD (iii).
\begin{lemma}\label{lemma:unif_cons_fYX}
Suppose that Assumptions S, K, and SQRD hold. In addition, let $f_{YX}$ be continuously differentiable on $\mathscr{Y}_1 \times [\underline{x},0)$, and $\mathscr{Y}_1 \times (0,\overline{x}]$ with bounded partial derivatives and assume kernel $K$ be symmetric. Let $a_n$ be such that $a_n\to 0$, $na_n \to \infty$, $\frac{na^2_n}{|\log a_n|}\to \infty$, $\frac{|\log a_n|}{\log \log a_n}\to \infty$, and $a^2_n\le c a^2_{2n}$ for some $c>0$.
then $\hat f_{Y|X} (y|0^\pm)$ such that $\sup_{y\in \mathscr{Y}_1}|\hat f_{Y|X} (y|0^\pm)-f_{Y|X} (y|0^\pm)|=o^{x}_p(1)$.
\begin{proof}
We will show only for $0^+$. The other side follows similarly.
Note the denominator $\hat g_X(0^+)=\frac{1}{2} f_X(0^+)+o^{x}_p(1)$ and $f_X(0)$ is bounded away from $0$ under Assumption \ref{a:BR}(ii).

Thus, it suffices to show $\sup_{y\in \mathscr{Y}_1}|\hat g_{YX}(y,0^+)-\frac{1}{2}f_{YX}(y,0^+)|=o^x_p(1)$.
Note $|\hat g_{YX}(y,0^+)-\frac{1}{2}f_{YX}(y,0^+)|\le |\hat g_{YX}(y,0^+)-E\hat g_{YX}(y,0^+)|+|E\hat g_{YX}(y,0^+)-\frac{1}{2}f_{YX}(y,0^+)|$.
To control the stochastic part, Theorem 2.3 of Gin\'e and Guillou (2002) suggests that
\begin{align*}
\sup_{y\in \mathscr{Y}_1}|\hat g_{YX}(y,0^+)-E\hat g_{YX}(y,0^+)|=O^x_{a.s.}(\sqrt{\frac{\log \frac{1}{a_n}}{na^2_n}}).
\end{align*}
For the deterministic part, under the smoothness assumptions made on $f_{YX}$, a mean value expansion gives
\begin{align*}
&E[\hat g_{YX}(y,0^+)]\\
=&\int_{\mathds{R}}\int_{\mathds{R}_+}K(u)K(v)(f_{YX}(y,0^+)+\frac{\partial}{\partial y}f_{YX}(y^*,x^*)ua_n + \frac{\partial}{\partial x}f_{YX}(y^*,x^*)va_n)dudv\\
=&\frac{1}{2}f_{YX}(y,0^+)+O(a_n)
\end{align*}
uniformly on $\mathscr{Y}_1$,
where $(y^*,x^*)$ is a linear combination of $(y,0)$ and $(y+uh_n,vh_n)$. This concludes the proof.
\end{proof}
\end{lemma}

In Section \ref{sec:a:SQRK}, we assumed the existence of an estimator that satisfies Assumption SQRK. Define
\begin{align*}
\hat{f}_{Y|X}(y|0)=\frac{\frac{1}{na^2_n}\sum_{i=1}^{n}K(\frac{Y_i-y}{a_n})K(\frac{X_i}{a_n})}{\frac{1}{ n a_n}\sum_{i=1}^{n}K(\frac{X_i}{a_n})}=\frac{\hat f_{YX}(y,0)}{\hat f_X(0)}
\end{align*}
The following lemma provides such estimator.
\begin{lemma}\label{lemma:super_consist_SQRK}
Assume Assumption S and SQRK (i) (b) hold. In addition, if the joint density $f_{YX}(y,x)$ is twice continuously differentiable on $\mathscr{Y}_2 \times [\underline{x},\overline{x}]$ with bounded partial derivatives for some open set $\mathscr{Y}_2$ such that $\mathscr{Y}_1 \subset\mathscr{Y}_2 \subset\mathscr{Y} $.
Let $K$ be a bounded kernel function such that $\{x\mapsto K(x-x'/a):a>0,x'\in \mathds{R} \}$ forms a VC type class (e.g. Epanechnikov) and $\int_{\mathds{R}}uK(u)du=0$, $\int_{\mathds{R}}u^2K^2(u)du<0$.
 The bandwidth $a_n\to0$ satisfies $\frac{na^2_n}{|\log a_n|}\to \infty$, $\frac{|\log a_n|}{\log \log a_n}\to \infty$, and $a^2_n\le c a^2_{2n}$ for some $c>0$, $\frac{h^3_n\log\frac{1}{a_n}}{a^2_n}\to 0$ and $nh^3_n a^4_n\to 0$.
Then, $\sup_{y\in \mathscr{Y}_1}\sqrt{nh^3_n}|\hat f_{Y|X}(y|0)- f_{Y|X}(y|0)|\underset{x}{\overset{p}{\to}}0$. Thus, Assumption SQRK (iv) is satisfied.
\end{lemma}

We remark that the condition $\frac{h^3_n\log\frac{1}{a_n}}{a^2_n}\to 0$ is easily satisfied by the bandwidth selectors $h^{MSE}_{1,n}\propto n^{-1/5}$ and $h^{ROT}_{1,n}\propto n^{-1/4}$ proposed in Section \ref{sec:practical_guideline} along with a Silverman's rule of thumb $a_n\propto n^{-1/6}$.
\begin{proof}
First note that the denominator is essentially $f_X(0)+o^x_p(1)$ and $f_X(0)>C>0$. It then suffices to show that the numerator $\hat f_{YX}(y,0)$ satisfies $\sup_{y \in \mathscr{Y}_1}\sqrt{nh_n}|\hat f_{YX}(y,0)-f_{YX}(y,0)|=o^x_p(1)$. Write $|\hat f_{YX}(y,0)-f_{YX}(y,0)|\le |\hat f_{YX}(y,0)-E\hat f_{YX}(y,0)|+|E\hat f_{YX}(y,0)- f_{YX}(y,0) |$. To control the stochastic part, Theorem 2.3 of Gin\'e and Guillou (2002) suggests that
\begin{align*}
\sup_{y\in \mathscr{Y}_1}|\hat{f}_{YX}(y,0)-E\hat f_{YX}(y,0)|=O^x_{a.s.}(\sqrt{\frac{\log \frac{1}{a_n}}{na^2_n}}).
\end{align*}
For the deterministic part, a mean value expansion gives
\begin{align*}
&E[\hat f_{YX}(y,0)]\\
=&\int_{\mathds{R}}\int_{\mathds{R}}K(u)K(v)(f_{YX}(y,0)+\frac{\partial}{\partial y}f_{YX}(y,0)ua_n + \frac{\partial}{\partial x}f_{YX}(y,0)va_n\\
&+\frac{\partial}{\partial y}\frac{\partial}{\partial y}f_{YX}(y^*,x^*)uv a^2_n+\frac{1}{2}\frac{\partial^2}{\partial y^2}f_{YX}(y^*,x^*)u^2 a^2_n+\frac{1}{2}\frac{\partial^2}{\partial x^2}f_{YX}(y^*,x^*)v^2 a^2_n)dudv\\
=&f_{YX}(y,0)+O(a^2_n)
\end{align*}
uniformly on $\mathscr{Y}_1$,
where $(y^*,x^*)$ is a linear combination of $(y,0)$ and $(y+uh_n,vh_n)$.
Thus, the conclusion follows from $\frac{\log\frac{1}{a_n}h^3_n}{a^2_n}\to 0$ and $nh^3_na^4_n\to0$.
\end{proof}

In Section \ref{sec:a:FQRD}, the densities in the denominator can be estimated in the following manner.
Define
\begin{align*}
&\hat{f}_{Y^d|C}(y)=
\frac{\hat{f}_{Y|XD}(y|0^+,d)\hat{\mu}_{2,2}(0^+,d)-
\hat{f}_{Y|XD}(y|0^-,d)\hat{\mu}_{2,2}(0^-,d)}
{\hat{\mu}_{2,2}(0^+,d)-\hat{\mu}_{2,2}(0^-,d)},
\end{align*}
where
$\hat{f}_{Y|XD}(y|0^\pm,1)=\frac{\frac{1}{na^2_{n}}\sum_{i=1}^{n} K(\frac{X_i}{a_{n}})K(\frac{Y_i-y}{a_{n}}) D_i \delta^\pm_i }{\frac{1}{na_{n}}\sum_{i=1}^{n}K(\frac{X_i}{a_{n}})  D_i \delta^\pm_i}$,
$\hat{f}_{Y|XD}(y|0^\pm,0)=\frac{\frac{1}{na^2_{n}}\sum_{i=1}^{n} K(\frac{X_i}{a_{n}})K(\frac{Y_i-y}{a_{n}}) (1-D_i) \delta^\pm_i }{\frac{1}{na_{n}}\sum_{i=1}^{n}K(\frac{X_i}{a_{n}}) (1- D_i) \delta^\pm_i}$,
with bandwidths $a_n$. 
The following lemma shows uniform consistency of these estimators.

\begin{lemma}\label{lemma:FQRDD}
Suppose that Assumptions, S, K, and FQRD hold. In addition, for each $d=0$, $1$, let $f_{YX|D}(\cdot,\cdot|d)$ be continuously differentiable on $\mathscr{Y}_1 \times [\underline{x},0)$, and $\mathscr{Y}_1 \times (0,\overline{x}]$ with bounded partial derivatives and assume kernel $K$ be symmetric. Let $a_n$ be such that $a_n\to 0$, $na_n \to \infty$, $\frac{na^2_n}{|\log a_n|}\to \infty$, $\frac{|\log a_n|}{\log \log a_n}\to \infty$, and $a^2_n\le c a^2_{2n}$ for some $c>0$.
Then, $\sup_{y\in \mathcal{Y}}|\hat{f}_{Y^0|C}(y)-f_{Y^0|C}(y)|=o^x_p(1)$ and $\sup_{y\in \mathcal{Y}}|\hat{f}_{Y^1|C}(y)-f_{Y^1|C}(y)|=o^x_p(1)$.
\end{lemma}
\begin{proof}
Note that
\begin{align*}
f_{Y^1|C}(y)&=\frac{\partial }{\partial y} F_{Y^1|C}(y)
=\frac{\frac{\partial }{\partial y}\mu_1(0^+,(y,1))-
\frac{\partial }{\partial y}\mu_{1}(0^-,(y,1))}
{\mu_{2}(0^+,1)-\mu_{2}(0^-,1)}
\end{align*}
The numerator can be estimated consistently with the local quadratic regression defined above.
For the denominator, we have
\begin{align*}
\frac{\partial }{\partial y}\mu_{1}(0^+,(y,1))=&\frac{\partial }{\partial y}E[\mathds{1}\{Y_i\le y\}\mathds{1}\{D_i=1\}|X_i=0^+]\\
=&\frac{\partial }{\partial y}\Big(E[\mathds{1}\{Y_i\le y\}|X_i=0^+,D_i=1]\mathds{P}^x(D_i=1|X_i=0^+)+0\Big)\\
=&\frac{\partial }{\partial y}F_{Y^1|XD}(y|0^+,1)\mu_2(0^+,1)\\
=&f_{Y^1|XD}(y|0^+,1)\mu_2(0^+,1)
\end{align*}
Uniform consistency of $\hat{f}_{Y^1|X,D}(y|0^\pm,1)$ can be shown by applying Theorem 2.3 of Gin\'e and Guillou (2002), as in Lemma \ref{lemma:unif_cons_fYX}.
Also $\hat{\mu}_{2,p}(0^\pm)$ is uniformly consistent by Corollary \ref{corollary:mu_unif_consist}.
Corresponding result for $\hat{f}_{Y^0|C}(y)$ can be shown similarly.
\end{proof}

\subsubsection{Proof of Lemma \ref{lemma:quantile_deriv}}\label{sec:lemma:quantile_deriv}
\begin{proof}
The Hadamard differentiability of the left-inverse operator and a mean value expansion give
\begin{align*}
\frac{\partial }{\partial x}Q_{Y|X}(\theta''|0^+)
&=\lim_{\delta\downarrow 0}\frac{Q_{Y|X}(\theta''|\delta)-Q_{Y|X}(\theta''|0^+)}{\delta}\\
&=\lim_{\delta\downarrow 0}\frac{\Phi\Big(F_{Y|X}(\cdot|\delta)\Big)(\theta'')-\Phi \Big(F_{Y|X}(\cdot|0^+)\Big)(\theta'')}{\delta}\\
&=\lim_{\delta\downarrow 0}\frac{\Phi\Big((F_{Y|X}(\cdot|0^+) + \delta \frac{\partial }{\partial x}F_{Y|X}(\cdot|x^*)\Big)(\theta'')-\Phi\Big(F_{Y|X}(\cdot|0^+)\Big)(\theta'')}{\delta}\\
&= \Phi'_{F_{Y|X}(\cdot|0)} \Big(\frac{\partial}{\partial x}F_{Y|X}(\cdot|0^+)\Big) (\theta'')\\
&=-\frac{ F^{(1)}_{Y|X}(Q_{Y|X}(\theta''|0)|0^+)}{f_{Y|X}(Q_{Y|X}(\theta''|0)|0)}\\
&=\phi(F^{(1)}_{Y|X}(\cdot|0^+))(\theta'')
\end{align*}
by the definition of Hadamard Derivative and then Lemma 3.9.23 (i) of van der Vaart and Wellner (1996).
\end{proof}


\subsubsection{Uniform Consistency of $\hat Q_{Y^d |C}$ under Assumption FQRD}\label{sec:lemma:uniform_consist_Q_{Y|C}}

The next lemma, which follows from Theorem \ref{theorem:weak_conv}, shows the weak convergence in probability of the estimators of the local conditional quantiles of the potential outcomes.

\begin{lemma}\label{lemma:uniform_consist_Q_{Y|C}}
Suppose that Assumptions S, K, and FQRD hold.
Then, $\sup_{\mathscr{Y}_1}|\hat{Q}_{Y^d|C}-Q_{Y^d |C}|\underset{ x \times \xi}{\overset{p}{\to}} 0$ for $d=1,0$.
\end{lemma}

\begin{proof}
Write $\Phi(F_{Y^d | C} ) (\theta''):= \inf\{y \in \mathscr{Y}_1 : F_{Y^d | C}(y) \geq \theta'' \}$
for $\theta''\in [a,1-a]$.
The weak convergence of $\nu^\pm_n$ from Theorem \ref{theorem:weak_conv},  Lemma 3.9.23 (i) of van der Vaart and Wellner (1996), and the functional delta method yield
\begin{align*}
\sqrt{nh_n}[\hat{Q}_{Y^d|C}(\cdot)-Q_{Y^d|C}(\cdot)]
=&\sqrt{nh_n}[\Phi(\Hat{F}_{Y^d|C})(\cdot)-\Phi(F_{Y^d|C})(\cdot)]=O^x_p(1)
\end{align*}
The result follows from Assumption FQRD and Slutsky's lemma.
\end{proof}

\subsection{Proofs of Corollaries for the Ten Examples}\label{sec:corollary_for_examples}
We use $\mathds{G}$ and $\mathds{G}_\pm$ to denote generic zero mean Gaussian processes that appear in some intermediate steps in the proofs without having to specify their covariance structure. They differ across different examples, but are fixed within each one.

\subsubsection{Proof of Corollary \ref{corollary:FQRD}}\label{sec:corollary:FQRD}
\begin{proof}
We first verify that the Assumptions required by Theorem 1 are satisfied.

To show Assumption \ref{a:BR}(ii)(a), it is true that the subgraphs of $\{(y,d)\mapsto\mathds{1}\{y\le y'\}:y'\in\mathscr{Y}_1 \}$ and $\{(y,d) \mapsto \mathds{1}\{d=d'\}:d'\in \mathscr{D}\}$ can not shatter any two-point sets (for definition, see Section 9.1.1 of Kosorok(2008)); to see this, let $(y_1,d_1,r_1)$, $(y_2,d_2,r_2)\in \mathscr{Y}_1\times\mathscr{D}\times \mathds{R}$ with $y_1\le y_2$. Notice that $\{(y_1,d_1,r_1)\}$ can never be picked out by any function in either of the classes. Therefore they are of VC-subgraph classes with VC index of 2. This implies that they are of VC type and Lemma \ref{lemma:VC type_stability} implies that the set of products $\{(t,y)\mapsto \mathds{1}\{y\le y'\}t:y'\in \mathscr{Y}_1 \}$ is of VC type with square integrable constant function $1$ as its envelope. On the other hand, $\{x\mapsto E[\mathds{1}\{Y_i\le y\}\mathds{1}\{t=d\}|X_i=x]:(y,d)\in \mathscr{Y}_1\times \mathscr{D}\}$ is of VC type with square integrable envelope $1$ because of Example 19.7 of van der Vaart (1998) under boundedness of $\mathscr{Y}_1$ and Lipschitz continuity from Assumption S (a) and FQRD (i). Assumption \ref{a:BR}(ii)(c) follows from Assumption FQRD (i) since for any $(y_1,d_1)$, $(y_2,d_2)\in \mathscr{Y}\times \mathscr{D}$, if $d_1=d_2=d$ $E[\mathds{1}\{Y\le y_1,D=d\}\mathds{1}\{Y\le y_2,D=d\}|X=x]=E[\mathds{1}\{Y\le y_1 \wedge y_2,D=d\}|X=x]$ and It equals zero if $d_1\ne d_2$. Assumption \ref{a:BR} (ii) (d) is implied by the right continuity of $(y',d')\mapsto \mathds{1}\{y\le y', d=d'\}$ on $\mathscr{Y}_1 \times \{0,1\}$ in $y$ for each $d$.
Assumption \ref{a:first_stage} is implied by Lemma \ref{lemma:epsilon_unif_const}. Assumption FQRD (iv) and Lemma 3.9.23 of van der Vaart and Wellner (1996) (applicable under Assumption FQRD (i), (iii) and (iv)) implies Hadamard differentiability of  Assumption \ref{a:cond_weak_conv}(i).

The rest are directly implied by Assumption S, K, M, and FQRD. Thus we may apply Theorem \ref{theorem:weak_conv} and acquire $\nu_n \leadsto \mathds{G}:= \mathds{G}^+ -\mathds{G}^-$ for a zero mean Gaussian process $\mathds{G}$,
where
\begin{align*}
\begin{bmatrix}
\nu_n(y,d_1,d_2,1)\\
\nu_n(y,d_1,d_2,2)
\end{bmatrix}
=
\begin{bmatrix}
\sqrt{nh_n}[(\hat{\mu}_{1,2}(0^+,y,d_1)-\hat{\mu}_{1,2} (0^-,y,d_1))-(\mu_{1} (0^+,y,d_1)- \mu_{1}(0^-,y,d_1) )]\\
\sqrt{nh_n}[(\hat{\mu}_{2,2}(0^+,d_2)-\hat{\mu}_{2,2} (0^-,d_2))-(\mu_2 (0^+,d_2)- \mu_2 (0^-,d_2) )]
\end{bmatrix}
\end{align*}
Theorem \ref{theorem:weak_conv} further implies
\begin{align*}
&\sqrt{nh}(\hat{\tau}_{FQRD}(\cdot)-\tau_{FQRD}(\cdot))\\
\leadsto&
\mathds{G}'_{FQRD}(\cdot)\\
:=&\Upsilon'_W\Big(\frac{[\mu_{2} (0^+,d)-\mu_{2} (0^-,d)]
\mathds{G}(\cdot,\cdot,\cdot,1)-[\mu_{1} (0^+,y,d)-\mu_{1} (0^-,y,d)]
\mathds{G}(\cdot,\cdot,\cdot,2)}{[\mu_{2} (0^+,d)-\mu_{2} (0^-,d)]^2}\Big)(\cdot)\\
=& -\frac{[\mu_{2} (0^+,d)-\mu_{2} (0^-,d)]
\mathds{G}(Q_{Y^1|C}(\cdot),1,1,1)-[\mu_{1} (0^+,Q_{Y^1|C}(\cdot),1)-\mu_{1} (0^-,Q_{Y^1|C}(\cdot),1)]
\mathds{G}(Q_{Y^1|C}(\cdot),1,1,2)}{f_{Y^1|C}(Q_{Y^1|C}(\cdot))[\mu_{2} (0^+,1)-\mu_{2} (0^-,1)]^2}\\
&+\frac{[\mu_{2} (0^+,d)-\mu_{2} (0^-,d)]
\mathds{G}(Q_{Y^0|C}(\cdot),0,0,1)-[\mu_{1} (0^+,Q_{Y^0|C}(\cdot),0)-\mu_{1} (0^-,Q_{Y^0|C}(\cdot),0)]
\mathds{G}(Q_{Y^1|C}(\cdot),0,0,2)}{f_{Y^0|C}(Q_{Y^0|C}(\cdot))[\mu_{2} (0^+,0)-\mu_{2} (0^-,0)]^2}
\end{align*}
For the conditional weak convergence part of the proof, define
\begin{align*}
&\Upsilon'_W(\widehat{\mathds{Y}}_n)(\theta'')=\\
-&\frac{[\mu_{2} (0^+,1)-\mu_{2} (0^-,1)]
\widehat{\mathds{X}}'_n(Q_{Y^1|C}(\theta''),1,1,1)-[\mu_{1}(0^+,Q_{Y^1|C}(\theta''),1)-\mu_{1}(0^-,Q_{Y^1|C}(\theta''),1)]
\widehat{\mathds{X}}'_n(Q_{Y^1|C}(\theta''),1,1,2)}{f_{Y^1|C}(Q_{Y^1|C}(\theta''))[\mu_{2} (0^+,1)-\mu_{2} (0^-,1)]^2}\\
+&\frac{[\mu_{2} (0^+,0)-\mu_{2} (0^-,0)]
\widehat{\mathds{X}}'_n(Q_{Y^0|C}(\theta''),0,0,1)-[\mu_{1}(0^+,Q_{Y^0|C}(\theta''),0)-\mu_{1}(0^-,Q_{Y^0|C}(\theta''),0)]
\widehat{\mathds{X}}'_n(Q_{Y^0|C}(\theta''),0,0,2)}{f_{Y^0|C}(Q_{Y^0|C}(\theta''))[\mu_{2} (0^+,0)-\mu_{2} (0^-,0)]^2}
\end{align*}
Theorem \ref{theorem:cond_weak_conv} then suggests that $\Upsilon'_W(\widehat{\mathds{Y}}_n)\underset{\xi}{\overset{p}{\leadsto}}\mathds{G}'_{FQRD}$. Thus it suffices to show
\begin{align*}
\sup_{\theta''\in [a,1-a]}|\widehat{\Upsilon}'_W(\widehat{\mathds{Y}}_n)(\theta'')-\Upsilon'_W(\widehat{\mathds{Y}}_n)(\theta'')|\underset{ x \times \xi}{\overset{p}{\to}} 0
\end{align*}
which is true by the asymptotic $\rho-$equicontinuity of $\mathds{X}'_n$ (which is inherited from the conditional weak convergence of $\hat{\nu}_{\xi,n}$), the uniform consistency of $\hat{Q}_{Y^d|C}$ and $\hat{f}_{Y^d|C}$ for $d=1,0$ by Lemmas \ref{lemma:FQRDD}, \ref{lemma:uniform_consist_Q_{Y|C}} and \ref{lemma:asymptotic_rho-conti_FQRD}.
\end{proof}

\subsubsection{Proof of Corollary \ref{corollary:FMRD}}\label{sec:corollary:FMRD}
\begin{proof}
To check the assumptions required by Theorems \ref{theorem:weak_conv} and \ref{theorem:cond_weak_conv}, for Assumption \ref{a:BR}(ii)(a), notice that\\
 $\{ \mu_2(x,0):[\underline x,\overline x]\mapsto \mathds{R}\}$ is a singleton, and therefore forms a VC type class with envelope $1$. Assumption \ref{a:first_stage} follows from Lemma \ref{lemma:epsilon_unif_const}. Other conditions can be checked as before.
Applying Theorem \ref{theorem:weak_conv} gives
\begin{align*}
&\frac{(\mu_{2}(0^+,0)-\mu_{2}(0^-,0))\widehat{\mathds{X}}'_n(0,1)-(\mu_{1}(0^+,0)-\mu_{1}(0^-,0))\widehat{\mathds{X}}'_n(0,2)}{(\mu_{2}(0^+,0)-\mu_{2}(0^-,0))^2}\leadsto N(0,\sigma^2_{FMRD})
\end{align*}
with probability approaching one.
According to Lemma \ref{lemma:cond_weak_conv_and in prob}, it remains to show
\begin{align*}
&\frac{(\mu_{2}(0^+,0)-\mu_{2}(0^-,0))\widehat{\mathds{X}}'_n(0,1)-(\mu_{1}(0^+,0)-\mu_{1}(0^-,0))\widehat{\mathds{X}}'_n(0,2)}{(\mu_{2}(0^+,0)-\mu_{2}(0^-,0))^2}\\
-&\frac{(\hat{\mu}_{2,2}(0^+,0)-\hat{\mu}_{2,2}(0^-,0))\widehat{\mathds{X}}'_n(0,1)-(\hat{\mu}_{1,2}(0^+,0)-\hat{\mu}_{1}(0^-,0))\widehat{\mathds{X}}'_n(0,2)}{(\hat{\mu}_{2,2}(0^+,0)-\hat{\mu}_{2,2}(0^-,0))^2}=o^{x\times \xi}_p(1).
\end{align*}
This is the case due to the uniform consistency of $\hat{\mu}_{1,2}(0^\pm,0)$, $\hat{\mu}_{2,2}(0^\pm,0)$ that follows from Corollary \ref{corollary:mu_unif_consist}, the independence between data and $\xi_i$ under Assumption M, and the fact that $|\mu_2(0^+,0)-\mu_2(0^-,0)|>0$ under Assumption FMRD (i) (d).
\end{proof}

\subsubsection{Proof of Corollary \ref{corollary:SMRD}}\label{sec:corollary:SMRD}
\begin{proof}
It suffices to show that Assumptions S, K, M, and SMRD imply Assumptions \ref{a:BR}, \ref{a:cond_weak_conv}, and \ref{a:first_stage}.
Most of these implications are direct.
For Assumption \ref{a:BR} (ii)(a), note that $\{\mu_1(x,0):[\underline x,\overline x]\mapsto \mathds{R}\}$ is a singleton, and therefore forms a VC type class with its sole element serving as an envelope, which is integrable by Assumption SMRD (i)(a). Assumption \ref{a:first_stage} follows from Lemma \ref{lemma:epsilon_unif_const}.
\end{proof}

\subsubsection{Proof of Corollary \ref{corollary:SCRD}}\label{sec:corollary:SCRD}
\begin{proof}
We check that Assumption SCRD implies the assumptions required by Theorems \ref{theorem:weak_conv} and \ref{theorem:cond_weak_conv},.
Most are direct, and we only need to check the following three points.
For Assumption \ref{a:BR} (ii) (a), $\{x\mapsto F_{Y|X}(\theta''|x):\theta''\in \mathscr{Y}_1\}$ and $\{y' \mapsto\mathds{1}\{y'\le \theta''\}:\theta''\in \mathscr{Y}_1 \}$ are increasing stochastic processes bounded by one, and they are of VC-subgraph classes according to Lemma 9.10 of Kosorok (2008), and thus of VC type.
For Assumption \ref{a:BR} (ii)(c), note that $E[(\mathds{1}\{Y_i\le y_1\}-F_{Y|X}(y_1|X_i))(\mathds{1}\{Y_i\le y_2\}-F_{Y|X}(y_2|X_i))|X_i]=F_{Y|X}(y_1\wedge y_2|X_i)-F_{Y|X}(y_1|X_i)F_{Y|X}(y_2|X_i)$ and thus it follows from Assumption SCRD (i). Assumption \ref{a:BR} (ii) (d) is implied by the right continuity of $y'\mapsto \mathds{1}\{y\le y'\}$.
Assumption \ref{a:first_stage} is implied by Lemma \ref{lemma:epsilon_unif_const}.
\end{proof}

\subsubsection{Proof of Corollary \ref{corollary:SQRD}}\label{sec:corollary:SQRD}
\begin{proof}
We need to check that Assumptions S, K, M, and SQRD imply all the assumptions required by Theorems \ref{theorem:weak_conv} and \ref{theorem:cond_weak_conv},.
The only non-trivial ones are Assumptions \ref{a:BR} (ii) (a), (d) \ref{a:cond_weak_conv} (i), and \ref{a:first_stage}.
For Assumptions \ref{a:BR} (ii) (a) and \ref{a:first_stage} (i), note that $\{y' \mapsto \mathds{1}\{y'\le y\}: y \in \mathscr{Y}_1\}$ and $\{x\mapsto F_{Y|X}(y|x):y\in \mathscr{Y}_1\}$ are collections of increasing stochastic processes, and Lemma 9.10 of Kosorok (2008) suggests that they are of VC-class and thus VC type with envelope one. Assumption \ref{a:BR} (ii) (d) is implied by the right continuity of $y'\mapsto \mathds{1}\{y\le y'\}$.
Assumption \ref{a:cond_weak_conv} (i) follows from Lemma 3.9.23 (i) of van der Vaart and Wellner (1996) and Assumption SQRD (i).
Assumption \ref{a:first_stage}(b) is implied by Lemma \ref{lemma:epsilon_unif_const}.

By Theorem \ref{theorem:weak_conv}, we have
$
\sqrt{nh_{n}}[\hat{\tau}_{SQRD}(\cdot)-\tau_{SQRD}(\cdot)]
\leadsto \mathds{G}'_{SQRD}
$,
where
\begin{align*}
\mathds{G}'_{SQRD}(\cdot)=&\phi'_{F_{Y|X}(\cdot|0^+)}\Big(\mathds{G}_{H+}\Big)(\cdot) - \phi'_{F_{Y|X}(\cdot|0^-)}\Big(\mathds{G}_{H-}\Big)(\cdot)\\
=&-\frac{\mathds{G}_{H+}(Q_{Y|X}(\cdot|0^+))}{f_{Y|X}(Q_{Y|X}(\cdot|0^+)|0^+)}+\frac{\mathds{G}_{H-}(Q_{Y|X}(\cdot|0^-))}{f_{Y|X}(Q_{Y|X}(\cdot|0^-)|0^-)}.
\end{align*}
Also, by Theorem \ref{theorem:cond_weak_conv},
$
\phi'_{F_{Y|X}(\cdot|0^+)}( \hat{\nu}^+_{\xi,n})-\phi'_{F_{Y|X}(\cdot|0^-)}( \hat{\nu}^-_{\xi,n})\underset{\xi}{\overset{p}{\leadsto}}
\mathds{G}'_{SQRD}
$,
where
\begin{align*}
\phi'_{F_{Y|X}(\cdot|0^\pm)}(\hat{\nu}^\pm_{\xi,n})(\theta'')
&=-\frac{\hat{\nu}^\pm_{\xi,n}(Q_{Y|X}(\theta''|0^\pm))}{f_{Y|X}(Q_{Y|X}(\theta''|0^\pm)|0^\pm)}.
\end{align*}
In the EMP, we replace $f_{Y|X}(\cdot|0^\pm)$, $Q_{Y|X}(\cdot|0^\pm)$ by their uniformly consistent estimators $\hat{f}_{Y|X}(\cdot|0^\pm)$, $\hat{Q}_{Y|X}(\cdot|0^\pm)$, where the uniform consistency of the former follows from Assumption SQRD (iii) (see Lemma \ref{lemma:asymptotic_rho-conti}) and the uniform consistency of the latter follows from Corollary \ref{corollary:operator_unif_consist}.

By Lemma \ref{lemma:cond_weak_conv_and in prob}, it suffices to show that $\sup_{\theta''\in[a,1-a]}\Big|\Big(\widehat{\phi}'_{F_{Y|X}(\cdot|0^+)}( \hat{\nu}^+_{\xi,n})(\theta'')-\widehat{\phi}'_{F_{Y|X}(\cdot|0^-)}( \hat{\nu}^-_{\xi,n})
(\theta'')\Big)
-\Big(\phi'_{F_{Y|X}(\cdot|0^+)}( \hat{\nu}^+_{\xi,n})(\theta'')-\phi'_{F_{Y|X}(\cdot|0^-)}( \hat{\nu}^-_{\xi,n})(\theta'')\Big)\Big|
\underset{x \times \xi}{\overset{p}{\to}}0
$.
We first show
\begin{align*}
&\norm{\widehat{\phi}'_{F_{Y|X}(\cdot|0^+)}(  \hat{\nu}^+_{\xi,n}(\hat{Q}_{Y|X}(\cdot|0^+)) )-\phi'_{F_{Y|X}(\cdot|0^+)}( \hat{\nu}^+_{\xi,n}(\hat{Q}_{Y|X}(\cdot|0^+)))}_{[a,1-a]}\\
=& \norm{-\frac{ \hat{\nu}^+_{\xi,n}(\hat{Q}_{Y|X}(\cdot|0^+))}{\hat{f}_{Y|X}(\hat{Q}_{Y|X}(\cdot|0^+)|0^+)}+\frac{ \hat{\nu}^+_{\xi,n}(Q_{Y|X}(\cdot|0^+))}{f_{Y|X}(Q_{Y|X}(\cdot|0^+)|0^+)}}_{[a,1-a]}
\underset{x \times \xi}{\overset{p}{\to}}0.
\end{align*}
Lemma \ref{lemma:asymptotic_rho-conti} and Corollary \ref{corollary:operator_unif_consist} along with the asymptotic equicontinuity of $ \hat{\nu}^+_{\xi,n}$ implied by its weak convergence in Theorem \ref{theorem:weak_conv} suggest that $ \hat{\nu}^+_{\xi,n}(\hat{Q}_{Y|X}(\cdot|0^+))- \hat{\nu}^+_{\xi,n}(Q_{Y|X}(\cdot|0^+))\underset{x \times \xi}{\overset{p}{\to}}0$ uniformly.
Assumption SQRD (i)(b) and the uniform consistency of both $\hat{Q}_{Y|X}(\cdot|0^+)$ and $\hat{f}_{Y|X}(\cdot|0^+)$ shows that $f_{Y|X}(\cdot|0^+)$ is bounded away from 0 uniformly, and with probability approaching one
\begin{align*}
&\sup_{\theta''\in[a,1-a]}|\hat{f}_{Y|X}(\hat{Q}_{Y|X}(\theta''|0^+)|0^+)-f_{Y|X}(Q_{Y|X}(\theta''|0^+)|0^+)|\\
\le& \sup_{y\in \mathscr{Y}_1}|\hat{f}_{Y|X}(y|0^+)-f_{Y|X}(y|0^+)|+
\sup_{\theta''\in [a,1-a]}L|\hat{Q}_{Y|X}(\theta''|0^+)-Q_{Y|X}(\theta''|0^+)|=o^x_p(1)+o^x_p(1)
\end{align*}
for a Lipschitz constant $L>0$.
Thus, $\norm{-\frac{ \hat{\nu}^+_{\xi,n}(\hat{Q}_{Y|X}(\cdot|0^+))}{\hat{f}_{Y|X}(\hat{Q}_{Y|X}(\cdot|0^+)|0^+)}+\frac{ \hat{\nu}^+_{\xi,n}(Q_{Y|X}(\cdot|0^+))}{f_{Y|X}(Q_{Y|X}(\cdot|0^+)|0^+)}}_{[a,1-a]}\underset{x \times \xi}{\overset{p}{\to}}0$.

Similar lines show $\norm{\widehat{\phi}'_{F_{Y|X}(\cdot|0^-)}(  \hat{\nu}^-_{\xi,n}(\hat{Q}_{Y|X}(\cdot|0^-)) )-\phi'_{F_{Y|X}(\cdot|0^-)}( \hat{\nu}^-_{\xi,n}(\hat{Q}_{Y|X}(\cdot|0^-)))}_{[a,1-a]}\underset{x \times \xi}{\overset{p}{\to}}0$ as well.
Therefore, we have
$$
\sup_{\theta''\in[a,1-a]}\Big|\Big(\widehat{\phi}'_{F_{Y|X}(\cdot|0^+)}( \hat{\nu}^+_{\xi,n})(\theta'')-\widehat{\phi}'_{F_{Y|X}(\cdot|0^-)}( \hat{\nu}^-_{\xi,n})
(\theta'')\Big)
-
\Big(\phi'_{F_{Y|X}(\cdot|0^+)}( \hat{\nu}^+_{\xi,n})(\theta'')-\phi'_{F_{Y|X}(\cdot|0^-)}( \hat{\nu}^-_{\xi,n})(\theta'')\Big)\Big|\underset{x \times \xi}{\overset{p}{\to}}0.
$$
We then apply Lemma \ref{lemma:cond_weak_conv_and in prob} to conclude the proof.
\end{proof}

\subsubsection{Proof of Corollary \ref{corollary:SQRK}}\label{sec:corollary:SQRK}
\begin{proof}
For the unconditional weak convergence, the assumptions required by Theorems \ref{theorem:weak_conv} and \ref{theorem:cond_weak_conv}, can be checked as in Corollary \ref{corollary:SQRD} except for Assumption \ref{a:cond_weak_conv}(i) now follows from Assumption SQRK (i)(a). Applying Theorem \ref{theorem:weak_conv}, we have $\sqrt{nh^3_n}[\tilde \tau_{SQRK} - \tau_{SQRK}]\leadsto \mathds{G}_{SQRD}$. It then suffices to show
\begin{align*}
\sqrt{nh^3_n}\norm{ \tilde \tau_{SQRK} -\hat \tau_{SQRK} }_{\Theta''} = o^x_p(1)
\end{align*}
By definition, we only need to show
\begin{align}
&\sqrt{nh^3_n} \norm{ \widehat \phi(\hat F^{(1)}_{Y|X}(\cdot|0^\pm) ) - \phi(\hat F^{(1)}_{Y|X}(\cdot|0^\pm)  ) }_{\Theta''}\nonumber\\
=&-\sqrt{nh^3_n} \norm{ \frac{\hat F^{(1)}_{Y|X}(\hat Q_{Y|X}(\theta''|0^\pm)|0)}{\hat f_{Y|X} ( \hat Q_{Y|X}(\theta''|0) |0)  }  -\frac{\hat F^{(1)}_{Y|X}(Q_{Y|X}(\theta''|0^\pm)|0)}{ f_{Y|X} (Q_{Y|X}(\theta''|0) |0)  } }_{\Theta''}= o^x_p(1)\label{eq:SQRK_phi approx}.
\end{align}
We first claim that
\begin{align*}
\sqrt{nh^3_n}[\hat F^{(1)}_{Y|X}(\hat{Q}_{Y|X}(\theta''|0)|0^\pm)-\hat F^{(1)}_{Y|X}(Q_{Y|X}(\theta''|0)|0^\pm)]=o^x_p(1)
\end{align*}
respectively. Then by Assumption SQRK (i) (b) and uniform super-consistency of $\hat f_{Y|X}(\cdot|0)$ and $\hat Q_{Y|X}(\cdot|0)$, $\sup_{\theta''\in [a,1-a]}|\hat f_{Y|X}(\hat Q_{Y|X}(\theta''|0)|0)-f_{Y|X}(Q_{Y|X}(\theta''|0)|0)|=o^x_p(1)$ and $\Big|\frac{1}{f_{Y|X}(Q_{Y|X}(\theta''|0)|0)}\Big|<C<\infty$. We can thus conclude that equation (\ref{eq:SQRK_phi approx}) is true.

To prove the claim, notice that
\begin{align*}
&\sqrt{nh^3_n}[\hat F^{(1)}_{Y|X}(\hat{Q}_{Y|X}(\theta''|0)|0^\pm)-\hat F^{(1)}_{Y|X}(Q_{Y|X}(\theta''|0)|0^\pm)]\\
\le & \sqrt{nh^3_n}[\hat F^{(1)}_{Y|X}( \hat Q_{Y|X}(\theta''|0)|0^\pm ) - F^{(1)}_{Y|X}(\hat Q_{Y|X}(\theta''|0)|0^\pm )]\\
&+\sqrt{nh^3_n}[ F^{(1)}_{Y|X}(\hat Q_{Y|X}(\theta''|0)|0^\pm ) - F^{(1)}_{Y|X}( Q_{Y|X}(\theta''|0)|0^\pm )]\\
&+\sqrt{nh^3_n}[F^{(1)}_{Y|X}(  Q_{Y|X}(\theta''|0)|0^\pm ) - \hat F^{(1)}_{Y|X}( Q_{Y|X}(\theta''|0)|0^\pm )]\\
=&(1)+(2)+(3)
\end{align*}
From Theorem \ref{theorem:weak_conv} and the Hadamard differentiability of left inverse operator, we have $\sqrt{nh_n}[\hat Q_{Y|X}(\theta''|0)- Q_{Y|X}(\theta''|0)]=O^x_p(1)$ uniformly. An application of Slutsky's lemma implies $|\hat Q_{Y|X}(\theta''|0)- Q_{Y|X}(\theta''|0)|=O^x(\frac{1}{\sqrt{nh_n}})$. This and Assumption SQRK (i) (a) imply $(2)=o^x_p(1)$ uniformly in $\theta''$. From Lemma \ref{lemma:BR}, $(3)=-\nu^\pm_n (Q_{Y|X}(\theta''|0))+o^x_p(1)$ and $(1)=\nu^\pm_n(\hat Q_{Y|X}(\theta''|0) )+o^x_p(1)=\nu^\pm_n (Q_{Y|X}(\theta''|0))+o^x_p(1)$ uniformly in $\theta''$, where the last equality is due to Lemma \ref{lemma:asymptotic_rho-conti_QRK} and asymptotic $\rho-$equicontinuity of $\nu^\pm_n$ implied by its weak convergence from Theorem \ref{theorem:weak_conv}. Thus we have $\sqrt{nh^3_n}[\hat F^{(1)}_{Y|X}(\hat{Q}_{Y|X}(\theta''|0)|0^\pm)-\hat F^{(1)}_{Y|X}(Q_{Y|X}(\theta''|0)|0^\pm)]=o^x_p(1)$.

As for the conditional weak convergence part of the statement, Theorem \ref{theorem:cond_weak_conv} shows
\begin{align*}
\phi'_{F_{Y|X}(\cdot|0^\pm)}(\hat{\nu}^\pm_{\xi,n})(\cdot)=-\frac{\hat{\nu}^\pm_{\xi,n}(Q_{Y|X}(\cdot|0))}{f_{Y|X}(Q_{Y|X}(\cdot|0))}\underset{\xi}{\overset{p}{\leadsto}} -\frac{\mathds{G}_\pm(Q_{Y|X}(\cdot|0))}{f_{Y|X}(Q_{Y|X}(\cdot|0)|0)}
\end{align*}
Uniform consistency of $\hat{f}_{Y|X}(\cdot|0)$, $\hat{Q}_{Y|X}(\cdot|0)$, Assumption SQRK (i), asymptotic $\rho-$equicontinuity of $\nu^\pm_{\xi,n}$, Lemmas \ref{lemma:asymptotic_rho-conti_QRK}, \ref{lemma:cond_weak_conv_and in prob} and \ref{lemma:super_consist_SQRK} then imply
\begin{align*}
\widehat \phi'_{F_{Y|X}(\cdot|0^\pm)}(\hat{\nu}^\pm_{\xi,n})(\cdot)=-\frac{\hat{\nu}^\pm_{\xi,n}(\hat Q_{Y|X}(\cdot|0))}{\hat f_{Y|X}(\hat Q_{Y|X}(\cdot|0))}\underset{\xi}{\overset{p}{\leadsto}} -\frac{\mathds{G}_\pm(Q_{Y|X}(\cdot|0))}{f_{Y|X}(Q_{Y|X}(\cdot|0)|0)}
\end{align*}
and thus by Assumption S (a) and the continuous mapping theorem, we conclude
\begin{align*}
\frac{\widehat \phi'_{F_{Y|X}(\cdot|0^+)}(\hat{\nu}^+_{\xi,n})(\cdot)-\widehat \phi'_{F_{Y|X}(\cdot|0^-)}(\hat{\nu}^-_{\xi,n})(\cdot)}{b^{(1)}(0^+)-b^{(1)}(0^-)}
=\frac{1}{b^{(1)}(0^+)-b^{(1)}(0^-)}[-\frac{ \hat{\nu}^+_{\xi,n}(\hat{Q}_{Y|X}(\cdot|0))}{\hat{f}_{Y|X}(\hat{Q}_{Y|X}(\cdot|0)|0)}+\frac{ \hat{\nu}^-_{\xi,n}(\hat{Q}_{Y|X}(\cdot|0))}{\hat{f}_{Y|X}(\hat{Q}_{Y|X}(\cdot|0)|0)}]
\\
\underset{\xi}{\overset{p}{\leadsto}} \mathds{G}'_{SQRK}(\cdot)
\end{align*}
\end{proof}

\subsubsection{Proof of Corollary \ref{corollary:GFMRD}}\label{sec:corollary:GFMRD}
\begin{proof}
It is direct to show that Assumption GFMRD and Lemma \ref{lemma:epsilon_unif_const} together imply the assumptions required by Theorems \ref{theorem:weak_conv} and \ref{theorem:cond_weak_conv},. Note for $k=1$, $2$, $\{\mu_k(\cdot,\theta):[\underline x, \overline x]\mapsto \mathds{R}: \theta \in \{1,...,K\}\}$ has finite elements and therefore is of VC-subgraph class with envelope $\underset{\theta\in\{1,...,K\}}{\max} \mu_k(x,\theta)$ and Assumption \ref{a:BR} (ii)(a) is satisfied.
The theorem then gives
\begin{align*}
&\frac{(\mu_{2} (0^+,\cdot)-\mu_{2} (0^-,\cdot))\widehat{\mathds{X}}'_n(\cdot,1)-(\mu_{1}(0^+,\cdot)-\mu_{1}(0^-,\cdot))\widehat{\mathds{X}}'_n(\cdot,2)}{(\mu_{2} (0^+,\cdot)-\mu_{2} (0^-,\cdot))^2}\leadsto N(0,\Sigma_{GFMRD})
\end{align*}
with probability approaching one.
Lemmas \ref{lemma:cond_weak_conv_and in prob} and \ref{lemma:epsilon_unif_const} then give the desired result.
\end{proof}

\section{Bandwidth Choice in Practice}\label{sec:practical_guideline}

While our theory prescribes asymptotic rates of bandwidths, empirical practitioners need to choose bandwidth for each finite $n$.
In this section, we provide a guide for this matter.
We emphasize that our robust inference procedure allows for large bandwidths such as the ones based on the MSE optimality.
Following the bias-robust approach from Calonico, Cattaneo and Titiunik (2014), we increment the degree of local polynomial estimation by one or more to $p$ for the purpose of bias correction while using the optimal bandwidths for the correct order $s$.
For instance, if we are interested in the local linear model ($s=1$), then we run a local quadratic regression ($p=2$) while using the optimal bandwidths for the local linear model ($s=1$).
Generally, when we want to estimate the $v$-th order derivative via a local $s$-th order polynomial estimation, we fix the degree $p$ such that $0\le v \le s \le p$, and additionally $s < p$ if one wants to implement a bias correction.
We now remind the readers of the following short-hand notations: $\Psi_s=\int_{\mathds{R}}r_s(u)r'_s(u)K^2(u)du$, $\Psi^\pm_s=\int_{\mathds{R}_\pm}r_s(u)r'_s(u)K^2(u)du$, $\Gamma_s=\int_{\mathds{R}}K(u)r_s(u)r'_s(u)du$, $\Gamma^\pm_s=\int_{\mathds{R}_\pm}K(u)r_s(u)r'_s(u)du$, $\Lambda_{s,s+1}=\int_{\mathds{R}}u^{s+1}r_s(u)K(u)du$, and $\Lambda^\pm_{s,s+1}=\int_{\mathds{R}_\pm}u^{s+1}r_s(u)K(u)du$.

For the main local polynomial estimation, we first derive the oracle MSE-optimal bandwidths for the $v$-th order derivative based on a local polynomial estimation of the $s$-th degree.
For the numerator, we have
\begin{align*}
&h^{orac}_{1,n}(\theta_1|s,v)=\left(\frac{2v+1}{2s+2-2v}\frac{C'_{1,\theta_1,s,v}}{C^2_{1,\theta_1,s,v}}\right)^{1/(2s+3)}n^{-1/(2s+3)},
\end{align*}
where $C'_{1,\theta_1,s,v}$ and $C^2_{1,\theta_1,s,v}$ are given by the leading terms of bias and variance, respectively:
\begin{align*}
&C_{1,\theta_1,s,v} = \frac{Bias(\mu^{(v)}_{1} (0^+,\theta_1)-\mu^{(v)}_{1} (0^-,\theta_1))}{h^{s+1-v}_{1,n}(\theta_1)} = e'_v \left[\frac{(\Gamma^+_s)^{-1}\Lambda^{+}_{s,s+1}}{(s+1)!}\mu^{(s+1)}_{1} (0^+,\theta_1)
-
\frac{(\Gamma^-_s)^{-1}\Lambda^{-}_{s,s+1}}{(s+1)!}\mu^{(s+1)}_{1} (0^-,\theta_1)\right]
\\
&C'_{1,\theta_1,s,v} = nh^{2v+1}_{1,n}(\theta_1) Var(\mu^{(v)}_{1} (0^\pm,\theta_1)-\mu^{(v)}_{1} (0^-,\theta_1))=
\\
& \ \ \frac{ e'_v [\sigma_{11}(\theta_1,\theta_1|0^+)(\Gamma^+_s)^{-1}\Psi^{+}_s((\theta_1,1)(\theta_1,1))(\Gamma^+_s)^{-1} +\sigma_{11}(\theta_1,\theta_1|0^-) (\Gamma^-_s)^{-1}\Psi^{-}_s((\theta_1,1)(\theta_1,1))(\Gamma^-_s)^{-1}]e_v }{f_X(0)}
\end{align*}
Likewise, for the denominator, we have
\begin{align*}
&h^{orac}_{2,n}(\theta_2|s,v)=\left(\frac{2v+1}{2s+2-2v}\frac{C'_{2,\theta_2,s,v}}{C^2_{2,\theta_2,s,v}}\right)^{1/(2s+3)}n^{-1/(2s+3)}
\end{align*}
where $C'_{2,\theta_2,s,v}$ and $C^2_{2,\theta_2,s,v}$ are given by the leading terms of bias and variance, respectively:
\begin{align*}
&C_{2,\theta_2,s,v} = \frac{Bias(\mu^{(v)}_{2} (0^+,\theta_2)-\mu^{(v)}_{2} (0^-,\theta_2))}{h^{s+1-v}_{2,n}(\theta_2)} = e'_v\left[\frac{(\Gamma^+_s)^{-1}\Lambda^{+}_{s,s+1}}{(s+1)!}\mu^{(s+1)}_{2} (0^+,\theta_2)
-
\frac{(\Gamma^-_s)^{-1}\Lambda^{-}_{s,s+1}}{(s+1)!}\mu^{(s+1)}_{2} (0^-,\theta_2)\right]
\\
&C'_{2,\theta_2,s,v} = nh^{2v+1}_{2,n}(\theta_2) Var(\mu^{(v)}_{2} (0^+,\theta_2)-\mu^{(v)}_{2} (0^-,\theta_2))=
\\
& \ \
\frac{e'_v[\sigma_{22}(\theta_2,\theta_2|0^+)(\Gamma^+_s)^{-1}\Psi^{+}_s((\theta_2,2)(\theta_2,2))(\Gamma^+_s)^{-1}+\sigma_{22}(\theta_2,\theta_2|0^-) (\Gamma^-_s)^{-1}\Psi^{-}_s((\theta_2,2)(\theta_2,2))(\Gamma^-_s)^{-1}]e_v}{f_X(0)}
\end{align*}
In practice, the unknowns in the above bandwidth selectors need to be replaced by their consistent estimates.
We propose the following three-step procedure.

\textbf{Step 1}: Estimate $f_X(0)$ by the kernel density estimator
\begin{align*}
\hat f_X(0)= \frac{1}{nc_n}\sum_{i=1}^{n}K(\frac{X_i}{c_n})
\end{align*}
with the bandwidth $c_n$ determined by Silverman's rule of thumb
\begin{align*}
c_n=1.06\hat{\sigma}_X n^{-1/5},
\end{align*}
where $\hat \sigma_X$ is the standard deviation of the sample $\{X_i\}_{i=1}^n$.
We then compute the preliminary bandwidths for first-stage estimates, $\mu^{(v)}_k$, $k=1,2$, by
\begin{align*}
h^{0}_{1,n}=(\frac{2v+1}{2s+2-2v}\frac{C'_{1,0}}{C^2_{1,0}})^{1/5}n^{-1/5},\\
h^{0}_{2,n}=(\frac{2v+1}{2s+2-2v}\frac{C'_{2,0}}{C^2_{2,0}})^{1/5}n^{-1/5},
\end{align*}
where the constant terms
\begin{align*}
&C_{k,0}= e'_v[\frac{(\Gamma^+_s)^{-1}\Lambda^+_{s,s+1}}{(s+1!)}\bar \mu^{(s+1)}_{k,+} - \frac{(\Gamma^-_s)^{-1}\Lambda^-_{s,s+1}}{(s+1)!}\bar \mu^{(s+1)}_{k,-} ]\\
&C'_{k,0}= e'_v[\bar \sigma^2_{k,+} (\Gamma^+_s)^{-1} \Psi^+_s (\Gamma^+_s)^{-1} +\bar \sigma^2_{k,-} (\Gamma^-_s)^{-1} \Psi^-_s (\Gamma^-_s)^{-1}  ]e_v/\hat f_X(0)
\end{align*}
depend on the preliminary estimates $\bar \mu^{(s+1)}_{k,\pm}$ and $\bar \sigma^2_{k,\pm}$ for $\mu^{(v)}_k$ and $\sigma_{kk}$, respectively.
These preliminary estimates may be obtained by global parametric polynomial regressions of order greater or equal to $s+1$ and the sample variance of $\bar \mu^{(s+1)}_{k,\pm}$.
Through simulations to be presented below, we find that simply setting $\bar \mu^{(s+1)}_{k,\pm}$ and $\bar \sigma^2_{k,\pm}$ to one in this first step also yields fine results, whereas $\hat f_X(0)$ should not be substituted by an arbitrary constant.

\textbf{Step 2}
Using the preliminary bandwidths obtained in Step 1, we next obtain the first stage estimates $[\check \mu_k(0^\pm,\theta_k),...,\check \mu^{(s)}_k(0^\pm,\theta_k)]'=\check \alpha_{k\pm,s}'\text{diag}[1,1!/h^{0}_{k,n},..,s!/(h^0_{k,n})^s]$ as follows.
Solve
\begin{align*}
&\check \alpha_{1\pm,s}':=\argmin_{\alpha \in \mathds{R}^{s+1}}\sum_{i=1}^{n}\delta^\pm_i(g_1(Y_i|\theta_1)-r_s(X_i/h_{0,n})' \alpha)K(\frac{X_i}{h^{0}_{1,n}}),\\
&\check \alpha_{2\pm,s}':=\argmin_{\alpha \in \mathds{R}^{s+1}}\sum_{i=1}^{n}\delta^\pm_i(g_2(D_i,\theta_2)-r_s(X_i/h_{0,n})' \alpha) K(\frac{X_i}{h^{0}_{2,n}}).
\end{align*}
Using these estimates of the local polynomial coefficients, we in turn compute first stage estimates based on $s$-th order expansion
\begin{align*}
\check \mu_k(x,\theta_k)=&[\mu_k(0^+,\theta_k)+\mu^{(1)}_k(0^+,\theta_k)x+...+\mu^{(s)}_k(0^+,\theta_k)\frac{x^s}{s!}]\delta^+_x \\
+& [\mu_k(0^-,\theta_k)+\mu^{(1)}_k(0^-,\theta_k)x+...+\mu^{(s)}_k(0^-,\theta_k)\frac{x^s}{s!}]\delta^-_x,
\end{align*}
for $k=1,2$.
The covariance estimates are in turn computed by
\begin{align*}
&\hat \sigma_{11}(\theta_1,\theta_1|0^\pm)=\Big(\frac{\sum_{i=1}^{n}(g_1(Y_i,\theta_1)-\check \mu_1 (X_i,\theta_1))^2 K(\frac{X_i}{h^{0}_{1,n}})\delta^\pm_i}{\sum_{i=1}^{n}K(\frac{X_i}{h^{0}_{1,n}})\delta^\pm_i}\Big)^{1/2},\\
&\hat \sigma_{22}(\theta_2,\theta_2|0^\pm)=\Big(\frac{\sum_{i=1}^{n}(g_2(D_i,\theta_2)-\check \mu_2 (X_i,\theta_2))^2 K(\frac{X_i}{h^{0}_{2,n}})\delta^\pm_i}{\sum_{i=1}^{n}K(\frac{X_i}{h^{0}_{2,n}})\delta^\pm_i}\Big)^{1/2}.
\end{align*}
The uniform consistency of $\check \mu_k (x,\theta_k) \mathds{1}\{|x|\le h^{0}_{k,n}\}$ in $(x,\theta_k)$ is implied by Lemma \ref{lemma:epsilon_unif_const} with bandwidths $h_{1,n}(\theta_1)=h_{2,n}(\theta_2)=h_{0,n}$ selected in Step 1 under $r=s$.
This further implies the uniform consistency of $\hat \sigma_{kk}$ for $k=1,2$.

\textbf{Step 3}:
We are now ready to derive a feasible version of the main bandwidths.
Let
\begin{align*}
&h^{MSE}_{k,n}(\theta_1|s,v)=\left(\frac{2v+1}{2s+2-2v}\frac{\hat C'_{k,\theta_k,s,v}}{\hat C^2_{k,\theta_k,s,v}}\right)^{1/(2s+3)}n^{-1/(2s+3)}
\end{align*}
where $C_{1,\theta_1,s,v}$, $ C'_{1,\theta_1,s,v}$, $C_{2,\theta_2,s,v}$ and $ C'_{2,\theta_2,s,v}$ are replaced by their estimates:
\begin{align*}
&\hat C_{1,\theta_1,s,v}  = e'_v \left[\frac{(\Gamma^+_s)^{-1}\Lambda^{+}_{s,s+1}}{(s+1)!}\check \mu^{(s+1)}_{s,1} (0^+,\theta_1)
-
\frac{(\Gamma^-_s)^{-1}\Lambda^{-}_{s,s+1}}{(s+1)!}\check \mu^{(s+1)}_{s,1} (0^-,\theta_1)\right],
\\
&\hat C'_{1,\theta_1,s,v} = \frac{ e'_v [\hat \sigma_{11}(\theta_1,\theta_1|0^+)(\Gamma^+_s)^{-1}\Psi^{+}_s(\Gamma^+_s)^{-1} +\hat \sigma_{11}(\theta_1,\theta_1|0^-) (\Gamma^-_s)^{-1}\Psi^{-}_s(\Gamma^-_s)^{-1}]e_v }{\hat f_X(0)}\\
&\hat C_{2,\theta_2,s,v}  = e'_v\left[\frac{(\Gamma^+_s)^{-1}\Lambda^{+}_{s,s+1}}{(s+1)!}\check\mu^{(s+1)}_{2} (0^+,\theta_2)
-
\frac{(\Gamma^-_s)^{-1}\Lambda^{-}_{s,s+1}}{(s+1)!}\check\mu^{(s+1)}_{2} (0^-,\theta_2)\right],
\\
&\hat C'_{2,\theta_2,s,v} =
\frac{e'_v[\hat\sigma_{22}(\theta_2,\theta_2|0^+)(\Gamma^+_s)^{-1}\Psi^{+}_s(\Gamma^+_s)^{-1}+\hat\sigma_{22}(\theta_2,\theta_2|0^-) (\Gamma^-_s)^{-1}\Psi^{-}_s(\Gamma^-_s)^{-1}]e_v}{\hat f_X(0)},
\end{align*}
respectively.
To these feasible MSE-optimal choices, we further apply the rule of thumb (ROT) bandwidth algorithm for optimal coverage error following Calonico, Cattaneo, and Farrell (2016ab):
\begin{align*}
h^{ROT}_{1,n}(\theta_1|s,v)&=h^{MSE}_{1,n}(\theta_1|s,v)n^{-s/(2s+3)(s+3)},\\
h^{ROT}_{2,n}(\theta_2|s,v)&=h^{MSE}_{2,n}(\theta_2|s,v)n^{-s/(2s+3)(s+3)}.
\end{align*}


\section*{Additional Tables}
\clearpage
\begin{table}[t]
	\centering
	  \caption{Simulated acceptance probabilities for treatment nullity under the fuzzy RDD across alternative values of $\beta_1 \in \{0.00, 0.25, 0.50, 0.75, 1.00\}$. The nominal acceptance probability is 95\%.}
		\begin{tabular}{cc|c|cccc}
		\hline\hline
		  \multicolumn{2}{c}{$n$} & \multicolumn{5}{c}{$\beta_1$} \\
			\cline{3-7}
			       && 0.00  & 0.25  & 0.50  & 0.75  & 1.00 \\
			\hline
			1000   && 0.934 & 0.876 & 0.686 & 0.403 & 0.210\\
			2000   && 0.956 & 0.846 & 0.516 & 0.230 & 0.090\\
			4000   && 0.949 & 0.762 & 0.336 & 0.103 & 0.017\\
		\hline\hline
		\\
		\end{tabular}
	\label{tab:FMRD}
\end{table}
\begin{table}[t]
	\centering
	  \caption{Simulated acceptance probabilities for treatment nullity under the sharp RDD across alternative values of $\beta_1 \in \{0.00, 0.25, 0.50, 0.75, 1.00\}$. The nominal acceptance probability is 95\%.}
		\begin{tabular}{cc|c|cccc}
		\hline\hline
		  \multicolumn{2}{c}{$n$} & \multicolumn{5}{c}{$\beta_1$} \\
			\cline{3-7}
			       && 0.00  & 0.25  & 0.50  & 0.75  & 1.00 \\
			\hline
			1000   && 0.939 & 0.842 & 0.641 & 0.366 & 0.190\\
			2000   && 0.949 & 0.812 & 0.481 & 0.204 & 0.075\\
			4000   && 0.955 & 0.738 & 0.304 & 0.082 & 0.019\\
		\hline\hline
		\\
		\end{tabular}
	\label{tab:SMRD}
\end{table}
\begin{table}[t]
	\centering
	  \caption{Simulated coverage probabilities under the sharp RDD under the data generating processes used in Calonico, Cattaneo and Titiunik (2014). The nominal acceptance probability is 95\%. FN stands for the fixed-neighborhood standard error estimators, PI stands for the plug-in residual standard error estimators, and MB stands for the multiplier bootstrap. The first four columns are copied from Table I of Calonico, Cattaneo and Titiunik (2014), whereas the last column is based on our simulations.}
		\begin{tabular}{lcccccccc}
		\hline\hline
		             && \multicolumn{2}{c}{Conventional} && \multicolumn{4}{c}{Robust} \\
								    \cline{3-4}\cline{6-9}
			DGP        && FN & PI && FN & PI && MB \\
			\hline
			Lee (2008)               && 89.4 & 88.4 && 91.6 & 90.7 && 91.3\\
			Ludwig and Miller (2007) && 87.3 & 80.8 && 93.2 & 90.5 && 90.9\\
		\hline\hline
		\\
		\end{tabular}
	\label{tab:SMRD_CCT}
\end{table}
\clearpage
\begin{table}[t]
	\centering
	  \caption{Simulated acceptance probabilities for treatment nullity under the fuzzy RKD across alternative values of $\beta_1 \in \{0.00, 0.25, 0.50, 0.75, 1.00\}$. The nominal acceptance probability is 95\%.}
		\begin{tabular}{cc|c|cccc}
		\hline\hline
		  \multicolumn{2}{c}{$n$} & \multicolumn{5}{c}{$\beta_1$} \\
			\cline{3-7}
			       && 0.00  & 0.25  & 0.50  & 0.75  & 1.00 \\
			\hline
			1000   && 0.913 & 0.864 & 0.760 & 0.642 & 0.554\\
			2000   && 0.925 & 0.838 & 0.686 & 0.551 & 0.414\\
			4000   && 0.936 & 0.786 & 0.588 & 0.419 & 0.290\\
		\hline\hline
		\\
		\end{tabular}
	\label{tab:FMRK}
\end{table}
\begin{table}[t]
	\centering
	  \caption{Simulated acceptance probabilities for treatment nullity under the sharp RKD across alternative values of $\beta_1 \in \{0.00, 0.25, 0.50, 0.75, 1.00\}$. The nominal acceptance probability is 95\%.}
		\begin{tabular}{cc|c|cccc}
		\hline\hline
		  \multicolumn{2}{c}{$n$} & \multicolumn{5}{c}{$\beta_1$} \\
			\cline{3-7}
			       && 0.00  & 0.25  & 0.50  & 0.75  & 1.00 \\
			\hline
			1000   && 0.914 & 0.843 & 0.731 & 0.598 & 0.469\\
			2000   && 0.917 & 0.812 & 0.641 & 0.465 & 0.324\\
			4000   && 0.932 & 0.772 & 0.516 & 0.348 & 0.194\\
		\hline\hline
		\\
		\end{tabular}
	\label{tab:SMRK}
\end{table}

\clearpage
\begin{table}[t]
	\centering
	  \caption{Simulated acceptance probabilities for (A) uniform treatment nullity and (B) treatment homogeneity under the sharp quantile RDD. The top panel (I) presents results across alternative values of $\beta_1 \in \{0.00, 0.25, 0.50, 0.75, 1.00\}$ while fixing $\gamma_1 = 0$. The bottom panel (II) presents results across alternative values of $\gamma_1 \in \{0.00, 0.25, 0.50, 0.75, 1.00\}$ while fixing $\beta_1 = 0$. The nominal acceptance probability is 95\%.}
		\begin{tabular}{cc|c|ccccccc|c|cccc}
		\hline\hline
		\multicolumn{7}{c}{(I) (A) Joint Treatment Nullity} && \multicolumn{7}{c}{(I) (B) Treatment Homogeneity} \\
			\cline{1-7}\cline{9-15}
		  \multicolumn{2}{c}{$n$} & \multicolumn{5}{c}{$\beta_1$} && \multicolumn{2}{c}{$n$} & \multicolumn{5}{c}{$\beta_1$} \\
			\cline{3-7}\cline{11-15}
			       && 0.00  & 0.25  & 0.50  & 0.75  & 1.00  &&        && 0.00  & 0.25  & 0.50  & 0.75  & 1.00  \\
			\cline{1-7}\cline{9-15}
			1000   && 0.966 & 0.917 & 0.798 & 0.620 & 0.484 && 1000   && 0.967 & 0.967 & 0.965 & 0.954 & 0.963 \\
			2000   && 0.959 & 0.859 & 0.633 & 0.414 & 0.269 && 2000   && 0.959 & 0.959 & 0.966 & 0.955 & 0.958 \\
			4000   && 0.950 & 0.740 & 0.400 & 0.161 & 0.074 && 4000   && 0.950 & 0.946 & 0.958 & 0.948 & 0.947 \\
		\hline\hline
		\\
		\hline\hline
		\multicolumn{7}{c}{(II) (A) Joint Treatment Nullity} && \multicolumn{7}{c}{(II) (B) Treatment Homogeneity} \\
			\cline{1-7}\cline{9-15}
		  \multicolumn{2}{c}{$n$} & \multicolumn{5}{c}{$\gamma_1$} && \multicolumn{2}{c}{$n$} & \multicolumn{5}{c}{$\gamma_1$} \\
			\cline{3-7}\cline{11-15}
			       && 0.00  & 0.25  & 0.50  & 0.75  & 1.00  &&        && 0.00  & 0.25  & 0.50  & 0.75  & 1.00  \\
			\cline{1-7}\cline{9-15}
			1000   && 0.966 & 0.928 & 0.842 & 0.742 & 0.647 && 1000   && 0.967 & 0.920 & 0.808 & 0.698 & 0.574 \\
			2000   && 0.959 & 0.866 & 0.669 & 0.500 & 0.378 && 2000   && 0.959 & 0.855 & 0.625 & 0.446 & 0.327 \\
			4000   && 0.950 & 0.718 & 0.362 & 0.206 & 0.138 && 4000   && 0.950 & 0.693 & 0.327 & 0.170 & 0.116 \\
		\hline\hline
		\\
		\end{tabular}
	\label{tab:SQRD}
\end{table}
\clearpage
\begin{table}[t]
	\centering
	  \caption{Simulated acceptance probabilities for (A) uniform treatment nullity and (B) treatment homogeneity under the fuzzy quantile RKD. The top panel (I) presents results across alternative values of $\beta_1 \in \{0.00, 0.25, 0.50, 0.75, 1.00\}$ while fixing $\gamma_1 = 0$. The bottom panel (II) presents results across alternative values of $\gamma_1 \in \{0.00, 0.25, 0.50, 0.75, 1.00\}$ while fixing $\beta_1 = 0$. The nominal acceptance probability is 95\%.}
		\begin{tabular}{cc|c|ccccccc|c|cccc}
		\hline\hline
		\multicolumn{7}{c}{(I) (A) Joint Treatment Nullity} && \multicolumn{7}{c}{(I) (B) Treatment Homogeneity} \\
			\cline{1-7}\cline{9-15}
		  \multicolumn{2}{c}{$n$} & \multicolumn{5}{c}{$\beta_1$} && \multicolumn{2}{c}{$n$} & \multicolumn{5}{c}{$\beta_1$} \\
			\cline{3-7}\cline{11-15}
			       && 0.00  & 0.25  & 0.50  & 0.75  & 1.00  &&        && 0.00  & 0.25  & 0.50  & 0.75  & 1.00  \\
			\cline{1-7}\cline{9-15}
			1000   && 0.945 & 0.929 & 0.903 & 0.868 & 0.842 && 1000   && 0.938 & 0.942 & 0.945 & 0.946 & 0.936 \\
			2000   && 0.941 & 0.911 & 0.873 & 0.836 & 0.815 && 2000   && 0.939 & 0.927 & 0.936 & 0.931 & 0.930 \\
			4000   && 0.935 & 0.904 & 0.846 & 0.799 & 0.802 && 4000   && 0.929 & 0.928 & 0.935 & 0.931 & 0.929 \\
		\hline\hline
		\\
		\hline\hline
		\multicolumn{7}{c}{(II) (A) Joint Treatment Nullity} && \multicolumn{7}{c}{(II) (B) Treatment Homogeneity} \\
			\cline{1-7}\cline{9-15}
		  \multicolumn{2}{c}{$n$} & \multicolumn{5}{c}{$\gamma_1$} && \multicolumn{2}{c}{$n$} & \multicolumn{5}{c}{$\gamma_1$} \\
			\cline{3-7}\cline{11-15}
			       && 0.00  & 0.25  & 0.50  & 0.75  & 1.00  &&        && 0.00  & 0.25  & 0.50  & 0.75  & 1.00  \\
			\cline{1-7}\cline{9-15}
			1000   && 0.945 & 0.956 & 0.950 & 0.952 & 0.928 && 1000   && 0.938 & 0.948 & 0.949 & 0.930 & 0.903 \\
			2000   && 0.941 & 0.944 & 0.936 & 0.919 & 0.921 && 2000   && 0.939 & 0.938 & 0.914 & 0.890 & 0.872 \\
			4000   && 0.935 & 0.941 & 0.928 & 0.905 & 0.888 && 4000   && 0.929 & 0.938 & 0.896 & 0.848 & 0.790 \\
		\hline\hline
		\\
		\end{tabular}
	\label{tab:FQRK}
\end{table}
\clearpage
\begin{table}[t]
	\centering
	  \caption{Simulated acceptance probabilities for (A) uniform treatment nullity and (B) treatment homogeneity under the sharp quantile RKD. The top panel (I) presents results across alternative values of $\beta_1 \in \{0.00, 0.25, 0.50, 0.75, 1.00\}$ while fixing $\gamma_1 = 0$. The bottom panel (II) presents results across alternative values of $\gamma_1 \in \{0.00, 0.25, 0.50, 0.75, 1.00\}$ while fixing $\beta_1 = 0$. The nominal acceptance probability is 95\%.}
		\begin{tabular}{cc|c|ccccccc|c|cccc}
		\hline\hline
		\multicolumn{7}{c}{(I) (A) Joint Treatment Nullity} && \multicolumn{7}{c}{(I) (B) Treatment Homogeneity} \\
			\cline{1-7}\cline{9-15}
		  \multicolumn{2}{c}{$n$} & \multicolumn{5}{c}{$\beta_1$} && \multicolumn{2}{c}{$n$} & \multicolumn{5}{c}{$\beta_1$} \\
			\cline{3-7}\cline{11-15}
			       && 0.00  & 0.25  & 0.50  & 0.75  & 1.00  &&        && 0.00  & 0.25  & 0.50  & 0.75  & 1.00  \\
			\cline{1-7}\cline{9-15}
			1000   && 0.930 & 0.928 & 0.906 & 0.881 & 0.851 && 1000   && 0.927 & 0.936 & 0.945 & 0.936 & 0.940 \\
			2000   && 0.941 & 0.917 & 0.883 & 0.850 & 0.837 && 2000   && 0.932 & 0.930 & 0.929 & 0.929 & 0.928 \\
			4000   && 0.941 & 0.902 & 0.850 & 0.835 & 0.829 && 4000   && 0.929 & 0.930 & 0.918 & 0.936 & 0.924 \\
		\hline\hline
		\\
		\hline\hline
		\multicolumn{7}{c}{(II) (A) Joint Treatment Nullity} && \multicolumn{7}{c}{(II) (B) Treatment Homogeneity} \\
			\cline{1-7}\cline{9-15}
		  \multicolumn{2}{c}{$n$} & \multicolumn{5}{c}{$\gamma_1$} && \multicolumn{2}{c}{$n$} & \multicolumn{5}{c}{$\gamma_1$} \\
			\cline{3-7}\cline{11-15}
			       && 0.00  & 0.25  & 0.50  & 0.75  & 1.00  &&        && 0.00  & 0.25  & 0.50  & 0.75  & 1.00  \\
			\cline{1-7}\cline{9-15}
			1000   && 0.930 & 0.954 & 0.955 & 0.943 & 0.936 && 1000   && 0.927 & 0.938 & 0.939 & 0.921 & 0.911 \\
			2000   && 0.941 & 0.931 & 0.936 & 0.926 & 0.902 && 2000   && 0.932 & 0.932 & 0.923 & 0.901 & 0.875 \\
			4000   && 0.941 & 0.931 & 0.928 & 0.905 & 0.879 && 4000   && 0.929 & 0.920 & 0.885 & 0.861 & 0.816 \\
		\hline\hline
		\\
		\end{tabular}
	\label{tab:SQRK}
\end{table}
\clearpage
\begin{table}[t]
	\centering
	  \caption{Simulated acceptance probabilities for (A) joint treatment nullity and (B) treatment homogeneity under the fuzzy RDD with group covariate across alternative values of $\beta_1 \in \{0.00, 0.25, 0.50, 0.75, 1.00\}$ while fixing $\beta_2 = 0$. The nominal acceptance probability is 95\%.}
		\begin{tabular}{cc|c|ccccccc|c|cccc}
		\hline\hline
		\multicolumn{7}{c}{(A) Joint Treatment Nullity} && \multicolumn{7}{c}{(B) Treatment Homogeneity} \\
			\cline{1-7}\cline{9-15}
		  \multicolumn{2}{c}{$n$} & \multicolumn{5}{c}{$\beta_1$} && \multicolumn{2}{c}{$n$} & \multicolumn{5}{c}{$\beta_1$} \\
			\cline{3-7}\cline{11-15}
			       && 0.00  & 0.25  & 0.50  & 0.75  & 1.00  &&        && 0.00  & 0.25  & 0.50  & 0.75  & 1.00  \\
			\cline{1-7}\cline{9-15}
			1000   && 0.980 & 0.971 & 0.957 & 0.940 & 0.885 && 1000   && 0.972 & 0.966 & 0.964 & 0.944 & 0.919 \\
			2000   && 0.977 & 0.959 & 0.942 & 0.881 & 0.803 && 2000   && 0.971 & 0.964 & 0.951 & 0.913 & 0.867 \\
			4000   && 0.972 & 0.959 & 0.902 & 0.786 & 0.601 && 4000   && 0.971 & 0.965 & 0.925 & 0.855 & 0.762 \\
		\hline\hline
		\\
		\end{tabular}
	\label{tab:GFMRD}
\end{table}

\end{document}